\documentclass[%
 reprint,
 superscriptaddress,
 amsmath,amssymb,
prb,
]{revtex4-2}
\usepackage{comment}
\usepackage{amsfonts}
\usepackage{float}
\usepackage{mathtools}
\usepackage{graphicx}
\usepackage{physics}
\usepackage{color}
\usepackage{xcolor}
\usepackage{amsthm}
\usepackage{mathdots}
\usepackage{tikz}
\usepackage{tikz-cd}
\usepackage{relsize}
\usepackage[colorlinks=true,linktoc=page,citecolor=red,linkcolor=blue]{hyperref}
\usepackage{caption}
\usepackage{subcaption}
\usepackage{dutchcal}
\usepackage{algpseudocode,algorithm}
\usepackage{CJK}
\usepackage{extarrows}
\usepackage{tikzit}
\tikzstyle{pointoperator}=[fill=blue, draw=none, shape=circle, minimum size=.1 cm, inner sep=0 pt]
\tikzstyle{pointred}=[fill=red, draw=none, shape=circle, minimum size=.1 cm, inner sep=0 pt]
\tikzstyle{pointblack}=[fill=black, draw=none, shape=circle, minimum size=.1 cm, inner sep=0 pt]

\tikzstyle{dashedline}=[dashed, -, thick]
\tikzstyle{blueline}=[-, draw=blue, thick]
\tikzstyle{redline}=[-, draw=red, thick]
\tikzstyle{greenline}=[-, draw=green, thick]
\tikzstyle{arrowline}=[<-, thick]
\tikzstyle{bluearrowline}=[<-, draw=blue, thick]
\tikzstyle{normalline}=[-, thick]
\tikzstyle{greenfillline}=[-, thick, fill=green, fill opacity=.5]
\tikzstyle{grayfillline}=[-, fill opacity=.5, fill={rgb,255: red,128; green,128; blue,128}, thick]
\tikzstyle{redfillline}=[-, fill opacity=.5, fill=red, thick]
\tikzstyle{whitefillline}=[-, fill opacity=.5, fill=white, thick]
\tikzstyle{bluefillline}=[-, thick, fill=blue, fill opacity=.5]
\tikzstyle{bluedashedline}=[-, draw=blue, thick, dashed]
\tikzstyle{reddashedline}=[-, draw=red, thick, dashed]
\tikzstyle{grayfilldashedline}=[dashed, -, fill opacity=.5, fill={rgb,255: red,128; green,128; blue,128}, thick]

\usetikzlibrary{positioning}

\newcommand{\beq}{\begin{equation}\begin{aligned}}
\newcommand{\eeq}{\end{aligned}\end{equation}}

\usepackage{tikz-cd}
\usetikzlibrary{decorations.markings}
\usetikzlibrary{decorations.pathmorphing}
\tikzset{snake it/.style={decorate, decoration=snake}}
\usetikzlibrary{calc}
\usepackage{tikz-3dplot} 
\tdplotsetmaincoords{70}{120}  
\tdplotsetrotatedcoords{+100}{0}{0} 
\pgfdeclarelayer{bg}
\pgfsetlayers{bg,main}

\newtheorem{theorem}{Theorem}
\newtheorem{lemma}[theorem]{Lemma}

\def\red#1{{\color{black} #1}}

\newcommand*\circled[1]{\tikz[baseline=(char.base)]{
  \node[shape=circle,draw,inner sep=1pt] (char) {#1};}}

\usepackage[normalem]{ulem} 
\usepackage{soul}
\usepackage{titlesec}
\titleformat{\paragraph}[block]{\normalfont\bfseries\filcenter}{\theparagraph}{1em}{}
\begin{document}

\begin{CJK*}{UTF8}{}
\title{
Higher-order topological phases protected by noninvertible and subsystem symmetries
}
\author{Aswin Parayil Mana}
\affiliation{C. N. Yang Institute for Theoretical Physics, State University of New York at Stony Brook, New York 11794-3840, USA}
\affiliation{Department of Physics and Astronomy, State University of New York at Stony Brook, New York 11794-3840, USA}

\author{Yabo Li (\CJKfamily{gbsn}李雅博)}
\affiliation{C. N. Yang Institute for Theoretical Physics, State University of New York at Stony Brook, New York 11794-3840, USA}
\affiliation{Department of Physics and Astronomy, State University of New York at Stony Brook, New York 11794-3840, USA}
\affiliation{Center for Quantum Phenomena, Department of Physics, New York University, 726 Broadway, New York, New York 10003, USA}

\author{Hiroki Sukeno (\CJKfamily{min}助野裕紀)}
\affiliation{C. N. Yang Institute for Theoretical Physics, State University of New York at Stony Brook, New York 11794-3840, USA}
\affiliation{Department of Physics and Astronomy, State University of New York at Stony Brook, New York 11794-3840, USA}

\author{Tzu-Chieh Wei (\CJKfamily{bsmi}魏子傑)}
\affiliation{C. N. Yang Institute for Theoretical Physics, State University of New York at Stony Brook, New York 11794-3840, USA}
\affiliation{Department of Physics and Astronomy, State University of New York at Stony Brook, New York 11794-3840, USA}

\date{\today}

\begin{abstract}
Higher-order topological phases with invertible symmetries have been extensively studied in recent years, revealing gapless modes localized on boundaries of higher codimension. In this work, we extend the framework of higher-order symmetry-protected topological (SPT) phases to include noninvertible symmetries. We construct a concrete model of a second-order SPT phase in $2+1$ dimensions that hosts symmetry-protected corner modes protected by a noninvertible symmetry. This construction is then generalized to a $d^{th}$-order SPT phase in $d+1$ dimensions, featuring similarly protected corner modes. Additionally, we demonstrate a second-order SPT phase in $3+1$ dimensions exhibiting hinge modes protected by a noninvertible symmetry.

\end{abstract}

\maketitle
\end{CJK*}

\tableofcontents

\section{Introduction}
In recent years, higher-order topological phases have garnered significant attention in condensed matter physics, offering new insights into the interplay between symmetry, topology, and dimensionality ~\cite{Benalcazar:2017udd, Benalcazar:2017dhp, Schindler:2017etn, Calugaru:2018xux, Song:2017uhz, Langbehn:2017rwy, Franca:2018bug, Benalcazar:2021tay,kunst2018lattice,ezawa2018magnetic,ezawa2018higher,geier2018second,khalaf2018higher,van2018higher,trifunovic2019higher,benalcazar2019quantization,ren2020engineering,yang2020type, Khalaf:2019iev,Manna:2022xmn}. Unlike conventional topological phases in $d$ dimensions, which host ($d-1$)-dimensional boundary states, a $k^{th}$-order topological phase exhibits boundary states localized on ($d-k$)-dimensional surfaces. This hierarchical framework generalizes the notion of phases, with conventional topological phases corresponding to first-order phases. Initially explored in the context of topological insulators~\cite{Benalcazar:2017udd, Benalcazar:2017dhp, Schindler:2017etn, Song:2017uhz, Langbehn:2017rwy}, higher-order phases have since been extended to symmetry-protected topological (SPT) phases~\cite{You:2018srv, Rasmussen:2018yjy} and subsystem symmetry-protected topological (SSPT) phases~\cite{You:2024syf}, where distinct phases are characterized by lower-dimensional boundary modes. The presence of higher-order modes in topological phases with crystalline and internal symmetry was studied in non-interacting fermionic models in~\cite{Benalcazar:2017udd, Benalcazar:2021tay, Langbehn:2017rwy, Song:2017uhz, Dwivedi:2018soz, Manna:2022xmn}. 
A variety of bosonic or interacting fermionic higher-order topological phases with gapless corner or hinge modes also appeared in~\cite{Tiwari:2019bzf, You:2018bmf, You:2023kev, You:2018srv, Rasmussen:2018yjy, You:2024syf}.  

A simplest example of higher-order SPT was constructed by~\cite{You:2018srv}, where the $2+1$D cluster state with $\mathbb{Z}_2\times\mathbb{Z}_2$ 0-form symmetry is considered. 
When placed on a rectangular region with open boundary conditions, there are gapless corner modes protected by the crystalline symmetries such as $C_4$ rotation (rotation by $\pi/2$ around the $z$ axis) or reflection (around the $x$ or $y$ axis). 
Since the gapless modes are on the corner (codimension $2$ boundary), this is a second-order SPT. 
In $d+1$ dimensions, one could expect to obtain similar phases with gapless modes on codimension $k$ surfaces ($k^{th}$ order SPTs). With this terminology, first-order SPTs are the ordinary SPTs that were considered in~\cite{Chen:2011pg,guwen2009tensor,pollmann2010entanglement,haldane1983nonlinear,chen2012symmetry,gu2014symmetry,kapustin2014symmetry}.

Parallel to these developments, noninvertible symmetries have become a vibrant area of research, bridging high-energy physics, condensed matter physics, and mathematics. These symmetries, which include operations like the Kramers-Wannier (KW) duality, generalize the concept of symmetry beyond group-theoretic frameworks. Gapped phases of matter with these symmetries have been explored in~\cite{Thorngren:2019iar,Cordova:2023bja,Inamura:2021wuo, Bhardwaj:2023idu, bhardwaj2024hasse, Inamura:2021szw, Cao:2024qjj, Li:2023ani, Chatterjee:2024ych, Seiberg:2024gek, Seifnashri:2025fgd, Seifnashri:2025fgd,Seo:2024its}. Recent work has extended the notion of SPT phases to systems with noninvertible symmetries on lattice models, revealing new classes of topological phases~\cite{fechisin2023non,Seifnashri:2024dsd,li2024non,inamura20241+,meng2024non,Choi:2024rjm, Jia:2024bng, pace2025spt, li2024domain, li2024noninvertible, jia2024weak, jia2024quantum, cao2025duality, Lu:2025rwd, Aksoy:2025rmg}. For example, \cite{Seifnashri:2024dsd} investigated the lattice realization of the Rep($D_8$) SPT phases in $1+1$ dimensions, where the symmetry category includes $\mathbb{Z}_2\times\mathbb{Z}_2$ symmetry and the KW duality, denoted by $\mathrm{\mathbf{D}}$. The
$\mathbb{Z}_2\times\mathbb{Z}_2$ cluster state, invariant under $\mathrm{\mathbf{D}}$, was shown to split into distinct noninvertible SPT phases, highlighting the rich structure of symmetry-protected topology in the presence of noninvertible symmetries. This framework has since been generalized to Rep($G$) for class 2-nilpotent groups $G$~\cite{li2024non}, to systems with spatially modulated symmetries \cite{Kim:2025ttl} and to systems with fusion category symmetries~\cite{inamura20241+,meng2024non}, opening new avenues for exploring topological phases.

A key tool in these studies has been the Kennedy-Tasaki (KT) transformation, a non-local mapping that connects symmetry breaking phases to SPT phases. Originally introduced to establish the Haldane phase as a non-trivial SPT~\cite{kennedy1992hidden,kennedy1992hidden2}, the KT transformation has been generalized to systems with arbitrary integer spins by Oshikawa~\cite{oshikawa1992hidden} and applied to SPT phases, where it maps distinct SPTs to distinct symmetry breaking phases. This approach has proven instrumental in classifying and understanding SPT phases, particularly in the context of noninvertible symmetries~\cite{Seifnashri:2024dsd,cao2025duality}. More recently, the KT transformation has been extended to subsystem symmetry-protected topological phases (SSPTs), enabling the mapping of SSPTs to spontaneous subsystem symmetry breaking (SSSB) phases in $2+1$ dimensions and higher~\cite{mana2024kennedy}.

Subsystem symmetries, which act on lower-dimensional subspaces of a system, have been a cornerstone in the study of SSPTs. First introduced for $2+1$D and $3+1$D in~\cite{You:2018oai}, SSPTs have since been systematically classified for linear~\cite{Devakul:2018fhz}, planar~\cite{Devakul:2019duj}, and fractal~\cite{Devakul:2019noj,Devakul:2018lmi} subsystem symmetries. The combination of subsystem symmetries with noninvertible symmetries, although studied in some recent works \cite{Cao:2023doz,Ebisu:2024lie,mana2024kennedy,Maity:2025fdb}, remains largely unexplored, presenting an exciting frontier for research.

In this manuscript, we investigate noninvertible higher-order SSPTs, focusing on phases protected by the interplay of subsystem symmetries and noninvertible symmetries. We begin by examining a cluster state that exhibits both subsystem symmetries and KW duality symmetry, a noninvertible symmetry. Using the KT transformation, we show the cluster phase is split into distinct equivalence classes of SSPTs protected by noninvertible symmetries. On the SSSB side, multiple ways of preserving the unbroken symmetry give rise to distinct SSPTs, some of which are characterized by higher-order corner or hinge modes. These modes become apparent at interfaces between distinct SSPTs, providing a robust signature of noninvertible higher-order \red{SPTs}.

We illustrate our framework with explicit examples, including a noninvertible second-order SSPT in $2+1$D, where two such phases are distinguished by corner modes, and a second-order SSPT in $3+1$D, distinguished by hinge modes. We further generalize the former construction to $d^{th}$-order SSPTs in $d+1$ dimensions, demonstrating the universality of our approach. Our results highlight the rich structure of noninvertible higher-order SSPTs and their potential for realizing novel topological phenomena in higher dimensions.

The structure of this paper is organized as follows. In Section~\ref{sec:SSPTs}, we review symmetry-protected topological phases (SPTs), focusing on SSPTs protected by $\mathbb{Z}_2$ or $\mathbb{Z}_2\times\mathbb{Z}_2$ linear subsystem symmetries in $2+1$ dimensions and higher, as well as noninvertible SPT phases in $1+1$D following~\cite{Seifnashri:2024dsd}.
Our main results on noninvertible higher-order SSPTs are presented in Section~\ref{sec:NSSPT}, where we explore noninvertible second-order SSPTs in $2+1$D (with corner modes), and generalize the (corner-mode) construction to $d^{th}$-order SSPTs in $d+1$ dimensions. In Section~\ref{sec:hinge}, we construct a model in 3D that hosts planar subsystem symmetry-protected topological phases in $3+1$ dimensions, and then explore the noninvertible second-order SSPTs.  In Section~\ref{sec:Conclusion}, we provide concluding remarks and discuss potential future directions. The appendices contain supplementary material: Appendix~\ref{sec:orderparameter} presents a lemma on the consistent choice of order parameters for symmetry breaking; Appendix~\ref{sec:anomaly} analyzes anomalies involving subsystem and 0-form symmetries; Appendix~\ref{sec:otherNSSPT} describes additional noninvertible SSPTs not covered in the main text; and Appendix~\ref{sec:Interfaceanalysis}, \ref{sec:interfaceanalysisotherspt} and~\ref{sec:Interfaceanalysis3D} provide detailed analyses of interface modes between distinct noninvertible SSPTs in $2+1$D and $3+1$D, respectively.

\section{Review of symmetry-protected topological phases: subsystem and noninvertible symmetries.}
\label{sec:SSPTs}
In this section, we review symmetry-protected topological phases protected by subsystem symmetry and noninvertible symmetry. First, we discuss SSPTs with symmetry groups $\mathbb{Z}_2$ and $\mathbb{Z}_2\times\mathbb{Z}_2$ in $2+1$D and then in $d+1$D. Then we discuss noninvertible SPTs in $1+1$D with Rep($D_8$) symmetry.
\subsection{Linear subsystem symmetry-protected topological phases in \texorpdfstring{$2+1$}{Lg}D}
In this section, we review subsystem symmetry-protected topological phases (SSPTs) in $2+1$D. We restrict our discussion to $\mathbb{Z}_2$ and $\mathbb{Z}_2\times\mathbb{Z}_2$ subsystem symmetry. SSPTs were first introduced by~\cite{You:2018oai} in $2+1$D and $3+1$D (we will also write $2+1$D or $3+1$D as 2D or 3D respectively). Later, a classification of linear SSPTs in $2+1$D was provided in~\cite{Devakul:2018fhz}. For studies on planar subsystem symmetries, see~\cite{Devakul:2019duj}, and on fractal subsystem symmetries, see~\cite{Devakul:2019noj,Devakul:2018lmi}. 

There are two different notions of SSPTs in $2+1$D: 1) weak SSPTs and 2) strong SSPTs. Weak SSPTs can be thought of as stacks of one-dimensional ($1+1$D) SPTs. Strong SSPTs are intrinsically $2+1$D phases. Strong equivalence of SSPTs are defined with respect to linearly symmetric local unitary evolution (LSLU) (see~\cite{Devakul:2018fhz} for a definition of strong SSPTs). 

According to Ref.~\cite{Devakul:2018fhz}, strong SSPTs protected by linear subsystem symmetries with onsite symmetry group $G_s$ in $2+1$D are classified by
\begin{align}
    \mathcal{C}[G_s]\equiv \mathcal{H}^2(G_s^2,U(1))/\mathcal{H}^2(G_s,U(1))^3,
\end{align}
where $\mathcal{H}^2(G,U(1))$ denotes the second group cohomology of $G$. We will use this formula in the following discussion.
\subsubsection{\texorpdfstring{$\mathbb{Z}_2$}{Lg} subsystem symmetry}
Let us consider the onsite symmetry group to be $G_s=\mathbb{Z}_2$. According to the classification, $\mathcal{C}[G_s]=\mathbb{Z}_2$. Hence, there are two different strong equivalence classes of $\mathbb{Z}_2$ SSPTs. One equivalence class is the trivial class, represented by the product state $\ket{+}^{\otimes_i s_i}$, where $\otimes_i s_i$ denotes the product over all sites. The other equivalence class is non-trivial and gives rise to a nontrivial $\mathbb{Z}_2$ SSPT phase. We write down a Hamiltonian for this phase at the fixed point. 

To do this, we consider a square lattice with $L$ sites in the horizontal and vertical direction. Qubits are placed on the vertices (sites) of the lattice. We denote the vertices by a pair of integers $(i,j)$ where $i,j=1,...,L$. The Hamiltonian for $\mathbb{Z}_2$ SSPT is given by
\begin{align} 
&\mathrm{H}_{\text{2D-SSPT}}^{\mathbb{Z}_2} = -\sum_{i,j}\begin{array}{ccc}
 & Z & Z \\
Z& X_{i,j} & Z \\
Z& Z &
\end{array}\,,  \label{eq:Z2SSPTHam}\\
&\, \, =-\sum_{i,j}X_{i,j}Z_{i+1,j}Z_{i,j+1}Z_{i+1,j+1}Z_{i,j-1}Z_{i-1,j}Z_{i-1,j-1}. \nonumber 
\end{align}
This Hamiltonian has horizontal, vertical, and diagonal linear subsystem symmetries. The symmetry operators are 
\begin{subequations}
\begin{align}
\eta^x_j &\equiv \prod_{i = 1 }^L X_{i,j} \quad (j \in \{1, ...,L\}), \\
\eta^y_i &\equiv \prod_{j = 1 }^L X_{i,j} \quad (i \in \{1, ...,L\}), \\
\eta^\text{diag}_k &\equiv \prod_{\ell = 1 }^L X_{ \ell,[\ell+k]_L} \quad (k \in \{1, ...,L\})\, .
\end{align}
\label{eq:2dsubsystemwithdiag}
\end{subequations}
The above symmetries satisfy a constraint $\prod_{j=1}^L\eta^x_j = \prod_{i=1}^L\eta^y_i = \prod_{k=1}^L\eta^\text{diag}_k$. We have in total $3L-2$ independent line symmetry generators \red{for odd $L$ and $3L-3$ independent line symmetry generators for even $L$. In any case, we emphasize that there is a macroscopic number (linearly depending on the system length $L$) of symmetry generators with $\mathbb{Z}_2$ subsystem symmetry.} In addition to this, the Hamiltonian is also symmetric  under the exchange
\begin{align}
    X\leftrightarrow \begin{array}{ccc}
 & Z & Z \\
Z&   & Z \\
Z& Z &
\end{array}\, .
\end{align}
This transformation can be implemented by the  operator $\mathbf{D}_\text{DPIM}^{(2)}\equiv \mathbf{T}_{1,1}^{-1}\mathbf{D}_\text{DPIM}$, where $\mathbf{T}_{1,1}^{-1}$ is diagonal translation by one site and $\mathbf{D}_\text{DPIM}$ is defined as~\cite{mana2024kennedy}
\begin{align}
\mathbf{D}_\text{DPIM}
\equiv 
\mathbf{P}_\text{DPIM}
\tilde{\mathbf{D}}_y^{(2)}
\mathbf{H}^{\otimes(2)}
\tilde{\mathbf{D}}_x^{(2)}
\mathbf{H}^{\otimes(2)}
\tilde{\mathbf{D}}_\text{diag}^{(2)}
\mathbf{P}_\text{DPIM} \,.
\label{eq:DDPIM}
\end{align}
Here $\mathbf{H}^{\otimes(2)}$ is the simultaneous Hadamard transformation on all the qubits and
\begin{subequations}
\begin{align}
&\tilde{\mathbf{D}}_x^{(2)}\equiv\prod_{j=1}^{L}\left(\left(\prod_{i=1}^{L-1}e^{i\frac{\pi}{4}X_{i,j}}e^{i\frac{\pi}{4}Z_{i,j}Z_{i+1j}}\right)\,e^{i\frac{\pi}{4}X_{L,j}}\right), \\
&\tilde{\mathbf{D}}_y^{(2)}\equiv\prod_{i=1}^{L}\left(\left(\prod_{j=1}^{L-1}e^{i\frac{\pi}{4}X_{i,j}}e^{i\frac{\pi}{4}Z_{i,j}Z_{i,j+1}}\right)e^{i\frac{\pi}{4}X_{i,L}}\right), \\
&\tilde{\mathbf{D}}_\text{diag}^{(2)}\equiv\prod_{k=1}^L\left(\prod_{\ell=1}^{L-1}e^{i\frac{\pi}{4}X_{\ell,[\ell+k]_L}}e^{i\frac{\pi}{4}Z_{\ell,[\ell+k]_L}Z_{\ell+1,[\ell+k+1]_L}}\right) \nonumber \\
&\qquad\qquad\qquad \times e^{i\frac{\pi}{4}X_{L,[L+k]_L}} \, ,\\
&\mathbf{P}_\text{DPIM}
\equiv \prod_{j = 1}^L \frac{1+\eta^x_j}{2}
\prod_{i = 1}^L \frac{1+\eta^y_i}{2}
\prod_{k = 1}^L \frac{1+\eta^\text{diag}_k}{2} \, .
\end{align}
\label{eq:Ddpimopdef}
\end{subequations}
We note that an explicit operator representation of KW was studied in 1+1 dimensions by~\cite{Seiberg:2023cdc,Chen:2023qst,Ho:2019hyv}.
$\mathbf{D}_\text{DPIM}^{(2)}$ satisfies the following algebra
\begin{subequations}
    \begin{align}
    &\left(\mathbf{D}_\text{DPIM}^{(2)}\right)^2\propto \mathbf{P}_\text{DPIM},\\
    &\mathbf{D}_\text{DPIM}^{(2)}\eta_i^y=\eta_i^y\mathbf{D}_\text{DPIM}^{(2)}=\mathbf{D}_\text{DPIM}^{(2)},\\
    &\mathbf{D}_\text{DPIM}^{(2)}\eta_j^x=\eta_j^x\mathbf{D}_\text{DPIM}^{(2)}=\mathbf{D}_\text{DPIM}^{(2)},\\
    &\mathbf{D}_\text{DPIM}^{(2)}\eta_k^{\text{diag}}=\eta_k^{\text{diag}}\mathbf{D}_\text{DPIM}^{(2)}=\mathbf{D}_\text{DPIM}^{(2)}.
\end{align}
\end{subequations}
We emphasize that $\mathbf{D}_\text{DPIM}^{(2)}$ is a noninvertible symmetry of the $\mathbb{Z}_2$ SSPT Hamiltonian~\eqref{eq:Z2SSPTHam}.  
\subsubsection{\texorpdfstring{$\mathbb{Z}_2\times\mathbb{Z}_2$}{Lg} subsystem symmetry}
Now let us consider the case where the onsite symmetry group $G_s=\mathbb{Z}_2\times\mathbb{Z}_2$. Using the classification result, $\mathcal{C}[\mathbb{Z}_2\times\mathbb{Z}_2]=\mathbb{Z}_2\times\mathbb{Z}_2\times\mathbb{Z}_2$. Hence, there are eight inequivalent SSPT phases generated by three SSPTs with $\mathbb{Z}_2\times\mathbb{Z}_2$ symmetry. To describe the generators, let us consider two square lattices that are dual to each other. We color them red and blue. The three generators are 1) $\mathbb{Z}_2$ SSPT on red sublattice, 2) $\mathbb{Z}_2$ SSPT on blue sublattice, and  3) phase of cluster state that entangles red and blue sublattices. We described $\mathbb{Z}_2$ SSPT before. Now we describe the cluster phase Hamiltonian at the fixed point. 

Let us denote the vertices and plaquettes of the red and blue square lattices by $v_r$, $v_b$, and $p_r$, $p_b$, respectively. The Hamiltonian of the cluster state is 
\begin{align}
    \bm \mathrm{H}_{\text{2D-cluster}}&=-\sum_{v_r}X_{v_r}\prod_{v_b\in \partial p_b}Z_{v_b}-\sum_{v_b}X_{v_b}\prod_{v_r\in \partial p_r}Z_{v_r} \, .
\label{eq:2dcluster}
\end{align}
Compared to the $\mathbb{Z}_2$ SSPT, the cluster state Hamiltonian does not have a diagonal linear subsystem symmetry. 

Let us consider a square lattice with $L_x$ vertices in the $x$-direction and $L_y$ vertices in the $y$-direction. We label the vertices of the red sublattice by integer coordinates $(i,j)$ and that of the blue sublattice by half-integer coordinates $(i+\frac{1}{2},j+\frac{1}{2})$ where $i=1,...,L_x$ and $j=1,...,L_y$. Subsystem symmetries are generated by
\begin{align}
    \eta^x_{r,j}&=\prod_{i}X_{i,j}\, ,\quad \eta^y_{r,i}=\prod_{j}X_{i,j}\, ,\nonumber\\
    \eta^x_{b,j}&=\prod_{i}X_{i+\frac{1}{2},j+\frac{1}{2}}\, ,\quad \eta^y_{b,i}=\prod_{j}X_{i+\frac{1}{2},j+\frac{1}{2}}\, .
    \label{eq:subsystemsymmetries2Dclstr}
\end{align}
In addition, there is another symmetry for the Hamiltonian~\eqref{eq:2dcluster} obtained by swapping
\begin{align}
    X_{v_r}\leftrightarrow \begin{array}{cc}
        Z_{v_b} & Z_{v_b}\\
        Z_{v_b} & Z_{v_b}
    \end{array},\qquad X_{v_b}\leftrightarrow \begin{array}{cc}
        Z_{v_r} & Z_{v_r}\\
        Z_{v_r} & Z_{v_r}
    \end{array}\, .
\end{align}
This transformation is a Kramers-Wannier duality that gauges the subsystem symmetries.
An operator representation of this symmetry up to a half lattice translation is given in~\cite{mana2024kennedy}. Here we define an operator with half lattice translation included:
\begin{align}
    \mathbf{D}^{(2)}=\mathbf{T}^{-1}_{\frac{1}{2},\frac{1}{2}}\mathbf{D}^{(2)}_r\mathbf{D}^{(2)}_b\, .
\end{align}
The operator $\mathbf{D}^{(2)}_{r(b)}$ on the red (blue) sublattice is defined as 
\begin{align}
    \mathbf{D}^{(2)}_{r(b)}\equiv\mathbf{P}^{(2)}_{r(b)} \Tilde{\mathbf{D}}_{x;r(b)}^{(2)}\mathbf{H}^{\otimes(2)}_{r(b)}\Tilde{\mathbf{D}}_{y;r(b)}^{(2)} \mathbf{P}^{(2)}_{r(b)}\, ,
\end{align}
where $\mathbf{H}^{\otimes(2)}_{r(b)}$ denotes the product of Hadamard operators on red (and repectively, blue) lattices and
\begin{subequations}
\begin{align}
\tilde{\mathbf{D}}_{x;r}^{(2)}&\equiv\prod_{j=1}^{L_y}\left(\left(\prod_{i=1}^{L_x-1}e^{i\frac{\pi}{4}X_{i,j}}e^{i\frac{\pi}{4}Z_{i,j}Z_{i+1j}}\right)\,e^{i\frac{\pi}{4}X_{L_x,j}}\right), \label{Dx-def}\\
\tilde{\mathbf{D}}_{y;r}^{(2)}&\equiv\prod_{i=1}^{L_x}\left(\left(\prod_{j=1}^{L_y-1}e^{i\frac{\pi}{4}X_{i,j}}e^{i\frac{\pi}{4}Z_{i,j}Z_{i,j+1}}\right)e^{i\frac{\pi}{4}X_{i,L_y}}\right),\label{Dy-def}\\
\mathbf{P}^{(2)}_{r}&\equiv\prod_{j=1}^{L_y}\frac{(1+\eta^x_{r,j})}{2}\prod_{i=1}^{L_x}\frac{(1+\eta^y_{r,i})}{2}\, ,
\end{align}
\label{eq:Dtilde2ddef}
\end{subequations}
with similar definitions for $\tilde{\mathbf{D}}_{x;b}^{(2)}$, $\tilde{\mathbf{D}}_{y;b}^{(2)}$ and $\mathbf{P}^{(2)}_{b}$.
\label{Dr(b)def}

The symmetry operator $\mathbf{D}^{(2)}$ is noninvertible and satisfies the fusion rules
\begin{align}
\begin{split}
    &(\mathbf{D}^{(2)})^2\propto \mathbf{P}^{(2)}_{r}\mathbf{P}^{(2)}_{b}\, ,\\
    &\eta^{x}_{r,j}\mathbf{D}^{(2)}=\mathbf{D}^{(2)}\eta^{x}_{r,j}=\mathbf{D}^{(2)}\, ,\\
    &\eta^{y}_{r,i}\mathbf{D}^{(2)}=\mathbf{D}^{(2)}\eta^{y}_{r,i}=\mathbf{D}^{(2)}\, ,\\
    &\eta^x_{b,j+\frac{1}{2}}\mathbf{D}^{(2)}=\mathbf{D}^{(2)}\eta^x_{b,j+\frac{1}{2}}=\mathbf{D}^{(2)}\, ,\\
    &\eta^y_{b,i+\frac{1}{2}}\mathbf{D}^{(2)}=\mathbf{D}^{(2)}\eta^y_{b,i+\frac{1}{2}}=\mathbf{D}^{(2)}\,.
    \end{split}
\end{align}

\subsection{Linear subsystem symmetry-protected topological phases in higher dimensions}
In this section, we will give some examples of higher-dimensional subsystem symmetry-protected topological phases (see~\cite{mana2024kennedy} for more details). Again, we will restrict our discussion to $\mathbb{Z}_2$ or $\mathbb{Z}_2\times\mathbb{Z}_2$ onsite symmetry groups.
\subsubsection{\texorpdfstring{$\mathbb{Z}_2$}{Lg} symmetry}
We give an example of a nontrivial $\mathbb{Z}_2$ symmetry-protected topological phase that is a generalization of~\eqref{eq:Z2SSPTHam}. Let us consider a hypercubic lattice in $d$ spatial dimensions. We denote the coordinate axis by $x_i$ for $i=1,...,d$. We denote the vertices and cube centers of the hypercubic lattice by $v$ and $c$, respectively. We take the lattice spacing to be of unit length and the number of vertices in each $x_i$ direction to be $L$. Hence, the vertices are at coordinates $(i_1,...,i_d)$ where $i_k=1,...,L$. The Hamiltonian for the SSPT is 
\begin{align}
    \mathrm{H}^{\mathbb{Z}_2}_{\text{dD-SSPT}}=-\sum_{v}X_v\prod_{\substack{v'\in\partial c\, ,\\
    c=v+(\frac{1}{2},...,\frac{1}{2})}}Z_{v'}\prod_{\substack{v'\in\partial c\, ,\\
    c=v-(\frac{1}{2},...,\frac{1}{2})}}Z_{v'}\, .
    \label{eq:Z2dDSSPT}
\end{align}
See Figure~6(b) of \cite{mana2024kennedy} for an illustration for $d=3$. 
This Hamiltonian has a rigid linear subsystem symmetry along all the $x_i$ directions and a diagonal line pointing in the $(1,...,1)$ direction, whose explicit formula is a straightforward generalization of~\eqref{eq:2dsubsystemwithdiag}. In addition to that, it also possesses a symmetry that exchanges 
\begin{align}
    X_v\leftrightarrow \prod_{\substack{v'\in\partial c\, ,\\
    c=v+(\frac{1}{2},...,\frac{1}{2})}}Z_{v'}\prod_{\substack{v'\in\partial c\, ,\\
    c=v-(\frac{1}{2},...,\frac{1}{2})}}Z_{v'}\, .
    \label{eq:ddimXtoZtransfZ2}
\end{align}
An explicit operator representation for this symmetry is 
\begin{align}
\begin{split}
\mathbf{D}_\text{DHCIM}^{(d)} 
&\equiv \mathbf{T}_{(1,1,...,1)}^{-1}\mathbf{P}_\text{DHCIM}^{(d)}
\tilde{\mathbf{D}}_{x_d}^{(d)}
\mathbf{H}^{\otimes(d)}
\tilde{\mathbf{D}}_{x_{d-1}}^{(d)}
\mathbf{H}^{\otimes(d)}\times\hdots\\
&\qquad\qquad\times\tilde{\mathbf{D}}_{x_1}^{(d)}
\mathbf{H}^{\otimes(d)}
\tilde{\mathbf{D}}_\text{diag}^{(d)}
\mathbf{P}_\text{DHCIM}^{(d)}, 
\end{split}
\end{align}
where the subscript DHCIM denotes double hypercube Ising model and $\tilde{\mathbf{D}}_{x_i}^{(d)}$, $\tilde{\mathbf{D}}_\text{diag}^{(d)}$, $\mathbf{P}_\text{DHCIM}^{(d)}$, $\mathbf{H}^{\otimes(d)}$ are  straightforward generalizations of the equations~\eqref{eq:Ddpimopdef}. $\mathbf{T}_{(1,1,...,1)}^{-1}$ is a diagonal lattice translation included to obtain the transformation \eqref{eq:ddimXtoZtransfZ2}.
The operator $\mathbf{D}_\text{DHCIM}^{(d)} $ is a noninvertible symmetry and satisfies
\begin{align}
    \left(\mathbf{D}_\text{DHCIM}^{(d)} \right)^2\propto \mathbf{P}_\text{DHCIM}^{(d)}\, .
\end{align}
\subsubsection{\texorpdfstring{$\mathbb{Z}_2\times\mathbb{Z}_2$}{Lg} symmetry}
Here, we give an example of $\mathbb{Z}_2\times\mathbb{Z}_2$ SSPT in higher dimensions. Let us consider two hypercubic lattices dual to each other in $d$ spatial dimensions. We color the lattices red and blue. As before, we denote the coordinate axis by $x_i$ for $i=1,..,d$. We take the lattice spacing to be of unit length and the number of vertices in $x_i$ direction to be $L_{x_i}$. We denote the vertices and cube centers of red and blue sublattices by $v_r$, $c_r$, and $v_b$, $c_b$, respectively. Since the red and blue sublattices are dual to each other, $v_r\equiv c_b$ and $v_b\equiv c_r$. Vertices of the red sublattice are at integer coordinates $(i_1,...,i_d)$ for $i_k=1,...,L_{x_k}$, and those of the blue sublattice are at half-integer coordinates. The Hamiltonian for the SSPT is 
\begin{widetext}
\begin{align}
\begin{split}
    \bm \mathrm{H}_{\text{dD-cluster}}&=-\sum_{v_r}X_{v_r}\prod_{v_b\in \partial c_b}Z_{v_b}-\sum_{v_b}X_{v_b}\prod_{v_r\in \partial c_r}Z_{v_r} \, .
\end{split}
\label{eq:dDcluster}
\end{align}
\end{widetext}
See Figure~6(a) of \cite{mana2024kennedy} for an illustration for $d=3$.
This Hamiltonian has rigid, linear subsystem symmetries
along $x_i$ direction on both red and blue sublattices. We do not provide an explicit expression here as it is a straightforward generalization of \eqref{eq:subsystemsymmetries2Dclstr}. In addition to this, it also possesses a symmetry that exchanges
\begin{align}
    X_{v_r}\leftrightarrow \prod_{v_b\in \partial c_b}Z_{v_b}\, ,\quad X_{v_b}\leftrightarrow \prod_{v_r\in \partial c_r}Z_{v_r} \, .
\end{align}
This is the Kramers-Wannier duality in $d$ dimensions obtained by gauging all the subsystem symmetries. We give an explicit operator representation of this symmetry 
\begin{align}
    \mathbf{D}^{(d)}=\mathbf{T}^{-1}_{(\frac{1}{2},...,\frac{1}{2})}\mathbf{D}_r^{(d)}\mathbf{D}_b^{(d)}
\end{align}
where the operator $\mathbf{D}_{r(b)}^{(d)}$ is defined as
\begin{align}
    \mathbf{D}_{r(b)}^{(d)}\equiv& \mathbf{P}_{r(b)}^{(d)}\tilde{\mathbf{D}}_{x_1;r(b)}^{(d)}\mathbf{H}_{r(b)}^{(d)}\tilde{\mathbf{D}}_{x_2;r(b)}^{(d)}\mathbf{H}_{r(b)}^{(d)}\times\nonumber\\
    &\qquad\qquad \cdots \mathbf{H}_{r(b)}^{(d)}\tilde{\mathbf{D}}_{x_d;r(b)}^{(d)}\mathbf{P}_{r(b)}^{(d)}\, .
\end{align}
The operator $\mathbf{H}_{r(b)}^{(d)}$ denotes the simultaneous action of Hadamard on all vertices. $\tilde{\mathbf{D}}_{x_k;r(b)}^{(d)}$ and $\mathbf{P}_{r(b)}^{(d)}$ are straightforward generalizations of \eqref{eq:Dtilde2ddef}.

\subsection{Noninvertible symmetry-protected topological phases in \texorpdfstring{$1+1$}{Lg}D}

In this subsection, we give a review of noninvertible symmetry-protected topological phases in $1+1$D following~\cite{Seifnashri:2024dsd}. 
Consider $\mathbb{Z}_2\times\mathbb{Z}_2$ SPT in $1+1$ D on a one-dimensional ring with $2L$ sites with periodic boundary conditions. We denote the position of the sites by the subscript $i$, for $i=1,...,2M$. Then $2L+1\equiv 1$. The Hamiltonian for the cluster state is
\begin{align}
    \mathrm{H}_{\text{1D-cluster}}=-\sum_{i=1}^{2L}Z_{i-1}X_iZ_{i+1}\, .
    \label{eq:1D-cluster}
\end{align}
This Hamiltonian~\eqref{eq:1D-cluster} has the following invertible symmetries
\begin{align}
    \eta^e=\prod_{j:\text{even}}X_j\, ,\quad \eta^o=\prod_{j:\text{odd}}X_j\, .
    \label{eq:symmetries1dcluster}
\end{align}
In addition to this, the Hamiltonian is symmetric under the Kramers-Wannier (KW) duality $X_i\leftrightarrow Z_{i-1}Z_{i+1}$, which is a noninvertible symmetry. An explicit operator representation  for the KW duality is 
\begin{align}
    \mathbf{D}=\mathbf{T}^{-1}\mathbf{D}_{e}\mathbf{D}_o  ,
\end{align}
where
\begin{subequations}
\begin{align}
    &\mathbf{D}_{e}\equiv\Big(\prod_{k=1}^{L-1} e^{i\frac{\pi}{4}X_{2k}}e^{i\frac{\pi}{4}Z_{2k}Z_{2k+2}}\Big) e^{i\frac{\pi}{4}X_{2L}}\frac{(1+\eta_{e})}{2},\\
    &\mathbf{D}_{o}\equiv\Big(\prod_{k=1}^{L-1} e^{i\frac{\pi}{4}X_{2k-1}}e^{i\frac{\pi}{4}Z_{2k-1}Z_{2k+1}}\Big) e^{i\frac{\pi}{4}X_{2L-1}}\frac{(1+\eta_{o})}{2}.
\end{align}
\end{subequations}
Since the cluster Hamiltonian is symmetric under $\mathbf{D}$, we can find the equivalence classes of SPTs protected by $\mathbf{D}$ inside the cluster phase. This problem can be tackled by mapping the SPT to SSB using the Kennedy-Tasaki (KT) transformation and finding the various symmetry breaking phases. 
The operator $\mathbf{KT}$ is defined as
\begin{align}
    \mathbf{KT}=\hat{\rm V}\mathbf{D}\hat{\rm V}\, ,
\end{align}
where $\hat{\rm V}=\prod_{i=1}^{2L}CZ_{i,i+1}$ is the cluster entangler. 
$\mathbf{KT}$ has the following action:
\begin{align}
    X_i&\xleftrightarrow{\mathbf{KT}} \hat{X}_i\, ,\nonumber\\
    Z_{i-1}X_iZ_{i+1}&\xleftrightarrow{\mathbf{KT}} \hat{Z}_{i-1}\hat{Z}_{i+1}\, .
\end{align}
Hence, the cluster state Hamiltonian~\eqref{eq:1D-cluster} is mapped to two copies of the Ising model,
\begin{align}
    \mathrm{H}_{\text{Ising}^2}=-\sum_{i:\text{even}}\hat{Z}_{i}\hat{Z}_{i+1}-\sum_{i:\text{odd}}\hat{Z}_{i}\hat{Z}_{i+1}\, .
    \label{eq:Ising^2}
\end{align}
Here, we used a hat on Pauli operators on the SSB side to distinguish them from Pauli operators on the SPT side. 
The order parameters for this SSB Hamiltonian~\eqref{eq:Ising^2} can be taken to be $Z_1$ and $Z_2$. 
The Hamiltonian possesses the following symmetries (obtained by applying $\mathbf{KT}$ on~\eqref{eq:symmetries1dcluster}):
\begin{align}
    \hat{\eta}_{e}=\prod_{j:\text{even}}\hat{X}_j\, ,\quad \hat{\eta}_o=\prod_{j:\text{odd}}\hat{X}_j\, .
\end{align}
These symmetries are spontaneously broken and lead to a ground state degeneracy of 4. There is an additional symmetry:  $\hat{\rm V}=\prod_{i=1}^{2L}CZ_{i,i+1}$ obtained by applying $\mathbf{KT}$ on $\mathbf{D}$. This can be seen from the identity~\cite{Seifnashri:2024dsd}
\begin{align}
   \mathbf{P} \left(\hat{\rm V}\mathbf{D}\hat{\rm V}\right)\mathbf{D}\mathbf{P}\propto \mathbf{P}\hat{\rm V}\left(\hat{\rm V}\mathbf{D}\hat{\rm V}\right)\mathbf{P}\, ,
   \label{eq:pvdvdp=pdvp}
\end{align}
where $\mathbf{P}=\frac{(1+\eta_e)}{2}\frac{(1+\eta_o)}{2}$ is the projection onto the $\mathbb{Z}_2\times\mathbb{Z}_2$ symmetric sector. This is the statement that the action of $\mathbf{D}$ on the symmetric sector is mapped to the action of $\hat{\rm V}$ on the symmetric sector after applying $\mathbf{KT}$. To argue this, we note that at the level of operators, on the symmetric sector
\begin{align}
    (X_i\xrightarrow{\mathbf{D}}Z_{i-1}Z_{i+1})\xrightarrow{\mathbf{KT}}(\hat{X}_i\xrightarrow{\hat{\rm V}}\hat{Z}_{i-1}\hat{X}_{i}\hat{Z}_{i+1})\,.
\end{align}
We note that if two invertible operators $O$ and $\tilde{O}$ act in the same way on all linear  operators, i.e. $OQO^{-1}=\tilde{O}Q\tilde{O}^{-1}$ for all linear operators $Q$, then  $O=c\tilde{O}$ for some constant $c$. In the symmetric subspace, the operators $\mathbf{P} \hat{\rm V}\mathbf{D}\hat{\rm V}\mathbf{D}\mathbf{P}$ and $\mathbf{P}\mathbf{D}\hat{\rm V}\mathbf{P}$ are invertible. On the symmetric subspace, a general linear operator can be taken to be a symmetric operator. Since they act in the same way on all the symmetric operators, they should be proportional. Any states orthogonal to the states in the symmetric subspace are annihilated by these two operators. So they should be proportional in the whole Hilbert space. 
Since the projection can be absorbed into $\mathbf{D}$, we have
\begin{align}
   \left(\hat{\rm V}\mathbf{D}\hat{\rm V}\right)\mathbf{D}\propto \hat{\rm V}\left(\hat{\rm V}\mathbf{D}\hat{\rm V}\right)\,.
   \label{eq:vdvd=dv}
\end{align}
The symmetry $\hat{\rm V}$ is still preserved in all symmetry-broken ground states of the particular Hamiltonian~\eqref{eq:Ising^2}.

Since we applied $\mathbf{KT}$ to the cluster phase, the $\mathbb{Z}_2\times \mathbb{Z}_2$ subsystem symmetries remain broken. Under this condition, could there be other phases? 
To answer this question, we need to examine the possible additional symmetry that is not spontaneously broken. We have seen one above: $\hat{\rm V}$, but  there could other possibilities, such as $\hat{\rm V}\hat{\eta}_o$ and $\hat{\rm V}\hat{\eta}_e$. However, the diagonal combination $\hat{\rm V}\hat{\eta}_e\hat{\eta}_o$ is an anomalous symmetry (this is the boundary symmetry of the CZX model~\cite{Chen:2011bcp}) and cannot be preserved. 
\subsubsection{\texorpdfstring{$\hat{\mathbb{Z}}_2^{\rm V}$}{Lg} preserving phase}
It can be seen that $\hat{\rm V}$ is the preserved symmetry for \eqref{eq:Ising^2} by explicitly checking the four broken ground states. In Appendix~\ref{sec:orderparameter}, we prove a Lemma (\ref{lemma:consistentorderparameters}) that allows us to check whether a symmetry is spontaneously broken or not by identifying order parameters. Now, we check whether all the conditions of Lemma~\ref{lemma:consistentorderparameters} are satisfied. The order parameters $Z_1$ and $Z_2$ commute with themselves, with the symmetry $\hat{\rm V}$, and with the Hamiltonian \eqref{eq:Ising^2}. The symmetry generators $\hat{\eta}_o$ and $\hat{\eta}_e$ satisfy: 1) $\{\hat{\eta}_o,Z_1
\}=0$, 2) $\{\hat{\eta}_e,Z_2
\}=0$, 3) $[\hat{\eta}_o,Z_2]=0$, and 4) $[\hat{\eta}_e,Z_1]=0$. We also know that the Hamiltonian \eqref{eq:Ising^2} has $2^2=4$ ground states. Therefore, all the conditions in Lemma~\ref{lemma:consistentorderparameters} are satisfied and $\hat{\rm V}$ is a preserved symmetry. Hence, for this case, the SSB Hamiltonian is \eqref{eq:Ising^2} and the corresponding SPT Hamiltonian after applying $\mathbf{KT}$ is \eqref{eq:1D-cluster}. 
\subsubsection{\texorpdfstring{$\text{diag}(\hat{\mathbb{Z}}_2^{\rm V}\times\hat{\mathbb{Z}}_2^e)$}{Lg} preserving phase}
Here we assume $L$ is a multiple of four.
The Hamiltonian for the SSB phase is~\cite{Seifnashri:2024dsd}
\begin{align}
    \hat{\mathrm{H}}_{\text{odd}}=\sum_{i=1}^{L/2}\hat{Z}_{2i-1}\hat{Z}_{2i+1}-\sum_{i=1}^{L/2}\hat{Y}_{2i}\hat{Y}_{2i+2}(1+\hat{Z}_{2i-1}\hat{Z}_{2i+3})\, .
\end{align}
We note that the Hamiltonian still has all the symmetries $\hat{\eta}_e$, $\hat{\eta}_o$ and $\hat{\rm V}$. However, the symmetries $\hat{\eta}_e$ and $\hat{\eta}_o$ are broken while the diagonal combination $\hat{\rm V}\hat{\eta}_e$ is preserved on the ground states.  The order parameters for this phase can be taken to be $\hat{Z}_1$ and $\hat{Y}_2(1-\hat{Z}_1\hat{Z}_3)$ that satisfy the conditions in Lemma~\ref{lemma:consistentorderparameters}, allowing us to conclude that $\hat{\rm V}\hat{\eta}_e$ is an unbroken symmetry. The original SPT Hamiltonian that gives rise to this SSB Hamiltonian is 
\begin{align}
    &\mathrm{H}_{\text{odd}}=\sum_{i=1}^{L/2}Z_{2i-1}X_{2i}Z_{2i+1}-\sum_{i=1}^{L/2}Y_{2i}X_{2i+1}Y_{2i+2}\nonumber\\
    &\qquad\qquad+\sum_{i=1}^{L/2}Z_{2i-1}Z_{2i}X_{2i+1}Z_{2i+2}Z_{2i+3}\, .
    \label{eq:oddHam}
\end{align}
This Hamiltonian has a unique ground state
\begin{align}
    \ket{\text{odd}}= \prod_{i=1}^{L/2}CZ_{2i-1,2i+1}\prod_{j=1}^LCZ_{j,j+1}\ket{-}^{\otimes_i s_i}\, ,
\end{align}
where $\otimes_i s_i$ denotes the tensor product over all the sites.
\subsubsection{\texorpdfstring{$\text{diag}(\hat{\mathbb{Z}}_2^{\rm V}\times\hat{\mathbb{Z}}_2^o)$}{Lg} preserving phase}
This is obtained by exchanging $o\longleftrightarrow e$.

\vspace{1cm}
The different phases can be distinguished from the analysis of the interface modes between them~\cite{Seifnashri:2024dsd}. We will use this technique in Appendix~\ref{sec:Interfaceanalysis},~\ref{sec:interfaceanalysisotherspt}, and \ref{sec:Interfaceanalysis3D} to distinguish between various phases that we obtain in higher dimensions.
\section{Noninvertible higher-order subsystem symmetry-protected topological phases: Corner modes}\label{sec:NSSPT}
We have reviewed SPT phases protected by subsystem symmetries in two and higher dimensions and a one-dimensional SPT phase protected by a noninvertible symmetry in the last section. We now move on to our main results in this paper. 
We show that there exist higher-order SPT phases protected by subsystem and noninvertible symmetries. 
In this section, we focus on linear subsystem symmetries.

\subsection{Noninvertible second-order SSPT phases from \texorpdfstring{$\mathbb{Z}_2\times\mathbb{Z}_2$}{Lg} cluster phase in \texorpdfstring{$2+1$}{Lg}D}
\label{sec:noninvertibleZ2xZ2}
To construct the noninvertible higher-order subsystem symmetry-protected topological phase, first we note that the $\mathbb{Z}_2\times\mathbb{Z}_2$ SSPT~\eqref{eq:2dcluster} is also invariant under $\mathbf{D}^{(2)}$. This noninvertible symmetry can further break the cluster phase into distinct phases protected by $\mathbf{D}^{(2)}$. 
Our analysis of interface modes in Appendix~\ref{sec:Interfaceanalysis} will show that this is indeed the case. We examine two types of interfaces on a torus: a line interface where two cylindrical regions are separated by two lines, and a rectangular interface where a rectangular region and it's complement are separated by a rectangle. Our results show that the line interface fails to reveal the presence of protected interface modes between the two phases. In contrast, the rectangular interface clearly distinguishes the two phases through the appearance of corner modes, which are protected by the noninvertible symmetry.
Hence, we obtain a second-order noninvertible SSPT.

We use the Kennedy-Tasaki (KT) transformation to study higher-order noninvertible SSPTs. KT map a symmetry-protected topological phase to a symmetry breaking phase. For the particular case we are going to analyze, we use the KT transformation provided in~\cite{mana2024kennedy}
\begin{align}
    \mathbf{KT}^{(2)}\equiv \hat{\rm V}^{(2)}\mathbf{D}^{(2)}\hat{\rm V}^{(2)}\, ,
\end{align}
where $\hat{\rm V}^{(2)}$ is the cluster entangler between red and blue sublattices. Explicitly,
\begin{align}
    \hat{\rm V}^{(2)}\equiv \prod_{v_b}\prod_{v_r\in\partial (p_r=v_b)}CZ_{v_r,v_b}\, .
\end{align}
$\mathbf{KT}^{(2)}$ acts in the following way:
\begin{align}
    &X_{v_r}\xleftrightarrow{\mathbf{KT}^{(2)}}\hat{X}_{v_r}\, ,\quad X_{v_b}\xleftrightarrow{\mathbf{KT}^{(2)}}\hat{X}_{v_b}\, ,\nonumber\\
    &\quad\begin{array}{ccc}
        Z_{v_b} & & Z_{v_b} \\
                & X_{v_r} & \\
        Z_{v_b} & & Z_{v_b}
    \end{array}\xleftrightarrow{\mathbf{KT}^{(2)}}
  \begin{array}{ccc}
        \hat{Z}_{v_b} & & \hat{Z}_{v_b} \\
        & & \\
        \hat{Z}_{v_b} &  & \hat{Z}_{v_b}
\end{array}  
    \, ,
    \nonumber\\
    &\quad\begin{array}{ccc}
        Z_{v_r} & & Z_{v_r} \\
                & X_{v_b} & \\
        Z_{v_r} & & Z_{v_r}
    \end{array}\xleftrightarrow{\mathbf{KT}^{(2)}}
    \begin{array}{ccc}
        \hat{Z}_{v_r} & & \hat{Z}_{v_r} \\
        & & \\
        \hat{Z}_{v_r} &  & \hat{Z}_{v_r}
    \end{array}
    \, .
\end{align}
Here, we have used a hat symbol $\hat{}$ on Pauli operators on the dual side to distinguish them from the Pauli operators on the SSPT side.
Hence, under $\mathbf{KT}^{(2)}$, we map the cluster state Hamiltonian~\eqref{eq:2dcluster} to spontaneous subsystem symmetry breaking (SSSB) Hamiltonian  
\begin{align}
    \mathrm{H}_{\text{SSSB}}^{(2)}=-\sum_{v_r}\begin{array}{cc}
        \hat{Z}_{v_r} & \hat{Z}_{v_r} \\
        \hat{Z}_{v_r} & \hat{Z}_{v_r}
    \end{array}-\sum_{v_b}\begin{array}{cc}
        \hat{Z}_{v_b} & \hat{Z}_{v_b} \\
        \hat{Z}_{v_b} & \hat{Z}_{v_b}
    \end{array}\, .
    \label{eq:2DSSSBtwocopiesPI}
\end{align}
We note that $\mathbf{KT}^{(2)}$ maps the subsystem symmetry line operators to dual subsystem symmetry line operators. The action of $\mathbf
D^{(2)}$ on the subsystem symmetric sector is mapped under $\mathbf{KT}^{(2)}$ to the action of $\hat{\rm V}^{(2)}$ on the symmetric sector. Following the same arguments around \eqref{eq:pvdvdp=pdvp} and \eqref{eq:vdvd=dv}, we conclude that 
\begin{align}
    \mathbf{KT}^{(2)}\mathbf{D}^{(2)}\propto \hat{\rm V}^{(2)}\mathbf{KT}^{(2)}\,.
\end{align} 
Hence, the symmetries of this dual SSSB Hamiltonian are
\begin{subequations}
    \begin{align}
    &\hat{\eta}^{x}_{r,j}=\prod_{i=1}^{L_x}\hat{X}_{i,j}\, ,\quad \hat{\eta}^{y}_{r,i}=\prod_{j=1}^{L_y}\hat{X}_{i,j}\,,\\
 &\hat{\eta}^{x}_{b,j}=\prod_{i=1}^{L_x}\hat{X}_{i+\frac{1}{2},j+\frac{1}{2}},\quad \hat{\eta}^{y}_{b,i}=\prod_{j=1}^{L_y}\hat{X}_{i+\frac{1}{2},j+\frac{1}{2}}\,,\\
&  \hat{\mathrm{V}}^{(2)}=\prod_{v_r}\prod_{v^r\in\partial p^r}\mathrm{CZ}_{v^r,p^r}\, .
\end{align}
\label{eq:symmetriesofSSSB}
\end{subequations}
These are dual $\hat{\mathbb{Z}}_2\times\hat{\mathbb{Z}}_2$ subsystem symmetries along with the $\hat{\mathbb{Z}}_2^{\rm V}$ 0-form symmetry. We note that only the $\hat{\mathbb{Z}}_2\times\hat{\mathbb{Z}}_2$ subsystem symmetry is broken while the $\hat{\mathbb{Z}}_2^{\rm V}$ 0-form symmetry is preserved. 
We also note that the global part of the subsystem symmetries that we define as 
\begin{subequations}
    \begin{align}
    \hat{\eta}_r=\prod_{j=1}^{L_y}\hat{\eta}^x_{r,j}=\prod_{i=1}^{L_x}\hat{\eta}^y_{r,i}\, ,\\
    \hat{\eta}_b=\prod_{j=1}^{L_y}\hat{\eta}^x_{b,j}=\prod_{i=1}^{L_x}\hat{\eta}^y_{b,i}\, ,
\end{align} 
\end{subequations}
plays an important role in the following discussion. 

To analyze other possible symmetry-protected topological phases under $\mathbf{D}^{(2)}$ in the cluster phase, it is necessary and sufficient to analyze the various possible symmetry-breaking patterns. Since we are restricted to the cluster phase as far as $\mathbb{Z}_2\times\mathbb{Z}_2$ subsystem symmetries are considered, on the symmetry breaking side, \textit{all} the subsystem symmetries have to be broken. Hence, we are led to analyze the various possible $0$-form symmetries that are preserved. 

It turns out that one cannot preserve all diagonal combinations of subsystem symmetries with $\hat{\rm V}^{(2)}$. Many diagonal combinations are anomalous. Detailed analysis of such anomalous symmetries are in Appendix~\ref{sec:anomaly}. 
Here we consider preserving diagonal combinations of $\hat{\rm V}^{(2)}$ with the global part of subsystem symmetries that are not anomalous.   
\subsubsection{\texorpdfstring{$\hat{\mathrm{V}}^{(2)}$}{Lg} is preserved}
In this case, on the symmetry breaking side, we obtain the Hamiltonian~\eqref{eq:2DSSSBtwocopiesPI}. 
The order parameters for this phase are
$\{ \hat{Z}_{i,1}, \hat{Z}_{i+\frac{1}{2}\frac{3}{2}}, \hat{Z}_{1,j}, \hat{Z}_{\frac{3}{2},j+\frac{1}{2}} \}_{i=1,...,L_x; j=2,...,L_y}$.
There are in total $2(L_x+L_y-1)$ independent order parameters. The order parameters commute with $\rm \hat{V}^{(2)}$ and therefore $\rm \hat{V}^{(2)}$ is unbroken according to Lemma~\ref{lemma:consistentorderparameters}. The original SSPT that gives rise to this Hamiltonian is in~\eqref{eq:2dcluster}.
\subsubsection{\texorpdfstring{$\hat{\mathrm{V}}^{(2)}\hat{\eta}_r$}{Lg} is preserved}
For this case, we assume $L_x$ and $L_y$ are even. First, we provide an SSSB Hamiltonian whose ground state(s) preserves the symmetry $\hat{\mathrm{V}}^{(2)}\hat{\eta}_r$: 
\begin{widetext}
\begin{align}
    \hat{\mathrm{H}}_{\text{blue}}=\sum_{v_b}\begin{array}{ccc}
         \hat{Z}_{v_b} &  &\hat{Z}_{v_b}\\
         & & \\
         \boxed{\hat{Z}_{v_b}} &  &\hat{Z}_{v_b}
    \end{array}-\sum_{v_r}\begin{array}{ccc}
         \hat{Y}_{v_r} &  &\hat{Y}_{v_r}\\
         & & \\
         \boxed{\hat{Y}_{v_r}} &  &\hat{Y}_{v_r}
    \end{array}-\sum_{v_r}\begin{array}{ccccc}
        \hat{Z}_{v_b} & & & & \hat{Z}_{v_b} \\
         & \hat{Y}_{v_r} & &  \hat{Y}_{v_r} & \\
         & & & & \\
         & \boxed{\hat{Y}_{v_r}} & &  \hat{Y}_{v_r} & \\
         \hat{Z}_{v_b} & & & & \hat{Z}_{v_b} \\
    \end{array}\, , 
    \label{eq:Hblue}
\end{align}
\end{widetext}
where we have boxed certain Pauli operators to indicate which vertex is summed. We note that despite the Hamiltonian commuting with $\hat{\rm V}^{(2)}$, the ground states break it. To verify that instead $\hat{\mathrm{V}}^{(2)}\hat{\eta}_r$ is the unbroken symmetry, we invoke Lemma~\ref{lemma:consistentorderparameters}. The order parameters for this phase are of the form 
$\hat{Z}_{v_b}$for $v_b$ of the form $(i+\frac{1}{2}, \frac{3}{2})$ and $( \frac{3}{2}, j+\frac{1}{2})$
and $\hat{Y}_{v_r}\left(1-\begin{array}{ccc}
    \hat{Z}_{v_b} & &  \hat{Z}_{v_b} \\
                  & & \\
     \hat{Z}_{v_b}& &  \hat{Z}_{v_b}            
\end{array}\right)$ for $v_r$ of the form $(i,1)$ and $(1,j)$ for $i=1,...,L_x$ and $j=1,...,L_y$. Then the ground state configuration is specified by the values of $\hat{Z}_{i+\frac{1}{2},\frac{3}{2}}=\pm 1$, $\hat{Z}_{\frac{3}{2},j+\frac{1}{2}}=\pm 1$,  $\hat{Y}_{i,1}=\pm 1$ and $\hat{Y}_{1,j}=\pm 1$ for $i=1,...,L_x$ and $j=1,...,L_y$. In total, there are $2^{2(L_x+Ly-1)}$ such possibilities, and hence there are that many ground states for the SSSB Hamiltonian~\eqref{eq:Hblue}. The order parameters commute with $ \hat{\mathrm{V}}^{(2)}\hat{\eta}_r$ and hence we apply Lemma~\ref{lemma:consistentorderparameters} in Appendix~\ref{sec:orderparameter} to conclude that $ \hat{\mathrm{V}}^{(2)}\hat{\eta}_r$ is unbroken. 

The SSPT Hamiltonian that gives rise to this particular SSSB symmetry breaking pattern can be found by applying $\mathbf{KT}^{(2)}$:
\begin{widetext}
    \begin{align}
    \mathrm{H}_{\text{blue}}=\sum_{v_r}\begin{array}{ccc}
         Z_{v_b} &  &Z_{v_b}\\
         & \boxed{X_{v_r}} & \\
         Z_{v_b} &  & Z_{v_b}
    \end{array}-\sum_{v_b}\begin{array}{ccc}
         Y_{v_r} &  &Y_{v_r}\\
         & \boxed{X_{v_b}} & \\
         Y_{v_r} &  & Y_{v_r}
    \end{array}-\sum_{v_b}\begin{array}{ccccc}
        Z_{v_b} & & & & Z_{v_b} \\
         & Z_{v_r} & &  Z_{v_r} & \\
         & & \boxed{X_{v_b}} & & \\
         & Z_{v_r} & &  Z_{v_r} & \\
         Z_{v_b} & & & & Z_{v_b} \\
    \end{array}\, .
    \label{eq:blueHam}
\end{align}
In the above equation, the boxed vertices are summed over. The above Hamiltonian describes a higher-order noninvertible subsystem symmetry-protected topological phase. As far as the $\mathbb{Z}_2\times\mathbb{Z}_2$ subsystem symmetries are concerned, this Hamiltonian belongs to the cluster phase \eqref{eq:2dcluster}.  However, this Hamiltonian is in a different phase from \eqref{eq:2dcluster} when $\mathbf{D}^{(2)}$ is also included in the set of symmetries. This can be seen from an edge mode analysis at the interface between the two phases described by \eqref{eq:blueHam} and \eqref{eq:2dcluster} (see Appendix~\ref{sec:Interface2DclusterHblue} for the analysis). It turns out that the two phases are distinct by corner modes that appear at the corners of the interface between the two phases (see Figure~\ref{fig:rectangularinterface}). The corner modes are robust to any symmetric local perturbations to the interface Hamiltonian (see Appendix~\ref{sec:stability2dclusterblue} for an argument).
\begin{figure*}
    \centering
    \includegraphics[scale=1]{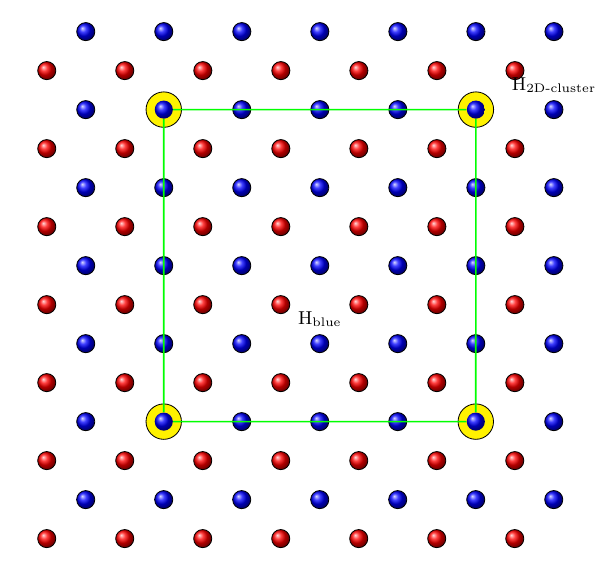}
    \caption{Rectangular interface region with the interface line colored green. The four yellow disks at the four corners indicate the corner modes obtained by putting $\rm H_{\text{blue}}$ inside the rectangular region and $\rm H_{\text{2D-cluster}}$ outside it.}
    \label{fig:rectangularinterface}
\end{figure*}
This is an indication of the fact that the two phases are distinct as a higher-order noninvertible symmetry-protected topological phase.  (We refer the readers to~\cite{You:2018srv} for discussions on higher-order bosonic and fermionic SPTs, \cite{Rasmussen:2018yjy} for a classification of higher-order bosonic SPTs, and~\cite{You:2024syf} for higher-order subsystem symmetric SPTs, all for invertible symmetries.)

The Hamiltonian~\eqref{eq:blueHam} has a unique ground state
\begin{align}
    \ket{\text{blue}}=\prod_{ v_b} CZ_{v_b,v_b+(1,1)}CZ_{v_b,v_b+(-1,1)}\prod_{v_r}\prod_{v_b\in\partial (p_b=v_r)}CZ_{v_r,v_b}\ket{+}^{\otimes \Delta_{v_b}}\ket{-}^{ \otimes\Delta_{v_r}}\, ,
    \label{eq:blusstate}
\end{align}
where $\Delta_{ v_b}$ and $\Delta_{v_r}$ denote the set of blue and red vertices respectively.
\end{widetext}
The state $\ket{\text{blue}}$ is related to the $2+1$D cluster state $\ket{\text{2D-cluster}}$ (ground state of \eqref{eq:2dcluster}) by a finite-depth circuit
\begin{align}
    \prod_{ v_r}Z_{v_r}\prod_{ v_b} CZ_{v_b,v_b+(1,1)}CZ_{v_b,v_b+(-1,1)}\, .
\end{align}
\subsubsection{\texorpdfstring{$\hat{\mathrm{V}}^{(2)}\hat{\eta}_b$}{Lg} is preserved}
This case is similar to the previous case where $\hat{\rm V}^{(2)}\hat{\eta}_r$ was preserved. On the symmetry breaking side, we get a Hamiltonian $\hat{\rm H}_{\text{red}}$ obtained by $v_b\leftrightarrow v_r$ in \eqref{eq:Hblue}. After applying the KT transformation, we get a Hamiltonian $\rm H_{\text{red}}$ and its unique ground state $\ket{\text{red}}$ obtained again by $v_b\leftrightarrow v_r$ and $p_b\leftrightarrow p_r$ in \eqref{eq:blueHam} and \eqref{eq:blusstate}. The Hamiltonian $\mathrm{H}_{\text{red}}$ is in a different phase from both $\mathrm{H}_{\text{2D-cluster}}$ and $\mathrm{H}_{\text{blue}}$. Any pair of them can be distinguished by the corner modes that appear at the interface Hamiltonian between the two phases. We refer the reader to Appendix~\ref{sec:Interface2DclusterHblue},\ref{sec:InterfaceHblueHred} for a detailed analysis.
\vspace{1cm}\\ 
\par We provide an analysis of anomalous symmetries in Appendix~\ref{sec:anomaly}. The general conclusion is that the product of the form $\mathcal{U}\hat{\rm V}$ where $\mathcal{U}$ contains a product of the adjacent subsystem symmetries from the different sublattices and $\mathcal{U}\neq \hat{\eta}_r\hat{\eta}_b$ are anomalous.  In Appendix~\ref{sec:nonanomalous}, we argue that $\hat{\rm V}^{(2)}\hat{\eta}_r\hat{\eta}_b$ is not anomalous in two spatial dimensions and give an example Hamiltonian with such symmetry and a unique short-range entangled ground state in two spatial dimensions.
 
We analyze some other non-anomalous diagonal combinations of preserved symmetries in Appendix~\ref{sec:otherNSSPT}. To be more precise, in Appendix~\ref{sec:otherNSSPTinclusterstate}, we analyze the cases of preserving the following diagonal combinations: 1) $\hat{\mathrm{V}}^{(2)}\hat{\eta}^x_{r,k}$ for $k\in\{1,...,L_y\}$, 2) $\hat{\rm V}^{(2)}\hat{\eta}^x_{r,k}\hat{\eta}^y_{ r,m}$ for $k\in\{1,...,L_y\}$ and $m\in\{1,...,L_x\}$, and 3) $\hat{\rm V}^{(2)}\hat{\eta}^x_{r,k}\hat{\eta}^y_{ b,m}$ for $k\in\{1,...,L_y\}$ and $m\in\{1,...,L_x\}$. The dual SSPTs for these cases are in different phases protected by $\mathbf{D}^{(2)}$. They differ from the cluster phase~\eqref{eq:2dcluster} by interface modes around the location where the subsystem symmetry lines (appearing in the preserved symmetry) intersect the interface lines. We analyze the interface mode between $\mathrm{H}_{\text{2D-cluster}}$ and the SPT dual of $\hat{\mathrm{V}}^{(2)}\hat{\eta}^x_{r,k}$ phase in Appendix~\ref{sec:interface2dclusterxkblue}. In addition to that, we also analyze the various noninvertible phases in another $\mathbb{Z}_2\times\mathbb{Z}_2$ SSPT that is invariant under $\mathbf{D}^{(2)}$ in Appendix~\ref{sec:Z2SSPTstackedtocluster}. We find that two such phases when separated by a straight-line interface differ by edge modes at the interface lines (see Appendix~\ref{sec:interfaceanalysisotherspt} for detailed analysis). Hence, they are in different phases in the usual sense (first order SPT) protected by $\mathbf{D}^{(2)}$.  
\subsection{Noninvertible \texorpdfstring{$d^{\text{th}}$}{Lg} order SSPT phases from \texorpdfstring{$\mathbb{Z}_2\times\mathbb{Z}_2$}{Lg} cluster phase in \texorpdfstring{$d+1$}{Lg}D}
In this section, we generalize the construction of the previous higher-order SSPTs protected by noninvertible symmetry to general dimensions $d\geq 3$. We start with $\mathbb{Z}_2\times\mathbb{Z}_2$ cluster state in general dimensions. This is also invariant under $\mathbf{D}^{(d)}$. Hence, we can further break the cluster phase into distinct phases protected by $\mathbf{D}^{(d)}$. Two of the resulting distinct phases differ from the cluster phase by corner modes. Hence, they are $d^{\text{th}}$ order noninvertible SSPT phases. To explore this possibility, we use the same strategy as before.

We use the Kennedy-Tasaki (KT) transformation to map SSPT to the SSSB phase. We use the following KT transformation provided in~\cite{mana2024kennedy}
\begin{align}
    \mathbf{KT}^{(d)}\equiv \hat{\rm V}^{(d)}\mathbf{D}^{(d)}\hat{\rm V}^{(d)}
\end{align}
where $\hat{\rm V}^{(d)}$ is the cluster entangler between red and blue sublattices. Explicitly,
\begin{align}
    \hat{\rm V}^{(d)}\equiv \prod_{v_b}\prod_{v_r\in\partial(p_r=v_b)}CZ_{v_r,v_b}\,.
\end{align}
$\mathbf{KT}^{(d)}$ acts in the following way:
\begin{subequations}
  \begin{align}
    X_{v_r}\xrightarrow{\mathbf{KT}^{(d)}} \hat{X}_{v_r}\, ,&\quad X_{v_b}\xrightarrow{\mathbf{KT}^{(d)}} \hat{X}_{v_b}\, ,\\
    X_{v_r}\prod_{v_b\in\partial (c_b=v_r)}Z_{v_b}&\xrightarrow{\mathbf{KT}^{(d)}}\prod_{v_b\in\partial (c_b=v_r)}\hat{Z}_{v_b}\, ,\\
    X_{v_b}\prod_{v_r\in\partial (c_r=v_b)}Z_{v_r}&\xrightarrow{\mathbf{KT}^{(d)}} \prod_{v_r\in\partial (c_r=v_b)}\hat{Z}_{v_r}\, .
\end{align}  
\end{subequations}
Hence, under $\mathbf{KT}^{(d)}$, we map the cluster state Hamiltonian~\eqref{eq:dDcluster} to spontaneous subsystem symmetry breaking (SSSB) Hamiltonian
\begin{align}
    \mathrm{H}_{\text{SSSB}}^{(d)}=-\sum_{c_r}\prod_{v_r\in \partial c_r}\hat{Z}_{v_r}-\sum_{c_b}\prod_{v_b\in\partial c_b}\hat{Z}_{v_b}\, .
    \label{eq:HdSSSB}
\end{align}
The dual SSSB Hamiltonian is again $\mathbb{Z}_2\times\mathbb{Z}_2$ subsystem symmetric. Under $\mathbf{KT}^{(d)}$, the noninvertible symmetry $\mathbf{D}^{(d)}$ is mapped to $\hat{\rm V}^{(d)}$ on the SSSB side. The Hamiltonian~\eqref{eq:HdSSSB} is symmetric under $\hat{\rm V}^{(d)}$. Now we define the global part of the subsystem symmetries
\begin{subequations}
   \begin{align}
    &\hat{\eta}_r=\prod_{v_r}\hat{X}_{v_r}\, ,\\
    &\hat{\eta}_b=\prod_{v_b}\hat{X}_{v_b}\, .
\end{align} 
\end{subequations}

We repeat the analysis of possible symmetry breaking patterns to find possible symmetry-protected topological phases. As far as the subsystem symmetry $\mathbb{Z}_2\times\mathbb{Z}_2$ is concerned, the cluster phase is mapped to $\mathbb{Z}_2\times\mathbb{Z}_2$ SSSB phase. However, there are various possible choices for the preserved symmetry.
\subsubsection{\texorpdfstring{$\hat{\rm V}^{(d)}$}{Lg} is preserved}
For this case, on the symmetry breaking side, we obtain the Hamiltonian~\eqref{eq:HdSSSB}. The original SSPT Hamiltonian that gives rise to this Hamiltonian is \eqref{eq:dDcluster}.
\subsubsection{\texorpdfstring{$\hat{\rm V}^{(d)}\hat{\eta}_r$}{} is preserved}
We assume all the $L_{x_i}$ are even. The SSSB Hamiltonian that preserves $\hat{\rm V}^{(d)}\hat{\eta}_r$ is 
\begin{widetext}
    \begin{align}
    \hat{\rm H}_{\text{blue}}^{(d)}=\sum_{c_b}\prod_{v_b\in\partial c_b}\hat{Z}_{v_b}-\sum_{c_r}\prod_{v_r\in \partial c_r}\hat{Y}_{v_r}\left(1+\prod_{v_b\in \partial(c_b=v_r)}\hat{Z}_{v_b}\right)\, .
    \label{eq:Hblued}
\end{align}
The SSPT Hamiltonian that gives rise to this SSSB Hamiltonian is found by applying $\mathbf{KT}^{(d)}$
\begin{align}
    \mathrm{H}_{\text{blue}}^{(d)}=\sum_{v_r}X_{v_r}\prod_{v_b\in\partial(c_b=v_r)}Z_{v_b}-\sum_{v_b}X_{v_b}\prod_{v_r\in\partial(c_r=v_b)}Y_{v_r}-\sum_{v_b}X_{v_b}\prod_{v_r\in\partial(c_r=v_b)}Z_{v_r}\left(\prod_{v_b\in\partial(c_b=v_r)}Z_{v_b}\right)\, .
    \label{eq:Hdblue}
\end{align}
\end{widetext}
See Figures~\ref{fig:3D-SSSB} and \ref{fig:3D-NI-SSPT} for illustrations of the terms in $\hat{\rm H}^{(3)}_{\text{blue}}$ and $\rm H^{(3)}_{\text{blue}}$.
This Hamiltonian describes an SSPT different from the cluster state protected by $\mathbf{D}^{(d)}$. The distinction arises from the corner modes that would appear when two of them are separated by a hypersurface interface region (see Figure~\ref{fig:3D-cornermode} for an illustration of corner modes in 3D). Although we do not explicitly give a detailed analysis of the interface modes in this manuscript, it is straightforward to generalize the analysis in Appendix~\ref{sec:Interfaceanalysis}. Hence, this is a $d^\text{th}$-order SSPT protected by $\mathbf{D}^{(d)}$, where the order $k$ indicates that the interface modes appear on $d-k$ cells on a hypercubic interface.
\begin{figure*}[]

    \begin{subfigure}[b]{2\columnwidth}
         \centering
   
    \includegraphics[width=0.5\linewidth]{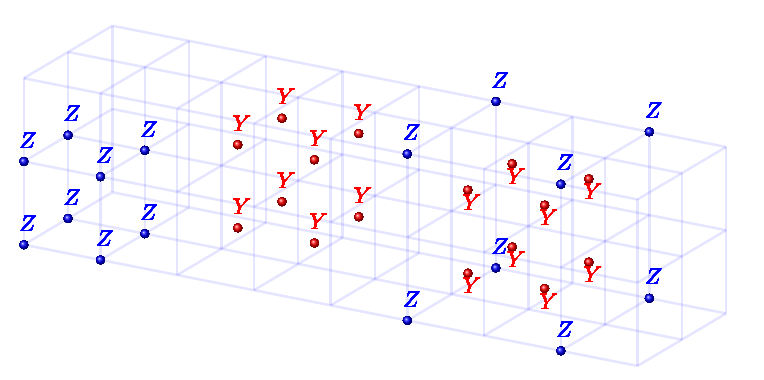}
    \caption{An illustration of the terms in the 3D Hamiltonian $\hat{\rm H}^{(3)}_{\text{blue}}$.}
     \label{fig:3D-SSSB}
    \end{subfigure}
    \begin{subfigure}[b]{2\columnwidth}
         \centering
          \includegraphics[width=0.5\linewidth]{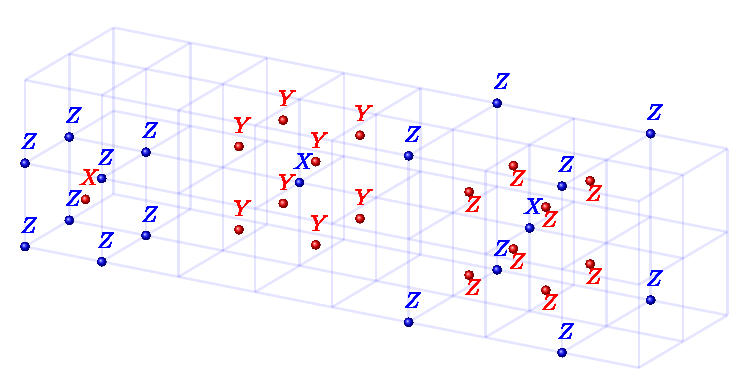}
         \caption{An illustration of the terms in the 3D Hamiltonian $\rm H^{(3)}_{\text{blue}}$.}
        \label{fig:3D-NI-SSPT}
    \end{subfigure}
    \caption{Illustrations of terms in the SSSB and SSPT Hamiltonian in 3D}
\end{figure*}
\begin{figure*}
    \centering
    \includegraphics[scale=1]{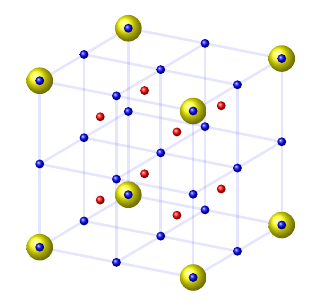}
    \caption{An illustration of corner modes that distinguish $\rm H^{(3)}_{\text{blue}}$ and $\rm H_{\text{3D-cluster}}$}
    \label{fig:3D-cornermode}
\end{figure*}
\subsubsection{\texorpdfstring{$\hat{\rm V}^{(d)}\hat{\eta}_b$}{Lg} is preserved}
This case is similar to the previous case where $\hat{\rm V}^{(d)}\hat{\eta}_r$ was preserved. On the symmetry breaking side, we get a Hamiltonian $\hat{\rm H}_{\text{red}}^{(d)}$ obtained by $v_b\leftrightarrow v_r$ in \eqref{eq:Hblued}. 
After applying KT, we get a Hamiltonian $\mathrm{H}_{\text{red}}^{(d)}$ that is an SSPT distinct from \eqref{eq:dDcluster} and \eqref{eq:Hdblue} by corner modes.
\vspace{1cm}\\
In Appendix~\ref{sec:nonanomalous}, we argue that $\hat{\rm V}^{(d)}\hat{\eta}_r\hat{\eta}_b$ is not anomalous in higher spatial dimensions. We leave the analysis of finding the phase for which $\hat{\rm V}^{(d)}\hat{\eta}_r\hat{\eta}_b$ is the unbroken symmetry to future directions. 

\red{We emphasize that no additional crystalline symmetries are required to protect the boundary zero modes in our construction. In conventional higher-order topological phases protected by invertible symmetries, crystalline symmetries often play an essential role in preventing the pairwise annihilation of corner zero modes by pinning them to specific spatial locations. In contrast, for the interface between the two phases that differ as a higher-order SSPT, we show that any symmetric perturbation supported within a fattened region around the interface necessarily preserves at least a two-fold degeneracy, regardless of whether the perturbation respects any crystalline symmetry (see Appendix~\ref{sec:stability2dclusterblue} for details). Thus, the boundary zero modes cannot be fully lifted by any symmetry-preserving perturbation along the interface. Furthermore, the KT-dual phases of the two higher-order SSPTs exhibit distinct symmetry-breaking patterns and lie in different SSB phases, implying that the original higher-order SSPTs belong to different phases without invoking any crystalline symmetry in the argument.} 
\subsection{Possibility of noninvertible higher-order SSPT phases from $\mathbb{Z}_2$ cluster phase}
 In this section, we analyze the possibility of noninvertible higher-order SSPT phases within the cluster phase with $\mathbb{Z}_2$ subsystem symmetry. First, we analyze the $2+1$D case and then generalize it to $d+1$D. 

\subsubsection{$2+1$D}
We note that $\mathbb{Z}_2$ SSPT \eqref{eq:Z2SSPTHam} is invariant under $\mathbf{D}_\text{DPIM}^{(2)}$. Hence, we could ask the following question: whether the cluster phase can be split to phases protected by the noninvertible symmetry $\mathbf{D}_\text{DPIM}^{(2)}$. We can again use KT transformation to map the SSPT to SSSB and study the various possible symmetry breaking patterns. KT transformation is given in~\cite{mana2024kennedy}
\begin{align}
    \mathbf{KT}^{(2)}_{\mathbb{Z}_2}\equiv \hat{\rm V}^{(2)}_{\mathbb{Z}_2}\mathbf{D}^{(2)}_{\text{DPIM}}\hat{\rm V}^{(2)}_{\mathbb{Z}_2}\, ,
\end{align}
where $\hat{\rm V}^{(2)}_{\mathbb{Z}_2}$ is the cluster entangler
\begin{align}
   \hat{\rm V}^{(2)}_{\mathbb{Z}_2} \equiv \prod_{v}CZ_{v,v+(1,0)}CZ_{v,v+(0,1)}CZ_{v,v+(1,1)}\, .
\end{align}
$\mathbf{KT}^{(2)}_{\mathbb{Z}_2}$ acts in the following way:
\begin{subequations}
   \begin{align}
    X_v&\xleftrightarrow{\mathbf{KT}^{(2)}_{\mathbb{Z}_2}}\hat{X}_v\, ,\\
    \begin{array}{ccc}
        & Z & Z  \\
        Z & X & Z  \\
        Z & Z &
    \end{array}&\xleftrightarrow{\mathbf{KT}^{(2)}_{\mathbb{Z}_2}}\begin{array}{ccc}
        & \hat{Z} & \hat{Z}  \\
        \hat{Z} &  & \hat{Z}  \\
        \hat{Z} & \hat{Z} &
    \end{array}\, .
\end{align} 
\end{subequations}
Under $\mathbf{KT}^{(2)}_{\mathbb{Z}_2}$, we map the cluster state Hamiltonian~\eqref{eq:Z2SSPTHam} to spontaneous subsystem symmetry breaking (SSSB) Hamiltonian  
\begin{align}
    \mathrm{H}^{(2)}_{\mathbb{Z}_2-\text{SSSB}}=-\sum_v\,  \begin{array}{ccc}
        & \hat{Z} & \hat{Z}  \\
        \hat{Z} &  & \hat{Z}  \\
        \hat{Z} & \hat{Z} &
    \end{array}\,,
    \label{eq:Z2SSSB}
\end{align}
where the hat indicates the Pauli operators are on the dual SSSB side. The symmetries of the dual Hamiltonian are
\begin{align}
&\hat{\eta}^x_j \equiv \prod_{i = 1 }^L \hat{X}_{i,j}\, , \quad 
\hat{\eta}^y_i \equiv \prod_{j = 1 }^L \hat{X}_{i,j} \, , 
\hat{\eta}^\text{diag}_k \equiv \prod_{\ell = 1 }^L \hat{X}_{ \ell,[\ell+k]_L} \, , \nonumber\\ 
&\hat{\rm V}^{(2)}_{\mathbb{Z}_2}=\prod_{v}CZ_{v,v+(1,0)}CZ_{v,v+(0,1)}CZ_{v,v+(1,1)}\, .
\end{align} 
The $\mathbb{Z}_2$ subsystem symmetries are spontaneously broken in the SSSB phase. When $L$ is even, the subsystem lines do not span the broken symmetries. For even system size, we note the following symmetries of the Hamiltonian \eqref{eq:Z2SSSB}
\begin{align}
    &\hat{\eta}_{(o,o)}=\prod_{\substack{(i,j)\\
i,j \text{ odd }}}\hat{X}_{i,j}\,\quad \hat{\eta}_{(e,o)}=\prod_{\substack{(i,j)\\
i \text{ even },j \text{ odd }}}\hat{X}_{i,j}\,,\nonumber\\
&\hat{\eta}_{(o,e)}=\prod_{\substack{(i,j)\\
i\text{ odd },j\text{ even } }}\hat{X}_{i,j}\,,\quad \hat{\eta}_{(e,e)}=\prod_{\substack{(i,j)\\
i,j \text{ even }}}\hat{X}_{i,j}\,.
\label{eq:2Detaoddoddsymmetry}
\end{align}
$\hat{\eta}_{(e,o)}$,$\hat{\eta}_{(o,e)}$, and $\hat{\eta}_{(e,e)}$ can be obtained from $\hat{\eta}_{(o,o)}$ by multiplying with subsystem lines. When $L=4k$, these symmetries can be obtained by taking the product of subsystem lines. \red{Hence, the symmetries in \eqref{eq:2Detaoddoddsymmetry} and subsystem lines are not enough to span the broken symmetries}. However, when $L=4k+2$, these symmetries are not obtained by taking the product of subsystem lines, and hence one of them can be taken as an independent symmetry. Together with subsystem lines they span the broken symmetries for $L=4k+2$. 

\red{From here on we restrict our discussion to even $L$}. We analyze the unbroken symmetries.  $\hat{\rm V}^{(2)}_{\mathbb{Z}_2}$ is a possible anomaly-free unbroken symmetry. It turns out that diagonal combinations that involve  $\hat{\rm V}^{(2)}_{\mathbb{Z}_2}$ and any adjacent parallel subsystem symmetry lines are anomalous (we refer the reader to Appendix~\ref{sec:anomaly} for an analysis of anomalous symmetries). However, $\hat{\rm V}^{(2)}_{\mathbb{Z}_2} \hat{\eta}_{o,o}$ is a possible anomaly free unbroken symmetry. Similarly, $\hat{\rm V}^{(2)}_{\mathbb{Z}_2} \hat{\eta}_{e,e}$, $\hat{\rm V}^{(2)}_{\mathbb{Z}_2} \hat{\eta}_{o,e}$, and $\hat{\rm V}^{(2)}_{\mathbb{Z}_2} \hat{\eta}_{e,o}$ are anomaly free and is a possible unbroken symmetry. We analyze the possibilities below.

{\paragraph{\texorpdfstring{$\hat{\rm V}^{(2)}_{\mathbb{Z}_2}$}{Lg} preserved}
On the symmetry breaking side, we obtain the Hamiltonian \eqref{eq:Z2SSSB}. The order parameters for this phase are $\{\hat{Z}_{i,1},\hat{Z}_{1,j},\hat{Z}_{k,k},\hat{Z}_{i+\frac{1}{2},\frac{3}{2}},\hat{Z}_{\frac{3}{2},j+\frac{1}{2}},\hat{Z}_{k+\frac{1}{2},k+\frac{1}{2}}\}\rvert_{\substack{i=1,...,L;j=2,...,L;\\k=2,...,L}}$. The original SSPT that gives rise to this Hamiltonian is \eqref{eq:Z2SSPTHam}.
\paragraph{\texorpdfstring{$\hat{\rm V}^{(2)}_{\mathbb{Z}_2}\hat{\eta}_{o,o}$}{Lg} preserved}
On the symmetry breaking side, we have the Hamiltonian
\begin{align}
  \hat{\rm H}_{\text{odd}}^{(2)}&=\sum_{v=(\text{odd},\text{odd})}\begin{array}{ccc}
       & \hat{Z} & \hat{Z}\\
       \hat{Z} &\hat{I}_v & \hat{Z} \\
       \hat{Z} & \hat{Z} &
  \end{array}-\sum_{v=(\text{even},\text{even})}\begin{array}{ccc}
       & \hat{Z} & \hat{Y}\\
       \hat{Z} &\hat{I}_v & \hat{Z} \\
       \hat{Y} & \hat{Z} &
  \end{array}\nonumber\\
  &\qquad-\sum_{v=(\text{odd},\text{even})}\begin{array}{ccc}
       & \hat{Y} & \hat{Z}\\
       \hat{Z} &\hat{I}_v & \hat{Z} \\
       \hat{Z} & \hat{Y} &
  \end{array}-\sum_{v=(\text{even},\text{odd})}\begin{array}{ccc}
       & \hat{Z} & \hat{Z}\\
       \hat{Y} &\hat{I}_v & \hat{Y} \\
       \hat{Z} & \hat{Z} &
  \end{array}\nonumber\\
 & -\sum_{v=(\text{even},\text{even})}\begin{array}{ccccc}
    & & & \hat{Z}& \hat{Z}\\
    &   &  & \hat{Y}&\hat{Z}\\
       & &\hat{I}_v & & \\
      \hat{Z} &\hat{Y} &  & &\\
       \hat{Z} & \hat{Z} & & &
  \end{array}-\sum_{v=(\text{odd},\text{even})}\begin{array}{ccc}
      & \hat{Z} & \hat{Z} \\
       \hat{Z} & \hat{Y} & \\
       &\hat{I}_v &\\
       & \hat{Y} &\hat{Z}\\
       \hat{Z} &\hat{Z}&
  \end{array}\nonumber\\
 & -\sum_{v=(\text{even},\text{odd})}\begin{array}{ccccc}
      & \hat{Z} & & &\hat{Z} \\
       \hat{Z}&\hat{Y} &\hat{I}_v &\hat{Y} &\hat{Z}\\
       \hat{Z} & & & \hat{Z} &
  \end{array}\,.
  \label{eq:H2odd}
\end{align}
The order parameters for this phase are of two types: 1) $\hat{Z}_{v}$ for $v$ of the form $(i,1)$,$(1,j)$, $(k,k)$, and $v\neq (\text{odd},\text{odd})$,i.e., both coordinates are not odd integers 2) $\hat{Y}_v\left(1-\begin{array}{ccc}
    & \hat{Z} & \hat{Z} \\
    \hat{Z} & & \hat{Z}\\
    \hat{Z} & \hat{Z}
\end{array}\right)$ for $v$ of the form $(i,1)$,$(1,j)$, $(k,k)$, and $v= (\text{odd},\text{odd})$, i.e., both coordinates are odd integers. Groundstates are in the SSSB phase and the ground state degeneracy is $2^{3L-2}$.

The SSPT Hamiltonian can be obtained by applying the $\mathbf{KT}^{(2)}_{\mathbb{Z}_2}$ transformation on the SSSB Hamiltonian \eqref{eq:H2odd}
\begin{align}
    \mathrm{H}_{\text{odd}}^{(2)}&=\sum_{v=(\text{odd},\text{odd})}\begin{array}{ccc}
       & \hat{Z} & \hat{Z}\\
       \hat{Z} &\hat{X}_v & \hat{Z} \\
       \hat{Z} & \hat{Z} &
  \end{array}-\sum_{v=(\text{even},\text{odd})}\begin{array}{ccc}
       & \hat{Z} & \hat{Y}\\
       \hat{Z} &\hat{X}_v & \hat{Z} \\
       \hat{Y} & \hat{Z} &
  \end{array}\nonumber\\
  &\qquad-\sum_{v=(\text{odd},\text{even})}\begin{array}{ccc}
       & \hat{Y} & \hat{Z}\\
       \hat{Z} &\hat{X}_v & \hat{Z} \\
       \hat{Z} & \hat{Y} &
  \end{array}-\sum_{v=(\text{even},\text{even})}\begin{array}{ccc}
       & \hat{Z} & \hat{Z}\\
       \hat{Y} &\hat{X}_v & \hat{Y} \\
       \hat{Z} & \hat{Z} &
  \end{array}\nonumber\\
 & +\sum_{v=(\text{even},\text{odd})}\begin{array}{ccccc}
      & & & \hat{Z} & \hat{Z}\\
      & & & \hat{Z} & \hat{Z}\\
       & &\hat{X}_v & & \\
       \hat{Z} & \hat{Z} & & &\\
        \hat{Z} & \hat{Z} & & &
  \end{array}+\sum_{v=(\text{odd},\text{even})}\begin{array}{ccc}
       & \hat{Z} &\hat{Z} \\
       \hat{Z}&\hat{Z} &\\
       & \hat{X} &\\
       & \hat{Z} &\hat{Z}\\
       \hat{Z} &\hat{Z}
  \end{array}\nonumber\\
  &+\sum_{v=(\text{even},\text{even})}\begin{array}{ccccc}
       & \hat{Z} &  &  &\hat{Z} \\
      \hat{Z} & \hat{Z} & \hat{X}_v & \hat{Z} & \hat{Z}\\
       \hat{Z} &  &  & \hat{Z} &
  \end{array}\,.
\end{align}
We consider an interface between $\mathrm{H}_{\text{2D-SSPT}}^{\mathbb{Z}_2}$ and $\mathrm{H}_{\text{odd}}^{(2)}$. We find that there are interface modes along the interface between the two Hamiltonians. Hence, the two Hamiltonian differ as a first-order SPT protected by the noninvertible symmetry. See Appendix-\ref{sec:H2oddandH2SSPT} for details.\\
\vspace{0.5cm}
\paragraph{\texorpdfstring{$\hat{\rm V}^{(2)}_{\mathbb{Z}_2}\hat{\eta}_{e,e}$}{Lg}, \texorpdfstring{$\hat{\rm V}^{(2)}_{\mathbb{Z}_2}\hat{\eta}_{o,e}$}{Lg} or \texorpdfstring{$\hat{\rm V}^{(2)}_{\mathbb{Z}_2}\hat{\eta}_{e,o}$}{Lg} are preserved}
These cases are similar to the case where $\hat{\rm V}^{(2)}_{\mathbb{Z}_2}\hat{\eta}_{o,o}$ is preserved.

\par $\mathrm{V}_{\mathbb{Z}_2}^{(2)} \hat{\eta}$ where $\hat{\eta}=\prod_{j=1}^L\hat{\eta}^x_j$ is the global part of the subsystem symmetry is anomaly-free and could realize a noninvertible higher-order SSPT distinct from the cluster state~\eqref{eq:Z2SSPTHam} in $2+1$D. See Appendix~\ref{sec:nonanomalousZ2} for the construction of a short-range entangled state symmetric under $\mathrm{V}_{\mathbb{Z}_2}^{(2)} \hat{\eta}$. We leave the construction of the corresponding noninvertible SSPT phase and its interface with cluster state to future exploration.
\subsubsection{$d+1$D}
There is a $\mathbb{Z}_2$ SSPT \eqref{eq:Z2dDSSPT} that is invariant under $\mathbf{D}_\text{DPIM}^{(d)}$. We look for phases in this cluster phase protected by $\mathbf{D}_\text{DPIM}^{(d)}$.  We use the KT transformation
\begin{align}
    \mathbf{KT}^{(d)}_{\mathbb{Z}_2}\equiv \hat{\rm V}^{(d)}_{\mathbb{Z}_2}\mathbf{D}^{(d)}_{\text{DPIM}}\hat{\rm V}^{(d)}_{\mathbb{Z}_2}\, ,
\end{align}
where $\hat{\rm V}^{(d)}_{\mathbb{Z}_2}$ is the cluster entangler
\begin{align}
   \hat{\rm V}^{(d)}_{\mathbb{Z}_2} \equiv \prod_{v}\prod_{\substack{(i_1,i_2,...,i_d)\\
   i_1,i_2,...,i_d=0,1}}CZ_{v,v+(i_1,...,i_d)}\, .
\end{align}
$\mathbf{KT}^{(d)}_{\mathbb{Z}_2}$ acts in the following way:
\begin{widetext}
\begin{subequations}
    \begin{align}
    X_v&\xleftrightarrow{\mathbf{KT}^{(d)}_{\mathbb{Z}_2}}\hat{X}_v\, ,\\
    X_v\prod_{\substack{v'\in\partial c\, ,\\
    c=v+(\frac{1}{2},...,\frac{1}{2})}}Z_{v'}\prod_{\substack{v'\in\partial c\, ,\\
    c=v-(\frac{1}{2},...,\frac{1}{2})}}Z_{v'}&\xleftrightarrow{\mathbf{KT}^{(d)}_{\mathbb{Z}_2}}\prod_{\substack{v'\in\partial c\, ,\\
    c=v+(\frac{1}{2},...,\frac{1}{2})}}\hat{Z}_{v'}\prod_{\substack{v'\in\partial c\, ,\\
    c=v-(\frac{1}{2},...,\frac{1}{2})}}\hat{Z}_{v'}\, .
\end{align}
\end{subequations}
\end{widetext}
$\mathbf{KT}^{(d)}_{\mathbb{Z}_2}$ maps cluster Hamiltonian \eqref{eq:Z2dDSSPT} to spontaneous subsystem symmetry breaking (SSSB) Hamiltonian. The dual Hamiltonian has the symmetry $\hat{\rm V}^{(d)}_{\mathbb{Z}_2}$ (via the KT transformation on $\mathbf{D}_\text{DPIM}^{(d)}$) in addition to the subsystem symmetry. 
Similar to the $2+1$D case, $\hat{\rm V}^{(d)}_{\mathbb{Z}_2}$ is a possible anomaly-free unbroken symmetry. Again, when the system size is even, we could define symmetries that are generalizations of \eqref{eq:2Detaoddoddsymmetry}. Then diagonal combinations of $\hat{\rm V}^{(d)}_{\mathbb{Z}_2}$ with such symmetries would be a possible anomaly-free unbroken symmetry. Then it is possible to analyze the interface modes between SSPT dual to $\hat{\rm V}^{(d)}_{\mathbb{Z}_2}$ preserved phase and diagonal combinations of $\hat{\rm V}^{(d)}_{\mathbb{Z}_2}$ with such symmetries. They would exhibit interface modes at the conventional boundary and would differ as a first-order SSPT protected by noninvertible symmetry. It would be an interesting future direction to analyze the $\hat{\rm V}^{(d)}_{\mathbb{Z}_2}\hat{\eta}$ preserved phase, where $\hat{\eta}$ is the global part of the subsystem symmetry, and its dual SSPT.

\section{Noninvertible higher-order subsystem symmetry-protected topological phases: Hinge modes}
\label{sec:hinge}
In this section, we demonstrate an example of second-order SPT with subsystem and noninvertible symmetries in 3D. 
The subsystem symmetries we consider \red{also include the} planar \red{symmetries}. 
We first construct a 3D SPT (cluster state) with subsystem symmetries and noninvertible symmetry (i.e., the Kramers-Wannier symmetry). We ask whether the cluster phase splits into multiple phases protected by the noninvertible symmetry. We study some of these possibilities and find that they differ by hinge modes protected by the noninvertible symmetry on an interface between the cluster and potential candidate phases.
\subsection{Planar subsystem symmetry-protected topological phases in $3+1$D}
In this section, we give a\red{n example}   construction of planar subsystem symmetry-protected topological phases in $3+1$D. 
See Ref.~\cite{Devakul:2019duj} for the classification of planar subsystem symmetric phases in $3+1$D. We restrict our discussion to a cluster state with $\mathbb{Z}_2\times\mathbb{Z}_2$ planar subsystem symmetry defined on a bipartite lattice colored red and blue. We label the vertices of the red sublattice with integer coordinates and the vertices of the blue sublattice with half-integer coordinates. The red sublattice forms a face-centered cubic (FCC) lattice spanned by primitive vectors
 \begin{align}
     \Vec{a}_1=(1,0,1)\, ,\quad \Vec{a}_2=(1,1,0)\, ,\quad \Vec{a}_3=(0,1,1)\, .
 \end{align}
 The blue sublattice also forms an FCC spanned by the same primitive vectors above but shifted from the red sublattice by $(\frac{1}{2},\frac{1}{2},\frac{1}{2})$ translation. See Figure~\ref{fig:FCC} for an illustration. We take the number of vertices in $x$,$y$, and $z$ directions along a straight line to be $L_x$, $L_y$, and $L_z$. We also assume periodic boundary conditions along the three directions.
 \begin{figure*}
     \centering
     \includegraphics[scale=1]{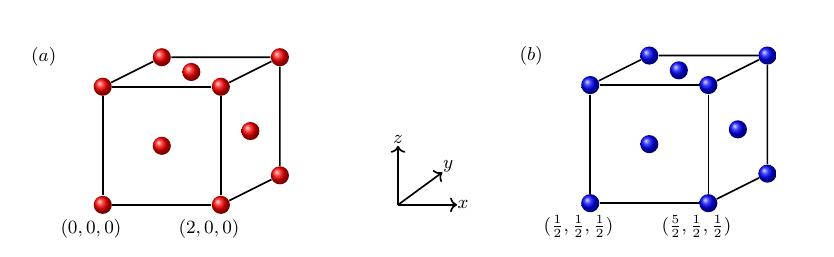}
     \caption{Face-centered cubic (FCC) lattice structure of red and blue sublattices. They are shifted by $(\frac{1}{2},\frac{1}{2},\frac{1}{2})$ vector. We take the edges of the FCC to have length 2.}
     \label{fig:FCC}
 \end{figure*}
 We denote the vertices of the red sublattice by $v_r$ and those of the blue sublattice by $v_b$. To define a cluster state, we define a graph $\mathcal{G}$ by connecting a blue vertex to its neighboring red vertices as shown in Figure~\ref{fig:graphG}.
 \begin{widetext}
 \begin{figure}
     \centering
     \includegraphics[scale=1]{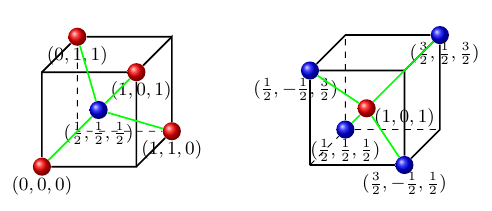}
     \caption{Graph $\mathcal{G}$ is defined by the set of red and blue vertices each forming an FCC lattice translated by $(\frac{1}{2},\frac{1}{2},\frac{1}{2})$ and the edges connecting them as shown here. The figure illustrates the connectivity of red and blue vertices locally.}
     \label{fig:graphG}
 \end{figure}
 \end{widetext}
 Let $\sigma$ be an adjacency matrix for this graph. Note that the matrix elements $\sigma_{v_1v_2}$ take values in $1$ or $0$ depending on whether the vertices $v_1$ and $v_2$ share an edge or not. 

 The cluster Hamiltonian based on the graph $\mathcal{G}$ defined above is 
 \begin{align}
     \mathrm{H}_{3\text{D-cluster}}^{\mathcal{G}}=-\sum_{v_r}X_{v_r}\prod_{v_b}Z_{v_b}^{\sigma_{v_rv_b}}-\sum_{v_b}X_{v_b}\prod_{v_r}Z_{v_r}^{\sigma_{v_rv_b}}\, .
     \label{eq:3DclusterG}
 \end{align}
 Let us denote the $x-y$, $y-z$ and $x-z$ lattice planes at fixed $z$, $x$ and $y$ by $P^z_{xy}$, $P^x_{yz}$ and $P^y_{xz}$, respectively. For example, $P^z_{xy}$ denotes a plane with fixed $z$ coordinate and varying $x$ and $y$ coordinates on the lattice. Note that  for fixed integer superscript $z$, $x$, or $y$ coordinate,  the plane passes through the red sublattice, and that for fixed half-integer coordinates, the plane passes through the blue sublattice.  The planar subsystem symmetries of \eqref{eq:3DclusterG} are
     \begin{align}\label{eq:planarsymmetriesred}
     \begin{split}
     \mathcal{P}^{z,r}_{xy}&=\prod_{v_r\in P^z_{xy}}X_{v_r}\,, \quad \mathcal{P}^{x,r}_{yz}=\prod_{v_r\in P^x_{yz}}X_{v_r}\,,\\
     &\qquad \mathcal{P}^{y,r}_{xz}=\prod_{v_r\in P^y_{xz}}X_{v_r}\, ,
     \end{split}
 \end{align}
 and 
 \begin{align}\label{eq:planarsymmetriesblue}
     \begin{split}
     \mathcal{P}^{z,b}_{xy}&=\prod_{v_b\in P^z_{xy}}X_{v_b}\,, \quad \mathcal{P}^{x,b}_{yz}=\prod_{v_b\in P^x_{yz}}X_{v_b}\,,\\
     &\qquad \mathcal{P}^{y,b}_{xz}=\prod_{v_b\in P^y_{xz}}X_{v_b}\, .
     \end{split}
 \end{align}
 Apart from these, the Hamiltonian is also symmetric under $X_{v_r}\leftrightarrow \prod_{v_b}Z_{v_b}^{\sigma_{v_rv_b}}$ and $X_{v_b}\leftrightarrow \prod_{v_r}Z_{v_r}^{\sigma_{v_rv_b}}$. 
 Such a symmetry transformation would be generated by a sequential circuit with projectors, just as in gauging linear subsystem symmetries. 
 Here, instead, we give a definition of this operator, following the construction given in~\cite{gorantla2024tensor} using the ZX calculus.

 Let us consider a bipartite graph $\mathcal{G}$ made of sets of vertices $\mathbb{V}$  and $\hat{\mathbb{V}}$. The connectivity of edges between $\mathbb{V}$ and $\hat{\mathbb{V}}$ is given by the adjacency matrix $\sigma$. We denote the vertices in $\mathbb{V}$ and $\hat{\mathbb{V}}$ by $v$ and $\hat{v}$. We place qubits on both $\mathbb{V}$ and $\hat{\mathbb{V}}$. We define the operators in terms of a ZX-diagram,
 \begin{subequations}
      \begin{align}
     \mathrm{D}_{\hat{v}\xleftarrow{}v}&:=\raisebox{-20pt}{\begin{tikzpicture}
        \draw[fill=green] (0,0) circle (.8ex) node (B1) {}; 
        \draw[fill=green] (0,1.5) circle (.8ex) node (B2) {};
        \draw[fill=red] (2,0) circle (.8ex) node (B3) {}; 
        \draw[fill=red] (2,1.5) circle (.8ex) node (B4) {};
        \draw[fill=yellow] (2.4,1.4) rectangle (2.6,1.6) node (B5) {};
        \draw[fill=yellow] (2.4,-0.1) rectangle (2.6,0.1) node (B6) {};
        \draw (1,0.75) node  (A) {$\sigma$};
        \draw[thick] (-1,0)--(B1)--(A)--(B3)--(2.4,0);
        \draw[thick] (2.6,0)--(3.5,0);
        \draw[thick] (-1,1.5)--(B2)--(A)--(B4)--(2.4,1.5);
        \draw[thick] (2.6,1.5)--(3.5,1.5);
        \draw (0,0.8) node () {$\vdots$};
        \draw (2,0.8) node () {$\vdots$};
     \end{tikzpicture}}\\
     \mathrm{D}_{v\xleftarrow{}\hat{v}}&:=\left(\mathrm{D}_{\hat{v}\xleftarrow{}v}\right)^{\dagger}\, .
 \end{align}
 \end{subequations}
The hermitian conjugate can be thought of as reflecting the ZX-diagram \red{about the vertical line in between the red and green spiders}. We note that the above definition~\cite{gorantla2024tensor} involves the ZX-calculus
and encourage the reader to refer to~\cite{gorantla2024tensor} for more details on this operator representation. We will not do any explicit calculations using the ZX-calculus,
but will directly use the results in~\cite{gorantla2024tensor}.

Let us denote the Hilbert space of the qubits on $\mathbb{V}$ by $\mathcal{H}^{\otimes \mathbb{V}}$ and that on $\hat{\mathbb{V}}$ by $\mathcal{H}^{\otimes \hat{\mathbb{V}}}$. Then, $\mathrm{D}_{\hat{v}\xleftarrow{}v}:\mathcal{H}^{\otimes \mathbb{V}}\longrightarrow \mathcal{H}^{\otimes \hat{\mathbb{V}}}$ and $\mathrm{D}_{v\xleftarrow{}\hat{v}}:\mathcal{H}^{\otimes \hat{\mathbb{V}}}\longrightarrow \mathcal{H}^{\otimes \mathbb{V}}$. We construct the tensor product map 
 \begin{align}
\mathrm{D}_{\hat{v}\xleftarrow{}v}\otimes \mathrm{D}_{v\xleftarrow{}\hat{v}}:\mathcal{H}^{\otimes \mathbb{V}}\otimes \mathcal{H}^{\otimes \hat{\mathbb{V}}}\longrightarrow \mathcal{H}^{\otimes \hat{\mathbb{V}}}\otimes \mathcal{H}^{\otimes \mathbb{V}}\, .
\end{align}
This tensor product map maps between two Hilbert spaces that are swapped in the tensor product operation. So, we define a swap operation 
\begin{align}
    \phi_s:\mathcal{H}^{\otimes \hat{\mathbb{V}}}\otimes \mathcal{H}^{\otimes \mathbb{V}}\longrightarrow \mathcal{H}^{\otimes \mathbb{V}}\otimes \mathcal{H}^{\otimes \hat{\mathbb{V}}} 
\end{align}
that swaps $\mathcal{H}^{\otimes \mathbb{V}}$ and $\mathcal{H}^{\otimes \hat{\mathbb{V}}}$ preserving the order of tensor product within it. Explicitly, if $\hat{v}\in\mathcal{H}^{\otimes \hat{\mathbb{V}}}$ and $v\in \mathcal{H}^{\otimes \mathbb{V}}$, then $\phi_s(\hat{v}\otimes v)=v\otimes\hat{v}$. Finally, we construct the operator mapping between the same Hilbert spaces
$\mathrm{\mathbf{D}}^{(3)}_{\rm pln}:\mathcal{H}^{\otimes \mathbb{V}}\otimes \mathcal{H}^{\otimes \hat{\mathbb{V}}}\longrightarrow \mathcal{H}^{\otimes \mathbb{V}}\otimes \mathcal{H}^{\otimes \hat{\mathbb{V}}}$ given by
 \begin{align}
\mathrm{\mathbf{D}}^{(3)}_{\rm pln}=\phi_s\circ \mathrm{D}_{\hat{v}\xleftarrow{}v}\otimes \mathrm{D}_{v\xleftarrow{}\hat{v}} \, .
 \end{align}

 Now let us take the graph $\mathcal{G}$ to be that we considered while defining \eqref{eq:3DclusterG}. Then, we choose $\mathbb{V}$ to be the set of red vertices and $\hat{\mathbb{V}}$ to be the set of blue vertices. Then, $\mathrm{\mathbf{D}}^{(3)}_{\rm pln}$ is the noninvertible symmetry that implements
 \begin{align}
    X_{v_r}\xleftrightarrow{\mathrm{\mathbf{D}}^{(3)}_{\rm pln}}\prod_{v_b}Z_{v_b}^{\sigma_{v_rv_b}}\, ,\qquad  X_{v_b}\xleftrightarrow{\mathrm{\mathbf{D}}^{(3)}_{\rm pln}}\prod_{v_r}Z_{v_r}^{\sigma_{v_rv_b}}\,. 
 \end{align}

 We note that $\phi_s\circ \mathrm{D}_{\hat{v}\xleftarrow{}v}\otimes \mathrm{D}_{v\xleftarrow{}\hat{v}}=\mathrm{D}_{v\xleftarrow{}\hat{v}}\otimes \mathrm{D}_{\hat{v}\xleftarrow{}v}\circ\phi_s$. Then,
 \begin{widetext}
     \begin{align}
     \left(\mathrm{\mathbf{D}}^{(3)}_{\rm pln}\right)^2&=\left(\mathrm{D}_{v\xleftarrow{}\hat{v}}\otimes \mathrm{D}_{\hat{v}\xleftarrow{}v}\circ\phi_s\right)\circ \left(\phi_s\circ \mathrm{D}_{\hat{v}\xleftarrow{}v}\otimes \mathrm{D}_{v\xleftarrow{}\hat{v}}\right)\nonumber\\
     &=\mathrm{D}_{v\xleftarrow{}\hat{v}}\otimes \mathrm{D}_{\hat{v}\xleftarrow{}v}\circ\left(\phi_s\circ \phi_s\right)\circ \mathrm{D}_{\hat{v}\xleftarrow{}v}\otimes \mathrm{D}_{v\xleftarrow{}\hat{v}}\nonumber\\
     &=\mathrm{D}_{v\xleftarrow{}\hat{v}}\otimes \mathrm{D}_{\hat{v}\xleftarrow{}v}\circ \mathrm{D}_{\hat{v}\xleftarrow{}v}\otimes \mathrm{D}_{v\xleftarrow{}\hat{v}}\nonumber\\
     &=\left(\mathrm{D}_{v\xleftarrow{}\hat{v}}\circ\mathrm{D}_{\hat{v}\xleftarrow{}v}\right)\otimes \left(\mathrm{D}_{\hat{v}\xleftarrow{}v}\circ \mathrm{D}_{v\xleftarrow{}\hat{v}}\right)\nonumber\\
     &\red{=\frac{1}{2^{2|\rm E|-|\mathbb{V}|-|\hat{\mathbb{V}}|}}\left(\sum_{a\in\{0,1\}^{|\rm V|}}\prod_{ v\in \rm V}X_{v}^{a_v}\prod_{\hat{v}\in\hat{\rm V}}\delta_{\sum_{v\in\rm V}\sigma_{\hat{v}v}a_v,0}\right)\left(\sum_{\hat{a}\in\{0,1\}^{|\rm \hat{V}|}}\prod_{ \hat{v}\in \rm \hat{V}}X_{\hat{v}}^{\hat{a}_{\hat{v}}}\prod_{v\in\rm V}\delta_{\sum_{\hat{v}\in\hat{\rm V}}\sigma_{v\hat{v}}\hat{a}_{\hat{v}},0}\right)}\nonumber\\
     &\red{=\frac{1}{2^{2|\rm E|-|\mathbb{V}|-|\hat{\mathbb{V}}|}}\left(\sum_{a\in\text{Ker}\sigma}\eta_a\right)\left(\sum_{\hat{a}\in\text{Ker}\sigma^T}\eta_{\hat{a}}\right)}
     \label{eq:D^3plansquared}
 \end{align}
  \end{widetext}
  \red{where $\eta_a:=\prod_{ v\in \rm V}X_{v}^{a_v}\prod_{\hat{v}\in\hat{\rm V}}\delta_{\sum_{v\in\rm V}\sigma_{\hat{v}v}a_v,0}$ and $\eta_{\hat{a}}:=\prod_{ \hat{v}\in \rm \hat{V}}X_{\hat{v}}^{\hat{a}_{\hat{v}}}\prod_{v\in\rm V}\delta_{\sum_{\hat{v}\in\hat{\rm V}}\sigma_{v\hat{v}}\hat{a}_{\hat{v}},0}$ and $a_v$, $\hat{a}_{\hat{v}}$ take values in 0 and 1.}
 In the second line of the above equation, we used the associativity of maps. In the third line, we used the fact that $\phi_s\circ\phi_s=I$ is the identity map. In the fourth line, we used the distributive property of the tensor product under composition of linear maps. In the fourth line we used the result $\mathrm{D}_{v\xleftarrow{}\hat{v}}\circ\mathrm{D}_{\hat{v}\xleftarrow{}v}:= \mathrm{C}$ where $\mathrm{C}$ is the condensation operator given in~\cite{gorantla2024tensor}. \red{The symbols $\eta_a$ and $\eta_{\hat{a}}$ for $a\in\text{Ker}\,\sigma$ and $\hat{a}\in\text{Ker}\,\sigma^T$ are invertible symmetries of \eqref{eq:3DclusterG} that include symmetries beyond the planar subsystem symmetries defined in \ref{eq:planarsymmetriesred} and \eqref{eq:planarsymmetriesblue}. To make the set of symmetries closed under multiplication, we  need to include all the invertible symmetries $\eta_{a}$ and $\eta_{\hat{a}}$ along with $\mathrm{\mathbf{D}}^{(3)}_{\rm pln}$ in the discussion below. }

 \red{We note that one may impose all invertible symmetries $\eta_a$ and $\eta_{\hat{a}}$ to define the SSPT phase containing the Hamiltonian~\eqref{eq:3DclusterG}. This corresponds to a smaller region in the space of allowed symmetric deformations that preserve the spectral gap. In the following discussion, we restrict our attention to this SPT phase. However, one could also restrict to planar subsystem symmetries along with $\mathrm{\mathbf{D}}^{(3)}_{\rm pln}$ for defining the SPT phases.}
\subsection{Noninvertible second-order SSPT phases from \texorpdfstring{$\mathbb{Z}_2\times\mathbb{Z}_2$}{Lg} cluster phase in \texorpdfstring{$3+1$}{Lg}D}
With the model constructed in the previous subsection, we are ready to look at a different type of higher-order noninvertible SSPT in $3+1$ dimensions, namely, that with hinge modes. This is referred to as second-order as the protected modes between two models are restricted to one dimension, and the spatial dimension of the system is three, and hence $3-1=2$, i.e., second-order. 

Since the cluster state~\eqref{eq:3DclusterG} is invariant under the noninvertible symmetry $\mathrm{\mathbf{D}}^{(3)}_{\rm pln}$, the \red{SSPT phase protected by the invertible symmetries} can be further broken into phases protected by the noninvertible symmetry. To find these new phases, we use the Kennedy-Tasaki transformation, which allows us to map SSPT phases to SSSB phases. We define it as
 \begin{align}
     \mathbf{KT}^{(3)}_{\rm pln}\equiv \hat{\rm U}^{(3)}\mathrm{\mathbf{D}}^{(3)}_{\rm pln}\hat{\rm U}^{(3)} \, ,
 \end{align}
 where $\hat{\rm U}^{(3)}$ is the cluster entangler between red and blue sublattices. Explicitly,
 \begin{align}
     \hat{\rm U}^{(3)}\equiv \prod_{v_b}\prod_{v_r}CZ^{\sigma_{v_rv_b}}_{v_b,v_r} \, .
 \end{align}
 $\mathbf{KT}^{(3)}_{\rm pln}$ acts as follows
 \begin{subequations}
  \begin{align}
    X_{v_r}\xrightarrow{\mathbf{KT}^{(3)}_{\rm pln}} \hat{X}_{v_r}\, ,&\quad X_{v_b}\xrightarrow{\mathbf{KT}^{(3)}_{\rm pln}} \hat{X}_{v_b}\, ,\\
    X_{v_r}\prod_{v_b}Z^{\sigma_{v_rv_b}}_{v_b}&\xleftrightarrow{\mathbf{KT}^{(3)}_{\rm pln}}\prod_{v_b}\hat{Z}^{\sigma_{v_rv_b}}_{v_b}\, ,\\
    X_{v_b}\prod_{v_r}Z^{\sigma_{v_bv_r}}_{v_r}&\xleftrightarrow{\mathbf{KT}^{(3)}_{\rm pln}} \prod_{v_r}\hat{Z}^{\sigma_{v_bv_r}}_{v_r}\, .
\end{align}  
\end{subequations}
$\mathbf{KT}^{(3)}_{\rm pln}$ thus maps from
the cluster state Hamiltonian to the spontaneous subsystem symmetry breaking (SSSB) Hamiltonian, which is two copies of the tetrahedral Ising model on the red and blue sublattices
\begin{align}
    \hat{\mathrm{H}}^{(3)}_{\text{Tet-I}}=-\sum_{v_b}\prod_{v_r}\hat{Z}^{\sigma_{v_bv_r}}_{v_r}-\sum_{v_r}\prod_{v_b}\hat{Z}^{\sigma_{v_rv_b}}_{v_b}\, .
    \label{eq:tetrahedral-Ising}
\end{align}
This dual SSSB Hamiltonian still has the \red{symmetries $\eta_a$ and $\eta_{\hat{a}}$ that we discussed in \eqref{eq:D^3plansquared}}. The noninvertible symmetry $\mathbf{D}^{(3)}_{\rm pln}$ is mapped to $\hat{\rm U}$, which is a symmetry of the Hamiltonian~\eqref{eq:tetrahedral-Ising}. We consider the \textit{global} part of the subsystem symmetry 
\begin{subequations}
  \begin{align}
    \hat{\mathcal{P}}_r=\prod_{v_r}\hat{X}_{v_r}\,, \\
    \hat{\mathcal{P}}_b=\prod_{v_b}\hat{X}_{v_b}\, .
\end{align}  
\end{subequations}
We could analyze the possible symmetry breaking patterns to find the possible symmetry-protected topological phases. On the SSSB side, all the \red{invertible} symmetries are broken. We look at \red{some} possible choices for preserved symmetry. 
\subsubsection{\texorpdfstring{$\hat{\rm U}^{(3)}$}{Lg} is preserved}
On the SSSB side, we have the Hamiltonian~\eqref{eq:tetrahedral-Ising}. The original SSPT Hamiltonian that gives rise to this Hamiltonian is \eqref{eq:3DclusterG}. 
\subsubsection{\texorpdfstring{$\hat{\rm U}^{(3)}\hat{\mathcal{P}}_r$}{Lg} is preserved}
The SSSB Hamiltonian that preserve $\hat{\rm U}^{(3)}\hat{\mathcal{P}}_r$ is 
\begin{align}
    \hat{\rm H}^{(3)\mathcal{G}}_{\text{blue}}=\sum_{v_r}\prod_{v_b}\hat{Z}_{v_b}^{\sigma_{v_rv_b}}-\sum_{v_b}\prod_{v_r}\left[\hat{Y}_{v_r}^{\sigma_{v_bv_r}}\left(1+\prod_{v_b'}\hat{Z}_{v_b'}^{\sigma_{v_rv_b'}}\right)\right]\,.
\end{align}
Here, like the cases \eqref{eq:Hblue} and \eqref{eq:Hblued} in the previous section, the first term is kept positive so that the order parameter for this phase is non-zero. The SSPT Hamiltonian that gives rise to this SSSB Hamiltonian is found by applying $\mathbf{KT}^{(3)}_{\rm pln}$
\begin{widetext}
    \begin{align}
    \mathrm{H}_{\text{blue}}^{(3)\mathcal{G}}=\sum_{v_r}X_{v_r}\prod_{v_b}Z_{v_b}^{\sigma_{v_rv_b}}-\sum_{v_b}X_{v_b}\prod_{v_r}Y_{v_r}^{\sigma_{v_bv_r}}-\sum_{v_b}X_{v_b}\prod_{v_r}Z_{v_r}^{\sigma_{v_bv_r}}\left(\prod_{v_b'}Z_{v_b'}^{\sigma_{v_rv_b'}}\right)\, .
\end{align}
\end{widetext}
See Figure~\ref{fig:HamiltonianG} for illustrations of terms in $\hat{\rm H}^{(3)\mathcal{G}}_{\text{blue}}$ and $\rm H_{\text{blue}}^{(3)\mathcal{G}}$. 
\begin{figure}
    \centering
    \includegraphics[scale=1]{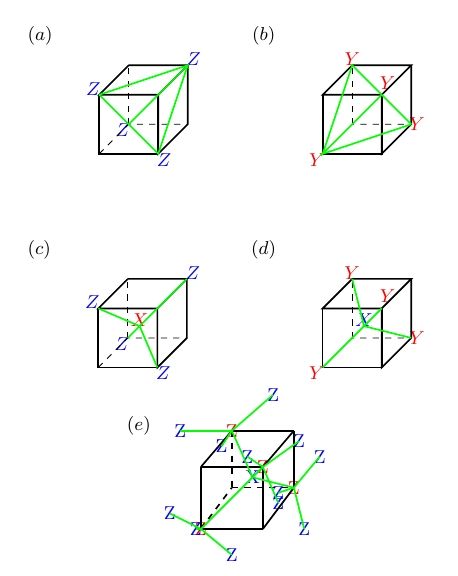}
    \caption{$(a)$,$(b)$ are the terms that appear in the Hamiltonian $\hat{\mathrm{H}}_{\text{blue}}^{(3)\mathcal{G}}$. They are tetrahedral interactions with opposite orientation. $(c)$,$(d)$ and $(e)$ are the terms appearing in the Hamiltonian $\mathrm{H}_{\text{blue}}^{(3)\mathcal{G}}$.}
    \label{fig:HamiltonianG}
\end{figure}
The Hamiltonian $\rm H_{\text{blue}}^{(3)\mathcal{G}}$ describe a distinct SSPT from cluster state protected by $\mathbf{D}^{(3)}_{\rm pln}$. The two SSPTs differ by hinge modes when considering the cubic interface between them (see Appendix~\ref{sec:Interfaceanalysis3D} for a detailed analysis). It can also be argued that the hinge modes are robust to symmetric local perturbations near the hinge (see Appendix~\ref{sec:stabilityhingemode} for a detailed analysis). Hence, this is a second-order SSPT protected by $\mathbf{D}^{(3)}_{\rm pln}$.
\subsubsection{\texorpdfstring{$\hat{\rm U}^{(3)}\hat{\mathcal{P}}_b$}{Lg} is preserved}
This case is similar to the above case when $\hat{\rm U}^{(3)}\hat{\mathcal{P}}_b$ is preserved. We need to replace all the analyses with $v_b\leftrightarrow v_r$. We could define $\hat{\rm H}^{(3)\mathcal{G}}_{\text{red}}$ and $\rm H_{\text{red}}^{(3)\mathcal{G}}$. The Hamiltonian $\rm H_{\text{red}}^{(3)\mathcal{G}}$ describe a distinct SSPT from the cluster state protected by $\mathbf{D}^{(3)}_{\rm pln}$, and the two SSPTs differ by hinge modes while considering cubic interface between them.  

\section{Conclusion}\label{sec:Conclusion}
In this work, we have presented explicit examples of noninvertible higher-order subsystem symmetry-protected topological phases (SSPTs). In particular, we constructed a noninvertible second-order SSPT in 
$2+1$ dimensions characterized by protected corner modes, as well as a second-order SSPT in $3+1$ dimensions featuring protected hinge modes. Furthermore, we extended our construction in 
$2+1$D to arbitrary dimensions, realizing noninvertible $d 
^{th}$-order SSPTs with protected corner modes.

We also analyzed anomalies analogous to type-III anomalies involving 0-form and subsystem symmetries, focusing on both $\mathbb{Z}_2$ and $\mathbb{Z}_2\times\mathbb{Z}_2$ symmetry groups (see Appendix~\ref{sec:anomaly}). These considerations allowed us to systematically exclude various patterns of spontaneous subsystem symmetry breaking (SSSB) that are incompatible with anomaly constraints. We presented additional examples of SSSB phases that preserve non-anomalous symmetries, along with their associated SSPT phases (see Appendix~\ref{sec:otherNSSPT}).

While we have provided a number of concrete examples, many open questions remain. For instance, it would be interesting to construct noninvertible higher-order SSPTs in $3+1$D that incorporate both planar and linear subsystem symmetries. A promising direction would also be to generalize our construction of SSPTs with hinge modes to arbitrary dimensions, potentially leading to novel noninvertible higher-order SSPTs.

More broadly, it would be valuable to develop a general framework for the systematic construction of noninvertible higher-order SSPTs. Our examples suggest a correlation between the dimensionality of the subsystem symmetry and the nature of the boundary modes: corner modes arise in the presence of linear subsystem symmetries, while hinge modes appear in the case of planar subsystem symmetries. This observation points to a deeper relationship between subsystem symmetry dimensionality and higher-order topological features, which warrants further investigation.

We discussed noninvertible higher-order $\mathbb{Z}_2\times\mathbb{Z}_2$ SSPTs distinguished by corner modes whose dual SSSB Hamiltonian preserves either $\hat{\rm V}^{(d)}$, $\hat{\rm V}^{(d)}\hat{\eta}_{r}$ or $\hat{\rm V}^{(d)}\hat{\eta}_{b}$. It would be interesting to study the SSPTs that correspond to the non-anomalous $\hat{\rm V}^{(d)}\hat{\eta}_{b}\hat{\eta}_r$ symmetry. An analogous analysis would also be worthwhile to explore in the context of $3+1$D model with planar subsystem symmetry that we study in our manuscript. Similarly, it would be interesting to analyze the noninvertible higher-order SSPT in the $\mathbb{Z}_2$ subsystem symmetric cluster state coming from preserving $\hat{\rm V}^{(d)}_{\mathbb{Z}_2}\hat{\eta}$ ($\hat{\eta}$ is the global part of the susbsytem symmetry) on the SSSB side. 

Finally, while our focus has been on noninvertible higher-order topological phases protected by subsystem symmetries, it would be intriguing to explore analogous phases protected by global symmetries in higher dimensions. \red{In particular, it would be interesting to construct examples of higher-order SPT phases protected jointly by noninvertible symmetries and global symmetries, where the latter may include higher-form symmetries in higher-dimensional settings, or to show that such phases cannot exist in the absence of subsystem symmetries.}

\medskip\noindent \textit{Notes Added}. During the final stage of preparing this manuscript, we noticed a similar and independent work on noninvertible SSPTs appeared on the arXiv~\cite{Furukawa:2025flp}.

\acknowledgments
The authors would like to thank an anonymous reviewer of our prior work~\cite{mana2024kennedy} in Physical Review B for pointing out a question on SPT phases protected by noninvertible symmetry in (2 + 1)D. 
APM thanks Mikhail Litvinov, Yunqin Zheng, and Sahand Seifnashri for useful discussions. 
This work was mainly supported by the U.S. National Science Foundation under Award No. PHY 2310614. T.-C. W. also acknowledges the support from the Center for Distributed Quantum Processing at Stony Brook University. Y. L. also acknowledges the support by the U.S. National Science Foundation under Grant No. NSF DMR-2316598.

\bibliography{Ref}
\appendix
\section{A consistent choice of order parameters for symmetry breaking}\label{sec:orderparameter}
In this Appendix, we state a lemma that captures a sufficient condition for a consistent choice of order parameters for a symmetry-breaking Hamiltonian. We will illustrate the gist of the lemma with an example later.
\begin{lemma}\label{lemma:consistentorderparameters}
    Suppose $\{\hat
{O}_i\}$ are a set of order parameters with non-zero eigenvalues restricted to the ground space for a symmetry breaking Pauli Hamiltonian $\rm H$, with symmetry group $G$, defined on a lattice with broken symmetry generators $\{\hat{g}_i\}$ for $i\in S$ where $S$ is a finite indexing set. Consider $\hat{\rm V}$ such that  $[\rm H, \hat{\rm V}]=0$ and $\hat{\rm V}^2=\rm I$ satisfying
\begin{itemize}
    \item $[\hat{O}_i,\mathrm{H}]=0$ $\forall i\in S$\, ,
    \item $[\hat{O}_i,\hat{\rm V}]=0$ $\forall i\in S$\, ,
    \item $[\hat{O}_i,\hat{O}_j]=0$ $\forall i,j\in S$\, ,
    \item $\exists$ linearly independent broken symmetry generators $\{\Tilde{g}_i\}$ such that $[\Tilde{g}_i,\hat{O}_j]=0$ when $i\neq j$ and $\{\Tilde{g}_i,\hat{O}_i\}=0$\, ,
    \item $\rm H$ has $2^{|S|}$ ground states.
\end{itemize}
Then $\hat{\rm V}$ is an unbroken symmetry.  
\end{lemma}
\begin{proof}
    Since we have $[\hat{O}_i,\mathrm{H}]=0$ $\forall i\in S$, the ground state subspace of the Hamiltonian $\rm H$ is preserved under the action of order parameters $\hat{O}_i$. Hence, we can think of $\hat{O}_i$ as $2^{|S|}\times 2^{|S|}$ dimensional matrices acting on the ground state subspace. Since $[\hat{O}_i,\hat{O}_j]=0$ $\forall i,j\in S$, they can be simultaneously diagonalized. Consider a simultaneous eigenvector $\ket{\Psi}$ of $\{\hat{O}_i\}$. Note that any arbitrary product of $\Tilde{g}_i$ acting on $\ket{\Psi}$ is also a simultaneous eigenvector of $\{\hat{O}_i\}$. In total, there are $2^{|S|}$ such products.  Since $[\Tilde{g}_i,\hat{O}_j]=0$ and $\{\Tilde{g}_i,\hat{O}_i\}=0$, the set of eigenvalues $\{\hat{o}_i\}$ for each of these states is distinct. Otherwise, the values $\{\hat{o}_i\}$ characterize the distinct simultaneous eigenvectors. Hence, the simultaneous eigenvectors should span the ground state subspace as there are $2^{|S|}$ ground states. Since $[\hat{O}_i,\hat{\rm V}]=0$ $\forall i\in S$, $\hat{\rm V}$ acting on the simultaneous eigenvectors does not change the values $\{\hat{o}_i\}$. Hence $\hat{\rm V}$ must be a scalar on the ground state subspace. Since $\hat{\rm V}^2=\rm I$, the scalar must be $\pm 1$. If it is 1, then $\hat{\rm V}$ is an unbroken symmetry. Otherwise, $-\hat{\rm V}$ is an unbroken symmetry. However, the minus sign is inconsequential.
\end{proof}
Now, let us illustrate the lemma with an example. Let us consider the symmetry-breaking Hamiltonian~\eqref{eq:oddHam}. The order parameters $\hat{Z}_1$ and $\hat{Y}_2(1-\hat{Z}_1\hat{Z}_3)$ commute with the Hamiltonian~\eqref{eq:oddHam},  commute with $\hat{\rm V}\eta_e$, commute with each other, and satisfy the following: $[\eta_e,\hat{Z}_1]=0$, $\{\eta_0,\hat{Z}_1\}=0$, $[\eta_o,\hat{Y}_2(1-\hat{Z}_1\hat{Z}_3)]=0$, $\{\eta_e,\hat{Y}_2(1-\hat{Z}_1\hat{Z}_3)\}=0$. We also note that $[\hat{\rm V}\eta_e,\rm H]=0$.  By our lemma, it says that $\hat{\rm V}\eta_e$ is the unbroken symmetry, which is indeed the case. We note that although $\hat{\rm V}$ commutes with $\rm H$ and is independent from $\eta_e$ and $\eta_o$, $\hat{\rm V}$ does not commute with $\hat{Y}_2(1-\hat{Z}_1\hat{Z}_3)$ and is not the preserved symmetry for this phase.
\section{Anomaly involving subsystem symmetries and 0-form symmetry}\label{sec:anomaly}
In this section, we discuss the anomalous symmetries, i.e., symmetries that cannot be realized on a unique gapped ground state. In the context of this manuscript, these symmetries can not be an unbroken symmetry of the symmetry breaking Hamiltonian. We will be discussing type-III anomaly between subsystem symmetries and 0-form symmetries.

We present three different methods to diagnose anomalous symmetries.  The first is the defect Hamiltonian method, where we compute the commutation relations of its symmetry operators and use the projective commutation as the signal of the anomaly. The second method is to use defect fusion from Ref.~\cite{Seifnashri:2023dpa}. We modify the original approach of a semi-infinite symmetry operator to a finite segment for computing the cohomology value to identify the anomaly. The third one is a direct generalization of the Else-Nayak method~\cite{else2014classifying}.

\subsection{Defect Hamiltonian method}
We diagnose type III anomalies with the symmetry generated by the CZ operator together with subsystem symmetries. 
We illustrate the defect Hamiltonian method both with the $\mathbb{Z}_2\times\mathbb{Z}_2$ subsystem symmetry and with the $\mathbb{Z}_2$ subsystem symmetry. 

\subsubsection{\texorpdfstring{$\mathbb{Z}_2\times\mathbb{Z}_2$}{Lg} subsystem symmetry}
To diagnose type III anomaly, let us consider three symmetry generators. We take 1) a subsystem symmetry on the blue sublattice, 2) a subsystem symmetry on the red sublattice, and  3) the symmetry $\hat{\rm V}$ given in \eqref{eq:symmetriesofSSSB}. First, we consider a Hamiltonian that is symmetric under all three symmetries. Then, we consider a subsystem-symmetry defect Hamiltonian for one of the subsystem symmetries. This defect Hamiltonian has modified symmetries that obey a projective algebra, indicating that the ground state can not be unique for the defect Hamiltonian. It turns out that if the defect Hamiltonian can not have a unique ground state, so does the Hamiltonian without any defect~\cite{Yao:2020xcm,Cheng:2022sgb} (see~\cite{Seifnashri:2023dpa} for an argument in $1+1$D). Since anomaly does not depend on the particular choice of the Hamiltonian and only depends on the symmetries that we consider, exhibiting the projective algebra for one Hamiltonian would be sufficient. The Hamiltonian we consider, which is symmetric under three symmetries, is the transverse-field cluster model
\begin{widetext}
\begin{align}
\begin{split}
    \bm \mathrm{H}_{\text{2D-TFCM}}&=-\sum_{v_r}\hat{X}_{v_r}\prod_{v_b\in \partial p_b}\hat{Z}_{v_b}-\sum_{v_b}\hat{X}_{v_b}\prod_{v_r\in \partial p_r}\hat{Z}_{v_r}-\sum_{v_r}\hat{X}_{v_r}-\sum_{v_b}\hat{X}_{v_b} \, .
\end{split}
\label{eq:2dTFcluster}
\end{align}
Here and below, the identifications $v_r =p_b$ etc. in cluster terms are understood and omitted.
We put it on a torus with $L_x$ and $L_y$ number of vertices on each sublattices in the $x$ and $y$ direction respectively. Now we consider a defect Hamiltonian
for the subsystem symmetry $\hat{\eta}^x_{r,k}$
\begin{align}
\begin{split}
     &\mathrm{H}_{\text{2D-TFCM},\hat{\eta}_{r,k}^x}=-\sum_{v_r}\hat{X}_{v_r}\prod_{v_b\in \partial p_b}\hat{Z}_{v_b}-\sum_{v_b}'\hat{X}_{ v_b}\prod_{v_r\in \partial p_r}\hat{Z}_{v_r}+\begin{array}{ccc}
        \hat{Z}_{l,k+1}  &  &\hat{Z}_{l+1,k+1} \\
                   & \hat{X}_{l+\frac{1}{2},k+\frac{1}{2}} & \\
        \hat{Z}_{l,k}    & & \hat{Z}_{l+1,k}
     \end{array}\\
    &\hspace{4cm}+\begin{array}{ccc}
        \hat{Z}_{l,k}  &  &\hat{Z}_{l+1,k} \\
         & \hat{X}_{l+\frac{1}{2},k-\frac{1}{2}} & \\
        \hat{Z}_{l,k-1} & & \hat{Z}_{l+1,k-1}
     \end{array} -\sum_{v_r}\hat{X}_{v_r}-\sum_{v_b}\hat{X}_{v_b}\, ,
     \label{eq:2dclstrtwistedH}
\end{split}
\end{align}
\end{widetext}
where the prime in the second term in the sum indicates that the third and fourth terms in the sum (with flipped signs) are subtracted. 
The symmetries of the Hamiltonian \eqref{eq:2dclstrtwistedH} are 
\begin{align}
    \begin{split}
        &(\hat{\eta}^{x}_{r,j})_{r,k}=\prod_{i=1}^{L_x}\hat{X}_{i,j},\quad (\hat{\eta}^{y}_{r,i})_{r,k}=\prod_{j=1}^{L_y}\hat{X}_{i,j}, \\
    &(\hat{\eta}^{x}_{b,j})_{r,k}=\prod_{i=1}^{L_x}\hat{X}_{i+\frac{1}{2},j+\frac{1}{2}},\quad (\hat{\eta}^{y}_{b,i})_{r,k}=\prod_{j=1}^{L_y}\hat{X}_{i+\frac{1}{2},j+\frac{1}{2}}, \\
    &\hat{\rm V}_{r,k}^{(2)}=\hat{Z}_{l+\frac{1}{2},k+\frac{1}{2}}\hat{Z}_{l+\frac{1}{2},k-\frac{1}{2}}\prod_{v_r}\prod_{v_r\in\partial p_r}\mathrm{CZ}_{v_r,p_r}\, ,
    \end{split}
    \label{eq:modifiedsymmetries}
\end{align}
where the subscript $r,k$ means that these are the symmetries of the Hamiltonian \eqref{eq:2dclstrtwistedH}, which possesses defects associated with $\hat{\eta}^x_{r,k}$. We note that these symmetries of the defect Hamiltonian do not depend on the choice of Hamiltonian~\eqref{eq:2dTFcluster}. The following symmetries of this defect Hamiltonian obey a projective algebra:
\begin{align}
    &(\hat{\eta}_{b,k}^x)_{r,k}\hat{\rm V}_{r,k}^{(2)}=-\hat{\rm V}_{r,k}^{(2)}(\hat{\eta}_{b,k}^x)_{r,k}\, ,\nonumber\\
    &  (\hat{\eta}_{b,k-1}^x)_{r,k}\hat{\rm V}_{r,k}^{(2)}=-\hat{\rm V}_{r,k}^{(2)}(\hat{\eta}_{b,k-1}^x)_{r,k}\, .
\end{align}
This indicates that two parallel and adjacent subsystem symmetries of different sublattices have a type III anomaly with $\hat{\rm V}^{(2)}$. Hence, for example, it is not possible to realize a symmetric gapped phase with symmetry of the form $\hat{\rm V}^{(2)}\hat{\eta}^x_{r,k}\hat{\eta}^x_{b,k}$. Applying this method again, we can see that any symmetries of the form $\hat{\rm V}^{(2)}\prod_{j=j_0}^{j_1}(\hat{\eta}^x_{r,j}\hat{\eta}^x_{b,j})$ is anomalous for $j_0,j_1\in\{1,...,L_y\}$ and the pair $(j_0,j_1)\neq(1,L_y)$. Similarly, $\hat{\rm V}^{(2)}\prod_{i=i_0}^{i_1}(\hat{\eta}^y_{r,i}\hat{\eta}^y_{b,i})$ is also anomalous for $i_0,i_1\in\{1,...,L_x\}$ and the pair $(i_0,i_1)\neq(1,L_x)$. 

Now we sketch a method to find the modified symmetries of the defect Hamiltonian without a particular choice of original Hamiltonian. Here we take an infinite 2D lattice for the argument. Later, we can put it on a finite lattice. The defect Hamiltonian for the line symmetry $\hat{\eta}^x_{r,k}$ is 
\begin{align}
    \mathrm{H}_{\hat{\eta}^x_{r,k}}=\hat{\mathcal{U}}_{\eta^x_{r,k}}^{\leq l}\mathrm{H}\hat{\mathcal{U}}_{\hat{\eta}^x_{r,k}}^{\leq l}
\end{align}
where $\hat{\mathcal{U}}_{\eta^x_{r,k}}^{\leq l}=\prod_{i'\leq l}\hat{X}_{i',j}$. 
Then it is straightforward to see that if $\rm \hat{V}^{(2)}$ and the subsystem symmetries are symmetries of $\rm H$, then  \eqref{eq:modifiedsymmetries} are the symmetries of $\mathrm{H}_{\hat{\eta}^x_{r,k}}$.
\subsubsection{\texorpdfstring{$\mathbb{Z}_2$}{Lg} subsystem symmetry}
Now let us consider a Hamiltonian that is symmetric under $\mathbb{Z}_2$ subsystem symmetry as well as $\hat{\rm V}^{(2)}_{\mathbb{Z}_2}$,
\begin{align}
    \mathrm{H}^{\mathbb{Z}_2}_{\text{TFSSPT}}= -\sum_{i,j}\begin{array}{ccc}
 & \hat{Z} & \hat{Z} \\
\hat{Z}& \hat{X} & \hat{Z} \\
\hat{Z}& \hat{Z} &
\end{array}-\sum_{i,j}\hat{X}\, .
\label{eq:2dZ_2TFSSPT}
\end{align}
Now, let us twist the Hamiltonian (defect Hamiltonian) \eqref{eq:2dZ_2TFSSPT} with a linear symmetry $\hat{\eta}^x_k$
\begin{widetext}
\begin{align}
    \mathrm{H}^{\mathbb{Z}_2}_{\text{TFSSPT},\hat{\eta}^x_k}&=-\sum_{i,j}'\begin{array}{ccc}
 & \hat{Z} & \hat{Z} \\
\hat{Z}& \hat{X} & \hat{Z} \\
\hat{Z}& \hat{Z} &
\end{array}+\begin{array}{ccc}
 & \hat{Z} & \hat{Z} \\
\hat{Z}& \hat{X}_{L,k+1} & \hat{Z} \\
\hat{Z}& \hat{Z} &
\end{array}+\begin{array}{ccc}
 & \hat{Z} & \hat{Z} \\
\hat{Z}& \hat{X}_{L,k} & \hat{Z} \\
\hat{Z}& \hat{Z} &
\end{array}+\begin{array}{ccc}
 & \hat{Z} & \hat{Z} \\
\hat{Z}& \hat{X}_{L-1,k} & \hat{Z} \\
\hat{Z}& \hat{Z} &
\end{array}+\begin{array}{ccc}
 & \hat{Z} & \hat{Z} \\
\hat{Z}& \hat{X}_{L-1,k-1} & \hat{Z} \\
\hat{Z}& \hat{Z} &
\end{array}-\sum_{i,j}\hat{X}\, ,
\label{eq:Z_2TFSSPT,k}
\end{align}
where the prime in the first term in the sum indicates that  
the particular four terms (with opposite signs) that follow the sum are removed. The symmetries of the twisted Hamiltonian are 
\begin{align}
\begin{split}
     &(\hat{\eta}^x_j)_k=\prod_{i = 1 }^L \hat{X}_{i,j}\, ,\quad (\hat{\eta}^y_i)_k=\prod_{j = 1 }^L \hat{X}_{i,j}\, ,\\
    & (\hat{\eta}^{\text{diag}}_l)_k=\prod_{m = 1 }^L \hat{X}_{ m,[m+l]_L}\,,\qquad  i,j,l,k\in\{1,...,L\}\\
    &  (\hat{\rm V}^{(2)}_{\mathbb{Z}_2})_{k}=\prod_{ v}CZ_{v,v+(1,0)}CZ_{v,v+(1,1)}CZ_{v,v+(0,1)}\\
    &\hspace{2cm}\times \hat{Z}_{L,k+1}\hat{Z}_{L-1,k-1}\hat{Z}_{L,k}\hat{Z}_{L-1,k}\, .
\end{split}
\end{align}
\end{widetext}
The subscript $k$ in the above symmetries indicates that the above given symmetries are the symmetries of \eqref{eq:Z_2TFSSPT,k}. $(\hat{\rm V}^{(2)}_{\mathbb{Z}_2})_k$ satisfies a projective algebra with the following symmetries:
\begin{align}
\begin{split}
     &(\hat{\rm V}^{(2)}_{\mathbb{Z}_2})_k(\hat{\eta}^x_{k+1})_{k}=-(\hat{\eta}^x_{k+1})_{k}(\hat{\rm V}^{(2)}_{\mathbb{Z}_2})_k\, ,\\
     &(\hat{\rm V}^{(2)}_{\mathbb{Z}_2})_k(\hat{\eta}^x_{k-1})_{k}=-(\hat{\eta}^x_{k-1})_{k}(\hat{\rm V}^{(2)}_{\mathbb{Z}_2})_k\, .
\end{split}
\end{align}
The above equations indicate that there is a type III anomaly between two parallel adjacent subsystem symmetry lines and $\hat{\rm V}^{(2)}_{\mathbb{Z}_2}$. (Note that the twist is at $k$, so  $k+1$ or $k-1$ is adjacent to it.) Hence, we cannot have a unique gapped ground state with the product of the three symmetries preserved. The argument is independent of the choice of Hamiltonian, as we argued in the previous case.

We note that the product of $\hat{\rm V}^{(2)}_{\mathbb{Z}_2}$ with any individual horizontal, vertical, or diagonal line symmetry is also anomalous. This anomaly is the $1+1$D CZX anomaly discussed in the literature~\cite{Chen:2011bcp}. Moreover, $\hat{\rm V}^{(2)}_{\mathbb{Z}_2}\prod_{j=1}^L\hat{\eta}^x_j$ is anomaly-free as we find $[\prod_{j\neq k}(\hat{\eta}^x_j)_k,(\hat{\rm V}^{(2)}_{\mathbb{Z}_2})_k]$=0, in the presence of  $\hat{\eta}^x_k$ defect.
\subsection{Defect fusion method}
Here, we establish the anomaly that we discussed before by the defect Hamiltonian method using a different method. We will compute the anomaly by fusing defects of various symmetries. This method works in $1+1$D and is developed in~\cite{Seifnashri:2023dpa}. Although we are working in $2+1$D, certain anomalies are $1+1$D anomalies and can be captured by this method. Now we briefly review the method described in~\cite{Seifnashri:2023dpa}.

Let $G$ be the symmetry group. If $g\in G$, then let $U_g$ be the unitary representation of the symmetry group element $g$. If we start with a Hamiltonian $\rm H$, then we can create a defect Hamiltonian $\mathrm{H}_g$ by applying a truncated unitary $U_g$ on the Hamiltonian $\rm H$. Let us assume we have a one-dimensional chain of sites labeled by integers. Then the defects are located on the links. A $g$-defect at a link between the vertices $j$ and $j+1$ can be denoted by $\mathrm{H}_g^{(j,j+1)}$. Now, if we have two defects on adjacent links and we want to fuse them, it can be fused by a unitary operator.
\begin{align}
    \lambda^j(g,h)\mathrm{H}_{g,h}^{(j-1,j);(j,j+1)}\lambda^j(g,h)^{-1}=  \mathrm{H}_{gh}^{(j,j+1)}.
\end{align}
Now, if we have three defects, we can fuse them in two different ways. They should be equivalent up to an overall phase factor. The phase factor would give the information about the anomaly. Explicitly, the phase factor is the $F$ symbol defined in the equation below:
\begin{align}
    &\lambda^j(g_1,g_2g_3)\lambda^{j-1}(g_1,1)\lambda^{j}(g_2,g_3)\nonumber\\
    &\qquad=F^j(g_1,g_2,g_3)\lambda^j(g_1g_2,g_3)\lambda^{j-1}(g_1,g_2)\,.
    \label{eq:lambdaassociativity}
\end{align}
According to~\cite{Seifnashri:2023dpa}, the anomaly of the symmetry group is captured by 
\begin{align}
\omega^j(g_1,g_2,g_3)\equiv \frac{F^j(g_1,g_2,g_3)}{F^j(g_1,g_2,1)}\, .
    \label{eq:cocycle}
\end{align}
It can be checked that $\omega^j(g_1,g_2,g_3)$ satisfies the cocycle condition and describes an anomaly if it is not equivalent to a coboundary.
\begin{figure*}[]
    \centering
    \includegraphics[scale=1]{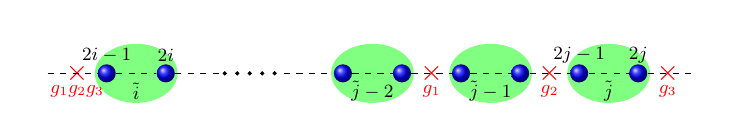}
    \caption{One-dimensional chain with two sites (colored blue) combined into an effective site (colored green). We stick to the convention that the left blue site in a green site is at an odd integer location and the right one is on an even integer location. The effective site is labeled with a tilde. The defects corresponding to $g_1=\eta_0$, $g_2=\eta_e$ and $g_3=\hat{\rm V}$ are located on the links $(\tilde{j}-2,\tilde{j}-1)$, $(\tilde{j}-1,\tilde{j})$ and $(\tilde{j},\tilde{j}+1)$ respectively. There is a defect $g_1g_2g_3$ on the link $(\tilde{i}-1,\tilde{i})$. All the defects are marked with a red cross.}
    \label{fig:1Ddefectfusion}
    \end{figure*}
\begin{figure*}[]

    \begin{subfigure}{2\columnwidth}
         \centering
    \includegraphics[scale=1]{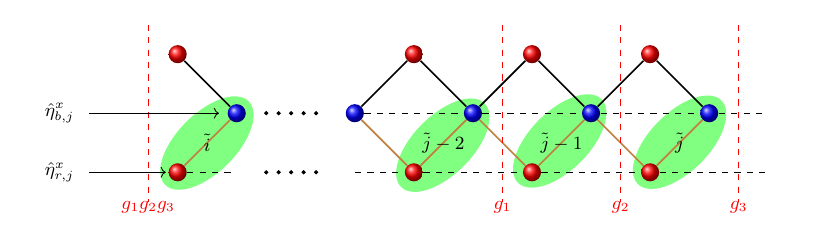}
    \caption{An illustration of starting position of defects of two symmetry lines $g_1=\hat{\eta}^x_{r,j}$, $g_2=\hat{\eta}^x_{b,j}$ and controlled-Z operator $g_3=\hat{\rm V}^{(2)}$ on union of lattices formed by a square lattice and its dual. For this, we combine the red and blue sites into an effective site that is colored green. Then we consider the $1+1$D brown zig-zag line for placing the defects and fusing them. Black and brown solid lines indicate the pattern in which the sites are entangled via the controlled-Z operator $\hat{\rm V}^{(2)}$. The black dashed lines indicate the subsystem symmetry lines for which we consider the defects.}
    \label{fig:twolinesinz2timesz2}
    \end{subfigure}
    \begin{subfigure}[b]{2\columnwidth}
        \centering
    \includegraphics[scale=1]{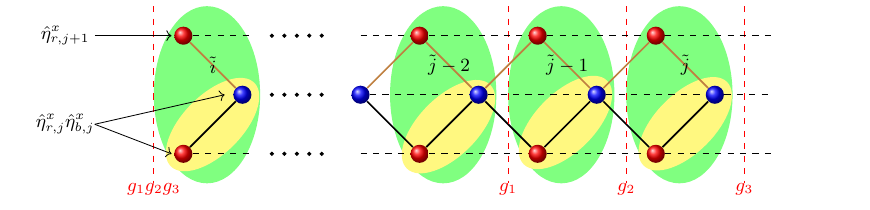}
    \caption{An illustration of position of defects for symmetry lines $g_1=\hat{\eta}^x_{r,j+1}$, $g_2=\hat{\eta}^x_{b,j}\hat{\eta}^x_{r,j}$ and controlled-Z operator $g_3=\hat{\rm V}^{(2)}$ on union of lattices formed by square lattice and its dual. For this, we combine the red and blue sites to form an effective site that is colored yellow, and further combine one more red site to form an effective site that is colored green. Then we consider the $1+1$D brown zig-zag line for placing the defects and fusing them. Black and brown solid lines indicate the pattern in which the sites are entangled via the controlled-Z operator $\hat{\rm V}^{(2)}$. The black dashed lines indicate the subsystem symmetry lines for which we consider the defects.}
    \label{fig:threelinedefectfusionz2xz2}
    \end{subfigure}
    \begin{subfigure}[b]{2\columnwidth}
        \centering
    \includegraphics[scale=1]{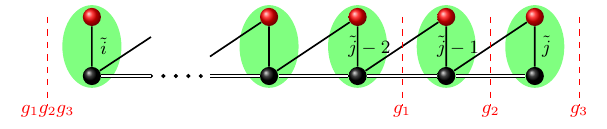}
    \caption{This figure illustrates blocking sites to form an effective site. Starting with a two-dimensional square lattice and its dual (colored red and blue), here we have blocked multiple sites in the vertical direction to form an effective site. The red sites indicate a site in the original two-dimensional lattice, while a black site indicates a blocked effective site. We further block the red and black sites to form another effective site that is colored green. The black solid line indicates a pattern in which the sites are entangled via the controlled-Z operator $\hat{\rm V}^{(2)}$. The black double solid line indicates that the effective black sites are entangled via $\hat{\rm V}^{(2)}$ by multiple connections.}
    \label{fig:multiplelinedefectfusion}
    \end{subfigure}
    \caption{An illustration of the defect fusion method to detect type III anomaly for various choices of $g_1$,$g_2$, and $g_3$, shown in (a), (b), and (c). The red dashed line indicates the link where the defects are placed. On the left side, all the defects are placed on the same link $(\tilde{i}-1,\tilde{i}),$ giving rise to $g_1g_2g_3$ defect. On the right side $g_1$, $g_2$ and $g_3$ defects are placed on the links $(\tilde{j}-2,\tilde{j}-1)$, $(\tilde{j}-1,\tilde{j})$ and $(\tilde{j},\tilde{j}+1)$ respectively. In all the choices, we reduce the $2+1$D defect fusion into an effective $1+1$D defect fusion. The green colored region can be thought of as an effective site, and we enumerate them with a tilde.}
\end{figure*}
\subsubsection{Anomaly from defect fusion method with truncated symmetry on a line segment}
Reference~\cite{Seifnashri:2023dpa} considered defects by truncating the symmetry in a semi-infinite way. 
However, we have considered defects that are obtained by truncating the symmetry to a finite line segment. Hence, we need to check whether fusing the defects as in \eqref{eq:lambdaassociativity} would give the anomaly cocycle as in~\eqref{eq:cocycle}. Here, we prove that this indeed gives rise to the anomaly cocycle.

Let us consider a truncated symmetry operator in the interval $[i,j]$ with reference site $i$. Suppose $U_g$ is a unitary symmetry operator for $g\in G$. We define $U^{[i,j]}_g$ as the truncated symmetry operator in the interval $[i,j]$. Then, according to~\cite{Seifnashri:2023dpa} and~\cite{else2014classifying}
\begin{align}
    U^{[i,j]}_{g_1}U^{[i,j]}_{g_2}=\left(\Omega_L^{i}(g_1,g_2)\Omega_R^{j}(g_1,g_2)\right)^{-1}U_{g_1g_2}^{[i,j]}
    \label{eq:projective symmetryoninterval}
\end{align}
In~\cite{else2014classifying}, $\left(\Omega_L^{i}(g_1,g_2)\Omega_R^{j}(g_1,g_2)\right)^{-1}$ is written as $\Omega_M(g_1,g_2)$ where $M$ is the interval $[i,j]$.  In terms of the unitary operator, the defect Hamiltonian takes the form
\begin{align}
    \mathrm{H}_{g_1,g_2}^{(j-1,j);(j,j+1)}=U_{g_1}^{[i,j-1]}U_{g_2}^{[i,j]}\mathrm{H}(U_{g_2}^{[i,j]})^{-1}(U_{g_1}^{[i,j-1]})^{-1}\, .
    \label{eq:defectHamiltoniandef}
\end{align}
The fusion operators can be explicitly derived from \eqref{eq:defectHamiltoniandef}
\begin{align}
    \lambda^j(g_1,g_2)=U_{g_1g_2}^{[i,j]}(U_{g_2}^{[i,j]})^{-1}(U_{g_1}^{[i,j-1]})^{-1}\, .
    \label{eq:fusionoperator}
\end{align}
Combining \eqref{eq:fusionoperator} and \eqref{eq:projective symmetryoninterval}, we find
\begin{align}\label{eq:omega-ij-lambda-j}
    \Omega_L^{i}(g_1,g_2)\Omega_R^{j}(g_1,g_2)=\lambda^{j}(g_1,g_2)(\lambda^{j}(g_1,1))^{-1}\, .
\end{align}
\begin{widetext}
Using \eqref{eq:lambdaassociativity} and \eqref{eq:cocycle}, we have
\begin{align}
    \omega^j(g_1,g_2,g_3)=\frac{\lambda^j(g_1,g_2g_3)\lambda^{j-1}(g_1,1)\lambda^j(g_2,g_3)(\lambda^j(g_2,1))^{-1}}{\lambda^j(g_1g_2,g_3)(\lambda^j(g_1g_2,1))^{-1}\lambda^j(g_1,g_2)\lambda^{j-1}(g_1,1)}\, ,
\end{align}
which can be rewritten as follows using \eqref{eq:omega-ij-lambda-j}:
\begin{align}
   \omega^j(g_1,g_2,g_3)=\frac{\Omega_L^i(g_1,g_2g_3)\Omega_R^j(g_1,g_2g_3)\lambda^j(g_1,1)\lambda^{j-1}(g_1,1)\Omega_L^i(g_2,g_3)\Omega_R^j(g_2,g_3)}{\Omega_L^i(g_1g_2,g_3)\Omega_R^j(g_1g_2,g_3)\Omega_L^i(g_1,g_2)\Omega_R^j(g_1,g_2)\lambda^j(g_1,1)\lambda^{j-1}(g_1,1)} \,.
   \label{eq:cocycle2}
\end{align}
We assume $U_g^j=\lambda^j(g,1)$ is supported on sites $j$ and $j+1$. On the other hand, we note that 
\begin{align}
    &\lambda^j(g_1,1)\lambda^{j-1}(g_1,1)\Omega_L^i(g_2,g_3)\Omega_R^j(g_2,g_3)\left(\lambda^j(g_1,1)\lambda^{j-1}(g_1,1)\right)^{-1}\nonumber\\
    &\hspace{3cm}=\lambda^j(g_1,1)\lambda^{j-1}(g_1,1)\Omega_R^j(g_2,g_3)\left(\lambda^j(g_1,1)\lambda^{j-1}(g_1,1)\right)^{-1}\Omega_L^i(g_2,g_3)
\end{align}
This follows from the fact that $\Omega_L^i(g_2,g_3)$ is unaffected by operators $\lambda^j(g_1,1)$ and $\lambda^{j-1}(g_1,1)$ as they are supported around site $j$ and $j-1$, and $\Omega_L^i(g_2,g_3)$ is supported around site $i$ that is far away from site $j$. Furthermore,
\begin{align}
   \lambda^j(g_1,1)\lambda^{j-1}(g_1,1)\Omega_R^j(g_2,g_3)\left(\lambda^j(g_1,1)\lambda^{j-1}(g_1,1)\right)^{-1}=U_{g_1}^{[i,j]}\Omega_R^j(g_2,g_3) (U_{g_1}^{[i,j]})^{-1}\, .
\end{align}
Clubbing the above equation with \eqref{eq:cocycle2}, we find
\begin{align}
    &\omega^j(g_1,g_2,g_3)\Omega_L^i(g_1g_2,g_3)\Omega_R^j(g_1g_2,g_3)\Omega_L^i(g_1,g_2)\Omega_R^j(g_1,g_2)\nonumber\\
    &\qquad=\Omega_L^i(g_1,g_2g_3)\Omega_R^j(g_1,g_2g_3)\Omega_L^i(g_2,g_3)U_{g_1}^{[i,j]}\Omega_R^j(g_2,g_3)(U_{g_1}^{[i,j]})^{-1}
\end{align}
We choose $\Omega_L$ such that all the phases are absorbed into $\Omega_R$, i.e., 
\begin{align}
    \Omega_L^i(g_1g_2,g_3)\Omega_L^i(g_1,g_2)=\Omega_L^i(g_1,g_2g_3)\Omega_L^i(g_2,g_3)\, .
\end{align}
Then 
\begin{align}
    \Omega_R^j(g_1g_2,g_3)\Omega_R^j(g_1,g_2)\omega^j(g_1,g_2,g_3)=\Omega_R^j(g_1,g_2g_3)U_{g_1}^{[i,j]}\Omega_R^j(g_2,g_3)(U_{g_1}^{[i,j]})^{-1}\, ,
\end{align}
which is the same as equation (5) in~\cite{else2014classifying}. In other words, the 3-cocycle computed using truncated symmetry operators on a line segment indeed gives the same anomaly as computed using the semi-infinite segment in~\cite{else2014classifying}.
\end{widetext}
\subsubsection{$1+1$D CZX anomaly}
First, let us analyze the $1+1$ dimensional CZX anomaly using this technique. We consider a one-dimensional ring with $2N$ sites. The system has the following symmetries: $\hat{\rm V}=\prod_{i\in\mathbb{Z}_{2N}}{\rm CZ}_{i,i+1}$, $\hat{\eta}_e=\prod_{i\in \mathbb{Z}_{2N}}\hat{X}_{2i}$ and $\hat{\eta}_o=\prod_{i\in \mathbb{Z}_{2N}}\hat{X}_{2i+1}$. 
These form a symmetry group $\mathbb{Z}_2^3$. Now, let us denote the three nontrivial generators of $\mathbb{Z}_2^3$ by $g_1$, $g_2$, and $g_3$. Let us combine the two sites into an effective single site. Sites at $2k-1$ and $2k$ are combined to a single site $\tilde{k}$ for $\tilde{k}=1,...,N$.  Then we can denote the operator $\mathcal{O}$ at site $\tilde{k}$ by
\begin{align}
   \mathcal{O}_{\tilde{k}}^o=\mathcal{O}_{2k-1}\, ,\qquad\mathcal{O}_{\tilde{k}}^e=\mathcal{O}_{2k}\, . 
\end{align}
We set $g_1=\hat{\eta}_o$, $g_2=\hat{\eta}_e$ and $g_3=\hat{\rm V}$. We place the $g_1$, $g_2$ and $g_3$ defects on the links $(\tilde{j}-2,\tilde{j}-1)$, $(\tilde{j}-1,\tilde{j})$ and $(\tilde{j},\tilde{j}+1)$ respectively. We also put the $g_1$, $g_2$, and $g_3$ defects far away from the previously defined defects, all at the same location on the other side. We take its location at the link $(\tilde{i}-1,\tilde{i})$ (see Figure~\ref{fig:1Ddefectfusion} for an illustration). We compute the unitary that fuses the defects:
\begin{align}
\begin{split}
    &\lambda^{\tilde{j}-1}(g_1,g_2)=\hat{X}_{2j-3}\, ,\quad \lambda^{\tilde{j}-1}(g_1,1)=\hat{X}_{2j-3}\, ,\\
    &\lambda^{\tilde{j}}(g_2,g_3)=\hat{Z}_{2i-1}\hat{X}_{2j}\, ,\quad\lambda^{\tilde{j}}(g_1,g_2g_3)=\hat{X}_{2j-1}\hat{Z}_{2j}\, ,\\
    & \lambda^{\tilde{j}}(g_1g_2,g_3)=-\hat{Z}_{2i-1}\hat{X}_{2j-1}\hat{Z}_{2j}\hat{X}_{2j}\, .
\end{split}  
\end{align}
In the above computations, we choose the  truncated symmetry operator on the interval $[\tilde{i},\tilde{j}]$ of $g_2g_3$ to be $\left(\overrightarrow{\prod}_{\tilde{k}=\tilde{i}}^{\tilde{j}-1}CZ_{2k-1,2k}CZ_{2k,2k+1}\hat{X}_{2k}\right)CZ_{2j-1,2j}\hat{X}_{2j}$ and that of $g_1g_2g_3$ to be $\left(\overrightarrow{\prod}_{\tilde{k}=\tilde{i}}^{\tilde{j}-1}CZ_{2k-1,2k}CZ_{2k,2k+1}\hat{X}_{2k-1}\hat{X}_{2k}\right)CZ_{2j-1,2j}\hat{X}_{2j-1}\hat{X}_{2j}$ where the vector arrow on top of the product indicate that the product is taken from left to right with increasing value of $\tilde{k}$. We choose this convention so that we can pull all the $CZ$ operators to the left, and in the infinite lattice limit, when $[\tilde{i},\tilde{j}]$ is taken to $[-\infty,\infty]$, we get the symmetry operator $\hat{\rm V}\hat{\eta}_o\hat{\eta}_e=\hat{\eta}_o\hat{\eta}_e\hat{\rm V}=g_1g_2g_3$.
Substituting the above operators to \eqref{eq:lambdaassociativity}, we find $F^{\tilde{j}}(g_1,g_2,g_3)=-1$. By repeating the calculation for a general element $g_1=\hat{\eta}_o^{i_1}\hat{\eta}_e^{i_2}\hat{\rm V}^{i_3}$, $g_2=\hat{\eta}_o^{j_1}\hat{\eta}_e^{j_2}\hat{\rm V}^{j_3}$ and $g_3=\hat{\eta}_o^{k_1}\hat{\eta}_e^{k_2}\hat{\rm V}^{k_3}$, we find $F^{\tilde{j}}(g_1,g_2,g_3)=(-1)^{i_1j_2k_3}$. Furthermore, we find $\omega^{\tilde{j}}(g_1,g_2,g_3)=(-1)^{i_1j_2k_3}$, which is a non-trivial 3-cocycle, which in turn indicates that there is a type III anomaly between the three symmetries.     
\subsubsection{\texorpdfstring{$\mathbb{Z}_2\times\mathbb{Z}_2$}{Lg} subsystem symmetry}
There are many possible combinations of anomalous symmetries. First, let us list a few examples, then state the general result. In all the examples below, we truncate the symmetry to a cylinder whose horizontal coordinates are denoted with a tilde. 
\begin{enumerate}
    \item $g_1=\hat{\eta}_{r,j}^x$, $g_2=\hat{\eta}_{b,j}^x$, $g_3=\hat{\rm V}^{(2)}$: This case is equivalent to the $1+1$D case that we discussed before, if we consider the one-dimensional zigzag line connecting the sites in the red and blue sublattice (see Figure~\ref{fig:twolinesinz2timesz2}). The product of controlled-Z along this zigzag line is contained in the bigger product of controlled-Z ($\hat{\rm V}^{(2)}$). 
    Then $g_1$, $g_2$, and $g_3$ have a type III anomaly that can also be verified by an explicit calculation using defect fusion. Therefore, $\hat{\eta}_{r,j}^x\hat{\eta}_{b,j}^x\hat{\rm V}^{(2)}$ is an anomalous symmetry with the anomaly originating from CZX anomaly in $1+1$D.
    \item $g_1=\hat{\eta}_{r,j+1}^x$, $g_2=\hat{\eta}_{r,j}^x\hat{\eta}_{b,j}^x$, $g_3=\hat{\rm V}^{(2)}$: In this case, the type III anomaly can be detected by an explicit computation of defect fusion. One can reduce this to a calculation of defect fusion along a one-dimensional line by blocking sites to form an effective site as given in Figure~\ref{fig:threelinedefectfusionz2xz2}. Again, we find $F^{\tilde{j}}(g_1,g_2,g_3)=-1$ and therefore $\omega^{\tilde{j}}(g_1,g_2,g_3)=-1$, indicating that there is a type III anomaly between the three symmetries.  
\item $g_1=\hat{\eta}_{r,j+1}^x$, $g_2=\prod_{l=j_0}^j\hat{\eta}_{r,l}^x\hat{\eta}_{b,l}^x$, $g_3=\hat{\rm V}^{(2)}$: Again we can reduce this to an effective one dimensional problem by blocking sites to form an effective site as given in Figure~\ref{fig:multiplelinedefectfusion}. It is straightforward to verify that $\omega^{\tilde{j}}(g_1,g_2,g_3)=-1$, indicating type III anomaly.
\item $g_1=\hat{\eta}_{r,j+1}^x$, $g_2=\hat{\eta}_{b,j_0-1}^x\prod_{l=j_0}^j\hat{\eta}_{r,l}^x\hat{\eta}_{b,l}^x$, $g_3=\hat{\rm V}^{(2)}$ for $j\neq j_0-1\,\text{mod}\, L_y$: Again we can reduce this to an effective one dimensional problem by blocking sites to form an effective site as given in Figure~\ref{fig:multiplelinedefectfusion}. It is straightforward to verify that $\omega^{\tilde{j}}(g_1,g_2,g_3)=-1$, indicating type III anomaly.
\end{enumerate}
In general, from defect fusion method, it can be argued that $\hat{\rm V}^{(2)}\prod_{k \in \mathcal{K}}\hat{\eta}^x_{r,k}\prod_{\ell \in \mathcal{L}}\hat{\eta}^x_{b,\ell}$ is anomalous when the integer sets $\mathcal{K}$ and $\mathcal{L}$, subsets of $\{1,...,L_y\}$, satisfy that $\emptyset\neq\mathcal{K}^+\cap \mathcal{L}\subsetneq \{1,...,L_y\}$ or $\emptyset\neq\mathcal{K}^-\cap \mathcal{L}\subsetneq \{1,...,L_y\}$, where $\mathcal{K}^+ = \mathcal{K}$  and $\mathcal{K}^- = \{k-1 \text{ mod }L_y\,| \, k \in \mathcal{K}\}$.
\subsubsection{\texorpdfstring{$\mathbb{Z}_2$}{Lg} subsystem symmetry}
As in the $\mathbb{Z}_2\times\mathbb{Z}_2$ case, we consider a few combinations of anomalous symmetries.
\begin{enumerate}
    \item $g_1=\hat{\eta}_{k}^x$, $g_2=\hat{\rm V}_{\mathbb{Z}_2}^{(2)}$: In $\hat{\rm V}_{\mathbb{Z}_2}^{(2)}$, there is a product of controlled-Z operators along the line in the $x$-direction at $y=k$. The product of this controlled-Z with $\hat{\eta}_{k}^x$ is anomalous and originates from the CZX anomaly in $1+1$ dimension. Therefore, $\hat{\rm V}_{\mathbb{Z}_2}^{(2)}\hat{\eta}_k^x$ is also an anomalous symmetry.
    \item $g_1=\hat{\eta}_{k+1}^x$,$g_2=\hat{\eta}_{k}^x$, $g_3=\hat{\rm V}_{\mathbb{Z}_2}^{(2)}$: Here, we block two sites into an effective site as shown in the Figure~\ref{fig:twolinedefectfusionz2}. The brown zig-zag line in \ref{fig:twolinedefectfusionz2} indicates the effective $1+1$D line in which we fuse the defects. We repeat the calculation of defect fusion and find that $\omega^{\tilde{j}}(g_1,g_2,g_3)=-1$, indicating a type III anomaly between the three symmetries.
    \item $g_1=\hat{\eta}_{k}^x$,$g_2=\prod_{i=l}^{k-1}\hat{\eta}_{i}^x$, $g_3=\hat{\rm V}_{\mathbb{Z}_2}^{(2)}$: The calculation proceeds in the same manner as before by blocking sites to form an effective site, as shown in Figure~\ref{fig:multiplelinedefectfusionZ2}. We again find $\omega^{\tilde{j}}(g_1,g_2,g_3)=-1$, which indicates a type III anomaly. 
\end{enumerate}
In general, from the defect fusion method, it can be argued that $\hat{\rm V}^{(2)}_{\mathbb{Z}_2}\prod_{k \in \mathcal{K}}\hat{\eta}^x_{k}$ is anomalous when the integer set $\mathcal{K}$ satisfies $\emptyset\neq\mathcal{K}\subsetneq \{1,...,L_y\}$. 
\begin{figure*}
\begin{subfigure}[b]{2\columnwidth}
    \centering
    \includegraphics[scale=1]{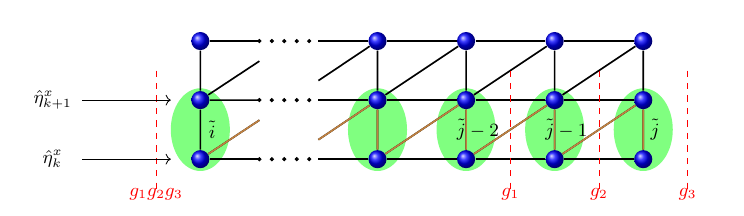}
    \caption{An illustration of the starting position of defects of two symmetry lines 
    $g_1=\hat{\eta}^x_{k+1}$, $g_2=\hat{\eta}^x_{k}$ and controlled-Z operator $g_3=\hat{\rm V}_{\mathbb{Z}_2}^{(2)}$ on a square lattice and its dual. For this, we combine two vertical blue sites into an effective site that is colored green. Then we consider the $1+1$D brown zig-zag line for placing the defects and fusing them. Black and brown solid lines indicate the pattern in which the sites are entangled via the controlled-Z operator $\hat{\rm V}_{\mathbb{Z}_2}^{(2)}$.}
    \label{fig:twolinedefectfusionz2}
    \end{subfigure}
    \begin{subfigure}[b]{2\columnwidth}
    \centering
    \includegraphics[scale=1]{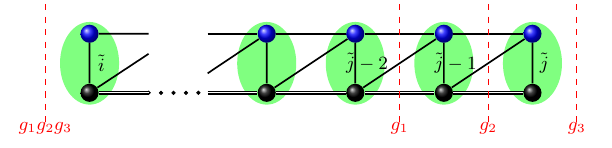}
    \caption{This figure illustrates blocking sites to form an effective site. Starting with a two-dimensional square lattice as in the $\mathbb{Z}_2$ subsystem symmetric case, here we have blocked multiple sites in the vertical direction to form an effective site. The blue sites indicate a site in the original two-dimensional lattice, while a black site indicates a blocked effective site. We further block blue and black sites to form another effective site that is colored green. The black solid line indicates a pattern in which the sites are entangled via the controlled-Z operator $\hat{\rm V}_{\mathbb{Z}_2}^{(2)}$. The black double solid line indicates that the effective black sites are entangled via $\hat{\rm V}_{\mathbb{Z}_2}^{(2)}$ by multiple connections in a more intricate way.}
    \label{fig:multiplelinedefectfusionZ2}
\end{subfigure}
\caption{An illustration of defect fusion method to detect type III anomaly involving $\mathbb{Z}_2$ subsystem symmetry and $\hat{\rm V}^{(2)}_{\mathbb{Z}_2}$. The red dashed line indicates the link where the defects are placed. On the left side, all the defects are placed on the same link $(\tilde{i}-1,\tilde{i})$, giving rise to $g_1g_2g_3$ defect. On the right side $g_1$, $g_2$ and $g_3$ defects are placed on the links $(\tilde{j}-2,\tilde{j}-1)$, $(\tilde{j}-1,\tilde{j})$ and $(\tilde{j},\tilde{j}+1)$ respectively. In both cases, we reduce the $2+1$D defect fusion into an effective $1+1$D defect fusion. The green colored region can be thought of as an effective site, and we enumerate them with a tilde.}
\end{figure*}

\subsection{Symmetry truncation method}
\label{app:mixedAnomaly}
There is another closely related method 
to establish a way to show the mixed anomaly involving linear subsystem symmetries in 2D. In Ref.~\cite{else2014classifying}, Else and Nayak offer a method to obtain a 3-cocycle from truncating the symmetry operators when the symmetries are local unitary transformations. If the 3-cocycle represents a non-trivial element in $H^3(G,U(1))$, there is a nontrivial $G$ anomaly. Below, we will first review their method when $G=\mathbb{Z}_2\times \mathbb{Z}_2$ in 1D, then discuss the generalization to 2D.

Let us start with two symmetry operators $V$ and $W$, which satisfy $VW=WV$ and $V^2=W^2=1$. They generate a $\mathbb{Z}_2\times \mathbb{Z}_2$ symmetry in the 1D system, where $V$ is the operator for $(1,0)$, and $W$ is the operator for $(0,1)$. Suppose we truncate (restrict) the symmetry operators on an interval $(a,b)$ as $V_{trunc}$ and $W_{trunc}$, and we let the truncation of operator $(1,1)$ as $W_{trunc}V_{trunc}$. Since the symmetry operators are given by finite-depth local unitaries, the following operator
\begin{equation}
 \Omega(g_1,g_2)\equiv U_{trunc}(g_1)U_{trunc}(g_2)U_{trunc}(g_1 g_2)^{-1}. 
\end{equation}
is a local unitary acting in the vicinity of $a$ and $b$. We focus only on truncations that satisfy $V_{trunc}^2=W_{trunc}^2=1$, which would simplify the discussion below. We note that this condition can be understood as the anomaly-free condition of each $\mathbb{Z}_2$ symmetry, which suits all the examples we consider in this work. From the associativity of operators $U_{trunc}(g)$, we have 
\begin{align}
    &\Omega(g_1,g_2)\Omega(g_1 g_2,g_3)= \nonumber
    \\
    &  U_{trunc}(g_1)\Omega(g_2,g_3)U_{trunc}(g_1)^{-1}\Omega(g_1,g_2 g_3)\,.
\end{align}
We take a large interval such that points $a$ and $b$ are far away. Hence, we can further restrict this local unitary as $\Omega_a (g_1,g_2) \equiv (\Omega(g_1,g_2))_a$, where $(\mathcal{O})_a$ is our notation for restriction of operator $\mathcal{O}$ to the vicinity of $a$. We note that the operator $\Omega_a (g_1,g_2)$ is well-defined up to a phase ambiguity. It satisfies
\begin{widetext}
\begin{align}
    &\Omega_a(g_1,g_2)\Omega_a(g_1 g_2,g_3)=\omega(g_1,g_2,g_3)U_{trunc}(g_1)\Omega_a(g_2,g_3)U_{trunc}(g_1)^{-1}\Omega_a(g_1,g_2 g_3)\,,
\end{align}
\end{widetext}
where $\omega(g_1,g_2,g_3)$ is a 3-cocycle. 
 
According to the above definitions, we can compute six components of the 3-cocycle when the arguments are generators of the group $\mathbb{Z}_2\times \mathbb{Z}_2$. There are only two non-trivial ones,
\begin{subequations}
\begin{align}
    \omega_1\equiv \omega((1,0),(1,0),(0,1))&=(V_{trunc} B_a V_{trunc} B_a)^{-1},\\ \omega_2\equiv \omega((0,1),(1,0),(0,1))&=(W_{trunc} B_a W_{trunc} B_a)^{-1},
\end{align}
\end{subequations}
where $B_a\equiv (V_{trunc}W_{trunc}V_{trunc}^{-1}W_{trunc}^{-1})_a$.
Furthermore, because of our extra condition on the truncations, the operators $\Omega_a(g_1,g_2)$ are either $B_a$ or trivial. Therefore, the potential coboundary ambiguity of the 3-cocycle is entirely due to the phase ambiguity of the operator $B_a$. Whenever $\omega_1=\omega_2$, we can always redefine the operator $B_a$ by a phase to make both phases trivial. Whenever $\omega_1\neq \omega_2$, we can never make the 3-cocycle trivial by redefining $B_a$.

It can be shown that when $\omega_1\neq \omega_2$, there is a mixed anomaly between $V$ and $W$ symmetries. The idea of the proof is: suppose there is a short-range entangled state $\ket{\psi}$ that is symmetric under both $V$ and $W$, then we can redefine the operator truncation such that $V_{trunc}$ and $W_{trunc}$ both stabilize the state, and still satisfy $V_{trunc}^2=W_{trunc}^2=1$. As a result, we can choose the operator $B_a$ out of the new truncated operators such that it also stabilizes the state $\ket{\psi}$. Under this new truncation, by applying the operators above on the state $\ket{\psi}$, we can show that the phases $\omega_1=\omega_2=1$. In the meantime, it can also be shown that, when redefining the truncated operators by an extra unitary on the endpoints, the phases $\omega_1$ and $\omega_2$ remain invariant. Since the redefinition of $B_a$ can only change $\omega_1$ and $\omega_2$ simultaneously, we have a contradiction. Therefore, there could not be any short-range entangled symmetric state $\ket{\psi}$, i.e., there is a mixed anomaly (see the appendix in Ref.~\cite{zhang2024long} for more details of the proof).

Now let us consider a 2D system, on which there is a symmetry operator $V$ defined as a finite-depth local unitary in the 2D bulk and a symmetry operator $V$ defined as a finite-depth local unitary on a line-like subsystem, which satisfy $VW=WV$ and $V^2=W^2=1$. They generate a $\mathbb{Z}_2$ 0-form symmetry and a  $\mathbb{Z}_2$ line-like symmetry. In the systems considered in this work, the line-like operators form a subsystem $\mathbb{Z}_2$ symmetry, and we are taking one of the symmetry operators. 

Suppose we truncate (restrict) the $W$ symmetry operator on an line with far apart endpoints $a$ and $b$ as $W_{trunc}$, while we truncate the 0-form symmetry operator in a region $R$ as $V_{trunc}$, with boundary $\partial R$ far away from both $a$ and $b$. The argument below also works for a 1-form symmetry. We again focus only on truncations that satisfy $V_{trunc}^2=W_{trunc}^2=1$. Just as in the 1D situation, we can consider the operators $\Omega(g_1,g_2)$ for these two truncated symmetry operators, which factorizes as a product of local operators in the vicinities of $a$ and $b$ respectively. For the part that is restricted to the vicinity of $a$, we can define a projective phase $\omega(g_1, g_2, g_3)$, which satisfies the 3-cocycle condition due to associativity. When this 3-cocycle is non-trivial up to a coboundary given by the phase ambiguity of operators restriction, there is a mixed-anomaly between $V$ and $W$ symmetries.

Since symmetries $V$ and $W$ are drastically different forms, it might not make sense to think of their product $VW$ as a well-defined symmetry. However, we note that the non-trivial components of 3-cocycle, $\omega_{V,V,W}$ and $\omega_{W,V,W}$ can be understood as some braiding and fusion data between the topological defect of $V$ and $W$ symmetries.

\subsection{Anomaly-free symmetries}
\subsubsection{$\hat{\rm V}^{(d)}\hat{\eta}_b\hat{\eta}_r$ is anomaly-free}\label{sec:nonanomalous}

We establish this using the defect Hamiltonian method. We consider a Hamiltonian $\rm H$ symmetric under all the symmetries $\hat{\rm V}^{(d)}$, $\hat{\eta}_b$, and $\hat{\eta}_r$. Consider the truncated semi-infinite symmetry operator of $\hat{\eta}_r$ denoted by $\hat{\eta}_r^{trunc}$. This would introduce a defect to the Hamiltonian and we call the defect Hamiltonian as $\rm H^{\hat{\eta}_r}$. The modified symmetries for this Hamiltonain $\rm H^{\hat{\eta}_r}$ are 
$\tilde{\hat{\eta}}_b=\hat{\eta}_r^{trunc}\tilde{\hat{\eta}}_b(\hat{\eta}_r^{trunc})^{-1}=\hat{\eta}_b$ 
and $\tilde{\hat{\mathrm{V}}}^{(d)}=\hat{\eta}_r^{trunc}\hat{\mathrm{V}}^{(d)}(\hat{\eta}_r^{trunc})^{-1}$. The modified symmetry operator $\tilde{\hat{\mathrm{V}}}^{(d)}\sim\hat{ \mathrm{V}}^{(d)}\prod Z_{v_b}$,  where the product is over an even number of $Z$ operators. Then, $[\tilde{\hat{\eta}}_b,\tilde{\hat{\mathrm{V}}}^{(d)}]=0$, and we find that there is no projective algebra. Similarly, we can argue that if we truncate $\hat{\eta}_b^{trunc}$, the modified symmetries $[\tilde{\hat{\eta}}_r,\tilde{\hat{\mathrm{V}}}^{(d)}]=0$. Now, we argue that in the presence of the $\hat{\rm V}^{(d)}$ defect, i.e., applying $(\hat{\rm V}^{(d)})^{trunc}$ on half-space, the modified symmetries $[\tilde{\hat{\eta}}_b,\tilde{\hat{\eta}}_r]=0$.  Conjugating with $(\hat{\rm V}^{(d)})^{trunc}$ on $\hat{\eta}_r$ or $\hat{\eta}_b$ produces products of $Z_{v_b}$ or products of $Z_{v_r}$ operators. However, we get either an odd number of $Z$ operators or an even number of $Z$ operators on both lattices. Then $[\tilde{\hat{\eta}}_b,\tilde{\hat{\eta}}_r]=0$, and there is no projective algebra. 

As another method, in $2+1$D, we notice that the symmetry under consideration here fits into the criterion discussed in
Ref.~\cite{else2014classifying}, i.e. it contains an on-site shift part ($\hat{\eta}_r \hat{\eta}_b$) and a non-on-site diagonal part ($\hat{\rm V}^{(2)}$). After restricting this symmetry operator $U$ in a region as $U_{trunc}$, an operator supported on the boundary of this region can be defined as
\beq
    \mathcal{N}^{(2)}\equiv U_{trunc}(g_1)U_{trunc}(g_2)U_{trunc}(g_1 g_2)^{-1}. 
\eeq
Further restricting this $\mathcal{N}^{(2)}$ operator on an open line gives rise to $\mathcal{N}^{(1)}$  that can be decomposed into a product of operators supported on the two endpoints of the open line. Another restriction of $\mathcal{N}^{(1)}$ on one of the endpoints is eventually a point-like operator, the charge of which under $U$ is a component of the 4-cocycle  that characterizes the potential anomaly. (We note that the 4-cocyle classifies whether the phase in consideration is the boundary of a (3+1)D nontrivial SPT phase.) 

If we now specialize to the $U=\hat{\mathrm{V}}^{(2)}\hat{\eta}_r\hat{\eta}_b$ symmetry operator in a square region $R$, then $\mathcal{N}^{(2)} = \prod_{v\in \partial R} Z_v$. Another restriction of $\mathcal{N}^{(2)}$ gives rise to $\mathcal{N}^{(1)}$ being just a phase factor. It is not charged under $U$, thus resulting in a trivial 4-cocycle. Therefore, we can conclude that the symmetry is anomaly-free. Such an argument generalizes to the higher-dimension of $U=\hat{\rm V}^{(d)}\hat{\eta}_r\hat{\eta}_b$ and leads to the conclusion that the higher-dimensional versions are also anomaly-free.

Now we write down an explicit Hamiltonian and it is the unique short-range entangled ground state that is symmetric with respect to the product $\hat{\mathrm{V}}\hat{\eta}_r\hat{\eta}_b$ in $2+1$D. Let us consider a bi-partite lattice colored red and blue in Figure~\ref{fig:short-range-entangled-state}. Consider the Hamiltonian
\begin{widetext}
\begin{align}
    \mathrm{H}_{XX-ZZ}&=-\sum_{\substack{v_b=(i+\frac{1}{2},j+\frac{1}{2})\\
    i+j=0\text{ mod }2}}\raisebox{-15pt} {\begin{tikzpicture}[scale=0.6, auto,
    redBall/.style={circle, ball color = red,
    radius=0.1em},blueBall/.style={circle, ball color = blue,
    radius=0.4em}, decoration={markings,
  mark=between positions 0 and 1 step 6pt
  with { \draw [fill] (0,0) circle [radius=0.8pt];}}]
        \fill[green!50] (0.7,0) ellipse (1.7 and 0.5);
        \node[draw, blueBall] (b0) at (0,0) {};
        \node[draw, blueBall, right of=b0] (b1) {};
        \node[] (a0) at (0,-0.5) {$X_{v_b}$};
        \node[] (a1) at (2,-0.5) {$X_{v_b+(1,0)}$};
    \end{tikzpicture}}-\sum_{\substack{v_b=(i+\frac{1}{2},j+\frac{1}{2})\\
    i+j=0\text{ mod }2}}\raisebox{-15pt} {\begin{tikzpicture}[scale=0.6, auto,
    redBall/.style={circle, ball color = red,
    radius=0.1em},blueBall/.style={circle, ball color = blue,
    radius=0.4em}, decoration={markings,
  mark=between positions 0 and 1 step 6pt
  with { \draw [fill] (0,0) circle [radius=0.8pt];}}]
        \fill[green!50] (0.7,0) ellipse (1.7 and 0.5);
        \node[draw, blueBall] (b0) at (0,0) {};
        \node[draw, blueBall, right of=b0] (b1) {};
        \node[] (a0) at (0,-0.5) {$Z_{v_b}$};
        \node[] (a1) at (2,-0.5) {$Z_{v_b+(1,0)}$};
    \end{tikzpicture}}-\sum_{\substack{v_b=(i+\frac{1}{2},j+\frac{1}{2})\\
    i+j=0\text{ mod }2}}\raisebox{-25pt} {\begin{tikzpicture}[scale=0.6, auto,
    redBall/.style={circle, ball color = red,
    radius=0.1em},blueBall/.style={circle, ball color = blue,
    radius=0.4em}, decoration={markings,
  mark=between positions 0 and 1 step 6pt
  with { \draw [fill] (0,0) circle [radius=0.8pt];}}]
        \fill[green!50] (0.7,0) ellipse (1.7 and 0.5);
        \fill[violet!50] (-1,0) ellipse (0.5 and 1.7);
        \fill[violet!50] (3,0) ellipse (0.5 and 1.7);
        \node[draw, blueBall] (b0) at (0.1,0) {};
        \node[draw, blueBall, right of=b0] (b1) {};
        \node[draw, redBall] (r0) at (-1,-1) {};
        \node[draw, redBall] (r1) at (-1,1) {};
        \node[draw, redBall] (r2) at (3,-1) {};
        \node[draw, redBall] (r3) at (3,1) {};
        \node[] (s0) at (-1,-1.5) {$Z_{v_r}$};
        \node[] (s1) at (-1.3,1.5) {$Z_{v_r+(0,1)}$};
        \node[] (s2) at (3.3,-1.5) {$Z_{v_r+(2,0)}$};
        \node[] (s3) at (3.3,1.5) {$Z_{v_r+(2,1)}$};
        \node[] (a0) at (0,-0.5) {$X_{v_b}$};
        \node[] (a1) at (2,-0.5) {$X_{v_b+(1,0)}$};
    \end{tikzpicture}}\nonumber\\
    &\hspace{3cm}-\sum_{\substack{v_r=(i,j)\\
    i+j=0\text{ mod }2}}\raisebox{-25pt} {\begin{tikzpicture}[scale=0.6, auto,
    redBall/.style={circle, ball color = red,
    radius=0.1em},blueBall/.style={circle, ball color = blue,
    radius=0.4em}, decoration={markings,
  mark=between positions 0 and 1 step 6pt
  with { \draw [fill] (0,0) circle [radius=0.8pt];}}]
        \fill[violet!50] (0,-0.7) ellipse (0.5 and 1.7);
        \node[draw, redBall] (r0) at (0,0) {};
        \node[draw, redBall, below of=r0] (r1) {};
        \node[] (s0) at (-0.1,-0.5) {$X_{v_r+(0,1)}$};
        \node[] (s1) at (0,-2.3) {$X_{v_r}$};
    \end{tikzpicture}}-\sum_{\substack{v_r=(i,j)\\
    i+j=0\text{ mod }2}}\raisebox{-25pt} {\begin{tikzpicture}[scale=0.6, auto,
    redBall/.style={circle, ball color = red,
    radius=0.1em},blueBall/.style={circle, ball color = blue,
    radius=0.4em}, decoration={markings,
  mark=between positions 0 and 1 step 6pt
  with { \draw [fill] (0,0) circle [radius=0.8pt];}}]
        \fill[violet!50] (0,-0.7) ellipse (0.5 and 1.7);
        \node[draw, redBall] (r0) at (0,0) {};
        \node[draw, redBall, below of=r0] (r1) {};
        \node[] (s0) at (-0.1,-0.5) {$Z_{v_r+(0,1)}$};
        \node[] (s1) at (0,-2.3) {$Z_{v_r}$};
    \end{tikzpicture}}-\sum_{\substack{v_r=(i,j)\\
    i+j=0\text{ mod }2}}\raisebox{-50pt} {\begin{tikzpicture}[scale=0.6, auto,
    redBall/.style={circle, ball color = red,
    radius=0.1em},blueBall/.style={circle, ball color = blue,
    radius=0.4em}, decoration={markings,
  mark=between positions 0 and 1 step 6pt
  with { \draw [fill] (0,0) circle [radius=0.8pt];}}]
        \fill[violet!50] (0,-0.8) ellipse (0.5 and 1.7);
        \fill[green!50] (0,1) ellipse (1.7 and 0.5);
        \fill[green!50] (0,-3) ellipse (1.7 and 0.5);
        \node[draw, redBall] (r0) at (0,0) {};
        \node[draw, redBall, below of=r0] (r1) {};
        \node[draw, blueBall] (b0) at (-1,1) {};
        \node[draw, blueBall, right of=b0] (r1) {};
        \node[draw, blueBall] (b2) at (-1,-3) {};
        \node[draw, blueBall, right of=b2] (b3) {};
        \node[] (s0) at (-0.1,-0.5) {$X_{v_r+(0,1)}$};
        \node[] (s1) at (0,-2.3) {$X_{v_r}$};
        \node[] (a0) at (-1.2,1.5) {$Z_{v_b+(0,2)}$};
        \node[] (a1) at (1.2,1.5){$Z_{v_b+(1,2)}$};
        \node[] (a2) at (-1,-3.5) {$Z_{v_b}$};
        \node[] (a3) at (1,-3.5) {$Z_{v_b+(1,0)}$};
    \end{tikzpicture}}\,,
    \label{eq:HXXZZ}
\end{align}
where the summed blue vertices are to the left of the pair of blue vertices in the green colored ellipse and the summed red vertices are to the bottom of the pair of red vertices in the violet colored ellipse. This Hamiltonian is symmetric under $\hat{\rm V}^{(2)}$, $\hat{\eta}_r$, and $\hat{\eta}_b$.
Let us denote the ground state of this Hamiltonian by $\ket{\Psi}$. It can be easily verified that $\ket{\Psi}$ is a product of Bell states of the form $\frac{1}{\sqrt{2}}(\ket{00}+\ket{11})$ on a pair of sites contained inside each colored ellipses, and is the unique ground sate. Hence, $\ket{\Psi}$ is symmetric with respect to $\hat{\rm V}^{(2)}\hat{\eta}_r\hat{\eta}_b$. Since we realized a Hamiltonian whose ground state is short-range entangled and respect the symmetry $\hat{\rm V}^{(2)}\hat{\eta}_r\hat{\eta}_b$, we conclude that this symmetry is not anomalous. 
Note that this construction works for other Bell states, by changing the signs of $XX$ and $ZZ$, and can be generalized to higher dimensions.
\subsubsection{$\hat{\rm V}_{\mathbb{Z}_2}^{(d)}\hat{\eta}$ is anomaly-free}\label{sec:nonanomalousZ2}
First, we construct a short-range entangled state in $2+1$D that is symmetric under $\hat{\rm V}_{\mathbb{Z}_2}^{(2)}\hat{\eta}$. Let us consider the product of GHZ state on the green ellipses as shown in Figure~\ref{fig:Z2symmetricstate}. This state can be easily verified to be symmetric under $\hat{\rm V}_{\mathbb{Z}_2}^{(2)}\hat{\eta}$. Since it is a product of GHZ states, it is a short-range entangled state. It is straightforward to write down a Hamiltonian whose unique ground states are the ones given in Figure~\ref{fig:Z2symmetricstate}. This type of construction can also be generalized to higher dimensions where the state $\hat{\rm V}_{\mathbb{Z}_2}^{(d)}\hat{\eta}$. 
\begin{figure}
    \centering
    \includegraphics[scale=1]{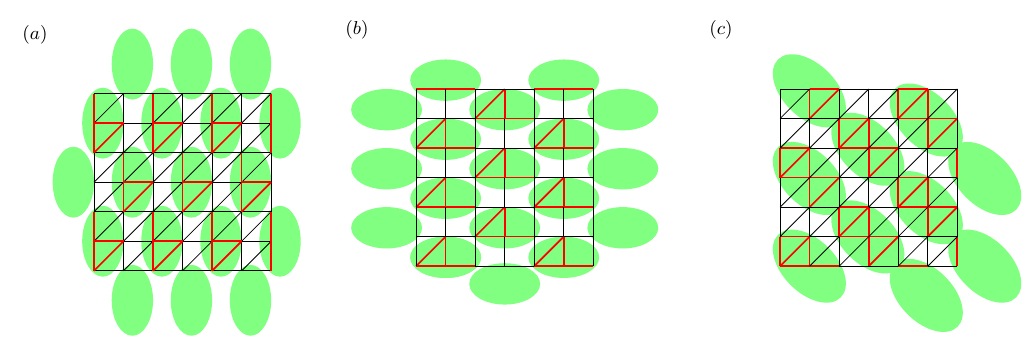}
    \caption{The Figure illustrate short-range entangled product state that is symmetric under $\hat{\rm V}_{\mathbb{Z}_2}^{(2)}\hat{\eta}$. The vertices inside the green ellipse are in the GHZ state $\frac{(\ket{0000}+\ket{1111})}{\sqrt{2}}$ (Figure $(a)$ or $(b)$) or $\frac{(\ket{000000}+\ket{111111})}{\sqrt{2}}$ (Figure $(c)$). The final state is the product of such GHZ in all the green ellipses. Inside each green ellipse, there are an even number of $CZ$ operators indicated by the red line segments. Each green ellipse is connected to any adjacent green ellipse by an even number of $CZ$ operators indicated by the black line segments. The final state is symmetric under $\hat{\rm V}_{\mathbb{Z}_2}^{(2)}$ and $\hat{\eta}$ and their product $\hat{\rm V}_{\mathbb{Z}_2}^{(2)}\hat{\eta}$. Figure $(a)$ is a symmetric state on a lattice with $L_x$ even and $L_y$ a multiple of four. Figure $(b)$ is a symmetric state on a lattice with $L_y$ even and $L_x$ a multiple of four.  Figure $(c)$ is a symmetric state on a lattice with $L_x$ a multiple of six and $L_y$ a multiple of three.  }
    \label{fig:Z2symmetricstate}
\end{figure}
\end{widetext}
\begin{figure}
    \centering
    \includegraphics[scale=1]{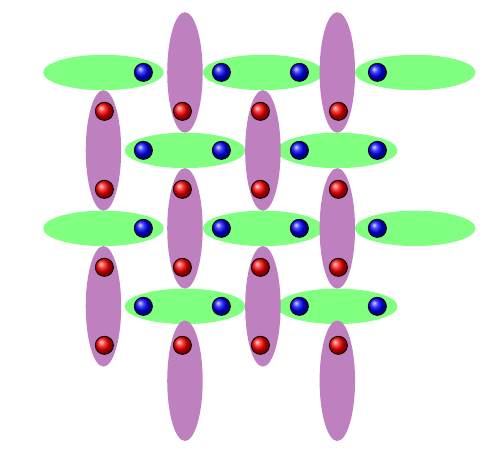}
    \caption{The Figure is an illustration of the short-range entangled state $\ket{\Psi}$ that is the ground state of the Hamiltonian \eqref{eq:HXXZZ}. A pair of sites in the colored ellipses are in the Bell state $\frac{1}{\sqrt{2}}\left(\ket{00}+\ket{11}\right)$ and there is no entanglement outside the colored eliptical regions. Hence, the state $\ket{\Psi}$ is the product of GHZ on all the colored ellipses. This state is symmetric with respect to $\hat{\mathrm{V}}$, $\hat{\eta}_b$, $\hat{\eta
    }_r$, and hence with respect to the product $\hat{\mathrm{V}}\hat{\eta
    }_r\hat{\eta
    }_b$. In particular, since we realized a symmetric short-range entangled state $\ket{\Psi}$, $\hat{\mathrm{V}}\hat{\eta
    }_r\hat{\eta
    }_b$ is not anomalous. }
    \label{fig:short-range-entangled-state}
\end{figure}

\section{Other noninvertible SSPTs in \texorpdfstring{$2+1$}{Lg}D}\label{sec:otherNSSPT}

In Sec.~\ref {sec:noninvertibleZ2xZ2} of the main text, we have discussed three specific noninvertible SSPT phases. 
In this section, we will discuss other noninvertible  $\mathbb{Z}_2\times\mathbb{Z}_2$ SSPTs in $2+1$D. 
\subsection{Cluster state}\label{sec:otherNSSPTinclusterstate}
Let us consider a few other examples of noninvertible SSPTs in the cluster phase. We recall that the cluster phase is mapped to $\mathbb{Z}_2\times\mathbb{Z}_2$ symmetry breaking phase under KT transformation $\mathbf{KT}^{(2)}$. There are several possible choices for the unbroken symmetry. 
In the main text, we discussed preserving $\hat{\rm V}^{(2)}$, $\hat{\rm V}^{(2)}\hat{\eta}_r$, or $\hat{\rm V}^{(2)}\hat{\eta}_b$ with $\hat{\eta}_{r}$ and $\hat{\eta}_b$ the global part of subsystem symmetries. 
Here, we discuss preserving products of $\hat{\rm V}^{(2)}$ with one or more linear subsystem-symmetry generators.

\subsubsection{\texorpdfstring{$\hat{\mathrm{V}}^{(2)}\hat{\eta}^x_{r,k}$}{Lg} preserving phase}\label{sec:Vetaxrk}
We take $L_x$ and $L_y$ to be even. Without any loss of generality, we assume $k\neq 1$. The Hamiltonian for the SSSB phase is given below, with a distinct structure in the $k^{th}$ row:
\begin{widetext}
    \begin{align}
    \begin{split}
     \hat{\mathrm{H}}_{\text{blue}}^{x;k}&=\sum_{v_b}\begin{array}{ccc}
         \hat{Z}_{v_b} &  &\hat{Z}_{v_b}\\
         & & \\
         \boxed{\hat{Z}_{v_b}} &  &\hat{Z}_{v_b}
    \end{array}-\sum_{v_r\neq (-,k),(-,k-1)}\begin{array}{ccc}
         \hat{Z}_{v_r} &  &\hat{Z}_{v_r}\\
         & & \\
         \boxed{\hat{Z}_{v_r}} &  &\hat{Z}_{v_r}
    \end{array}-\sum_{v_r= (-,k)}\begin{array}{ccc}
         \hat{Z}_{v_r} &  &\hat{Z}_{v_r}\\
         & & \\
         \boxed{\hat{Y}_{v_r}} &  &\hat{Y}_{v_r}
         \end{array}-\sum_{v_r= (-,k)}\begin{array}{ccc}
         \boxed{\hat{Y}_{v_r}} &  &\hat{Y}_{v_r}\\
         & & \\
         \hat{Z}_{v_r} &  &\hat{Z}_{v_r}
         \end{array}\\\\
         &\hspace{2cm}-\sum_{v_r= (-,k)}\begin{array}{ccccc}
         &\hat{Z}_{v_r} &  &\hat{Z}_{v_r}&\\
        \hat{Z}_{b} & & & &  \hat{Z}_{b}\\
         & \boxed{\hat{Y}_{v_r}} &  &\hat{Y}_{v_r} & \\
         \hat{Z}_{b} & & & & \hat{Z}_{b}
         \end{array}-\sum_{v_r= (-,k)}\begin{array}{ccccc}
          \hat{Z}_{b} & & & & \hat{Z}_{b}\\
         &\boxed{\hat{Y}_{v_r}} &  &\hat{Y}_{v_r}& \\
        \hat{Z}_{b} & & & &  \hat{Z}_{b}\\
         & \hat{Z}_{v_r} &  &\hat{Z}_{v_r} & \\
         \end{array}\, ,
         \end{split}
\end{align}
where $(-,k)$ denotes any vertex with coordinate $y=k$. The boxed vertices are the ones that are summed over.
The order parameters for this phase can be chosen to be of the form $\{\hat{Z}_{i+\frac{1}{2},\frac{3}{2}},\hat{Z}_{\frac{3}{2},j+\frac{1}{2}}\}$,  $\{\hat{Z}_{i,1},\hat{Z}_{1,j}\}\rvert_{j\neq k}$ and $\hat{Y}_{ (1,k)}\left(1-\begin{array}{ccc}
    \hat{Z}_{v_b} & &  \hat{Z}_{v_b} \\
                  & & \\
     \hat{Z}_{v_b}& &  \hat{Z}_{v_b}            
\end{array}\right)$ for $i=1,...,L_x$, $j=1,...,L_y$. There are in total $2(L_x+L_y-1)$ order parameters, as required. The order parameters satisfy the properties given in Lemma~\ref{lemma:consistentorderparameters}, and $\hat{\rm V}^{(2)}\hat{\eta}^x_{r,k}$ is unbroken. By applying $\mathbf{KT}^{(2)}$, we find the original SSPT Hamiltonian that gives rise to this SSSB Hamiltonian: 
\begin{align}
\begin{split}
\mathrm{H}_{\text{blue}}^{x;k}&=\sum_{v_r}\begin{array}{ccc}
         Z_{v_b} &  &Z_{v_b}\\
         & X_{v_r} & \\
         Z_{v_b} &  & Z_{v_b}
    \end{array}-\sum_{v_b\neq (-,k+\frac{1}{2}),(-,k-\frac{1}{2})}\begin{array}{ccc}
         Z_{v_r} &  &Z_{v_r}\\
         & X_{v_b} & \\
         Z_{v_r} &  & Z_{v_r}
    \end{array}-\sum_{v_b= (-,k+\frac{1}{2})}\begin{array}{ccc}
         Z_{v_r} &  &Z_{v_r}\\
         & X_{v_b} & \\
         Y_{v_r} &  & Y_{v_r}
    \end{array}-\sum_{v_b= (-,k-\frac{1}{2})}\begin{array}{ccc}
         Y_{v_r} &  &Y_{v_r}\\
         & X_{v_b} & \\
         Z_{v_r} &  & Z_{v_r}
    \end{array}\\\\
    &\hspace{2cm}+\sum_{v_b=(-,k+\frac{1}{2})}\begin{array}{ccccc}
         & & & &  \\
         & Z_{v_r} & &  Z_{v_r} & \\
        Z_{v_b} & & X_{v_b} & & Z_{v_b}\\
         & Z_{v_r} & &  Z_{v_r} & \\
         Z_{v_b} & & & & Z_{v_b} \\
    \end{array}+\sum_{v_b=(-,k-\frac{1}{2})}\begin{array}{ccccc}
    Z_{v_b} & & & & Z_{v_b} \\
         & Z_{v_r} & &  Z_{v_r} & \\
        Z_{v_b} & & X_{v_b} & & Z_{v_b}\\
         & Z_{v_r} & &  Z_{v_r} & 
    \end{array}  \, ,
    \end{split}
\end{align}
where the vertices that are summed over are the ones where  a Pauli $X$ is placed. This Hamiltonian has a unique ground state
\begin{align}
    \ket{\text{blue},(x;k)}&=\prod_{v_b=(-,k+ \frac{1}{2})}CZ_{v_b,v_b+(1,0)}CZ_{v_b,v_b+(1,- 1)}\prod_{v_b=(-,k- \frac{1}{2})}CZ_{v_b,v_b+(1,0)}CZ_{v_b,v_b+(1, 1)}\nonumber\\
    &\hspace{3cm}\times\prod_{v_r}\prod_{v_r\in\partial p_r}CZ_{v_r,p_r}\ket{+}^{\tilde{\Delta}_{v_b}}\ket{-}^{\Delta_{v_r}}\ket{-}^{\Delta_{v_b}\setminus \tilde{\Delta}_{v_b}
    }\, ,
\end{align}
where $\tilde{\Delta}_{v_b}\equiv \Delta_{v_b}\setminus \{v_b|v_b=(-,k+\frac{1}{2})\text{ or }v_b=(-,k-\frac{1}{2})\}$. 
$\ket{\text{blue},(x;k)}$ is related to the 2D-cluster state $\ket{\text{2D-cluster}}$ (ground state of \eqref{eq:2dcluster}) by a finite-depth circuit
\begin{align}
    \prod_{v_r}Z_{v_r}\prod_{\substack{ v_b=(-,k+\frac{1}{2})\\
    ,(-,k-\frac{1}{2})}}Z_{v_b}\prod_{v_b=(-,k+ \frac{1}{2})}CZ_{v_b,v_b+(1,0)}CZ_{v_b,v_b+(1,- 1)}\prod_{v_b=(-,k-\frac{1}{2})}CZ_{v_b,v_b+(1,0)}CZ_{v_b,v_b+(1,1)}\, .
\end{align}
We examine the interface modes between $\mathrm
{H}_{\text{2D-cluster}}$ and $\mathrm{H}_{\text{blue}}^{x;k}$ in Appendix~\ref{sec:interface2dclusterxkblue}. We put a line interface between the two Hamiltonians and find that there are four interface modes protected by $\mathbf{D}^{(2)}$ that distinguish between the two phases.
\subsubsection{\texorpdfstring{$\hat{\rm V}^{(2)}\hat{\eta}^x_{r,k}\hat{\eta}^y_{ r,m}$}{Lg} preserving phase}
We take $L_x$ and $L_y$ to be even. Without loss of generality, we assume $k,m\neq 1$. The Hamiltonian for the SSSB phase is
\begin{align}
\begin{split}
    \hat{\rm H}_{\text{red};k,m}^{x,y}&=\sum_{v_b}\begin{array}{ccc}
         \hat{Z}_{v_b} &  &\hat{Z}_{v_b}\\
         & & \\
         \boxed{\hat{Z}_{v_b}} &  &\hat{Z}_{v_b}
    \end{array}-\sum_{\substack{v_r\neq (m,-),(m-1,-),\\
    (-,k),(-,k-1)}}\begin{array}{ccc}
         \hat{Z}_{v_r} &  &\hat{Z}_{v_r}\\
         & & \\
         \boxed{\hat{Z}_{v_r}} &  &\hat{Z}_{v_r}
    \end{array}-\sum_{\substack{v_r= (m-1,-)\\
    v_r\neq (m-1,k),\\
    v_r\neq (m-1,k-1)}}\begin{array}{ccc}
         \hat{Z}_{v_r} &  &\hat{Y}_{v_r}\\
         & & \\
         \boxed{\hat{Z}_{v_r}} &  &\hat{Y}_{v_r}
    \end{array}-\sum_{\substack{v_r= (m,-)\\
    v_r\neq (m,k),\\
    v_r\neq (m,k-1)}}\begin{array}{ccc}
         \hat{Y}_{v_r} &  &\hat{Z}_{v_r}\\
         & & \\
         \boxed{\hat{Y}_{v_r}} &  &\hat{Z}_{v_r}
    \end{array}\\\\
    &-\sum_{\substack{v_r= (-,k)\\
    v_r\neq (m-1,k),\\
    v_r\neq (m,k)}}\begin{array}{ccc}
         \hat{Z}_{v_r} &  &\hat{Z}_{v_r}\\
         & & \\
         \boxed{\hat{Y}_{v_r}} &  &\hat{Y}_{v_r}
    \end{array}-\sum_{\substack{v_r= (-,k-1)\\
    v_r\neq (m-1,k-1),\\
    v_r\neq (m,k-1)}}\begin{array}{ccc}
         \hat{Y}_{v_r} &  &\hat{Y}_{v_r}\\
         & & \\
        \boxed{\hat{Z}_{v_r}} &  &\hat{Z}_{v_r}
    \end{array}-\begin{array}{ccc}
         \hat{Y}_{v_r} &  &\hat{Z}_{v_r}\\
         & & \\
        \boxed{\hat{Z}_{v_r}} &  &\hat{Y}_{v_r}
    \end{array}-\begin{array}{ccc}
         \boxed{\hat{Z}_{v_r}} &  &\hat{Y}_{v_r}\\
         & & \\
        \hat{Y}_{v_r} &  &\hat{Z}_{v_r}
    \end{array}-\begin{array}{ccc}
         \hat{Y}_{v_r} &  &\boxed{\hat{Z}_{v_r}}\\
         & & \\
        \hat{Z}_{v_r} &  &\hat{Y}_{v_r}
    \end{array}\\\\
   &\hspace{3cm}-\begin{array}{ccc}
         \hat{Z}_{v_r} &  &\hat{Y}_{v_r}\\
         & & \\
        \hat{Y}_{v_r} &  &\boxed{\hat{Z}_{v_r}}
    \end{array} +\hat{\rm V}^{(2)}\text{ conjugated terms }\, ,
    \end{split}
\end{align}
where boxes denote the summed over vertices in the terms with summation. Boxes in the terms without summation denote the coordinate $ (m,k)$.

The order parameters for this phase can be chosen to be of the form $\{\hat{Z}_{i+\frac{1}{2},\frac{3}{2}},\hat{Z}_{\frac{3}{2},j+\frac{1}{2}}\}$, $\{\hat{Z}_{1,j},\hat{Z}_{i,1}\}\rvert_{j\neq k, i\neq m}$, $\hat{Y}_{ (1,k)}\left(1-\begin{array}{ccc}
    \hat{Z}_{v_b} & &  \hat{Z}_{v_b} \\
                  & & \\
     \hat{Z}_{v_b}& &  \hat{Z}_{v_b}            
\end{array}\right)$, $\hat{Y}_{ (m,1)}\left(1-\begin{array}{ccc}
    \hat{Z}_{v_b} & &  \hat{Z}_{v_b} \\
                  & & \\
     \hat{Z}_{v_b}& &  \hat{Z}_{v_b}            
\end{array}\right)$ for $i=1,...,L_x$, $j=1,...,L_y$. This gives in total of $2(L_x+L_y-1)$ order parameters. The order parameters satisfy the properties in the Lemma~\ref{lemma:consistentorderparameters}, and $ \hat{\rm V}^{(2)}\hat{\eta}^x_{r,k}\hat{\eta}^y_{r,m}$ is unbroken. 

The original SSPT Hamiltonian that gives rise to this SSSB Hamiltonian is
\begin{align}
\begin{split}
    \mathrm{H}_{\text{red};k,m}^{x,y}&=\sum_{v_r}\begin{array}{ccc}
         Z_{v_b} &  &Z_{v_b}\\
         &X_{v_r} & \\
         Z_{v_b} &  &Z_{v_b}
    \end{array}-\sum_{\substack{v_b\neq (m+\frac{1}{2},-),(m-\frac{1}{2},-),\\
    (-,k+\frac{1}{2}),(-,k-\frac{1}{2})}}\begin{array}{ccc}
         Z_{v_r} &  &Z_{v_r}\\
         & X_{v_b} & \\
         Z_{v_r} &  & Z_{v_r}
    \end{array}-\sum_{\substack{v_b= (m-\frac{1}{2},-)\\
    v_b\neq (m-\frac{1}{2},k+\frac{1}{2}),\\
   v_b\neq (m-\frac{1}{2},k-\frac{1}{2})}}\begin{array}{ccc}
         Z_{v_r} &  &Y_{v_r}\\
         & X_{v_b} & \\
         Z_{v_r} &  &Y_{v_r}
    \end{array}\\\\
    &-\sum_{\substack{v_b= (m+\frac{1}{2},-)\\
    v_b\neq(m+\frac{1}{2},k+\frac{1}{2}),\\
    v_b\neq (m+\frac{1}{2},k-\frac{1}{2})}}\begin{array}{ccc}
         Y_{v_r} &  & Z_{v_r}\\
         & X_{v_b} & \\
         Y_{v_r} &  & Z_{v_r}
    \end{array}-\sum_{\substack{v_b= (-,k+\frac{1}{2})\\
    v_b\neq (m-\frac{1}{2},k+\frac{1}{2}),\\
    v_b\neq (m+\frac{1}{2},k+\frac{1}{2})}}\begin{array}{ccc}
         Z_{v_r} &  & Z_{v_r}\\
         & X_{v_b} & \\
         Y_{v_r} &  & Y_{v_r}
    \end{array}-\sum_{\substack{v_b= (-,k-\frac{1}{2})\\
    v_b\neq (m-\frac{1}{2},k-\frac{1}{2}),\\
    v_b\neq (m+\frac{1}{2},k-\frac{1}{2})}}\begin{array}{ccc}
         Y_{v_r} &  & Y_{v_r}\\
         & X_{v_b} & \\
         Z_{v_r} &  & Z_{v_r}
    \end{array}\\\\
    &-\begin{array}{ccc}
         Y_{v_r} &  & Z_{v_r}\\
         & X_{v_b} & \\
        \boxed{Z_{v_r}} &  & Y_{v_r}
    \end{array}-\begin{array}{ccc}
         \boxed{Z_{v_r}} &  &Y_{v_r}\\
         & X_{v_b}& \\
        Y_{v_r} &  &Z_{v_r}
    \end{array}-\begin{array}{ccc}
         Y_{v_r} &  &\boxed{Z_{v_r}}\\
         &X_{v_b} & \\
        Z_{v_r} &  &Y_{v_r}
    \end{array}-\begin{array}{ccc}
         Z_{v_r} &  &Y_{v_r}\\
         & X_{v_b} & \\
        Y_{v_r} &  &\boxed{Z_{v_r}}
    \end{array}
    +\mathbf{D}^{(2)}\text{ conjugated terms }\, .
    \end{split}
\end{align}
Boxes in the terms without summation denote the coordinate $(m,k)$. We could study an interface between $\mathrm{H}_{\text{2D-cluster}}$ and $\mathrm{H}_{\text{red};k,m}^{x,y}$. If we consider a rectangular interface between the two Hamiltonian with $\mathrm{H}_{\text{red};k,m}^{x,y}$ inside the rectangle and $\mathrm{H}_{\text{2D-cluster}}$ outside it, we expect to find eight interface modes that are protected by $\mathbf{D}^{(2)}$. These interface modes distinguish between the two phases.
\subsubsection{\texorpdfstring{$\hat{\rm V}^{(2)}\hat{\eta}^x_{r,k}\hat{\eta}^y_{ b,m}$}{Lg} preserving phase}
We take $L_x$ and $L_y$ to be even. Without loss of generality, we assume $k,m\neq 1$. The Hamiltonian for the SSSB phase is
\begin{align}
\begin{split}
    \hat{\rm H}_{\text{red},\text{blue};\rm,k,m}^{x,y}&=\sum_{ \substack{v_r\neq (-,k),\\(-,k-1)}}\begin{array}{ccc}
        \hat{Z}_{v_r} & & \hat{Z}_{v_r}  \\
        & & \\
        \boxed{\hat{Z}_{v_r}} & & \hat{Z}_{v_r} 
    \end{array}+\sum_{ \substack{v_b\neq (m+\frac{1}{2},-),\\(m-\frac{1}{2},-)}}\begin{array}{ccc}
        \hat{Z}_{v_b} & & \hat{Z}_{v_b}  \\
        & & \\
        \boxed{\hat{Z}_{v_b}} & & \hat{Z}_{v_b} 
    \end{array}+\sum_{ \substack{v_r= (-,k),\\v_r\neq (m,k),\\
    (m-1,k),(m+1,k)}}\begin{array}{ccc}
        \hat{Z}_{v_r} & & \hat{Z}_{v_r}  \\
        & &\\
        \boxed{\hat{Y}_{v_r}} & & \hat{Y}_{v_r} 
    \end{array}\\
    \vspace{1cm}\\
    &+\sum_{ \substack{v_r= (-,k-1),\\v_r\neq (m,k-1),\\
    (m-1,k-1),(m+1,k-1)}}\begin{array}{ccc}
        \hat{Y}_{v_r} & & \hat{Y}_{v_r}  \\
        & & \\
        \boxed{\hat{Z}_{v_r}} & & \hat{Z}_{v_r} 
    \end{array}+\sum_{ \substack{v_b= (m-\frac{1}{2},-),\\v_b\neq (m-\frac{1}{2},k+\frac{1}{2}),\\
    (m-\frac{1}{2},k-\frac{1}{2}),(m-\frac{1}{2},k-\frac{3}{2})}}\begin{array}{ccc}
        \hat{Z}_{v_b} & & \hat{Y}_{v_b}  \\
        & & \\
        \boxed{\hat{Z}_{v_b}} & & \hat{Y}_{v_b} 
    \end{array}\\
    \vspace{1cm}\\
    &\hspace{2cm}+\sum_{ \substack{v_b= (m+\frac{1}{2},-),\\v_b\neq (m+\frac{1}{2},k+\frac{1}{2}),\\
    (m+\frac{1}{2},k-\frac{1}{2}),(m+\frac{1}{2},k-\frac{3}{2})}}\begin{array}{ccc}
        \hat{Y}_{v_b} & & \hat{Z}_{v_b}  \\
        & & \\
        \boxed{\hat{Y}_{v_b}} & & \hat{Z}_{v_b} 
    \end{array}
    \\
    \vspace{1cm}\\&+\begin{array}{ccc}
         Z_{v_r} &  & Z_{v_r} \\
         &  & \\
        Y_{v_r} & & \boxed{Y_{v_r}}
    \end{array}+\begin{array}{ccc}
         Y_{v_r} &  & \boxed{Y_{v_r}} \\
         &  & \\
        Z_{v_r} & & Z_{v_r}
    \end{array}+\begin{array}{ccc}
         \boxed{Y_{v_b}} &  & Z_{v_b} \\
         &  & \\
        Y_{v_b} & & Z_{v_b}
    \end{array}+\begin{array}{ccc}
         Z_{v_b} &  & \boxed{Y_{v_b}} \\
         &  & \\
        Z_{v_b} & & Y_{v_b}
    \end{array}\nonumber\\
    \vspace{1cm}\\
    &+\begin{array}{cccc}
       Z_{v_r} & \quad & \quad & Z_{v_r} \\
         & \qquad & \quad  & \\
         Y_{v_r} &  &  & \circled{$Y_{v_r}$} 
    \end{array}+ \begin{array}{cccc}
       Y_{v_r} & \qquad  & \quad & \circled{$Y_{v_r}$}\\
         & & &\\
         Z_{v_r} &  &  & Z_{v_r}  
    \end{array}+\begin{array}{cc}
        \circled{$Y_{v_b}$} & Z_{v_b} \\
             &  \\ 
             &  \\
        Y_{v_b}& Z_{v_b}
    \end{array}+\begin{array}{cc}
        Z_{v_b}& \circled{$Y_{v_b}$} \\
               &  \\ 
               &  \\
        Z_{v_b}& Y_{v_b}
    \end{array}+\nonumber\\
    \vspace{1cm}\\
    &+\begin{array}{cccc}
       Z_{v_r} & Z_{v_r} & Z_{v_r} & Z_{v_r} \\
         & & &\\
         Y_{v_r} & Y_{v_r} & \boxed{Y_{v_r}} & Y_{v_r} 
    \end{array}+ \begin{array}{cccc}
       Y_{v_r} & Y_{v_r} & \boxed{Y_{v_r}} & Y_{v_r}\\
         & & &\\
         Z_{v_r} & Z_{v_r} & Z_{v_r} & Z_{v_r}  
    \end{array}+\begin{array}{cc}
        Y_{v_b}& Z_{v_b} \\
        \boxed{Y_{v_b}}& Z_{v_b} \\ 
        Y_{v_b}& Z_{v_b} \\
        Y_{v_b}& Z_{v_b}
    \end{array}+\begin{array}{cc}
        Z_{v_b}& Y_{v_b} \\
        Z_{v_b}& \boxed{Y_{v_b}} \\ 
        Z_{v_b}& Y_{v_b} \\
        Z_{v_b}& Y_{v_b}
    \end{array}+\rm \hat{V}^{(2)}\text{ conjugated terms }\, .
    \end{split}
\end{align}
In the above equation, the box in the terms with summation indicates the vertices that are summed over. Boxes in the terms without summation denote the coordinate $(m+1,k)$ or $(m+\frac{1}{2},k+\frac{1}{2})$ depending on red or blue vertex respectively, and the circles in the terms without summation denote the coordinate $(m+2,k)$ or $(m+\frac{1}{2},k+\frac{3}{2})$ depending on red or blue vertex respectively.

The order parameters for this phase can be chosen to be of the form $\{\hat{Z}_{i+\frac{1}{2},\frac{3}{2}},\hat{Z}_{\frac{3}{2},j+\frac{1}{2}}\}\rvert_{i\neq m}$, $\{\hat{Z}_{1,j},\hat{Z}_{i,1}\}\rvert_{j\neq k}$, $\hat{Y}_{(1,k)}\left(1-\begin{array}{ccc}
    \hat{Z}_{v_b} & &  \hat{Z}_{v_b} \\
                  & & \\
     \hat{Z}_{v_b}& &  \hat{Z}_{v_b}            
\end{array}\right)$, $\hat{Y}_{ (m+\frac{1}{2},\frac{3}{2})}\left(1-\begin{array}{ccc}
    \hat{Z}_{v_r} & &  \hat{Z}_{v_r} \\
                  & & \\
     \hat{Z}_{v_r}& &  \hat{Z}_{v_r}            
\end{array}\right)$ for $i=1,...,L_x$, $j=1,...,L_y$. The order parameters satisfy the properties in Lemma~\ref{lemma:consistentorderparameters}, and hence $ \hat{\rm V}^{(2)}\hat{\eta}^x_{r,k}\hat{\eta}^y_{b,m}$ is unbroken. The original SSPT Hamiltonian can be found by applying the $\rm \mathbf{KT}^{(2)}$ transformation.
\subsection{\texorpdfstring{$\mathbb{Z}_2$}{Lg} SSPT stacked onto cluster state}\label{sec:Z2SSPTstackedtocluster}
Let us consider another $\mathbb{Z}_2\times\mathbb{Z}_2$ SSPT in $2+1$ dimensions that is symmetric under $\rm D^{(2)}$
    \begin{align}
    \mathrm{H}_{\text{2D-}\widetilde{\text{clstr}}}=-\sum_{v_r}\begin{array}{ccc}
       Z_{v_b} & & Z_{v_b}  \\
         & X_{v_r} &  \\
       Z_{v_b} & & Z_{v_b}
    \end{array}+\sum_{v_b}\begin{array}{ccc}
      Z_{v_r} &  & Y_{v_r} \\
       & X_{v_b} &  \\
      Y_{v_r} &  & Z_{v_r}
    \end{array}-\sum_{v_b}\begin{array}{ccccc}
        & & Z_{v_b} & & Z_{v_b} \\
        & Z_{v_r} & & Z_{v_r} & \\
        Z_{v_b} & & X_{v_b} & & Z_{v_b}\\
        & Z_{v_r} & & Z_{v_r} & \\
        Z_{v_b} & & Z_{v_b} & &
    \end{array}\, .
    \label{eq:2dtildeclstr}
\end{align}
\end{widetext}
We note that the last term depends on the first two terms. However, we can still minimize all terms simultaneously. Now we can apply $\mathbf{KT}^{(2)}$ onto this Hamiltonian
\begin{align}
    \hat{\rm H}_{\text{PI-Wen}}&=-\sum_{ v_b}\begin{array}{ccc}
        \hat{Z}_{v_b} & & \hat{Z}_{v_b}  \\
        & &\\
        \hat{Z}_{v_b} & &\hat{Z}_{v_b}
    \end{array}+\sum_{v_r}\begin{array}{ccc}
        \hat{Z}_{v_r} & & \hat{Y}_{v_r}  \\
        & & \\
        \hat{Y}_{ v_r} & &\hat{Z}_{v_r}
    \end{array}\nonumber\\\nonumber\\
    &\qquad\quad+\sum_{v_b}\begin{array}{ccccc}
        & & \hat{Z}_{v_b} & & \hat{Z}_{v_b} \\
        & \hat{Z}_{v_r} & & \hat{Y}_{v_r} & \\
        \hat{Z}_{v_b} & &  & & \hat{Z}_{v_b}\\
        & \hat{Y}_{v_r} & & \hat{Z}_{v_r} & \\
        \hat{Z}_{v_b} & & \hat{Z}_{v_b} & &
    \end{array}\, .
    \label{eq:PI-Wen}
\end{align}
On the blue colored sublattice, the above Hamiltonian is in a $\mathbb{Z}_2$ SSSB phase. The ground state degeneracy in this phase is $2^{L_x+L_y-1}$. On the red sublattice, the Hamiltonian is that of the Wen-plaquette model and hence is in a 1-form symmetry broken phase. The ground state degeneracy depends on whether $L_x$ and $L_y$ are even or odd. For the case $L_x$ and $L_y$ even numbers, the ground state degeneracy is 4. For all other cases, the ground state degeneracy is 2. See~\cite{Wen:2003yv} for more details on the Wen-plaquette model and its ground state degeneracy. 

Again, the dual symmetries after $\mathbf{KT}^{(2)}$ transformation are same as that in the SSSB Hamiltonian~\eqref{eq:2DSSSBtwocopiesPI} case:
\begin{subequations}
   \begin{align}
    &\hat{\eta}^{x}_{r,j}=\prod_{i=1}^{L_x}\hat{X}_{i,j}\, ,\quad \hat{\eta}^{y}_{r,i}=\prod_{j=1}^{L_y}\hat{X}_{i,j}\,,\\
    &\hat{\eta}^{x}_{b,j}=\prod_{i=1}^{L_x}\hat{X}_{i+\frac{1}{2},j+\frac{1}{2}}\,,\quad \hat{\eta}^{y}_{b,i}=\prod_{j=1}^{L_y}\hat{X}_{i+\frac{1}{2},j+\frac{1}{2}}\,, \\
    &\hat{\rm V}^{(2)}=\prod_{v_r}\prod_{v^r\in\partial p^r}CZ_{v^r,p^r}\, .
\end{align} 
\end{subequations}
We repeat the same analysis as before and look at various possible symmetry preserved phases. For simplicity, we restrict our discussion to even by even lattice, i.e., $L_x$ and $L_y$ are even. 

We note that since we start with $\mathbb{Z}_2\times\mathbb{Z}_2$ SSPT~\eqref{eq:2dtildeclstr}, the dual SSB is fixed as far as the $\mathbb{Z}_2\times\mathbb{Z}_2$ symmetries are concerned. All the subsystem symmetries on the blue sublattice are broken, while on the red sublattice, one form symmetries is broken. On the red sublattice, symmetries of the form $\hat{\eta}^x_{r,j}\hat{\eta}^x_{r,j+1}$, $\hat{\eta}^y_{r,i}\hat{\eta}^y_{r,i+1}$ $\forall\, i\in\{1,...,L_x\}\,,\forall\, j\in\{1,...,L_y\}$ and their arbitrary products are preserved while single $\hat{\eta}^x_{r,j}$ and $\hat{\eta}^y_{r,i}$ are broken.
\begin{widetext}
\subsubsection{\texorpdfstring{$\hat{\rm V}$, $\hat{\eta}_{r,j}^x\hat{\eta}_{r,j+1}^x$ $\forall j$}{Lg} and \texorpdfstring{$\hat{\eta}_{r,i}^y\hat{\eta}_{r,i+1}^y$ $\forall i$}{Lg} are preserved}
The Hamiltonian in this phase is described in \eqref{eq:PI-Wen}. The order parameters for this phase are $\{\hat{Z}_{\frac{3}{2},j+\frac{1}{2}},\hat{Z}_{i+\frac{1}{2},\frac{3}{2}}\}$ for $i=1,...,L_x$ and $j=1,...,L_y$ on the blue sublattice and non-local order parameters of the form 
\begin{align}
   \left( \begin{array}{c}
         \vdots  \\
          \hat{Z}_{v_r}      \\
                 \\
          \hat{Y}_{v_r}      \\
                 \\
          \hat{Z}_{v_r}      \\
                 \\
          \hat{Y}_{v_r}      \\
          \vdots
    \end{array}
    +\begin{array}{ccc}
    &\vdots & \\
      & \hat{Z}_{v_r} &  \\
        \hat{Z}_{v_b} & & \hat{Z}_{v_b}\\
        & \hat{Y}_{v_r} & \\
        \hat{Z}_{v_b} & & \hat{Z}_{v_b}\\
        & \hat{Z}_{v_r}  &  \\
        \hat{Z}_{v_b} & & \hat{Z}_{v_b}\\
        & \hat{Y}_{v_r} & \\
        & \vdots & 
    \end{array}\right)\, \text{ and } \left(
    \begin{array}{cccccccccc}
         \hdots & \hat{Z}_{v_r} & & \hat{Y}_{v_r} & & \hat{Z}_{v_r} & & \hat{Y}_{v_r} & & \hdots    
    \end{array}+\begin{array}{cccccccccc}
        & & \hat{Z}_{v_b} & & \hat{Z}_{v_b} & & \hat{Z}_{v_b} & & \hat{Z}_{v_b} & \\
      \hdots & \hat{Z}_{v_r} & & \hat{Y}_{v_r} & & \hat{Z}_{v_r} & & \hat{Y}_{v_r} & & \hdots  \\
         & & \hat{Z}_{v_b} & & \hat{Z}_{v_b} & & \hat{Z}_{v_b} & & \hat{Z}_{v_b} & \\
    \end{array}\right)\, .
\end{align}
It is straightforward to verify that the order parameters satisfy all the conditions stated in Lemma~\ref{lemma:consistentorderparameters}, and that the symmetries $\hat{\rm V}$, $\hat{\eta}_{r,j}^x\hat{\eta}_{r,j+1}^x$ $\forall j$, and $\hat{\eta}_{r,i}^y\hat{\eta}_{r,i+1}^y$ $\forall i$ are preserved. 
The corresponding SSPT is the same as \eqref{eq:2dtildeclstr}. 
\subsubsection{\texorpdfstring{$\hat{\rm V}\hat{\eta}^x_{r,j}$, $\hat{\eta}_{r,j}^x\hat{\eta}_{r,j+1}^x$ $\forall j$}{Lg} and \texorpdfstring{$\hat{\eta}_{r,i}^y\hat{\eta}_{r,i+1}^y$ $\forall i$}{Lg} are preserved }
We take $L_y=4n+2$ and $L_x=4m$. The Hamiltonian in this phase is 
\begin{align}
    \hat{\rm H}_{\text{PI-Wen}}&=\sum_{ v_b}\begin{array}{ccc}
        \hat{Z}_{v_b} & & \hat{Z}_{v_b}  \\
        & &\\
        \hat{Z}_{v_b} & &\hat{Z}_{v_b}
    \end{array}+\sum_{v_r}\begin{array}{ccc}
        \hat{Z}_{v_r} & & \hat{Y}_{v_r}  \\
        & & \\
        \hat{Y}_{ v_r} & &\hat{Z}_{v_r}
    \end{array}+\sum_{v_b}\begin{array}{ccccc}
        & & \hat{Z}_{v_b} & & \hat{Z}_{v_b} \\
        & \hat{Z}_{v_r} & & \hat{Y}_{v_r} & \\
        \hat{Z}_{v_b} & &  & & \hat{Z}_{v_b}\\
        & \hat{Y}_{v_r} & & \hat{Z}_{v_r} & \\
        \hat{Z}_{v_b} & & \hat{Z}_{v_b} & &
    \end{array}\, .
    \label{eq:PI-Wenx}
\end{align}
The order parameters for this phase are $\{\hat{Z}_{\frac{3}{2},j+\frac{1}{2}},\hat{Z}_{i+\frac{1}{2},\frac{3}{2}}\}$ for $i=1,...,L_x$ and $j=1,...,L_y$ on the blue sublattice and non-local order parameters of the form 
\begin{align}
   \left( \begin{array}{c}
         \vdots  \\
          \hat{Z}_{v_r}      \\
                 \\
          \hat{Y}_{v_r}      \\
                 \\
          \hat{Z}_{v_r}      \\
                 \\
          \hat{Y}_{v_r}      \\
          \vdots
    \end{array}
    -\begin{array}{ccc}
    &\vdots & \\
      & \hat{Z}_{v_r} &  \\
        \hat{Z}_{v_b} & & \hat{Z}_{v_b}\\
        & \hat{Y}_{v_r} & \\
        \hat{Z}_{v_b} & & \hat{Z}_{v_b}\\
        & \hat{Z}_{v_r}  &  \\
        \hat{Z}_{v_b} & & \hat{Z}_{v_b}\\
        & \hat{Y}_{v_r} & \\
        & \vdots & 
    \end{array}\right)\, \text{ and } \left(
    \begin{array}{cccccccccc}
         \hdots & \hat{Z}_{v_r} & & \hat{Y}_{v_r} & & \hat{Z}_{v_r} & & \hat{Y}_{v_r} & & \hdots    
    \end{array}+\begin{array}{cccccccccc}
        & & \hat{Z}_{v_b} & & \hat{Z}_{v_b} & & \hat{Z}_{v_b} & & \hat{Z}_{v_b} & \\
      \hdots & \hat{Z}_{v_r} & & \hat{Y}_{v_r} & & \hat{Z}_{v_r} & & \hat{Y}_{v_r} & & \hdots  \\
         & & \hat{Z}_{v_b} & & \hat{Z}_{v_b} & & \hat{Z}_{v_b} & & \hat{Z}_{v_b} & \\
    \end{array}\right)\, .
\end{align}
For the order parameters to be nonzero on the ground space, we need to take $L_y=4n+2$ and $L_x=4m$ for some $n$ and $m$. 
The corresponding SSPT is
\begin{align}
    \mathrm{H}_{\text{2D-}\widetilde{\text{clstr}}}^{x}=\sum_{v_r}\begin{array}{ccc}
       Z_{v_b} & & Z_{v_b}  \\
         & X_{v_r} &  \\
       Z_{v_b} & & Z_{v_b}
    \end{array}+\sum_{v_b}\begin{array}{ccc}
      Z_{v_r} &  & Y_{v_r} \\
       & X_{v_b} &  \\
      Y_{v_r} &  & Z_{v_r}
    \end{array}-\sum_{v_b}\begin{array}{ccccc}
        & & Z_{v_b} & & Z_{v_b} \\
        & Z_{v_r} & & Z_{v_r} & \\
        Z_{v_b} & & X_{v_b} & & Z_{v_b}\\
        & Z_{v_r} & & Z_{v_r} & \\
        Z_{v_b} & & Z_{v_b} & &
    \end{array}\, .
    \label{eq:2dtildeclstrx}
\end{align}
We note that \eqref{eq:2dtildeclstrx} is in a different phase from \eqref{eq:2dtildeclstr}. We analyze the interface modes between these noninvertible SSPTs in Appendix~\ref{sec:interfaceanalysisotherspt}.  
One could also do the above analysis for $L_y=4n$ and $L_x=4m$; writing down a Hamiltonian similar to \eqref{eq:PI-Wenx} by flipping the sign of the plaquette term on the blue sublattice on two adjacent rows. 
\subsubsection{\texorpdfstring{$\hat{\rm V}\hat{\eta}^y_{r,i}$, $\hat{\eta}_{r,i}^y\hat{\eta}_{r,i+1}^y$ $\forall i$}{Lg} and \texorpdfstring{$\hat{\eta}_{r,j}^x\hat{\eta}_{r,j+1}^x$ $\forall j$}{Lg} preserved} 
We take $L_y=4n$ and $L_x=4m+2$. The Hamiltonian in this phase is the same as \eqref{eq:PI-Wenx}.
The order parameters for this phase are $\{\hat{Z}_{\frac{3}{2},j+\frac{1}{2}},\hat{Z}_{i+\frac{1}{2},\frac{3}{2}}\}$ for $i=1,...,L_x$ and $j=1,...,L_y$ on the blue sublattice and non-local order parameters of the form 
\begin{align}
   \left( \begin{array}{c}
         \vdots  \\
          \hat{Z}_{v_r}      \\
                 \\
          \hat{Y}_{v_r}      \\
                 \\
          \hat{Z}_{v_r}      \\
                 \\
          \hat{Y}_{v_r}      \\
          \vdots
    \end{array}
    +\begin{array}{ccc}
    &\vdots & \\
      & \hat{Z}_{v_r} &  \\
        \hat{Z}_{v_b} & & \hat{Z}_{v_b}\\
        & \hat{Y}_{v_r} & \\
        \hat{Z}_{v_b} & & \hat{Z}_{v_b}\\
        & \hat{Z}_{v_r}  &  \\
        \hat{Z}_{v_b} & & \hat{Z}_{v_b}\\
        & \hat{Y}_{v_r} & \\
        & \vdots & 
    \end{array}\right)\, \text{ and } \left(
    \begin{array}{cccccccccc}
         \hdots & \hat{Z}_{v_r} & & \hat{Y}_{v_r} & & \hat{Z}_{v_r} & & \hat{Y}_{v_r} & & \hdots    
    \end{array}-\begin{array}{cccccccccc}
        & & \hat{Z}_{v_b} & & \hat{Z}_{v_b} & & \hat{Z}_{v_b} & & \hat{Z}_{v_b} & \\
      \hdots & \hat{Z}_{v_r} & & \hat{Y}_{v_r} & & \hat{Z}_{v_r} & & \hat{Y}_{v_r} & & \hdots  \\
         & & \hat{Z}_{v_b} & & \hat{Z}_{v_b} & & \hat{Z}_{v_b} & & \hat{Z}_{v_b} & \\
    \end{array}\right)\, .
\end{align}
For the order parameters to be nonzero on the ground space, we need to take $L_y=4n$ and $L_x=4m+2$ for some $n$ and $m$. 
The corresponding SSPT Hamiltonian is the same as \eqref{eq:2dtildeclstrx}. This does not contradict the fact that two different SSB phases should come from two different SPTs, since the SPT Hamiltonians are the same for different system sizes. One could also do the above analysis for $L_y=4n$ and $L_x=4m$; writing down a Hamiltonian similar to \eqref{eq:PI-Wenx} by flipping the sign of the plaquette term on the blue sublattice on two adjacent columns. 
\subsubsection{\texorpdfstring{$\hat{\rm V}\hat{\eta}^x_{r,i}\hat{\eta}^y_{r,j}$, $\hat{\eta}_{r,j}^x\hat{\eta}_{r,j+1}^x$ $\forall j$}{Lg} and \texorpdfstring{$\hat{\eta}_{r,i}^y\hat{\eta}_{r,i+1}^y$ $\forall i$}{Lg} preserved }
We take $L_y=4n+2$ and $L_x=4m+2$ for some $n$ and $m$. The Hamiltonian in this phase is the same as \eqref{eq:PI-Wenx}. 
The order parameters for this phase are $\{\hat{Z}_{\frac{3}{2},j+\frac{1}{2}},\hat{Z}_{i+\frac{1}{2},\frac{3}{2}}\}$ for $i=1,...,L_x$ and $j=1,...,L_y$ on the blue sublattice and non-local order parameters of the form 
\begin{align}
   \left( \begin{array}{c}
         \vdots  \\
          \hat{Z}_r      \\
                 \\
          \hat{Y}_r      \\
                 \\
          \hat{Z}_r      \\
                 \\
          \hat{Y}_r      \\
          \vdots
    \end{array}
    -\begin{array}{ccc}
    &\vdots & \\
      & \hat{Z}_r &  \\
        \hat{Z}_b & & \hat{Z}_b\\
        & \hat{Y}_r & \\
        \hat{Z}_b & & \hat{Z}_b\\
        & \hat{Z}_r  &  \\
        \hat{Z}_b & & \hat{Z}_b\\
        & \hat{Y}_r & \\
        & \vdots & 
    \end{array}\right)\, \text{ and } \left(
    \begin{array}{cccccccccc}
         \hdots & \hat{Z}_r & & \hat{Y}_r & & \hat{Z}_r & & \hat{Y}_r & & \hdots    
    \end{array}-\begin{array}{cccccccccc}
        & & \hat{Z}_b & & \hat{Z}_b & & \hat{Z}_b & & \hat{Z}_b & \\
      \hdots & \hat{Z}_r & & \hat{Y}_r & & \hat{Z}_r & & \hat{Y}_r & & \hdots  \\
         & & \hat{Z}_b & & \hat{Z}_b & & \hat{Z}_b & & \hat{Z}_b & \\
    \end{array}\right)\, .
\end{align}
For the order parameters to be nonzero on the ground space, we need to take $L_y=4n+2$ and $L_x=4m+2$ for some $n$ and $m$. 
The corresponding SSPT Hamiltonian is the same as \eqref{eq:2dtildeclstrx}. This does not contradict the fact that two different SSB phases should come from two different SPTs because the same SPT Hamiltonian~\eqref{eq:PI-Wenx} is for different system sizes in each case.
One could do a similar analysis for all other cases when $L_y$ and $L_x$ are even. 
\end{widetext}
\section{Interface between two distinct noninvertible SSPTs in 2D: corner modes}\label{sec:Interfaceanalysis}
\subsection{Interface between \texorpdfstring{$\rm H_{\text{2D-cluster}}$}{Lg} and \texorpdfstring{$\rm H_{\text{blue}}$}{Lg}}\label{sec:Interface2DclusterHblue}
\subsubsection{Line interface}
Let us consider an interface of two Hamiltonians $\rm H_{\text{2D-cluster}}$ and $\rm H_{\text{blue}}$ on a torus with interface along the line $x=l+\frac{1}{2}$ and $x=L_x-1+\frac{1}{2}$ for some $l\neq L_x-1\in \mathbb{Z}_{L_x}$ such that $L_x-l$ is odd and $L_y$ is even. These two lines divide the torus into two regions. Let us call the region that contains $(L_x,0)$, including the boundary (the two interface lines), to be $A$ and the region that contains $(L_x-1,0)$, including the boundary, to be $B$ as given in the Figure~\ref{fig:lineinterface}. Explicitly,
\red{
\begin{subequations}
    \begin{align}
    A&=\{(x,y)\in (\mathbb{Z}/2,\mathbb{Z}/2)|\, x\leq l\text{ or }x=L_x,L_x+\frac{1}{2}\} \, , \\
    B&=\{(x,y)\in (\mathbb{Z}/2,\mathbb{Z}/2)|\, l+1\leq x\leq L_x-1\} \, .
\end{align}
\label{eq:regionABline}
\end{subequations}}
The interface Hamiltonian is obtained by restricting the terms in the Hamiltonian $\rm H_{\text{2D-cluster}}$ and $\rm H_{\text{blue}}$ onto the respective regions $A$ and $B$. 
\begin{figure*}
    \centering
    \includegraphics[scale=1]{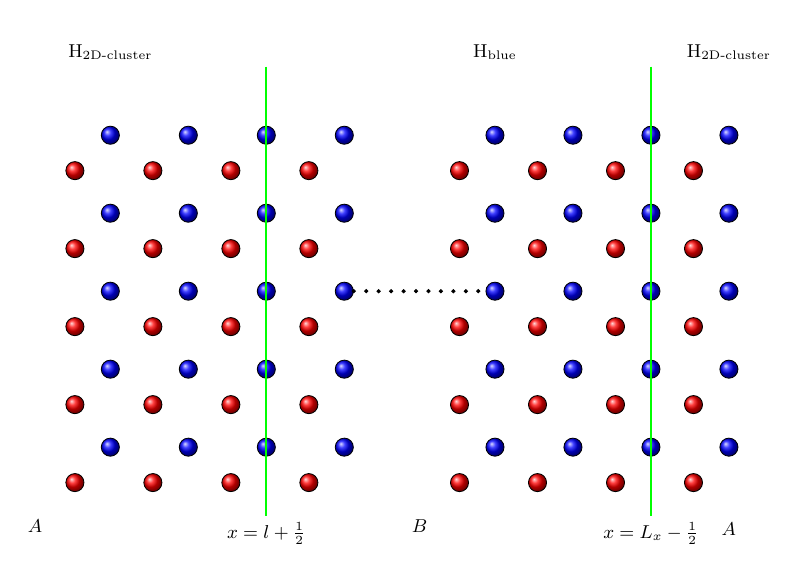}
    \caption{Line interface between $\rm H_{\text{blue}}$ (region $B$) and $\rm H_{\text{2D-cluster}}$ (region $A$). The vertical interface line is colored green and passes through blue vertices at $x=l+\frac{1}{2}$ and $x=L_x-\frac{1}{2}$. The lattice is put on a torus, and the green lines are non-contractable cycles along the $y$ direction. }
    \label{fig:lineinterface}
\end{figure*}
\begin{widetext}
\begin{align}
    \mathrm{H}_{\text{2D-cluster}|\text{blue}}&=-\sum_{v_r\in A}\begin{array}{ccc}
        Z_{v_b} & & Z_{v_b}  \\
        & X_{v_r} & \\
        Z_{v_b} & & Z_{v_b}  
    \end{array}-\sum_{v_b\in A}\begin{array}{ccc}
        Z_{v_r} & & Z_{v_r}  \\
        & X_{v_b} & \\
        Z_{v_r} & & Z_{v_r}  
    \end{array}+\sum_{v_r\in B}\begin{array}{ccc}
        Z_{v_b} & & Z_{v_b}  \\
        & X_{v_r} & \\
        Z_{v_b} & & Z_{v_b}  
    \end{array}\nonumber\\
    &\hspace{3cm}-\sum_{v_b\in B}\begin{array}{ccc}
        Y_{v_r} & & Y_{v_r}  \\
        & X_{v_b} & \\
        Y_{v_r} & & Y_{v_r}  
    \end{array}-\sum_{v_b\in B}\begin{array}{ccccc}
       Z_{v_b} & & & &  Z_{v_b}\\
         & Z_{v_r} & & Z_{v_r} & \\
       & & X_{v_b} & & \\
       & Z_{v_r} & & Z_{v_r} &  \\
       Z_{v_b} & & & &  Z_{v_b}
    \end{array}\, .
    \label{eq:H2Dclusterblue}
\end{align}
\end{widetext}
We note that while restricting $\mathrm{H}_{\text{2D-cluster}}$ and $\mathrm{H}_{\text{blue}}$ to the regions $A$ and $B$, we do not have any terms in the interface Hamiltonian $\mathrm{H}_{\text{2D-cluster}|\text{blue}}$ that is not supported entirely on $A$ or $B$. Hence, $\mathrm{H}_{\text{2D-cluster}|\text{blue}}$ has degenerate ground states. Let $|\Psi\rangle$ be a ground state. Then we verify the following:
\begin{widetext}
\begin{subequations}
    \begin{align}
    &\eta^x_{r,j}\ket{\Psi}=\ket{\Psi}\quad \forall j\in\mathbb{Z}_{L_y}\, ,\qquad \eta^y_{r,i}\ket{\Psi}=\ket{\Psi}\quad \forall i\in\mathbb{Z}_{L_x}\, ,\qquad
    \eta^y_{b,i}\ket{\Psi}=\ket{\Psi}, \text{ for } i\neq L_x-1,l,i\in\mathbb{Z}_{L_x}\quad\, ,\\
    &\qquad\eta^x_{b,j}\ket{\Psi}=\eta^{x(L)}_{b,j}\eta^{x(R)}_{b,j}\ket{\Psi}\text{ where } \eta^{x(R)}_{b,j}=\begin{array}{ccc}
        Y & & Z  \\
        & X_{L_x-\frac{1}{2},j+\frac{1}{2}} & \\
        Y & & Z  
    \end{array}\text{ and    }\eta^{x(L)}_{b,j}=\begin{array}{ccc}
        Z & & Y  \\
        & X_{l+\frac{1}{2},j+\frac{1}{2}} & \\
        Z & & Y  
    \end{array}\, .
    \end{align}
\end{subequations} 
The operators $\eta^{x(L)}_{b,j}$ and $\eta^{x(R)}_{b,j}$ commute with $\mathbf{D}^{(2)}$ on the ground space
\begin{align}
    [\eta^{x(L)}_{b,j},\mathbf{D}^{(2)}]|\Psi\rangle=0\, ,\quad [\eta^{x(R)}_{b,j},\mathbf{D}^{(2)}]|\Psi\rangle=0\, .
\end{align}
Then we can add $\eta^{x(L)}_{b,j}$ and $\eta^{x(R)}_{b,j}$ and their symmetrized version under $\mathbf{D}^{(2)}$ to the Hamiltonian to lift the ground state degeneracy
\begin{align}
    \mathrm{H}_{\text{2D-cluster}|\text{blue}}'&=\mathrm{H}_{\text{2D-cluster}|\text{blue}}+\sum_{v_b=(L_x-\frac{1}{2},-)}\begin{array}{ccc}
        Y & & Z  \\
        & X_{v_b} & \\
        Y & & Z  
    \end{array}\left(1+\begin{array}{ccc}
       Z &   & Z \\
         & X &  \\
    Z^2  &   & Z^2_{v_b}\\
         & X & \\
       Z &   & Z
    \end{array}\right)\nonumber\\
    &\hspace{4cm}+\sum_{v_b=(l+\frac{1}{2},-)}\begin{array}{ccc}
        Z & & Y  \\
        & X_{v_b} & \\
        Z & & Y  
    \end{array}\left(1+\begin{array}{ccc}
       Z &   & Z \\
         & X &  \\
    Z^2_{ v_b}  &   & Z^2\\
         & X & \\
       Z &   & Z
    \end{array}\right)\, .
\end{align}
\end{widetext}
This Hamiltonian is $\mathbb{Z}_2\times\mathbb{Z}_2$ subsystem symmetric~\eqref{eq:subsystemsymmetries2Dclstr} as well as $\mathbf{D}^{(2)}$ symmetric. This is a stabilizer Hamiltonian with a priori no constraint. Hence, it has a unique ground state on a torus. This implies that there are no edge modes at the interface. This fact leads us to consider a different type of interface. 
\subsubsection{Rectangular interface}
We consider a rectangular interface placed on the torus between the Hamiltonians $\rm H_{\text{2D-cluster}}$ and $\rm H_{\text{blue}}$. We choose the interface line to run along the blue sublattice with corners at $(i_0+\frac{1}{2},j_0+\frac{1}{2})$, $(i_0+\frac{1}{2},j_1+\frac{1}{2})$, $(i_1+\frac{1}{2},j_0+\frac{1}{2})$ and $(i_1+\frac{1}{2},j_1+\frac{1}{2})$ as given in Figure~\ref{fig:rectangularinterface}.

Let us take $L_x$, $L_y$, $j_1-j_0$ and $i_1-i_0$ to be even. We consider $\rm H_{\text{blue}}$ inside the rectangular region and $\rm H_{\text{2D-cluster}}$ outside the rectangular region. In this interface Hamiltonian, we do not have any term that is supported both inside and outside of the rectangular region. We keep all the terms supported inside or outside the rectangular region, including the boundary. To define the Hamiltonian explicitly, we define the following sets
\begin{widetext}
\red{
\begin{subequations}
\begin{align}
    A&=\{(x,y)\in(\mathbb{Z}/2,\mathbb{Z}/2)|\, y\leq j_0\}\cup\{(x,y)\in(\mathbb{Z}/2,\mathbb{Z}/2)|\, y> j_1+\frac{1}{2}\}\nonumber\\
    &\qquad \cup\{(x,y)\in(\mathbb{Z}/2,\mathbb{Z}/2)|\, x\leq i_0\}\cup\{(x,y)\in(\mathbb{Z}/2,\mathbb{Z}/2)|\, x> i_1+\frac{1}{2}\}\, , \\
    B&=\{(x,y)\in(\mathbb{Z}/2,\mathbb{Z}/2)|\, i_0+\frac{1}{2}<x<i_1+\frac{1}{2}, j_0+\frac{1}{2}<y<j_1+\frac{1}{2}\}\, .  
\end{align}
\label{eq:regionABrectangular}
\end{subequations}}
The interface Hamiltonian is the same as \eqref{eq:H2Dclusterblue} but now with regions $A$ and $B$ as given above. From the previous analysis, we could add terms in the Hamiltonian along the interface everywhere except at the corners that respect subsystem and noninvertible symmetries and commute with each term in the interface Hamiltonian. The new Hamiltonian is
\begin{align}
    \tilde{\mathrm{H}}_{\text{2D-cluster}|\text{blue}}&=\mathrm{H}_{\text{2D-cluster}|\text{blue}}-\sum_{\substack{v_b=(i+\frac{1}{2},j_1+\frac{1}{2})\\
    i_0<i<i_1}}\begin{array}{ccc}
        Z & & Z  \\
        & X_{v_b} & \\
        Y & & Y  
    \end{array}\left(1+\begin{array}{ccccc}
        Z & & Z_{v_b}^2 & & Z \\
        & X & & X & \\
        Z & & Z^2 & & Z  
    \end{array}\right)\nonumber\\
    &-\sum_{\substack{v_b=(i+\frac{1}{2},j_0+\frac{1}{2})\\
    i_0<i<i_1}}\begin{array}{ccc}
        Y & & Y  \\
        & X_{v_b} & \\
        Z & & Z  
    \end{array}\left(1+\begin{array}{ccccc}
        Z & & Z^2 & & Z \\
        & X & & X & \\
        Z & & Z_{v_b}^2 & & Z  
    \end{array}\right)-\sum_{\substack{v_b=(i_1+\frac{1}{2},j+\frac{1}{2})\\
    j_0<j<j_1}}\begin{array}{ccc}
        Y & & Z  \\
        & X_{v_b} & \\
        Y & & Z  
    \end{array}\left(1+\begin{array}{ccc}
       Z &   & Z \\
         & X &  \\
    Z^2  &   & Z^2_{v_b}\\
         & X & \\
       Z &   & Z
    \end{array}\right)\nonumber\\
    &\hspace{4cm}-\sum_{\substack{v_b=(i_0+\frac{1}{2},j+\frac{1}{2})\\
    j_0<j<j_1}}\begin{array}{ccc}
        Z & & Y  \\
        & X_{v_b} & \\
        Z & & Y  
    \end{array}\left(1+\begin{array}{ccc}
       Z &   & Z \\
         & X &  \\
    Z^2_{ v_b}  &   & Z^2\\
         & X & \\
       Z &   & Z
    \end{array}\right)\, .
    \label{eq:Htilde2Dclusterblue}
\end{align}
Now let us call the ground state of this Interface Hamiltonian $\tilde{\mathrm{H}}_{\text{2D-cluster}|\text{blue}}$ with terms added along the interface except at the corners to be $|\tilde{\Psi}\rangle$.
\begin{align} \eta_r\ket{\tilde{\Psi}}=\ket{\tilde{\Psi}}\, ,\quad \eta_b\ket{\tilde{\Psi}}=\raisebox{-60pt}{\begin{tikzpicture}
    \node at (0,0) {$\color{blue}{X}$};
    \node at (3,0) {$\color{blue}{X}$};
    \node at (0,3) {$\color{blue}{X}$};
    \node at (3,3) {$\color{blue}{X}$};
    \node at (0.5,0.5) {$\color{red}{Y}$};
    \node at (-0.5,0.5) {$\color{red}{Z}$};
    \node at (0.5,-0.5) {$\color{red}{Z}$};
    \node at (-0.5,-0.5) {$\color{red}{Z}$};
    \node at (3.5,0.5) {$\color{red}{Z}$};
    \node at (3.5,-0.5) {$\color{red}{Z}$};
    \node at (2.5,-0.5) {$\color{red}{Z}$};
    \node at (2.5,0.5) {$\color{red}{Y}$};
    \node at (0.5,2.5) {$\color{red}{Y}$};
    \node at (0.5,3.5) {$\color{red}{Z}$};
    \node at (-0.5,2.5) {$\color{red}{Z}$};
    \node at (-0.5,3.5) {$\color{red}{Z}$};
    \node at (2.5,2.5) {$\color{red}{Y}$};
    \node at (2.5,3.5) {$\color{red}{Z}$};
    \node at (3.5,3.5) {$\color{red}{Z}$};
    \node at (3.5,2.5) {$\color{red}{Z}$};
    \draw[thick] (0,0)--(3,0)--(3,3)--(0,3)--(0,0);
\end{tikzpicture}}\ket{\tilde{\Psi}}\equiv\begin{array}{cc}
    \eta_b^{TL} &\eta_{b}^{TR}  \\
    \eta_b^{BL} & \eta_b^{BR}
\end{array}\ket{\tilde{\Psi}}
\end{align}
and 
\begin{align}
    &\eta_{r,j}^x|\tilde{\Psi}\rangle =|\tilde{\Psi}\rangle \quad\forall j ,\qquad \eta_{r,i}^y|\tilde{\Psi}\rangle =|\tilde{\Psi}\rangle\quad\forall i ,\qquad \eta_{b,j}^x|\tilde{\Psi}\rangle =|\tilde{\Psi}\rangle\,\quad \forall j\neq j_0,j_1\, ,\qquad\eta_{b,i}^y|\tilde{\Psi}\rangle =|\tilde{\Psi}\rangle\,\quad \forall i\neq i_0,i_1\, ,\nonumber\\
    &\eta_{b,j_0}^x|\tilde{\Psi}\rangle =\eta_b^{BL}\eta_b^{BR}|\tilde{\Psi}\rangle\, ,\quad \eta_{b,j_1}^x|\tilde{\Psi}\rangle =\eta_b^{TL}\eta_b^{TR}|\tilde{\Psi}\rangle\, ,\quad \eta_{b,i_0}^y|\tilde{\Psi}\rangle =\eta_b^{BL}\eta_b^{TL}|\tilde{\Psi}\rangle\, ,\quad \eta_{b,i_1}^x|\tilde{\Psi}\rangle =\eta_b^{BR}\eta_b^{TR}|\tilde{\Psi}\rangle\, . 
\end{align}
Explicitly
\begin{align}
    \eta_b^{TL}=\begin{array}{ccc}
        \color{red}{Z} & & \color{red}{Z} \\
           & \color{blue}{X}_{i_0+\frac{1}{2},j_1+\frac{1}{2}} &\\
           \color{red}{Z} & & \color{red}{Y} 
    \end{array},\quad \eta_b^{TR}=\begin{array}{ccc}
        \color{red}{Z} & & \color{red}{Z} \\
           & \color{blue}{X}_{i_1+\frac{1}{2},j_1+\frac{1}{2}} &\\
           \color{red}{Y} & & \color{red}{Z} 
    \end{array}\, ,\quad \eta_b^{BL}=\begin{array}{ccc}
        \color{red}{Z} & & \color{red}{Y} \\
           & \color{blue}{X}_{i_0+\frac{1}{2},j_0+\frac{1}{2}} &\\
           \color{red}{Z} & & \color{red}{Z} 
    \end{array}\, ,\quad \eta_b^{BR}=\begin{array}{ccc}
        \color{red}{Y} & & \color{red}{Z} \\
           & \color{blue}{X}_{i_1+\frac{1}{2},j_0+\frac{1}{2}} &\\
           \color{red}{Z} & & \color{red}{Z} 
    \end{array}\, .
    \label{eq:etabluemodes}
\end{align}
\end{widetext}
Let us define
\begin{align}
    &Z^{TL}=Z_{i_0+\frac{1}{2},j_1+\frac{1}{2}}\, ,\quad Z^{TR}=Z_{i_1+\frac{1}{2},j_1+\frac{1}{2}}\, ,\nonumber\\
    &Z^{BL}=Z_{i_0+\frac{1}{2},j_0+\frac{1}{2}}\, ,\quad Z^{BR}=Z_{i_1+\frac{1}{2},j_0+\frac{1}{2}}\, .
    \label{eq:Zops}
\end{align}
The above-defined operators anti-commute with operators defined in \eqref{eq:etabluemodes}
\begin{align}
    \{\eta^{TL},Z^{TL}\}=0\, ,\quad \{\eta^{TR},Z^{TR}\}=0\, ,\nonumber\\
    \{\eta^{BL},Z^{BL}\}=0\, ,\quad \{\eta^{BR},Z^{BR}\}=0\, .
    \label{eq:etaZanticommutation}
\end{align}
Hence, these operators form a basis of operators acting on the ground space.

We find that the corner operators satisfy projective algebra with $\mathbf{D}^{(2)}$ on the ground space:
    \begin{align}
\label{eq:cornermodealgebra1}
    &\mathbf{D}^{(2)}\eta_b^{TL}=-\eta_b^{TL}\mathbf{D}^{(2)}\, ,\quad\mathbf{D}^{(2)}\eta_{b}^{TR}=-\eta_{b}^{TR}\mathbf{D}^{(2)}\, ,\nonumber\\
    &\mathbf{D}^{(2)}\eta_b^{BL}=-\eta_b^{BL}\mathbf{D}^{(2)}\, ,\quad \mathbf{D}^{(2)}\eta_b^{BR}=-\eta_b^{BR}\mathbf{D}^{(2)}\, .
\end{align}
 We also have the following relation 
\begin{align}
\label{eq:D^24Z}
    \mathbf{D}^{(2)}\begin{array}{cc}
       Z^{TL}  & Z^{TR} \\
        Z^{BL} & Z^{BR}
    \end{array}\ket{\tilde{\Psi}}=\begin{array}{cc}
       Z^{TL}  & Z^{TR} \\
        Z^{BL} & Z^{BR}
    \end{array}\mathbf{D}^{(2)}\ket{\tilde{\Psi}}\, .
\end{align}
Let us consider $(\mathbf{D}^{(2)})^2$ on the ground space
\begin{widetext}
\begin{align}
\label{eq:D2^2psitilde}
    \begin{split}
        (\mathbf{D}^{(2)})^2\ket{\tilde{\Psi}}&\sim\prod_{j=1}^{L_y}\frac{(1+\eta^x_{r,j})}{2}\prod_{i=1}^{L_x}\frac{(1+\eta^y_{r,i})}{2}\prod_{j=1}^{L_y}\frac{(1+\eta^x_{b,j})}{2}\prod_{i=1}^{L_x}\frac{(1+\eta^y_{b,i})}{2}\ket{\tilde{\Psi}}\\
    &\sim\frac{(1+\eta^x_{b,j_1})}{2}\frac{(1+\eta^x_{b,j_0})}{2}\frac{(1+\eta^y_{b,i_1})}{2}\frac{(1+\eta^y_{b,i_0})}{2}\ket{\tilde{\Psi}}\\
    &\sim\frac{(1+\eta_b^{TL}\eta_b^{TR})}{2}\frac{(1+\eta_b^{BL}\eta_b^{BR})}{2}\frac{(1+\eta_b^{TR}\eta_b^{BR})}{2}\frac{(1+\eta_b^{TL}\eta_b^{BL})}{2}\ket{\tilde{\Psi}}
    \end{split}
\end{align}
From \eqref{eq:cornermodealgebra1} and \eqref{eq:D^24Z}, we conclude the most general form of $\mathbf{D}^{(2)}$ up to an overall constant is 
\begin{align}
    \mathbf{D}^{(2)}\ket{\tilde{\Psi}}&\sim Z^{TL}Z^{TR}Z^{BL}Z^{BR}\left(\alpha_0+\alpha_1 \eta^{TL}\eta^{TR}+\alpha_2\eta^{TL}\eta^{BL}+\alpha_3\eta^{TL}\eta^{BR}+\alpha_4\eta^{TR}\eta^{BL}+\alpha_5\eta^{TR}\eta^{BR}\right.\nonumber\\
    &\left.\hspace{8cm}+\alpha_6\eta^{BL}\eta^{BR}+\alpha_7\eta^{TL}\eta^{TR}\eta^{BL}\eta^{BR}\right)\ket{\tilde{\Psi}}\, .
\end{align}
Imposing the constraint~\eqref{eq:D2^2psitilde} fix $\alpha_0=\alpha_1=\alpha_2=\alpha_3=\alpha_4=\alpha_5=\alpha_6=\alpha_7$. Then we can equivalently write
\begin{align}
    \begin{split}
        \mathbf{D}^{(2)}\ket{\tilde{\Psi}}&\sim Z^{TL}Z^{TR}Z^{BL}Z^{BR}\frac{(\eta_b^{TL}+\eta_b^{TR})}{2}\frac{(\eta_b^{BL}+\eta_b^{BR})}{2}\frac{(\eta_b^{TR}+\eta_b^{BR})}{2}\frac{(\eta_b^{TL}+\eta_b^{BL})}{2}\ket{\tilde{\Psi}}\\
        \begin{split}
            &=\frac{1}{8}\left(Z^{TR}Z^{TL}Z^{BR}Z^{BL}+\mathrm{D}^{TL}\mathrm{D}^{TR}Z^{BL}Z^{BR}+\mathrm{D}^{TL}\mathrm{D}^{BR}Z^{TR}Z^{BL}+\mathrm{D}^{BR}\mathrm{D}^{TR}Z^{BL}Z^{TL}\right.\\
        &\left.+\mathrm{D}^{TL}\mathrm{D}^{BL}Z^{TR}Z^{BR}+\mathrm{D}^{BL}\mathrm{D}^{TR}Z^{BR}Z^{TL}+\mathrm{D}^{BR}\mathrm{D}^{BL}Z^{TL}Z^{TR}+\mathrm{D}^{TL}\mathrm{D}^{TR}\mathrm{D}^{BL}\mathrm{D}^{BR}\right)\ket{\tilde{\Psi}}
        \end{split}
            \end{split}
\end{align}
where
\begin{align}
    \mathrm{D}^{TL}=Z^{TL}\eta_b^{TL}\, ,\quad \mathrm{D}^{TR}=Z^{TR}\eta_b^{TR}\, ,\quad \mathrm{D}^{BL}=Z^{BL}\eta_b^{BL}\, ,\quad \mathrm{D}^{BR}=Z^{BR}\eta_b^{BR}\, .
\end{align}
We have the projective algebra
\begin{subequations}
    \begin{align}
        &\{\mathrm{D}^{TL}, Z^{TL}\}=0\, ,\quad \{\mathrm{D}^{TL},\eta_b^{TL}\}=0\, ,\quad \{Z^{TL},\eta_b^{TL}\}=0\, ,\\
        &\{\mathrm{D}^{TR}, Z^{TR}\}=0\, ,\quad \{\mathrm{D}^{TR},\eta_b^{TR}\}=0\, ,\quad \{Z^{TR},\eta_b^{TR}\}=0\, ,\\
        &\{\mathrm{D}^{BL}, Z^{BL}\}=0\, ,\quad \{\mathrm{D}^{BL},\eta_b^{BL}\}=0\, ,\quad \{Z^{BL},\eta_b^{BL}\}=0\, ,\\
        &\{\mathrm{D}^{BR}, Z^{BR}\}=0\, ,\quad \{\mathrm{D}^{BR},\eta_b^{BR}\}=0\, ,\quad \{Z^{BR},\eta_b^{BR}\}=0\, . 
    \end{align}
    \label{eq:2Dnoninvertibleprojectivealgebra}
\end{subequations}
This projective algebra indicate that the corner modes cannot be gapped out and they distinguish the two phases represented by the Hamiltonian $\mathrm{H}_{\text{2D-cluster}}$ and $\mathrm{H}_{\text{blue}}$. See Appendix~\ref{sec:stability2dclusterblue} for a rigourous argument for robustness of corner modes with symmetric perturbations to the interface Hamiltonian.
\subsection{Interface between \texorpdfstring{$\rm H_{\text{blue}}$}{Lg} and \texorpdfstring{$\rm H_{\text{red}}$}{Lg}}\label{sec:InterfaceHblueHred}
Let us consider the rectangular interface on a torus with corners at $(i_0+\frac{1}{2},j_0+\frac{1}{2})$, $(i_0+\frac{1}{2},j_1+\frac{1}{2})$, $(i_1+\frac{1}{2},j_0+\frac{1}{2})$ and $(i_1+\frac{1}{2},j_1+\frac{1}{2})$ as before. Let us take $L_x$, $L_y$, $j_1-j_0$ and $i_1-i_0$ to be even. Now, let us define the sets
\begin{subequations}
\red{
\begin{align}
     A&=\{(x,y)\in(\mathbb{Z}/2,\mathbb{Z}/2)|\, y< j_0\}\cup\{(x,y)\in(\mathbb{Z}/2,\mathbb{Z}/2)|\, y> j_1+1\}\nonumber\\
    &\qquad \cup\{(i,j)\in(\mathbb{Z}/2,\mathbb{Z}/2)|\, x< i_0\}\cup\{(x,y)\in(\mathbb{Z}/2,\mathbb{Z}/2)|\, x> i_1+1\}\, , \\
    B&=\{(x,y)\in(\mathbb{Z}/2,\mathbb{Z}/2)|\, i_0+\frac{1}{2}<x<i_1+\frac{1}{2}, j_0+\frac{1}{2}<y<j_1+\frac{1}{2}\}\, .
\end{align}}
\end{subequations}
We define the interface Hamiltonian
\begin{align}
    \mathrm{H}_{\text{blue}|\text{red}}&=\sum_{v_r\in B}\begin{array}{ccc}
        Z_{v_b} & & Z_{v_b}  \\
        & X_{v_r} & \\
        Z_{v_b} & & Z_{v_b}  
    \end{array}+\sum_{v_b\in B}\begin{array}{ccc}
        Y_{v_r} & & Y_{v_r}  \\
        & X_{v_b} & \\
        Y_{v_r} & & Y_{v_r}  
    \end{array}+\sum_{v_b\in B}\begin{array}{ccccc}
       Z_{v_b} & & & &  Z_{v_b}\\
         & Z_{v_r} & & Z_{v_r} & \\
       & & X_{v_b} & & \\
       & Z_{v_r} & & Z_{v_r} &  \\
       Z_{v_b} & & & &  Z_{v_b}
    \end{array}\nonumber\\
    &+\sum_{v_b\in A}\begin{array}{ccc}
        Z_{v_r} & & Z_{v_r}  \\
        & X_{v_b} & \\
        Z_{v_r} & & Z_{v_r}  
    \end{array}-\sum_{v_r\in A}\begin{array}{ccc}
        Y_{v_b} & & Y_{v_b}  \\
        & X_{v_r} & \\
        Y_{v_b} & & Y_{v_b}  
    \end{array}-\sum_{v_r\in A}\begin{array}{ccccc}
       Z_{v_r} & & & &  Z_{v_r}\\
         & Z_{v_b} & & Z_{v_b} & \\
       & & X_{v_r} & & \\
       & Z_{v_b} & & Z_{v_b} &  \\
       Z_{v_r} & & & &  Z_{v_r}
    \end{array}\,.
\end{align}
We note that we applied a finite depth local unitary conjugation (that is the product $\prod_{v_b\in B}Z_{v_b}$) on the Hamiltonian $\mathrm{H}_{\text{blue}}$ to change the sign in the second and third term.
We note that we could add terms in the Hamiltonian along the interface everywhere except at the corners that respect the symmetries and commute with each term in the interface Hamiltonian. The new Hamiltonain is 
\begin{align}
    \tilde{\rm H}_{\text{blure}|\text{red}}&=\mathrm{H}_{\text{blue}|\text{red}}-\sum_{\substack{v_b=(i+\frac{1}{2},j_1+\frac{1}{2})\\
    i_0<i<i_1}}\begin{array}{ccc}
        Z & & Z  \\
        & X_{v_b} & \\
        Y & & Y  
    \end{array}\left(1+\begin{array}{ccccc}
        Z & & Z_{v_b}^2 & & Z \\
        & X & & X & \\
        Z & & Z^2 & & Z  
    \end{array}\right)\nonumber\\
    &-\sum_{\substack{v_b=(i+\frac{1}{2},j_0+\frac{1}{2})\\
    i_0<i<i_1}}\begin{array}{ccc}
        Y & & Y  \\
        & X_{v_b} & \\
        Z & & Z  
    \end{array}\left(1+\begin{array}{ccccc}
        Z & & Z^2 & & Z \\
        & X & & X & \\
        Z & & Z_{v_b}^2 & & Z  
    \end{array}\right)-\sum_{\substack{v_b=(i_1+\frac{1}{2},j+\frac{1}{2})\\
    j_0<j<j_1}}\begin{array}{ccc}
        Y & & Z  \\
        & X_{v_b} & \\
        Y & & Z  
    \end{array}\left(1+\begin{array}{ccc}
       Z &   & Z \\
         & X &  \\
    Z^2  &   & Z^2_{v_b}\\
         & X & \\
       Z &   & Z
    \end{array}\right)\nonumber\\
    &-\sum_{\substack{v_b=(i_0+\frac{1}{2},j+\frac{1}{2})\\
    j_0<j<j_1}}\begin{array}{ccc}
        Z & & Y  \\
        & X_{v_b} & \\
        Z & & Y  
    \end{array}\left(1+\begin{array}{ccc}
       Z &   & Z \\
         & X &  \\
    Z^2_{ v_b}  &   & Z^2\\
         & X & \\
       Z &   & Z
    \end{array}\right)-\sum_{\substack{v_r=(i,j_0)\\
    i_0<i\leq i_1}}\begin{array}{ccc}
        Z & & Z  \\
        & X_{v_r} & \\
        Y & & Y  
    \end{array}\left(1+\begin{array}{ccccc}
        Z & & Z_{v_r}^2 & & Z \\
        & X & & X & \\
        Z & & Z^2 & & Z  
    \end{array}\right)\nonumber\\
    &-\sum_{\substack{v_r=(i,j_1+1)\\
    i_0<i\leq i_1}}\begin{array}{ccc}
        Y & & Y  \\
        & X_{v_r} & \\
        Z & & Z  
    \end{array}\left(1+\begin{array}{ccccc}
        Z & & Z^2 & & Z \\
        & X & & X & \\
        Z & & Z_{v_r}^2 & & Z  
    \end{array}\right)-\sum_{\substack{v_r=(i_0,j)\\
    j_0<j\leq j_1}}\begin{array}{ccc}
        Y & & Z  \\
        & X_{v_r} & \\
        Y & & Z  
    \end{array}\left(1+\begin{array}{ccc}
       Z &   & Z \\
         & X &  \\
    Z^2  &   & Z^2_{v_r}\\
         & X & \\
       Z &   & Z
    \end{array}\right)\nonumber\\
    &-\sum_{\substack{v_r=(i_1+1,j)\\
    j_0<j\leq j_1}}\begin{array}{ccc}
        Z & & Y  \\
        & X_{v_r} & \\
        Z & & Y  
    \end{array}\left(1+\begin{array}{ccc}
       Z &   & Z \\
         & X &  \\
    Z^2_{ v_r}  &   & Z^2\\
         & X & \\
       Z &   & Z
    \end{array}\right).
\end{align} 
The number of stabilizers in this Hamiltonian is $2L_xL_y-8$. So, naively we would expect there would be eight gapless modes contributing to $2^8$ fold degeneracy. However, we can add the following additional terms to gap out four among them.
\begin{align}
   \tilde{\rm H}'_{\text{blue}|\text{red}}&= \tilde{\rm H}_{\text{blue}|\text{red}}-\begin{array}{cccc}
         & Z_{v_r}& & Y_{v_r}  \\
        Y_{v_b} & & \boxed{Z_{v_b}X_{v_b}} & \\
        &X_{v_r} Z_{v_r}& &Z_{v_r} \\
        Y_{v_b} & & Y_{v_{b}}
    \end{array}\left(1+\begin{array}{cccccc}
         &  &  & Z_{v_b} & & Z_{v_b} \\
         Z_{v_r} & & Z_{v_r} & & X_{v_r}  & \\
         & X_{v_b} & & Z_{v_b} & & Z_{v_b}\\
         & & Z_{v_r} & & Z_{v_r} & \\
         & X_{v_b} & & X_{v_b} & & \\
         Z_{v_r} & & & & Z_{v_r}
    \end{array}\right)
    \nonumber\\
    &-\begin{array}{cccc}
        Y_{v_r} & & Z_{v_r} &  \\
         & \boxed{Z_{v_b}X_{v_b}} & & Y_{v_b}\\
         Z_{v_r} & & X_{v_r}Z_{v_r} & \\
         & Y_{v_b} & & Y_{v_b}
    \end{array}\left(1+\begin{array}{cccccc}
        Z_{v_b} & & Z_{v_b} & & &  \\
          & X_{v_r} & & Z_{v_r} & & Z_{v_r}\\
          Z_{v_b} & & Z_{v_b} & & X_{v_b} & \\
          & Z_{v_r} & & Z_{v_r} & & \\
          & & X_{v_b} & & X_{v_b} & \\
          Z_{v_r} & & & & &Z_{v_r}
    \end{array}\right)\nonumber\\
      &-\begin{array}{cccc}
        Y_{v_b} & & Y_{v_b} &  \\
         &  X_{v_r}Z_{v_r} & & Z_{v_r}\\
         Y_{v_b} & & \boxed{Z_{v_b}X_{v_b}} & \\
         & Z_{v_r} & & Y_{v_r}
    \end{array}\left(1+\begin{array}{cccccc}
        Z_{v_r} & & & & Z_{v_r}& \\
         & X_{v_b} & & X_{v_b} & & \\ 
         & & Z_{v_r} & & Z_{v_r} & \\
         & X_{v_b} & & Z_{v_b} & & Z_{v_b}\\
         Z_{v_r} & & Z_{v_r} & & X_{v_r} & \\
         & & & Z_{v_b} & & Z_{v_b}
    \end{array}\right)\nonumber\\
    &-\begin{array}{cccc}
         & Y_{v_b} & & Y_{v_b} \\
        Z_{v_r} & & X_{v_r}Z_{v_r} & \\
        & \boxed{Z_{v_b}X_{v_b}} & & Y_{v_b}\\
        Y_{v_r} & & Z_{v_r} & 
    \end{array}\left(1+\begin{array}{cccccc}
         & Z_{v_r} & & & & Z_{v_r} \\
         & & X_{v_b} & & X_{v_b} & \\
         & Z_{v_r} & & Z_{v_r} & & \\
         Z_{v_b} & & Z_{v_b} & & X_{v_b} & \\
         & X_{v_r} & & Z_{v_r} & & Z_{v_r}\\
         Z_{v_b}& & Z_{v_b} & & & 
    \end{array}\right)\,.
    \label{eq:H'blue|red}
\end{align}
where the boxed vertices in the four lines in the above equation are at $(i_0+\frac{1}{2},j_0+\frac{1}{2})$,$(i_1+\frac{1}{2},j_0+\frac{1}{2})$, $(i_0+\frac{1}{2},j_1+\frac{1}{2})$, and $(i_1+\frac{1}{2},j_1+\frac{1}{2})$ respectively. Suppose $\ket{\Psi}$ is a ground state of the Hamiltonian. Then 
\begin{align}
    \eta_b\ket{\tilde{\Psi}}=\raisebox{-60pt}{\begin{tikzpicture}
    \node at (0,0) {$\color{blue}{X}$};
    \node at (3,0) {$\color{blue}{X}$};
    \node at (0,3) {$\color{blue}{X}$};
    \node at (3,3) {$\color{blue}{X}$};
    \node at (0.5,0.5) {$\color{red}{Y}$};
    \node at (-0.5,0.5) {$\color{red}{Z}$};
    \node at (0.5,-0.5) {$\color{red}{Z}$};
    \node at (-0.5,-0.5) {$\color{red}{Z}$};
    \node at (3.5,0.5) {$\color{red}{Z}$};
    \node at (3.5,-0.5) {$\color{red}{Z}$};
    \node at (2.5,-0.5) {$\color{red}{Z}$};
    \node at (2.5,0.5) {$\color{red}{Y}$};
    \node at (0.5,2.5) {$\color{red}{Y}$};
    \node at (0.5,3.5) {$\color{red}{Z}$};
    \node at (-0.5,2.5) {$\color{red}{Z}$};
    \node at (-0.5,3.5) {$\color{red}{Z}$};
    \node at (2.5,2.5) {$\color{red}{Y}$};
    \node at (2.5,3.5) {$\color{red}{Z}$};
    \node at (3.5,3.5) {$\color{red}{Z}$};
    \node at (3.5,2.5) {$\color{red}{Z}$};
    \draw[thick] (0,0)--(3,0)--(3,3)--(0,3)--(0,0);
\end{tikzpicture}}\ket{\Psi}\,,\qquad \eta_b\ket{\Psi}=\raisebox{-60pt}{\begin{tikzpicture}
    \node at (0,0) {$\color{blue}{Z}$};
    \node at (3,0) {$\color{blue}{Z}$};
    \node at (0,3) {$\color{blue}{Z}$};
    \node at (3,3) {$\color{blue}{Z}$};
    \node at (-0.5,-0.5) {$\color{red}{X}$};
     \node at (-1,0) {$\color{blue}{Y}$};
     \node at (-1,-1) {$\color{blue}{Y}$};
     \node at (0,-1) {$\color{blue}{Y}$};
     \node at (3.5,-0.5) {$\color{red}{X}$};
     \node at (4,0) {$\color{blue}{Y}$};
     \node at (4,-1) {$\color{blue}{Y}$};
     \node at (3,-1) {$\color{blue}{Y}$};
      \node at (-0.5,3.5) {$\color{red}{X}$};
     \node at (0,4) {$\color{blue}{Y}$};
     \node at (-1,3) {$\color{blue}{Y}$};
     \node at (-1,4) {$\color{blue}{Y}$};
     \node at (3.5,3.5) {$\color{red}{X}$};
     \node at (4,4) {$\color{blue}{Y}$};
     \node at (4,3) {$\color{blue}{Y}$};
     \node at (3,4) {$\color{blue}{Y}$};
    \draw[thick] (0,0)--(3,0)--(3,3)--(0,3)--(0,0);
\end{tikzpicture}}\ket{\Psi}\,.
\end{align}
In \eqref{eq:H'blue|red}, we added the product of local operators around each corner with it's $\mathbf{D}^{(2)}$ conjugated terms. Now, let us call the ground state of $\tilde{\rm H}'_{\text{blue}|\text{red}}$ to be $\ket{\tilde{\Psi}}$. We find
\begin{align} \eta_r\ket{\tilde{\Psi}}= \eta_b\ket{\tilde{\Psi}}=\raisebox{-60pt}{\begin{tikzpicture}
    \node at (0,0) {$\color{blue}{X}$};
    \node at (3,0) {$\color{blue}{X}$};
    \node at (0,3) {$\color{blue}{X}$};
    \node at (3,3) {$\color{blue}{X}$};
    \node at (0.5,0.5) {$\color{red}{Y}$};
    \node at (-0.5,0.5) {$\color{red}{Z}$};
    \node at (0.5,-0.5) {$\color{red}{Z}$};
    \node at (-0.5,-0.5) {$\color{red}{Z}$};
    \node at (3.5,0.5) {$\color{red}{Z}$};
    \node at (3.5,-0.5) {$\color{red}{Z}$};
    \node at (2.5,-0.5) {$\color{red}{Z}$};
    \node at (2.5,0.5) {$\color{red}{Y}$};
    \node at (0.5,2.5) {$\color{red}{Y}$};
    \node at (0.5,3.5) {$\color{red}{Z}$};
    \node at (-0.5,2.5) {$\color{red}{Z}$};
    \node at (-0.5,3.5) {$\color{red}{Z}$};
    \node at (2.5,2.5) {$\color{red}{Y}$};
    \node at (2.5,3.5) {$\color{red}{Z}$};
    \node at (3.5,3.5) {$\color{red}{Z}$};
    \node at (3.5,2.5) {$\color{red}{Z}$};
    \draw[thick] (0,0)--(3,0)--(3,3)--(0,3)--(0,0);
\end{tikzpicture}}\equiv\begin{array}{cc}
    \eta_b^{TL} &\eta_{b}^{TR}  \\
    \eta_b^{BL} & \eta_b^{BR}
\end{array}\ket{\tilde{\Psi}}
\end{align}
\begin{align}
    &\eta_{r,j}^x|\tilde{\Psi}\rangle =|\tilde{\Psi}\rangle\quad \forall j\neq j_0,j_1+1 \, ,\quad \eta_{r,i}^y|\tilde{\Psi}\rangle =|\tilde{\Psi}\rangle\quad \forall i\neq i_0,i_1+1\,,\quad \eta_{b,j}^x|\tilde{\Psi}\rangle =|\tilde{\Psi}\rangle\,\quad \forall j\neq j_0,j_1\, ,\quad\eta_{b,i}^y|\tilde{\Psi}\rangle =|\tilde{\Psi}\rangle\,\quad \forall i\neq i_0,i_1\, ,\nonumber\\
    &\eta_{b,j_0}^x|\tilde{\Psi}\rangle =\eta_b^{BL}\eta_b^{BR}|\tilde{\Psi}\rangle\, ,\quad \eta_{b,j_1}^x|\tilde{\Psi}\rangle =\eta_b^{TL}\eta_b^{TR}|\tilde{\Psi}\rangle\, ,\quad \eta_{b,i_0}^y|\tilde{\Psi}\rangle =\eta_b^{BL}\eta_b^{TL}|\tilde{\Psi}\rangle\, ,\quad \eta_{b,i_1}^x|\tilde{\Psi}\rangle =\eta_b^{BR}\eta_b^{TR}|\tilde{\Psi}\rangle\, ,
    \nonumber\\
    &\eta_{r,j_0}^x|\tilde{\Psi}\rangle =\eta_b^{BL}\eta_b^{BR}|\tilde{\Psi}\rangle\, ,\quad \eta_{r,j_1+1}^x|\tilde{\Psi}\rangle =\eta_b^{TL}\eta_b^{TR}|\tilde{\Psi}\rangle\, ,\quad \eta_{r,i_0}^y|\tilde{\Psi}\rangle =\eta_b^{BL}\eta_b^{TL}|\tilde{\Psi}\rangle\, ,\quad \eta_{r,i_1+1}^x|\tilde{\Psi}\rangle =\eta_b^{BR}\eta_b^{TR}|\tilde{\Psi}\rangle\,. 
\end{align}
where $\eta_b^{TL}$, $\eta_b^{TR}$, $\eta_b^{BL}$, and $\eta_b^{BR}$ are defined in \eqref{eq:etabluemodes}.
Now, we could define $Z$ operators as in \eqref{eq:Zops} and the relations \eqref{eq:etaZanticommutation},\eqref{eq:cornermodealgebra1}, and \eqref{eq:D^24Z}  hold in this case as well. Let us consider $(\mathbf{D}^{(2)})^2$ on the ground space
\begin{align}
\label{eq:D2^2psitilde2}
    \begin{split}
        (\mathbf{D}^{(2)})^2\ket{\tilde{\Psi}}&\sim\prod_{j=1}^{L_y}\frac{(1+\eta^x_{r,j})}{2}\prod_{i=1}^{L_x}\frac{(1+\eta^y_{r,i})}{2}\prod_{j=1}^{L_y}\frac{(1+\eta^x_{b,j})}{2}\prod_{i=1}^{L_x}\frac{(1+\eta^y_{b,i})}{2}\ket{\tilde{\Psi}}\\
    &\sim\frac{(1+\eta^x_{b,j_1})}{2}\frac{(1+\eta^x_{b,j_0})}{2}\frac{(1+\eta^y_{b,i_1})}{2}\frac{(1+\eta^y_{b,i_0})}{2}\frac{(1+\eta^x_{r,j_1+1})}{2}\frac{(1+\eta^x_{r,j_0})}{2}\frac{(1+\eta^y_{r,i_1+1})}{2}\frac{(1+\eta^y_{r,i_0})}{2}\ket{\tilde{\Psi}}\\
    &\sim\frac{(1+\eta_b^{TL}\eta_b^{TR})^2}{2}\frac{(1+\eta_b^{BL}\eta_b^{BR})^2}{2}\frac{(1+\eta_b^{TR}\eta_b^{BR})^2}{2}\frac{(1+\eta_b^{TL}\eta_b^{BL})^2}{2}\ket{\tilde{\Psi}}\\
    &\sim\frac{(1+\eta_b^{TL}\eta_b^{TR})}{2}\frac{(1+\eta_b^{BL}\eta_b^{BR})}{2}\frac{(1+\eta_b^{TR}\eta_b^{BR})}{2}\frac{(1+\eta_b^{TL}\eta_b^{BL})}{2}\ket{\tilde{\Psi}}.
    \end{split}
\end{align}
The rest of the analysis is exactly the same as in the previous interface mode analysis between $\mathrm{H}_{\text{2D-cluster}}$ and $\mathrm{H}_{\text{blue}}$. We obtain the projective algebra \eqref{eq:2Dnoninvertibleprojectivealgebra}. This projective algebra indicate that the corner modes cannot be gapped out without breaking the symmetry and they distinguish the two phases represented by the Hamiltonian $\mathrm{H}_{\text{blue}}$ and $\mathrm{H}_{\text{red}}$.
\end{widetext}
\subsection{Interface between \texorpdfstring{$\rm H_{\text{2D-cluster}}$}{Lg} and \texorpdfstring{$\rm H_{\text{blue}}^{x;k}$}{Lg}}\label{sec:interface2dclusterxkblue}
In this case, it is sufficient to consider two line interfaces. We place two line interfaces at $x=l+\frac{1}{2}$ and $x=L_x-\frac{1}{2}$ for some $l\neq L_x-1\in\mathbb{Z}_{L_x}$ such that $L_x-l$ is odd and $L_y$ is even. As before, these two lines divide the torus into regions $A$ and $B$ (including the boundary interface lines) containing $(L_x,0)$ and $(L_x-1,0)$ respectively. The interface Hamiltonian is
\begin{align}
    \mathrm{H}_{\text{2D-cluster}\rvert \text{blue}^{x;k}}=\mathrm{H}_{\text{2D-cluster}}\rvert_A+\mathrm{H}_{\text{blue}}^{x;k}\rvert_B
\end{align}
In the interface Hamiltonian, we remove all the terms that are not supported entirely on $A$ or $B$. Let $|\Psi\rangle$ be a ground state among the degenerate ground states of the interface Hamiltonian. Then, we have the following:
\begin{widetext}
    \begin{subequations}
    \begin{align}
    &\eta^x_{r,j}\ket{\Psi}=\ket{\Psi}\, ,\qquad \eta^y_{r,i}\ket{\Psi}=\ket{\Psi}\, ,\qquad
    \eta^y_{b,i}\ket{\Psi}=\ket{\Psi}, \text{ for } i\neq L_x-1,l\, ,\\
    &\eta^x_{b,j}\ket{\Psi}=\eta^{x(L)}_{b,j}\eta^{x(R)}_{b,j}\ket{\Psi}\text{ where } \eta^{x(R)}_{b,j}=\begin{array}{ccc}
        Z & & Z  \\
        & X_{L_x-\frac{1}{2},j+\frac{1}{2}} & \\
        Z & & Z  
    \end{array}\text{ and    }\eta^{x(L)}_{b,j}=\begin{array}{ccc}
        Z & & Z  \\
        & X_{l+\frac{1}{2},j+\frac{1}{2}} & \\
        Z & & Z  
    \end{array}\, , \forall j\neq k,k-1\in\mathbb{Z}_{L_y}\\
    &\eta^x_{b,k}\ket{\Psi}=\eta^{x(L)}_{b,k}\eta^{x(R)}_{b,k}\ket{\Psi}\text{ where } \eta^{x(R)}_{b,k}=\begin{array}{ccc}
        Z & & Z  \\
        & X_{L_x-\frac{1}{2},k+\frac{1}{2}} & \\
        Y & & Z  
    \end{array}\text{ and    }\eta^{x(L)}_{b,k}=\begin{array}{ccc}
        Z & & Z  \\
        & X_{l+\frac{1}{2},j+\frac{1}{2}} & \\
        Z & & Y  
    \end{array}\, ,\\
    &\eta^x_{b,k-1}\ket{\Psi}=\eta^{x(L)}_{b,k-1}\eta^{x(R)}_{b,k-1}\ket{\Psi}\text{ where } \eta^{x(R)}_{b,k-1}=\begin{array}{ccc}
        Y & & Z  \\
        & X_{L_x-\frac{1}{2},k-\frac{1}{2}} & \\
        Z & & Z  
    \end{array}\text{ and    }\eta^{x(L)}_{b,k-1}=\begin{array}{ccc}
        Z & & Y  \\
        & X_{l+\frac{1}{2},j+\frac{1}{2}} & \\
        Z & & Z  
    \end{array}\, ,
\end{align}
\end{subequations}
We note that $\eta^{x(L)}_{b,j}$ and $\eta^{x(R)}_{b,j}$ for $j\neq k,k-1$ commute with both $\mathbb{Z}_2\times\mathbb{Z}_2$ subsystem symmetries and $\mathbf{D}^{(2)}$ and hence can be added to $\mathrm{H}_{\text{2D-cluster}\rvert \text{blue}^{x;k}}$. \red{Similarly we can add the term $\eta^{x(L)}_{b,k}\eta^{x(L)}_{b,k-1}$ and $\eta^{x(R)}_{b,k}\eta^{x(R)}_{b,k-1}$} to obtain a new Hamitlonian
\begin{align}
    \mathrm{H}_{\text{2D-cluster}\rvert \text{blue}^{x;k}}^{'}&=\mathrm{H}_{\text{2D-cluster}\rvert \text{blue}^{x;k}}-\sum_{j\neq k,k-1}\begin{array}{ccc}
        Z & & Z  \\
        & X_{L_x-\frac{1}{2},j+\frac{1}{2}} & \\
        Z & & Z  
    \end{array}-\sum_{j\neq k,k-1}\begin{array}{ccc}
        Z & & Z  \\
        & X_{l+\frac{1}{2},j+\frac{1}{2}} & \\
        Z & & Z  
    \end{array}\nonumber\\
    &\hspace{2cm}\red{-\begin{array}{ccc}
        Z & & Z \\
         &X_{l+\frac{1}{2},k} &\\
         &X_{l+\frac{1}{2},k-1} & \\
         Z& & Z
    \end{array}-\begin{array}{ccc}
        Z & & Z \\
         &X_{L_x-\frac{1}{2},k} &\\
         &X_{L_x-\frac{1}{2},k-1} & \\
         Z& & Z
    \end{array}\,.}
\end{align}
\end{widetext}
Let us denote a generic ground state of the Hamiltonian $\mathrm{H}_{\text{2D-cluster}\rvert \text{blue}^{x;k}}^{'}$ by $|\tilde{\Psi}\rangle$. Then, we have
\begin{subequations}
    \begin{align}
    &\eta^x_{r,j}\ket{\tilde{\Psi}}=\ket{\tilde{\Psi}}\, ,\qquad \eta^y_{r,i}\ket{\tilde{\Psi}}=\ket{\tilde{\tilde{\Psi}}}\, ,\\
    &\eta^y_{b,i}\ket{\tilde{\Psi}}=\ket{\tilde{\Psi}}, \text{ for } i\neq L_x-1,l\, ,\\
    &\eta^x_{b,j}\ket{\tilde{\Psi}}=\ket{\tilde{\Psi}}\, ,\text{ for } j\neq k,k-1\\
    &\eta^x_{b,k}\ket{\tilde{\Psi}}=\eta^{x(L)}_{b,k\red{-1}}\eta^{x(R)}_{b,k\red{-1}}\ket{\tilde{\Psi}}\, ,\\
    &\eta^x_{b,k-1}\ket{\tilde{\Psi}}=\eta^{x(L)}_{b,k}\eta^{x(R)}_{b,k}\ket{\tilde{\Psi}}\, .
\end{align}
\end{subequations}
\\
Let us define \red{the following operators for $j=k,k-1$}
\begin{subequations}
\begin{align}
    &Z^{R}_{j}=Z_{L_x-\frac{1}{2},j+\frac{1}{2}}\, , \quad Z^{L}_{j}=Z_{l+\frac{1}{2},j+\frac{1}{2}}\\
    &X^{R}_{j}=\eta^{x(R)}_{b,j}\, , \quad X^{L}_{j}=\eta^{x(L)}_{b,j}\,.
\end{align}
\end{subequations}
These localized operators on the left and right interface lines anti-commute
\begin{align}
    \{Z_j^L,X_j^L\}=0\, ,\quad \{Z_j^R,X_j^R\}=0\, .
\end{align}
Hence, they form an operator basis on the ground space.

We note that the following localized operators anti-commute with $\mathbf{D}^{(2)}$ on the ground space of $\mathrm{H}_{\text{2D-cluster}\rvert \text{blue}^{x;k}}^{'}$.
\begin{widetext}
\begin{subequations}
    \begin{align}
    &\mathbf{D}^{(2)}\eta^{x(L)}_{b,k-1}\ket{\tilde{\Psi}}=-\eta^{x(L)}_{b,k-1}\mathbf{D}^{(2)}\ket{\tilde{\Psi}}\, \qquad \mathbf{D}^{(2)}\eta^{x(R)}_{b,k-1}\ket{\tilde{\Psi}}=-\eta^{x(R)}_{b,k-1}\mathbf{D}^{(2)}\ket{\tilde{\Psi}}\\
    &\mathbf{D}^{(2)}\eta^{x(L)}_{b,k}\ket{\tilde{\Psi}}=-\eta^{x(L)}_{b,k}\mathbf{D}^{(2)}\ket{\tilde{\Psi}}\, \qquad \mathbf{D}^{(2)}\eta^{x(R)}_{b,k}\ket{\tilde{\Psi}}=-\eta^{x(R)}_{b,k}\mathbf{D}^{(2)}\ket{\tilde{\Psi}}\, .
\end{align}
\label{eq:D2etaprojective}
\end{subequations}
On the other hand, 
\begin{align}
    \mathbf{D}^{(2)}\begin{array}{cc}
       Z^{L}_k  & Z^{R}_k \\
        Z^{L}_{k-1} & Z^{R}_{k-1}
    \end{array}\ket{\tilde{\Psi}}=\begin{array}{cc}
       Z^{L}_k  & Z^{R}_k \\
        Z^{L}_{k-1} & Z^{R}_{k-1}
    \end{array}\mathbf{D}^{(2)}\ket{\tilde{\Psi}}\, .
    \label{eq:D2onfourZ}
\end{align}
We note that 
\begin{align}
    (\mathbf{D}^{(2)})^2\ket{\tilde{\Psi}}&=\prod_{j=1}^{L_y}\frac{(1+\eta^x_{r,j})}{2}\prod_{i=1}^{L_x}\frac{(1+\eta^y_{r,i})}{2}\prod_{j=1}^{L_y}\frac{(1+\eta^x_{b,j})}{2}\prod_{i=1}^{L_x}\frac{(1+\eta^y_{b,i})}{2}\ket{\tilde{\Psi}}\nonumber\\
    &=\frac{(1+\eta^{x(L)}_{b,k}\eta^{x(R)}_{b,k})}{2}\frac{(1+\eta^{x(L)}_{b,k-1}\eta^{x(R)}_{b,k-1})}{2}\frac{(1+\eta^{x(L)}_{b,k}\eta^{x(L)}_{b,k-1})}{2}\frac{(1+\eta^{x(R)}_{b,k}\eta^{x(R)}_{b,k-1})}{2}\ket{\tilde{\Psi}}\nonumber\\
    &=\frac{(1+X^{L}_{k}X^{R}_{k})}{2}\frac{(1+X^{L}_{k-1}X^{R}_{k-1})}{2}\frac{(1+X^{L}_{k}X^{L}_{k-1})}{2}\frac{(1+X^{R}_{k}X^{R}_{k-1})}{2}\ket{\tilde{\Psi}}\, .
    \label{eq:D2^2ontildepsi}
\end{align}
From the properties~\eqref{eq:D2etaprojective}, \eqref{eq:D2onfourZ}, and \eqref{eq:D2^2ontildepsi} of $\mathbf{D}^{(2)}$, we conclude
\begin{subequations}
   \begin{align}
    \mathbf{D}^{(2)}\ket{\tilde{\Psi}}&\sim Z^{L}_kZ^{R}_k Z^{L}_{k-1}Z^{R}_{k-1}\frac{(X^{L}_{k}+X^{R}_{k})}{2}\frac{(X^{L}_{k-1}+X^{R}_{k-1})}{2}\frac{(X^{L}_{k}+X^{L}_{k-1})}{2}\frac{(X^{R}_{k}+X^{R}_{k-1})}{2}\ket{\tilde{\Psi}}\\
    \begin{split}
            &=\frac{1}{8}\left(Z^{R}_kZ^{L}_kZ^{R}_{k-1}Z^{L}_{k-1}+\mathrm{D}^{L}_k\mathrm{D}^{R}_{k}Z^{L}_{k-1}Z^{R}_{k-1}+\mathrm{D}^{L}_k\mathrm{D}^{R}_{k-1}Z^{R}_{k}Z^{L}_{k-1}+\mathrm{D}^{R}_{k-1}\mathrm{D}^{R}_kZ^{L}_{k-1}Z^{L}_k\right.\\
        &\left.+\mathrm{D}^{L}_k\mathrm{D}^{L}_{k-1}Z^{R}_kZ^{R}_{k-1}+\mathrm{D}^{L}_{k-1}\mathrm{D}^{R}_kZ^{R}_{k-1}Z^{L}_k+\mathrm{D}^{R}_{k-1}\mathrm{D}^{L}_{k-1}Z^{L}_kZ^{R}_k+\mathrm{D}^{L}_k\mathrm{D}^{R}_k\mathrm{D}^{L}_{k-1}\mathrm{D}^{R}_{k-1}\right)\ket{\tilde{\Psi}}
        \end{split}
\end{align} 
\end{subequations}
where
\begin{align}
\mathrm{D}^{L}_k=Z^{L}_k X_k^{L}\, ,\quad \mathrm{D}^{R}_k=Z^{R}_k X_k^{R}\, ,\quad \mathrm{D}^{L}_{k-1}=Z^{L}_{k-1}X_{k-1}^{L}\, ,\quad \mathrm{D}^{R}_{k-1}=Z^{R}_{k-1}X_{k-1}^{R}\, .
\end{align}
\red{We note that on the groundspace, $X_k^L=X_{k-1}^L$ and $X_k^R=X_{k-1}^R$. So effectively }we have the projective algebra
\begin{subequations}
\red{
    \begin{align}
        &\{\mathrm{D}^{L}_k, Z^{L}_k\}=0\, ,\quad \{\mathrm{D}^{L}_k,X^{L}_k\}=0\, ,\quad \{Z^{L}_k,X_{k}^{L}\}=0\, ,\\
        &\{\mathrm{D}^{R}_k, Z^{R}_k\}=0\, ,\quad \{\mathrm{D}^{R}_k,X_k^{R}\}=0\, ,\quad \{Z^{R}_k,X_k^{R}\}=0\, .
    \end{align}}
\end{subequations}
 The corresponding edge modes cannot be gapped out and distinguish between $\rm H_{\text{2D-cluster}}$ and $\rm H_{\text{blue}}^{x;k}$. See Figure~\ref{fig:lineinterfacemode} for an illustration.
\begin{figure*}
    \centering
    \includegraphics[scale=1]{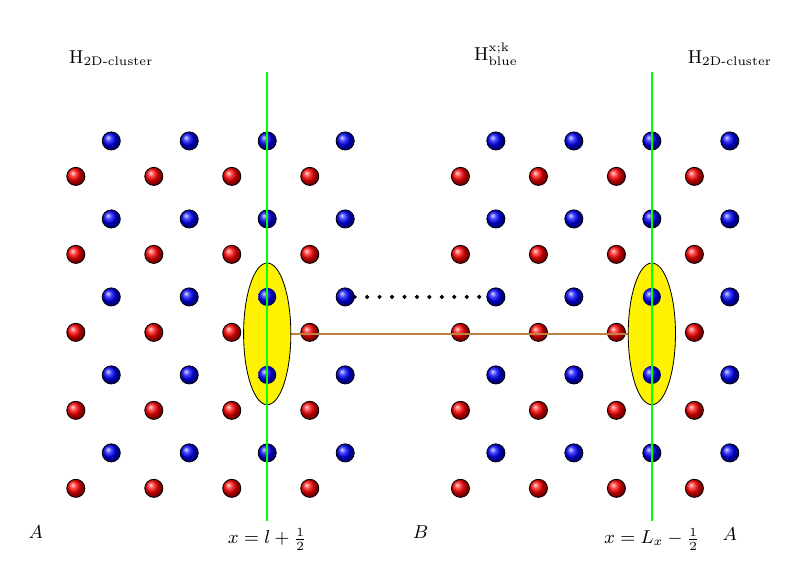}
    \caption{Interface between $\rm H_{\text{2D-cluster}}$ and $\rm H_{\text{blue}}^{x;k}$. The $x=k$ line is colored brown inside region $B$ where $\rm H_{\text{blue}}^{x;k}$ is placed. The green lines indicate the interface lines, and the yellow ellipses are the edge modes that can not be gapped out.}
    \label{fig:lineinterfacemode}
\end{figure*}
\end{widetext}
\section{Interface between two distinct noninvertible SSPTs in 2D: edge modes}\label{sec:interfaceanalysisotherspt}
\subsection{Line interface between \texorpdfstring{$\mathrm{H}_{\text{2D-}\widetilde{\text{clstr}}}$}{Lg} and \texorpdfstring{$ \mathrm{H}_{\text{2D-}\widetilde{\text{clstr}}}^{x}$}{Lg}}
Let us consider the line interface between the two Hamiltonians. We take the interface line to be along $x=l+\frac{1}{2}$ and $x=L_x-\frac{1}{2}$ for some $l\neq L_x-1\in \mathbb{Z}_{L_x}$ such that $L_x$ is even, $L_x-l-1$ to be a multiple of four (this choice is made so that the horizontal non-local order parameter is nonzero) and $L_y=4k+2$. We consider the regions $A$ and $B$ defined in \eqref{eq:regionABline} with $\mathrm{H}_{\text{2D-}\widetilde{\text{clstr}}}^{x}$ defined in \eqref{eq:2dtildeclstrx} in the region $B$ and $\mathrm{H}_{\text{2D-}\widetilde{\text{clstr}}}$ defined in \eqref{eq:2dtildeclstr} in region $A$.
\begin{align}
    \mathrm{H}_{\text{2D-}\widetilde{\text{clstr}}|\text{2D-}\widetilde{\text{clstr}}^x}=\mathrm{H}_{\text{2D-}\widetilde{\text{clstr}}}\rvert_{A}+\mathrm{H}_{\text{2D-}\widetilde{\text{clstr}}}^x\rvert_{B}
\end{align}
We note that there are no terms in the Hamiltonian supported on both regions $A$ and $B$. Let $\ket{\Psi}$ be a ground state. Then we find
\begin{widetext}
\begin{subequations}
    \begin{align}
    &\eta^x_{r,j}\ket{\Psi}=\ket{\Psi}\quad \forall j\in\mathbb{Z}_{L_y}\, ,\qquad \eta^y_{r,i}\ket{\Psi}=\ket{\Psi}\quad \forall i\in\mathbb{Z}_{L_x}\, ,\qquad
    \eta^y_{b,i}\ket{\Psi}=\ket{\Psi}, \text{ for } i\neq L_x-1,l,i\in\mathbb{Z}_{L_x}\quad\, ,\\
    &\qquad\eta^x_{b,j}\ket{\Psi}=\eta^{x(L)}_{b,j}\eta^{x(R)}_{b,j}\ket{\Psi}\text{ where } \eta^{x(L)}_{b,j}=\begin{array}{ccc}
        Y & & Z  \\
        & X_{l+\frac{1}{2},j+\frac{1}{2}} & \\
        Z & & Y  
    \end{array}\text{ and    }\eta^{x(R)}_{b,j}=\begin{array}{ccc}
        Y & & Z  \\
        & X_{L_x-\frac{1}{2},j+\frac{1}{2}} & \\
        Z & & Y  
    \end{array}\, .
    \end{align}
\end{subequations}
\end{widetext}
Let us define 
\begin{subequations}
\begin{align}
    &Z^{R}_{j}=Z_{L_x-\frac{1}{2},j+\frac{1}{2}}\, , \quad Z^{L}_{j}=Z_{l+\frac{1}{2},j+\frac{1}{2}}\\
    &X^{R}_{j}=\eta^{x(R)}_{b,j}\, , \quad X^{L}_{j}=\eta^{x(L)}_{b,j}\,.
\end{align}
\end{subequations}
These localized operators on the left and right interface lines anti-commute
\begin{align}
    \{Z_j^L,X_j^L\}=0\, ,\quad \{Z_j^R,X_j^R\}=0\, .
\end{align}
Hence, they form an operator basis on the ground space.
The operators $X^{L}_{j}$ and $X^{R}_{j}$ anti-commute with $\mathbf{D}^{(2)}$ on the ground space
\begin{align}
    \{X^{L}_{j},\mathbf{D}^{(2)}\}|\Psi\rangle=0\, ,\quad \{X^{R}_{j},\mathbf{D}^{(2)}\}|\Psi\rangle=0\, .
    \label{eq:XjD2anticommuteotherSPT}
\end{align}
This is because the first terms in \eqref{eq:2dtildeclstr} and \eqref{eq:2dtildeclstrx} have opposite signs. On the other hand, 
\begin{align}
    \mathbf{D}^{(2)}\begin{array}{cc}
       Z^{L}_j  & Z^{R}_j \\
        Z^{L}_{j-1} & Z^{R}_{j-1}
    \end{array}\ket{\Psi}=\begin{array}{cc}
       Z^{L}_j  & Z^{R}_j \\
        Z^{L}_{j-1} & Z^{R}_{j-1}
    \end{array}\mathbf{D}^{(2)}\ket{\Psi}\, .
    \label{eq:D2onfourZotherSPT}
\end{align}
\red{We note that although the individual $X_j^L$ and $X_j^R$ anti-commute with $\mathbf{D}^{(2)}$, their products, for example, $X_j^LX_{j-1}^L$ commute with $\mathbf{D}^{(2)}$. Hence, we can add such terms to the interface Hamiltonian to lift the degeneracy}
\begin{widetext}
\red{
\begin{align}
    \mathrm{H}'_{\text{2D-}\widetilde{\text{clstr}}|\text{2D-}\widetilde{\text{clstr}}^x}&=\mathrm{H}_{\text{2D-}\widetilde{\text{clstr}}|\text{2D-}\widetilde{\text{clstr}}^x}-\sum_{\substack{j=1\\
    v_b^*=(l+\frac{1}{2},j+\frac{1}{2})}}^{L_y}\begin{array}{ccc}
        Y_{v_r} &  & Z_{v_r}\\
         & X_{v_b^*} &\\
         X_{v_r} & & X_{v_r}\\
         & X_{v_b}& \\
         Z_{v_r} & & Y_{v_r}
    \end{array}-\sum_{\substack{j=1\\
    v_b^*=(L_x-\frac{1}{2},j+\frac{1}{2})}}^{L_y}\begin{array}{ccc}
        Y_{v_r} &  & Z_{v_r}\\
         & X_{v_b^*} &\\
         X_{v_r} & & X_{v_r}\\
         & X_{v_b}& \\
         Z_{v_r} & & Y_{v_r}
    \end{array}\nonumber\\
    \vspace{1cm}\nonumber\\
    &-\sum_{\substack{j=1\\
    v_b^*=(l+\frac{1}{2},j+\frac{1}{2})}}^{L_y}\begin{array}{ccccc}
        Z_{v_b} & & Z_{v_b} & & \\
         & Z_{v_r} & & Z_{v_r} & \\
        & & Y_{v_b^*} & & Z_{v_b}\\
        & & & & \\
       Z_{v_b} & & Y_{v_b} & & \\
        & Z_{v_r} & & Z_{v_r} & \\
         & & Z_{v_b} & & Z_{v_b}\\
    \end{array}-\sum_{\substack{j=1\\
    v_b^*=(L_x-\frac{1}{2},j+\frac{1}{2})}}^{L_y}\begin{array}{ccccc}
        Z_{v_b} & & Z_{v_b} & & \\
         & Z_{v_r} & & Z_{v_r} & \\
        & & Y_{v_b^*} & & Z_{v_b}\\
        & & & & \\
       Z_{v_b} & & Y_{v_b} & & \\
        & Z_{v_r} & & Z_{v_r} & \\
         & & Z_{v_b} & & Z_{v_b}\\
    \end{array}\,,
\end{align}
where the last two summation are added to ensure that the Hamiltonian is symmetric under $\mathbf{D}^{(2)}$.}
\end{widetext}
\color{black}
\begin{widetext}
   \red{Let $\ket{\tilde{\Psi}}$ be ground state of the above Hamiltonian.} We note that 
\red{\begin{align}
    (\mathbf{D}^{(2)})^2\ket{\tilde{\Psi}}&=\prod_{j=1}^{L_y}\frac{(1+\eta^x_{r,j})}{2}\prod_{i=1}^{L_x}\frac{(1+\eta^y_{r,i})}{2}\prod_{j=1}^{L_y}\frac{(1+\eta^x_{b,j})}{2}\prod_{i=1}^{L_x}\frac{(1+\eta^y_{b,i})}{2}\ket{\tilde{\Psi}}\nonumber\\
    &=\prod_{j=1}^{L_y}\frac{(1+\eta^{x(L)}_{b,j}\eta^{x(R)}_{b,j})}{2}\frac{(1+\prod_{j=1}^{L_y}\eta^{x(L)}_{b,j})}{2}\frac{(1+\prod_{j=1}^{L_y}\eta^{x(R)}_{b,j})}{2}\ket{\tilde{\Psi}}\nonumber\\
    &=\prod_{j=1}^{L_y}\frac{(1+X^{L}_{j}X^{R}_{j})}{2}\frac{(1+\prod_{j=1}^{L_y}X^{L}_{j})}{2}\frac{(1+\prod_{j=1}^{L_y}X^{R}_{j})}{2}\ket{\tilde{\Psi}}\, \nonumber\\
    &=\frac{\left(1+\frac{1}{L_y}\sum_{j=1}^{L_y}X_j^LX_j^R\right)}{2}\ket{\tilde{\Psi}}
    \label{eq:D2^2inappendE}
\end{align}}
From \eqref{eq:XjD2anticommuteotherSPT}, \eqref{eq:D2onfourZotherSPT}, and \red{\eqref{eq:D2^2inappendE}}, we deduce
\red{
\begin{align}
    \mathbf{D}^{(2)}\ket{\tilde{\Psi}}\sim\prod_{j=1}^{L_y}Z_j^{L}Z_j^{R}\frac{\left(1+\frac{1}{L_y}\sum_{j=1}^{L_y}X_j^LX_j^R\right)}{2}\ket{\tilde{\Psi}}\,.
\end{align}}
Let us define
\begin{align}
    \mathrm{D}_j^L=Z_j^LX_j^L\,,\qquad\mathrm{D}_j^R=Z_j^RX_j^R\,.
\end{align}
Then
\red{
\begin{align}
    \mathbf{D}^{(2)}\ket{\tilde{\Psi}}\sim \frac{1}{2}\left(\prod_{j=1}^{L_y}Z_j^LZ_j^R+\frac{1}{L_y}\sum_{j=1}^{L_y}D_j^LD_j^R\prod_{k\neq j}Z_k^LZ_k^R\right)\ket{\tilde{\Psi}}\,.
\end{align}}
 We have the projective algebra
\begin{align}
    \{D_j^L,X_j^L\}=0\,,\quad \{D_j^L,Z_j^L\}=0\,,\quad \{Z_j^L,X_j^L\}=0\,,\nonumber\\
    \{D_j^R,X_j^R\}=0\,,\quad \{D_j^R,Z_j^R\}=0\,,\quad \{Z_j^R,X_j^R\}=0\,.
\end{align}
\red{Only one of the $X_j^L$ is independant on the ground space of $\mathrm{H}'_{\text{2D-}\widetilde{\text{clstr}}|\text{2D-}\widetilde{\text{clstr}}^x}$. Hence, effectively we have one projective representation from both $L$ and $R$.}
The corresponding edge modes can not be gapped out and distinguish between $\mathrm{H}_{\text{2D-}\widetilde{\text{clstr}}}$ and $\mathrm{H}_{\text{2D-}\widetilde{\text{clstr}}}^x$.
\end{widetext}
\subsection{Line interface between $\mathrm{H}^{\mathbb{Z}_2}_{\text{2D-SSPT}}$ and $\mathrm{H}^{(2)}_{\text{odd}}$}\label{sec:H2oddandH2SSPT}
Let us again consider the line interface between the two Hamiltonians. We put the interface Hamiltonian on the torus. We take the interface line to be along $x=l$ and $x=L$ for some $l< L-1$ such that $l$ is even and $L$ and $L-l$ are multiples of four. Let us define the region $A$ and $B$ as
\begin{subequations}
\red{
    \begin{align}
    A&=\{(x,y)\in (\mathbb{Z},\mathbb{Z})|l<x<L-1\}\\
    B&=\{(x,y)\in(\mathbb{Z},\mathbb{Z})|x<l\}\,.
\end{align}}
\end{subequations}
We define the interface Hamiltonian
\begin{align}
    \mathrm{H}_{\mathbb{Z}_2-\text{SSPT}|\text{odd}}^{(2)}=\mathrm{H}_{\text{odd}}^{(2)}\rvert_{A}+\mathrm{H}_{\text{2D-SSPT}}^{\mathbb{Z}_2}\rvert_{B}\,.
\end{align}
Let $\ket{\Psi}$ be a groundstate of the interface Hamiltonian. Then
\begin{widetext}
\begin{subequations}
    \begin{align}
    \eta^y_i\ket{\Psi}&=\ket{\Psi}\qquad \forall i\neq l, L \in\mathbb{Z}_L\,,\\
    \eta^x_j\ket{\Psi}&=\eta_j^{x(L)}\eta_j^{x(R)}\ket{\Psi}\,, \quad \text{where } \eta_j^{x(L)}=\begin{array}{ccc}
        & Z & Y \\
       Z & X_{(l,j)} & Z \\
        Z & Z &
    \end{array} \text{ for } j \text{ even, }\eta_j^{x(L)}=\begin{array}{ccc}
        & Z & Z \\
       Z & X_{(l,j)} & Y \\
        Z & Z &
    \end{array} \text{ for } j \text{ odd,} \nonumber\\
    &\hspace{3cm}\eta_j^{x(R)}=\begin{array}{ccc}
        & Z & Z \\
       Z & X_{(L,j)} & Z \\
        Y & Z &
    \end{array} \text{ for } j \text{ even, and }\eta_j^{x(R)}=\begin{array}{ccc}
        & Z & Z \\
       Y & X_{(L,j)} & Z \\
        Z & Z &
    \end{array} \text{ for } j \text{ odd}\,,\\
    &\eta^{\text{diag}}_k\ket{\Psi}=\eta^{\text{diag}(L)}_k\eta^{\text{diag}(R)}_k\ket{\Psi}\, \text{ where }\eta^{\text{diag}(L)}_k=\begin{array}{ccc}
         &  Z & Y \\
        Z & X_{(l,[l+k]_{L})} & Z\\
        Z & Z &
    \end{array}\text{ for } k \text{ even, }\eta_k^{\text{diag}(L)}=\begin{array}{ccc}
        & Z & Z \\
       Z & X_{(l,[l+k]_{L})} & Y \\
        Z & Z &
    \end{array} \text{ for } k \text{ odd,} \nonumber\\
    &\hspace{3cm}\eta_k^{\text{diag}(R)}=\begin{array}{ccc}
        & Z & Z \\
       Z & X_{(L,[L+k]_{L})} & Z \\
        Y & Z &
    \end{array} \text{ for } k \text{ even, and }\eta_k^{\text{diag}(R)}=\begin{array}{ccc}
        & Z & Z \\
       Y & X_{(L,[L+k]_{L})} & Z \\
        Z & Z &
    \end{array} \text{ for } k \text{ odd}\,.
\end{align}
\end{subequations}
\end{widetext}
Let us define 
\begin{align}
    \eta^{x(L)}_j=\eta^{\text{diag}(L)}_{[l+j]}=X_j^L\,,\quad \eta^{x(R)}_j=\eta^{\text{diag}(R)}_{[L+j]}=X_j^R\,.
\end{align}
The operators $X_j^L$ and $X_j^R$ anticommute with $\mathbf{D}_{\text{DPIM}}^{(2)}$
\begin{align}
    \{\mathbf{D}_{\text{DPIM}}^{(2)},X_j^L\}\ket{\Psi}=0\,,\quad \{\mathbf{D}_{\text{DPIM}}^{(2)},X_j^R\}\ket{\Psi}=0\,.
    \label{eq:D2DPIMXanticommute}
\end{align}
We note that
\begin{widetext}
\begin{align}
    (\mathbf{D}_{\text{DPIM}}^{(2)})^2\ket{\Psi}&\propto\prod_{j=1}^L\frac{(1+\eta^x_j)}{2}\prod_{i=1}^L\frac{(1+\eta^y_i)}{2}\prod_{k=1}^{L}\frac{(1+\eta^{\text{diag}}_k)}{2}\ket{\Psi}\nonumber\\
    &=\prod_{j=1}^L\frac{(1+\eta^{x(L)}_j\eta^{x(R)}_j)}{2}\prod_{k=1}^L\frac{(1+\eta^{\text{diag}(L)}_{k}\eta^{\text{diag}(R)}_{k})}{2}\frac{(1+\prod_j\eta_j^{x(L)})}{2}\frac{(1+\prod_j\eta_j^{x(R)})}{2}\ket{\Psi}\nonumber\\
    &=\prod_j\frac{(1+X_j^{L}X_j^{R})}{2}\prod_k\frac{(1+X_{[l+k]}^{L}X_{[L+k]}^{R})}{2}\frac{(1+\prod_jX_j^{L})}{2}\frac{(1+\prod_jX_j^{R})}{2}\ket{\Psi}\,.
\end{align}
\label{eq:D2DPIM^2}
Let $\mathcal{A}$ be a subset of the set of coordinates of the vertices on the left interface line $\mathcal{L}=\{(l,j)|1\leq j\leq L\}$. Similarly $\mathcal{B}$ be a set of coordinates of the vertices on the right interface line $\mathcal{R}=\{(L,j)|1\leq j\leq L\}$. From \ref{eq:D2DPIMXanticommute} and \ref{eq:D2DPIM^2}, we deduce that on the ground space
\begin{align}
    \mathbf{D}_{\text{DPIM}}^{(2)}\sim \prod_jZ_j^LZ_j^R\left(\sum_{\mathcal{A}\subset\mathcal{L},\mathcal{B}\subset\mathcal{R}}\alpha_{\mathcal{A}\mathcal{B}}\prod_{(l,j)\in \mathcal{A}}X_j^{L}\prod_{(L,k)\in\mathcal{B}}X_k^{R}\right)\,.
    \label{eq:DDPIM2onGS}
\end{align}
Let us define 
\begin{align}
D_j^{L}=Z_j^{L}X_j^{L}\,,\qquad D_j^{R}=Z_j^RX_j^R\,.
\end{align}
We can rewrite \eqref{eq:DDPIM2onGS} as
    \begin{align}
    \mathbf{D}_{\text{DPIM}}^{(2)}\sim\sum_{\mathcal{A}\subset\mathcal{L},\mathcal{B}\subset\mathcal{R}}\alpha_{\mathcal{A}\mathcal{B}}\prod_{(l,j)\in \mathcal{L}\setminus\mathcal{A}}Z_j^{L}\prod_{(L,k)\in \mathcal{R}\setminus\mathcal{B}}Z_k^{R}\prod_{(l,j)\in \mathcal{A}}D_j^{L}\prod_{(L,k)\in \mathcal{B}}D_k^{R}\,.
\end{align}
\end{widetext}
We find the projective representation
\begin{subequations}
    \begin{align}
    \{D_j^{L},X_j^{L}\}=0\,, \quad \{X_j^{L},Z_j^{L}\}=0\,,\quad \{D_j^{L},Z_j^{L}\}=0\,,\\
    \{D_j^{R},X_j^{R}\}=0\,, \quad \{X_j^{R},Z_j^{R}\}=0\,,\quad \{D_j^{R},Z_j^{R}\}=0\,.
\end{align}
\label{eq:projectivealgD^2DPIM}
\end{subequations}
\red{We note that we can add terms to the interface Hamiltonian $\mathrm{H}_{\mathbb{Z}_2-\text{SSPT}|\text{odd}}^{(2)}$ that are symmetric and commuting with the interface Hamiltonian $\mathrm{H}_{\mathbb{Z}_2-\text{SSPT}|\text{odd}}^{(2)}$. These terms are the products of the form $X_j^{L}X_{j+1}^L$ and $X_j^{R}X_{j+1}^R$ for all $j$.
\begin{widetext}
\begin{align}
    \tilde{\mathrm{H}}_{\mathbb{Z}_2-\text{SSPT}|\text{odd}}^{(2)}=\mathrm{H}_{\mathbb{Z}_2-\text{SSPT}|\text{odd}}^{(2)}-\sum_{j=1}^L \left(X_j^LX_{j+1}^L+ X_j^RX_{j+1}^R\right)-\sum_{j=1}^L\mathbf{D}_{\text{DPIM}}^{(2)}\text{conj}\left[X_j^LX_{j+1}^L+X_j^RX_{j+1}^R\right]
\end{align}
where $\mathbf{D}_{\text{DPIM}}^{(2)}\text{conj}\left[\cdot\right]$ is the $\mathbf{D}_{\text{DPIM}}^{(2)}$ symmetric counterpart of the term inside the argument. 
If $\ket{\tilde{\Psi}}$ is a ground state of the Hamiltonian $\tilde{\mathrm{H}}_{\mathbb{Z}_2-\text{SSPT}|\text{odd}}^{(2)}$, then 
\begin{subequations}
\begin{align}
   (\mathbf{D}_{\text{DPIM}}^{(2)})^2\ket{\tilde{\Psi}}&\sim \frac{1}{2}\left(1+\frac{1}{2L}\sum_{j=1}^LX_j^{L}X_j^{R}+\frac{1}{2L}\sum_{j=1}^LX_{[l+k]}^{L}X_{[L+k]}^{R}\right) \ket{\tilde{\Psi}}\,\\
   \mathbf{D}_{\text{DPIM}}^{(2)}\ket{\tilde{\Psi}}&\sim \prod_{j=1}^LZ_j^LZ_j^R\frac{1}{2}\left(1+\frac{1}{2L}\sum_{j=1}^LX_j^{L}X_j^{R}+\frac{1}{2L}\sum_{j=1}^LX_{[l+k]}^{L}X_{[L+k]}^{R}\right) \ket{\tilde{\Psi}}\,.
\end{align}
\end{subequations}
\end{widetext}
We note that on the groundspace of $\tilde{\mathrm{H}}_{\text{SSPT}|\text{odd}}^{(2)}$ only one projective algebra each for left interface and right interface in \eqref{eq:projectivealgD^2DPIM} is independant.
}
Hence, the ground state degeneracy is $2^{2}=4$ and the two Hamiltonians $\mathrm{H}^{\mathbb{Z}_2}_{\text{2D-SSPT}}$ and $\mathrm{H}^{(2)}_{\text{odd}}$ differ as a first-order SSPT protected by noninvertible symmetry.
\section{Interface between two distinct noninvertible SSPTs in 3D: Hinge modes}\label{sec:Interfaceanalysis3D}
We consider a cubic interface placed on the 3-torus between the Hamiltonians $\mathrm{H}_{\text{3D-cluster}}^{\mathcal{G}}$ and $\mathrm{H}_{\text{blue}}^{(3)\mathcal{G}}$.
We choose the interface surface to lie along the planes in the blue sublattice with corners $(i+\frac{1}{2},j+\frac{1}{2},k+\frac{1}{2})$ with $i\in\{i_0,i_1\}$, $j\in\{j_0,j_1\}$ and $k\in\{k_0,k_1\}$. \red{We take $i_0$, $i_1$, $j_0$, $j_1$, $k_0$, $k_1$ $\in 4\mathbb{Z}$}. 
Now, let us define the following regions
\begin{widetext}
\begin{subequations}
\begin{align}
    A&=\{(x,y,z)\in(\mathbb{Z}/2,\mathbb{Z}/2,\mathbb{Z}/2)|x\leq i_0\}\cup \{(x,y,z)\in(\mathbb{Z}/2,\mathbb{Z}/2,\mathbb{Z}/2)|x> i_1+\frac{1}{2}\}\cup \{(x,y,z)\in(\mathbb{Z}/2,\mathbb{Z}/2,\mathbb{Z}/2)|y\leq j_0\}\nonumber\\
    & \cup \{(x,y,z)\in(\mathbb{Z}/2,\mathbb{Z}/2,\mathbb{Z}/2)|y> j_1+\frac{1}{2}\}\cup \{(x,y,z)\in(\mathbb{Z}/2,\mathbb{Z}/2,\mathbb{Z}/2)|z\leq k_0\}\cup \{(x,y,z)\in(\mathbb{Z}/2,\mathbb{Z}/2,\mathbb{Z}/2)|z> k_1+\frac{1}{2}\}\, ,\\
    B&=\{(x,y,z)\in(\mathbb{Z}/2,\mathbb{Z}/2,\mathbb{Z}/2)|i_0+\frac{1}{2}<x<i_1+\frac{1}{2}\, ,j_o+\frac{1}{2}<y<j_1+\frac{1}{2}\, ,k_0+\frac{1}{2}<z<k_1+\frac{1}{2}\}\, .
    \end{align}
    \end{subequations}
    Now, we define the boundary surfaces
    \begin{subequations}
    \begin{align}  
    S_T&=\{(x,y,z)\in(\mathbb{Z}+\frac{1}{2},\mathbb{Z}+\frac{1}{2},\mathbb{Z}+\frac{1}{2})|z=k_1+\frac{1}{2}\,, i_0+\frac{1}{2}<x<i_1+\frac{1}{2}\,, j_0+\frac{1}{2}<y<j_1+\frac{1}{2}\}\, ,\\
    S_B&=\{(x,y,z)\in(\mathbb{Z}+\frac{1}{2},\mathbb{Z}+\frac{1}{2},\mathbb{Z}+\frac{1}{2})|z=k_0+\frac{1}{2}\,, i_0+\frac{1}{2}<x<i_1+\frac{1}{2}\,, j_0+\frac{1}{2}<y<j_1+\frac{1}{2}\}\, ,\\
    S_N&=\{(x,y,z)\in(\mathbb{Z}+\frac{1}{2},\mathbb{Z}+\frac{1}{2},\mathbb{Z}+\frac{1}{2})|y=j_1+\frac{1}{2}\,, i_0+\frac{1}{2}<x<i_1+\frac{1}{2}\,, k_0+\frac{1}{2}<z<k_1+\frac{1}{2}\}\, ,\\
    S_S&=\{(x,y,z)\in(\mathbb{Z}+\frac{1}{2},\mathbb{Z}+\frac{1}{2},\mathbb{Z}+\frac{1}{2})|y=j_0+\frac{1}{2}\,, i_0+\frac{1}{2}<x<i_1+\frac{1}{2}\,, k_0+\frac{1}{2}<z<k_1+\frac{1}{2}\}\, ,\\
    S_R&=\{(x,y,z)\in(\mathbb{Z}+\frac{1}{2},\mathbb{Z}+\frac{1}{2},\mathbb{Z}+\frac{1}{2})|x=i_1+\frac{1}{2}\,, j_0+\frac{1}{2}<y<j_1+\frac{1}{2}\,, k_0+\frac{1}{2}<z<k_1+\frac{1}{2}\}\, ,\\
    S_L&=\{(x,y,z)\in(\mathbb{Z}+\frac{1}{2},\mathbb{Z}+\frac{1}{2},\mathbb{Z}+\frac{1}{2})|x=i_0+\frac{1}{2}\,, j_0+\frac{1}{2}<y<j_1+\frac{1}{2}\,, k_0+\frac{1}{2}<z<k_1+\frac{1}{2}\}\, .
\end{align}
\end{subequations}
Their union is the whole boundary surface
\begin{align}
    S=S_T\cup S_B\cup S_N\cup S_S\cup S_L\cup S_R\,.
\end{align}
Consider the Hamiltonian $\mathrm{H}_{\text{blue}}^{(3)\mathcal{G}}$ inside the cubic interface in region $B$ and $\mathrm{H}_{\text{3D-cluster}}^{\mathcal{G}}$ outside the cubic interface in region $A$. In this interface Hamiltonian, we do not have any term that is supported both inside and outside of the cubic interface. We keep all the terms to have support entirely inside the cubic interface or outside the cubic interface, where inside and outside also include their boundary. 
\begin{align}
    \mathrm{H}_{\text{3D-cluster}|\text{blue}}^{\mathcal{G}}&=-\sum_{v_r\in A\cap \mathcal{G}}X_{v_r}\prod_{v_b}Z_{v_b}^{\sigma_{v_rv_b}}-\sum_{v_b\in A\cap \mathcal{G}}X_{v_b}\prod_{v_r}Z_{v_r}^{\sigma_{v_rv_b}}+\sum_{v_r\in B\cap \mathcal{G}}X_{v_r}\prod_{v_b}Z_{v_b}^{\sigma_{v_rv_b}}\nonumber\\
    &\hspace{3cm}-\sum_{v_b\in B\cap \mathcal{G}}X_{v_b}\prod_{v_r}Y_{v_r}^{\sigma_{v_bv_r}}-\sum_{v_b\in B\cap \mathcal{G}}X_{v_b}\prod_{v_r}Z_{v_r}^{\sigma_{v_bv_r}}\left(\prod_{v'_b}Z_{v'_b}^{\sigma_{v_rv'_b}}\right)\,.
\end{align}
It is a straightforward exercise to see that we can add terms in the interface Hamiltonian along the interfacial surface everywhere except at hinges that respect \red{invertible} and noninvertible symmetries and commute with each term in the Hamiltonian. Explicitly, it is
\begin{align}
    \tilde{\mathrm{H}}_{\text{3D-cluster}|\text{blue}}^{\mathcal{G}}&=\mathrm{H}_{\text{3D-cluster}|\text{blue}}^{\mathcal{G}}-\sum_{v_b\in S\cap \mathcal{G}}X_{v_b}\prod_{v_r\in A\cap\mathcal{G}}Z_{v_r}^{\sigma_{v_rv_b}}\prod_{v_r\in B\cap\mathcal{G}}Y_{v_r}^{\sigma_{v_rv_b}}\nonumber\\
    &+\sum_{v_b\in S\cap \mathcal{G}}X_{v_b}\prod_{v_r\in A\cap\mathcal{G}}Z_{v_r}^{\sigma_{v_rv_b}}\prod_{v_r\in B\cap\mathcal{G}}\left(Z_{v_r}^{\sigma_{v_rv_b}}\prod_{v_b'}Z_{v_b'}^{\sigma_{v_rv_b'}}\right)-\sum_{\substack{v_b=(i_0+\frac{1}{2},j_0+\frac{1}{2},k_0+\frac{1}{2})\\
    v_b=(i_1+\frac{1}{2},j_0+\frac{1}{2},k_1+\frac{1}{2})\\
    v_b=(i_1+\frac{1}{2},j_1+\frac{1}{2},k_0+\frac{1}{2})\\
    v_b=(i_0+\frac{1}{2},j_1+\frac{1}{2},k_1+\frac{1}{2})}}X_{v_b}\prod_{v_r}Z_{v_r}^{\sigma_{v_bv_r}}\, .
    \label{eq:interfaceham3D}
\end{align}
\end{widetext}
Now, let us call the ground state of this interface Hamiltonian \eqref{eq:interfaceham3D}  with terms added along the interface, except at hinges, to be $\ket{\tilde{\Psi}}$. We find that 
\begin{align}
    &\mathcal{P}^{z,r}_{xy}\ket{\tilde{\Psi}}=\ket{\tilde{\Psi}}\, ,\quad \mathcal{P}^{y,r}_{xz}\ket{\tilde{\Psi}}=\ket{\tilde{\Psi}}\, ,\nonumber\\
    &\qquad \mathcal{P}^{x,r}_{yz}\ket{\tilde{\Psi}}=\ket{\tilde{\Psi}}\, ,
    \label{eq:subsystem_planarred}
\end{align}
and 
\begin{align}
    \mathcal{P}_r\ket{\tilde{\Psi}}=\ket{\tilde{\Psi}}\,.
\end{align}
\red{The other invertible symmetries $\eta_a$ and $\eta_{\hat{a}}$ also must satisfy
\begin{align}
    \eta_a\ket{\tilde{\Psi}}=\ket{\tilde{\Psi}}\,,\quad \eta_{\hat{a}}\ket{\tilde{\Psi}}=\ket{\tilde{\Psi}}\,.
\end{align}}
\begin{widetext}
Let us define the following operators along the hinge, with the $x$, $y$, and $z$ directions the same as in Figure~\ref{fig:FCC}:
\begin{align}
    \mathcal{X}_{(x,y,z)}=\begin{cases}
        \raisebox{-25pt}{\begin{tikzpicture}
        \node at (0.1,0.1) {$\color{blue}{X_{v}}$};
        \node at (1,0) {$\color{red}{Y}$};
        \node at (0,1) {$\color{red}{Z}$};
        \node at (-0.5,-0.5) {$\color{red}{Z}$};
        \node at (0.5,0.5) {$\color{red}{Z}$};
        \draw[thick] (0,1)--(1,1)--(0.5,0.5)--(-0.5,0.5)--(0,1);
        \draw[thick] (0.5,0.5)--(0.5,-0.5)--(-0.5,-0.5)--(-0.5,0.5);
        \draw[thick] (0.5,0.5)--(0.5,-0.5)--(1,0)--(1,1); 
        \end{tikzpicture}}\qquad\text{ for }&(x,y,z)=(i_0+\frac{1}{2},j_0+\frac{1}{2},k+\frac{1}{2})\, , k_0<k\leq k_1 \text{ and }\\
        &(x,y,z)=(i_0+\frac{1}{2},j+\frac{1}{2},k_1+\frac{1}{2})\, , j_0\leq j<j_1\text{ and }\\
        &(x,y,z)=(i+\frac{1}{2},j_0+\frac{1}{2},k_1+\frac{1}{2})\, , i_0\leq i<i_1\, ,\\
        \raisebox{-25pt}{\begin{tikzpicture}
        \node at (0.1,0.1) {$\color{blue}{X_{v}}$};
        \node at (1,0) {$\color{red}{Z}$};
        \node at (0,1) {$\color{red}{Z}$};
        \node at (-0.5,-0.5) {$\color{red}{Z}$};
        \node at (0.5,0.5) {$\color{red}{Y}$};
        \draw[thick] (0,1)--(1,1)--(0.5,0.5)--(-0.5,0.5)--(0,1);
        \draw[thick] (0.5,0.5)--(0.5,-0.5)--(-0.5,-0.5)--(-0.5,0.5);
        \draw[thick] (0.5,0.5)--(0.5,-0.5)--(1,0)--(1,1); 
        \end{tikzpicture}}\qquad\text{ for }&(x,y,z)=(i_0+\frac{1}{2},j_1+\frac{1}{2},k+\frac{1}{2})\, , k_0\leq k< k_1 \text{ and }\\
        &(x,y,z)=(i_0+\frac{1}{2},j+\frac{1}{2},k_0+\frac{1}{2})\, , j_0< j\leq j_1\text{ and }\\
        &(x,y,z)=(i+\frac{1}{2},j_1+\frac{1}{2},k_0+\frac{1}{2})\, , i_0\leq i<i_1\, ,\\
        \raisebox{-25pt}{\begin{tikzpicture}
        \node at (0.1,0.1) {$\color{blue}{X_{v}}$};
        \node at (1,0) {$\color{red}{Z}$};
        \node at (0,1) {$\color{red}{Y}$};
        \node at (-0.5,-0.5) {$\color{red}{Z}$};
        \node at (0.5,0.5) {$\color{red}{Z}$};
        \draw[thick] (0,1)--(1,1)--(0.5,0.5)--(-0.5,0.5)--(0,1);
        \draw[thick] (0.5,0.5)--(0.5,-0.5)--(-0.5,-0.5)--(-0.5,0.5);
        \draw[thick] (0.5,0.5)--(0.5,-0.5)--(1,0)--(1,1); 
        \end{tikzpicture}}\qquad\text{ for }&(x,y,z)=(i_1+\frac{1}{2},j_0+\frac{1}{2},k+\frac{1}{2})\, , k_0\leq k< k_1 \text{ and }\\
        &(x,y,z)=(i_1+\frac{1}{2},j+\frac{1}{2},k_0+\frac{1}{2})\, , j_0\leq j< j_1\text{ and }\\
        &(x,y,z)=(i+\frac{1}{2},j_0+\frac{1}{2},k_0+\frac{1}{2})\, , i_0< i\leq i_1\, ,\\
       \raisebox{-25pt}{\begin{tikzpicture}
        \node at (0.1,0.1) {$\color{blue}{X_{v}}$};
        \node at (1,0) {$\color{red}{Z}$};
        \node at (0,1) {$\color{red}{Z}$};
        \node at (-0.5,-0.5) {$\color{red}{Y}$};
        \node at (0.5,0.5) {$\color{red}{Z}$};
        \draw[thick] (0,1)--(1,1)--(0.5,0.5)--(-0.5,0.5)--(0,1);
        \draw[thick] (0.5,0.5)--(0.5,-0.5)--(-0.5,-0.5)--(-0.5,0.5);
        \draw[thick] (0.5,0.5)--(0.5,-0.5)--(1,0)--(1,1); 
        \end{tikzpicture}}\qquad\text{ for }&(x,y,z)=(i_1+\frac{1}{2},j_1+\frac{1}{2},k+\frac{1}{2})\, , k_0<k\leq k_1 \text{ and }\\
        &(x,y,z)=(i_1+\frac{1}{2},j+\frac{1}{2},k_1+\frac{1}{2})\, , j_0< j\leq j_1\text{ and }\\
        &(x,y,z)=(i+\frac{1}{2},j_1+\frac{1}{2},k_1+\frac{1}{2})\, , i_0< i\leq i_1\, 
    \end{cases}
    \label{eq:cubeoperator}
\end{align}
On the blue sublattice, planar subsystem symmetries act on the ground state as
\begin{subequations}
\begin{align}
    \mathcal{P}_{xy}^{z,b}\ket{\tilde{\Psi}}=\begin{cases}
        \ket{\tilde{\Psi}} \text{ for } z<k_0+\frac{1}{2}\, ,\,  k_1+\frac{1}{2}<z\, \text{ and } z\in2\mathbb{Z}+1\\
        \mathcal{X}_{(i_0+\frac{1}{2},j_0+\frac{1}{2},z)}\mathcal{X}_{(i_1+\frac{1}{2},j_0+\frac{1}{2},z)}\mathcal{X}_{(i_1+\frac{1}{2},j_1+\frac{1}{2},z)}\mathcal{X}_{(i_0+\frac{1}{2},j_1+\frac{1}{2},z)}\ket{\tilde{\Psi}} \text{ for } k_0+\frac{1}{2}<z<k_1+\frac{1}{2}\,\text{ and } z\in2\mathbb{Z} \\
        \prod\limits_{\substack{i=i_0+2\\
        i\in 2\mathbb{Z}}}^{i_1}\mathcal{X}_{(i+\frac{1}{2},j_0+\frac{1}{2},k_0+\frac{1}{2})}\prod\limits_{\substack{i=i_0\\
        i\in 2\mathbb{Z}}}^{i_1-2}\mathcal{X}_{(i+\frac{1}{2},j_1+\frac{1}{2},k_0+\frac{1}{2})}\prod\limits_{\substack{j=j_0+2\\
        j\in 2\mathbb{Z}}}^{j_1}\mathcal{X}_{(i_0+\frac{1}{2},j+\frac{1}{2},k_0+\frac{1}{2})}\prod\limits_{\substack{j=j_0\\
        j\in 2\mathbb{Z}}}^{j_1-2}\mathcal{X}_{(i_1+\frac{1}{2},j+\frac{1}{2},k_0+\frac{1}{2})}\ket{\tilde{\Psi}}\\
        \hspace{11cm}\text{for }z=k_0+\frac{1}{2} \\
        \prod\limits_{\substack{i=i_0\\
        i\in2\mathbb{Z} }}^{i_1-2}\mathcal{X}_{(i+\frac{1}{2},j_0+\frac{1}{2},k_1+\frac{1}{2})}\prod\limits_{\substack{i=i_0+2\\
        i\in2\mathbb{Z} }}^{i_1}\mathcal{X}_{(i+\frac{1}{2},j_1+\frac{1}{2},k_1+\frac{1}{2})}\prod\limits_{\substack{j=j_0\\
        j\in2\mathbb{Z} }}^{j_1-2}\mathcal{X}_{(i_0+\frac{1}{2},j+\frac{1}{2},k_1+\frac{1}{2})}\prod\limits_{\substack{j=j_0+2\\
        j\in2\mathbb{Z} }}^{j_1}\mathcal{X}_{(i_1+\frac{1}{2},j+\frac{1}{2},k_1+\frac{1}{2})}\ket{\tilde{\Psi}}\\
        \hspace{11cm}\text{for }z=k_1+\frac{1}{2} \,.
    \end{cases}
\end{align}
\begin{align}
    \mathcal{P}_{xz}^{y,b}\ket{\tilde{\Psi}}=\begin{cases}
        \ket{\tilde{\Psi}} \text{ for } y<j_0+\frac{1}{2}\, ,\,  j_1+\frac{1}{2}<y \text{ and } y\in2\mathbb{Z}+1\\
        \mathcal{X}_{(i_0+\frac{1}{2},j+\frac{1}{2},k_0+\frac{1}{2})}\mathcal{X}_{(i_1+\frac{1}{2},j+\frac{1}{2},k_0+\frac{1}{2})}\mathcal{X}_{(i_1+\frac{1}{2},j+\frac{1}{2},k_1+\frac{1}{2})}\mathcal{X}_{(i_0+\frac{1}{2},j+\frac{1}{2},k_1+\frac{1}{2})}\ket{\tilde{\Psi}} \text{ for } j_0+\frac{1}{2}<y<j_1+\frac{1}{2}\text{ and } y\in2\mathbb{Z} \\
        \prod\limits_{\substack{i=i_0+2\\
       i\in2\mathbb{Z} }}^{i_1}\mathcal{X}_{(i+\frac{1}{2},j_0+\frac{1}{2},k_0+\frac{1}{2})}\prod\limits_{\substack{i=i_0\\
       i\in2\mathbb{Z}
       }}^{i_1-2}\mathcal{X}_{(i+\frac{1}{2},j_0+\frac{1}{2},k_1+\frac{1}{2})}\prod\limits_{\substack{k=k_0+2\\
       k\in2\mathbb{Z}}}^{k_1}\mathcal{X}_{(i_0+\frac{1}{2},j_0+\frac{1}{2},k+\frac{1}{2})}\prod\limits_{\substack{k=k_0\\
       k\in2\mathbb{Z}}}^{k_1-2}\mathcal{X}_{(i_1+\frac{1}{2},j_0+\frac{1}{2},k+\frac{1}{2})}\ket{\tilde{\Psi}}\\
        \hspace{11cm}\text{for }y=j_0+\frac{1}{2} \\
        \prod\limits_{\substack{i=i_0\\
        i\in2\mathbb{Z}}}^{i_1-2}\mathcal{X}_{(i+\frac{1}{2},j_1+\frac{1}{2},k_0+\frac{1}{2})}\prod\limits_{\substack{i=i_0+2\\
        i\in2\mathbb{Z}}}^{i_1}\mathcal{X}_{(i+\frac{1}{2},j_1+\frac{1}{2},k_1+\frac{1}{2})}\prod\limits_{\substack{k=k_0\\
        k\in2\mathbb{Z}}}^{k_1-2}\mathcal{X}_{(i_0+\frac{1}{2},j_1+\frac{1}{2},k+\frac{1}{2})}\prod\limits_{\substack{k=k_0+2\\
        k\in2\mathbb{Z}}}^{k_1}\mathcal{X}_{(i_1+\frac{1}{2},j_1+\frac{1}{2},k+\frac{1}{2})}\ket{\tilde{\Psi}}\\
        \hspace{11cm}\text{for }y=j_1+\frac{1}{2} \,.
    \end{cases}
\end{align}
\begin{align}
    \mathcal{P}_{yz}^{x,b}\ket{\tilde{\Psi}}=\begin{cases}
        \ket{\tilde{\Psi}} \text{ for } x<i_0+\frac{1}{2}\,,\,  i_1+\frac{1}{2}<x \text{ and } x\in2\mathbb{Z}+1\\
        \mathcal{X}_{(x,j_0+\frac{1}{2},k_0+\frac{1}{2})}\mathcal{X}_{(x,j_0+\frac{1}{2},k_1+\frac{1}{2})}\mathcal{X}_{(x,j_1+\frac{1}{2},k_0+\frac{1}{2})}\mathcal{X}_{(x,j_1+\frac{1}{2},k_1+\frac{1}{2})}\ket{\tilde{\Psi}} \text{ for } i_0+\frac{1}{2}<x<i_1+\frac{1}{2}\text{ and } x\in2\mathbb{Z} \\
        \prod\limits_{\substack{k=k_0+2\\
        k\in2\mathbb{Z}}}^{k_1}\mathcal{X}_{(i_0+\frac{1}{2},j_0+\frac{1}{2},k+\frac{1}{2})}\prod\limits_{\substack{k=k_0\\
        k\in2\mathbb{Z}}}^{k_1-2}\mathcal{X}_{(i_0+\frac{1}{2},j_1+\frac{1}{2},k+\frac{1}{2})}\prod\limits_{\substack{j=j_0+2\\
        j\in2\mathbb{Z}}}^{j_1}\mathcal{X}_{(i_0+\frac{1}{2},j+\frac{1}{2},k_0+\frac{1}{2})}\prod\limits_{\substack{j=j_0\\
        j\in2\mathbb{Z}}}^{j_1-2}\mathcal{X}_{(i_0+\frac{1}{2},j+\frac{1}{2},k_1+\frac{1}{2})}\ket{\tilde{\Psi}}\\
        \hspace{11cm}\text{for }x=i_0+\frac{1}{2} \\
        \prod\limits_{\substack{k=k_0\\
        k\in2\mathbb{Z}}}^{k_1-2}\mathcal{X}_{(i_1+\frac{1}{2},j_0+\frac{1}{2},k+\frac{1}{2})}\prod\limits_{\substack{k=k_0+2\\
        k\in2\mathbb{Z}}}^{k_1}\mathcal{X}_{(i_1+\frac{1}{2},j_1+\frac{1}{2},k+\frac{1}{2})}\prod\limits_{\substack{j=j_0\\
        j\in2\mathbb{Z}}}^{j_1-2}\mathcal{X}_{(i_1+\frac{1}{2},j+\frac{1}{2},k_0+\frac{1}{2})}\prod\limits_{\substack{j=j_0+2\\
        j\in2\mathbb{Z}}}^{j_1}\mathcal{X}_{(i_1+\frac{1}{2},j+\frac{1}{2},k_1+\frac{1}{2})}\ket{\tilde{\Psi}}\\
        \hspace{11cm}\text{for }x=i_1+\frac{1}{2} \,.
    \end{cases}
\end{align}
\label{eq:planaractiononPsi}
\end{subequations}
\red{
Let us define the following set
\begin{align}
    \text{Hinge}&=\Big\{(i_0+\frac{1}{2},j_0+\frac{1}{2},k+\frac{1}{2})|k_0+2\leq k\leq k_1\,, k\in2\mathbb{Z}\Big\}\cup\Big\{(i_1+\frac{1}{2},j_0+\frac{1}{2},k+\frac{1}{2})|k_0\leq k\leq k_1-2\,, k\in2\mathbb{Z}\Big\}\nonumber\\
    &\cup\Big\{(i_1+\frac{1}{2},j_1+\frac{1}{2},k+\frac{1}{2})|k_0+2\leq k\leq k_1\,, k\in2\mathbb{Z}\Big\} \cup \Big\{(i_0+\frac{1}{2},j_1+\frac{1}{2},k+\frac{1}{2})|k_0\leq k\leq k_1-2\,, k\in2\mathbb{Z}\Big\}\nonumber\\
    &\cup \Big\{(i+\frac{1}{2},j_0+\frac{1}{2},k_0+\frac{1}{2})|i_0+2\leq k\leq i_1\,, i\in2\mathbb{Z}\Big\}\cup \Big\{(i+\frac{1}{2},j_0+\frac{1}{2},k_1+\frac{1}{2})|i_0\leq k\leq i_1-2\,, i\in2\mathbb{Z}\Big\}\nonumber\\
    &\cup \Big\{(i+\frac{1}{2},j_1+\frac{1}{2},k_1+\frac{1}{2})|i_0+2\leq k\leq i_1\,, i\in2\mathbb{Z}\Big\}\cup \Big\{(i+\frac{1}{2},j_1+\frac{1}{2},k_0+\frac{1}{2})|i_0\leq k\leq i_1-2\,, i\in2\mathbb{Z}\Big\}\nonumber\\
    &\cup\Big\{(i_0+\frac{1}{2},j+\frac{1}{2},k_0+\frac{1}{2})|j_0+2\leq j\leq j_1\,, j\in2\mathbb{Z}\Big\}\cup\Big\{(i_0+\frac{1}{2},j+\frac{1}{2},k_1+\frac{1}{2})|j_0\leq j\leq j_1-2\,, j\in2\mathbb{Z}\Big\}\nonumber\\
    &\cup \Big\{(i_1+\frac{1}{2},j+\frac{1}{2},k_1+\frac{1}{2})|j_0+2\leq j\leq j_1\,, j\in2\mathbb{Z}\Big\}\cup \Big\{(i_1+\frac{1}{2},j+\frac{1}{2},k_0+\frac{1}{2})|j_0\leq j\leq j_1-2\,, j\in2\mathbb{Z}\Big\}\,.
\end{align}}
\color{black}
The global part of the subsystem symmetry acts on the ground state as
    \begin{align}
    \mathcal{P}_b\ket{\tilde{\Psi}}=\prod_{\substack{v\in\text{Hinge}\\
    v=(x,y,z)}}\mathcal{X}_{(x,y,z)}\ket{\tilde{\Psi}},
\end{align}
where the product is over hinge modes appearing on the edges of the cubic interface as shown in Figure~\ref{fig:hingemode}.
\begin{figure}
    \centering
    \includegraphics[scale=1]{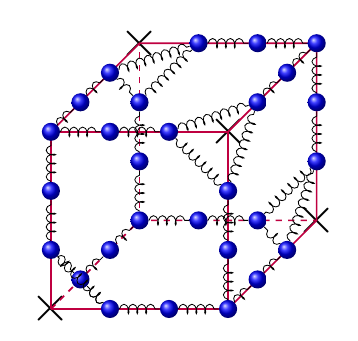}
    \caption{\red{Hinge mode is delocalized on the 12 hinges colored purple. The corners marked with "$\times$" can be gapped out by adding suitable terms in the interface Hamiltonian $\mathrm{H}_{\text{3D-cluster}|\text{blue}}^{\mathcal{G}}$ that commute with symmetry and the Hamiltonian. In addition, we can also add products of two cube operators $\mathcal{X}_{v_b}\mathcal{X}_{v_b'}$ as in \eqref{eq:tildeHG3Dcluster|blue} and their associated $\mathbf{D}^{(3)}_{\rm pln}$ conjugated operators. These products are represented by a curly line connecting $v_b$ and $v_b'$. There are $|\text{Hinge}|-1$ such independant operators. Hence, the Hamiltonian $\tilde{\mathrm{H}}_{\text{3D-cluster}|\text{blue}}^{'\mathcal{G}}$ would have a two fold degeneracy.}}
    \label{fig:hingemode}
\end{figure}
We note that each of these hinge modes satisfies the following relations on the ground space:
\begin{align}
    &\left\{\mathbf{D}^{(3)}_{\rm pln},\raisebox{-25pt}{\begin{tikzpicture}
        \node at (0.1,0.1) {$\color{blue}{X_v}$};
        \node at (1,0) {$\color{red}{Y}$};
        \node at (0,1) {$\color{red}{Z}$};
        \node at (-0.5,-0.5) {$\color{red}{Z}$};
        \node at (0.5,0.5) {$\color{red}{Z}$};
        \draw[thick] (0,1)--(1,1)--(0.5,0.5)--(-0.5,0.5)--(0,1);
        \draw[thick] (0.5,0.5)--(0.5,-0.5)--(-0.5,-0.5)--(-0.5,0.5);
        \draw[thick] (0.5,0.5)--(0.5,-0.5)--(1,0)--(1,1);
    \end{tikzpicture}}\right\}\ket{\tilde{\Psi}}=0\, ,\quad \left\{\mathbf{D}^{(3)}_{\rm pln},\raisebox{-25pt}{\begin{tikzpicture}
        \node at (0.1,0.1) {$\color{blue}{X_v}$};
        \node at (1,0) {$\color{red}{Z}$};
        \node at (0,1) {$\color{red}{Y}$};
        \node at (-0.5,-0.5) {$\color{red}{Z}$};
        \node at (0.5,0.5) {$\color{red}{Z}$};
        \draw[thick] (0,1)--(1,1)--(0.5,0.5)--(-0.5,0.5)--(0,1);
        \draw[thick] (0.5,0.5)--(0.5,-0.5)--(-0.5,-0.5)--(-0.5,0.5);
        \draw[thick] (0.5,0.5)--(0.5,-0.5)--(1,0)--(1,1);
    \end{tikzpicture}}\right\}\ket{\tilde{\Psi}}=0\, ,\nonumber\\
    &\left\{\mathbf{D}^{(3)}_{\rm pln},\raisebox{-25pt}{\begin{tikzpicture}
        \node at (0.1,0.1) {$\color{blue}{X_v}$};
        \node at (1,0) {$\color{red}{Z}$};
        \node at (0,1) {$\color{red}{Z}$};
        \node at (-0.5,-0.5) {$\color{red}{Y}$};
        \node at (0.5,0.5) {$\color{red}{Z}$};
        \draw[thick] (0,1)--(1,1)--(0.5,0.5)--(-0.5,0.5)--(0,1);
        \draw[thick] (0.5,0.5)--(0.5,-0.5)--(-0.5,-0.5)--(-0.5,0.5);
        \draw[thick] (0.5,0.5)--(0.5,-0.5)--(1,0)--(1,1);
    \end{tikzpicture}}\right\}\ket{\tilde{\Psi}}=0\, ,\quad \left\{\mathbf{D}^{(3)}_{\rm pln},\raisebox{-25pt}{\begin{tikzpicture}
        \node at (0.1,0.1) {$\color{blue}{X_v}$};
        \node at (1,0) {$\color{red}{Z}$};
        \node at (0,1) {$\color{red}{Z}$};
        \node at (-0.5,-0.5) {$\color{red}{Z}$};
        \node at (0.5,0.5) {$\color{red}{Y}$};
        \draw[thick] (0,1)--(1,1)--(0.5,0.5)--(-0.5,0.5)--(0,1);
        \draw[thick] (0.5,0.5)--(0.5,-0.5)--(-0.5,-0.5)--(-0.5,0.5);
        \draw[thick] (0.5,0.5)--(0.5,-0.5)--(1,0)--(1,1);
    \end{tikzpicture}}\right\}\ket{\tilde{\Psi}}=0\, .
    \label{eq:D3hatXanticommutation}
\end{align}
Now let us calculate $\mathbf{D}^{(3)}_{\rm pln}\ket{\tilde{\Psi}}$. First, we note using \eqref{eq:D^3plansquared} that
\red{
\begin{subequations}
\begin{align}
(\mathbf{D}^{(3)}_{\rm pln})^2\ket{\tilde{\Psi}}&=\frac{1}{2^{2|\rm E|-|\mathbb{V}|-|\hat{\mathbb{V}}|}}\left(\sum_{a\in\text{Ker}\sigma}\eta_a\right)\left(\sum_{\hat{a}\in\text{Ker}\sigma^T}\eta_{\hat{a}}\right)\ket{\tilde{\Psi}}\\
    &=\frac{2^{\text{dim}\,\text{Ker}\,\sigma}}{2^{2|\rm E|-|\mathbb{V}|-|\hat{\mathbb{V}}|}}\left(\sum_{\hat{a}\in\text{Ker}\sigma^T}\eta_{\hat{a}}\right)\ket{\tilde{\Psi}}\,.
    \label{eq:D3plnsquare}
\end{align}
\end{subequations}}
\color{black}
The second line follows from \eqref{eq:subsystem_planarred}. We also note that
\begin{align}
   \mathbf{D}^{(3)}_{\rm pln}\prod_{\substack{v\in\text{Hinge}\\
    v=(x,y,z)}}Z_{(x,y,z)}\ket{\tilde{\Psi}}= \prod_{\substack{v\in\text{Hinge}\\
    v=(x,y,z)}}Z_{(x,y,z)}\mathbf{D}^{(3)}_{\rm pln}\ket{\tilde{\Psi}}\,.
   \label{eq:D3plnZcommutation}
\end{align}
Let $\mathcal{A}$ and $\mathcal{B}$ be subsets of the total set of vertices on the hinge such that the product $\prod_{\mathcal{A}}\prod_{\mathcal{B}}=\prod_{v\in\text{Hinge}}$ and $\mathcal{A}\cap \mathcal{B}=\emptyset$. From \eqref{eq:D3hatXanticommutation} and \eqref{eq:D3plnZcommutation}, we infer that on the ground space subspace
\begin{align}
   \mathbf{D}^{(3)}_{\rm pln}&\sim \prod_{\substack{v\in\text{Hinge}\\
    v=(x,y,z)}}Z_{(x,y,z)}\left(\sum_{|\mathcal{A}|=\text{even}} \alpha_{\mathcal{A}}\prod_{(x,y,z)\in\mathcal{A}} \mathcal{X}_{(x,y,z)}\right)\,.
\end{align} 
Now applying the constraint~
\eqref{eq:D3plnsquare}, we see that 
\red{
\begin{align}
 \mathbf{D}^{(3)}_{\rm pln}&\sim \prod_{\substack{v\in\text{Hinge}\\
    v=(x,y,z)}}Z_{(x,y,z)}\left(\sum_{\hat{a}\in\text{Ker}\sigma^T}\eta_{\hat{a}}\right)
\end{align}}
\color{black}
is a possible solution that also satisfy the constraints \eqref{eq:D3hatXanticommutation} and \eqref{eq:D3plnZcommutation}.
Let us define
\begin{align}
    \mathrm{D}_{(x,y,z)}=Z_{(x,y,z)}\mathcal{X}_{(x,y,z)}\, .
\end{align}
Then $\mathbf{D}^{(3)}_{\rm pln}$ can be equivalently written as
\begin{align}
    \mathbf{D}^{(3)}_{\rm pln}\sim \sum_{\red{|\mathcal{A}|=\text{even}}}\red{\alpha_{\mathcal{A}}} \prod_{\mathcal{A}} \mathrm{D}_{(x,y,z)}\prod_{\mathcal{B}} Z_{(x,y,z)}
\end{align}
in a schematic form.  
Hence $\mathrm{D}_{(x,y,z)}$ and $Z_{(x,y,z)}$ are localized symmetry operators of $\mathbf{D}^{(3)}_{\rm pln}$ near $(x,y,z)$.
Then we see the projective algebra between localized symmetry operators of $\mathbf{D}^{(3)}_{\rm pln}$ and the localized symmetry operators of planar symmetries~\eqref{eq:planaractiononPsi}
\begin{align}
    \{\mathrm{D}_{(x,y,z)},\mathcal{X}_{(x,y,z)}\}=0\, ,\quad \{\mathrm{D}_{(x,y,z)},Z_{(x,y,z)}\}=0\, ,\quad \{\mathcal{X}_{(x,y,z)},Z_{(x,y,z)}\}=0\, .
    \label{eq:fractionalizedsymopalgebra3D}
\end{align}
\red{
Now let us add terms obtained by taking products of \eqref{eq:cubeoperator} to the Hamiltonian \eqref{eq:interfaceham3D} as described in Figure~\ref{fig:hingemode}. We note that a product of two such cube operators commute with $(\mathbf{D}^{(3)}_{\rm pln})^2$ and adding such terms can further lift the degeneracy. Let us define the following sets
\begin{subequations}
\begin{align}
    l_x&=\{(i+\frac{1}{2},j_0+\frac{1}{2},k_0+\frac{1}{2})|i_0<i<i_1\,, i\in2\mathbb{Z}\}\cup \{(i+\frac{1}{2},j_1+\frac{1}{2},k_1+\frac{1}{2})|i_0<i<i_1\,, i\in2\mathbb{Z}\}\nonumber\\
    &\{(i+\frac{1}{2},j_1+\frac{1}{2},k_0+\frac{1}{2})|i_0\leq i<i_1-2\,, i\in2\mathbb{Z}\}\cup\{(i+\frac{1}{2},j_0+\frac{1}{2},k_1+\frac{1}{2})|i_0\leq i<i_1-2\,, i\in2\mathbb{Z}\}\\
    l_y&=\{(i_0+\frac{1}{2},j+\frac{1}{2},k_0+\frac{1}{2})|j_0<j<j_1\,, j\in2\mathbb{Z}\}\cup \{(i_1+\frac{1}{2},j+\frac{1}{2},k_1+\frac{1}{2})|j_0<j<j_1\,, j\in2\mathbb{Z}\}\nonumber\\
    &\{(i_0+\frac{1}{2},j+\frac{1}{2},k_1+\frac{1}{2})|j_0\leq j<j_1-2\,, i\in2\mathbb{Z}\}\cup\{(i_1+\frac{1}{2},j+\frac{1}{2},k_0+\frac{1}{2})|j_0\leq j<j_1-2\,, j\in2\mathbb{Z}\}\\
    l_z&=\{(i_0+\frac{1}{2},j_0+\frac{1}{2},k+\frac{1}{2})|k_0<k<k_1\,, k\in2\mathbb{Z}\}\cup \{(i_1+\frac{1}{2},j_1+\frac{1}{2},k+\frac{1}{2})|k_0<k<k_1\,, k\in2\mathbb{Z}\}\nonumber\\
    &\{(i_1+\frac{1}{2},j_0+\frac{1}{2},k+\frac{1}{2})|k_0\leq k<k_1-2\,, k\in2\mathbb{Z}\}\cup\{(i_0+\frac{1}{2},j_1+\frac{1}{2},k+\frac{1}{2})|k_0\leq k<k_1-2\,, k\in2\mathbb{Z}\}\\
    \mathcal{C}&=\Bigg\{\left((i_1+\frac{1}{2},j_0+\frac{1}{2},k_1-\frac{3}{2}),(i_1-\frac{3}{2},j_0+\frac{1}{2},k_1+\frac{1}{2})\right),\left((i_1-\frac{3}{2},j_0+\frac{1}{2},k_1+\frac{1}{2}),(i_0+\frac{1}{2},j_0+\frac{5}{2},k_1+\frac{1}{2})\right)\nonumber\\
    &\left((i_0+\frac{1}{2},j_1+\frac{5}{2},k_1+\frac{1}{2}),(i_1+\frac{1}{2},j_0+\frac{1}{2},k_1-\frac{3}{2})\right),\left((i_0+\frac{5}{2},j_0+\frac{1}{2},k_0+\frac{1}{2}),(i_0+\frac{1}{2},j_0+\frac{5}{2},k_0+\frac{1}{2})\right)\nonumber\\
    &\left((i_0+\frac{1}{2},j_0+\frac{5}{2},0_1+\frac{1}{2}),(i_0+\frac{1}{2},j_0+\frac{1}{2},k_0+\frac{5}{2})\right),\left((i_0+\frac{1}{2},j_0+\frac{1}{2},k_0+\frac{5}{2}),(i_0+\frac{5}{2},j_0+\frac{1}{2},k_0+\frac{1}{2})\right)\nonumber\\
    &\left((i_0+\frac{1}{2},j_1-\frac{3}{2},k_1+\frac{1}{2}),(i_0+\frac{5}{2},j_1+\frac{1}{2},k_1+\frac{1}{2})\right),\left((i_0+\frac{5}{2},j_1+\frac{1}{2},k_1+\frac{1}{2}),(i_0+\frac{1}{2},j_0+\frac{5}{2},k_0-\frac{3}{2})\right)\nonumber\\
    &\left((i_0+\frac{1}{2},j_0+\frac{1}{2},k_0-\frac{3}{2}),(i_0+\frac{1}{2},j_1-\frac{3}{2},k_1+\frac{1}{2})\right),\left((i_1+\frac{1}{2},j_1-\frac{3}{2},k_0+\frac{1}{2}),(i_1-\frac{3}{2},j_0+\frac{1}{2},k_0+\frac{1}{2})\right)\nonumber\\
    &\left((i_1-\frac{1}{3},j_0+\frac{1}{2},k_0+\frac{1}{2}),(i_1+\frac{1}{2},j_1+\frac{1}{2},k_0+\frac{5}{2})\right),\left((i_1+\frac{1}{2},j_1+\frac{1}{2},k_0+\frac{5}{2}),(i_1+\frac{1}{2},j_1-\frac{3}{2},k_0+\frac{1}{2})\right)\Bigg\}\,.
\end{align}
\end{subequations}
With the help of these sets, we define the interface Hamiltonian with the added stabilizer terms
\begin{align}
    \tilde{\mathrm{H}}_{\text{3D-cluster}|\text{blue}}^{'\mathcal{G}}&=\tilde{\mathrm{H}}_{\text{3D-cluster}|\text{blue}}^{\mathcal{G}}-\sum_{v_b\in l_x}\mathcal{X}_{v_b}\mathcal{X}_{v_b+(2,0,0)}-\mathbf{D}^{(3)}_{\rm pln}\text{conj}\left[\sum_{v_b\in l_x}\mathcal{X}_{v_b}\mathcal{X}_{v_b+(2,0,0)}\right]\nonumber\\
    &-\sum_{v_b\in l_y}\mathcal{X}_{v_b}\mathcal{X}_{v_b+(0,2,0)}-\mathbf{D}^{(3)}_{\rm pln}\text{conj}\left[\sum_{v_b\in l_y}\mathcal{X}_{v_b}\mathcal{X}_{v_b+(0,2,0)}\right]\nonumber\\
    &-\sum_{v_b\in l_z}\mathcal{X}_{v_b}\mathcal{X}_{v_b+(0,0,2)}-\mathbf{D}^{(3)}_{\rm pln}\text{conj}\left[\sum_{v_b\in l_z}\mathcal{X}_{v_b}\mathcal{X}_{v_b+(0,0,2)}\right]\,\nonumber\\
    &-\sum_{(v_b,v_b')\in \mathcal{C}}\mathcal{X}_{v_b}\mathcal{X}_{v_b'}-\mathbf{D}^{(3)}_{\rm pln}\text{conj}\left[\sum_{(v_b,v_b')\in \mathcal{C}}\mathcal{X}_{v_b}\mathcal{X}_{v_b'}\right]\,,
    \label{eq:tildeHG3Dcluster|blue}
\end{align}
where $\mathbf{D}^{(3)}_{\rm pln}\text{conj}\left[\cdot\right]$ is the $\mathbf{D}^{(3)}_{\rm pln}$ symmetric counterpart of the term inside the argument. We note that we have lifted the degeneracy of $\tilde{\mathrm{H}}_{\text{3D-cluster}|\text{blue}}^{\mathcal{G}}$ that scale with the hinge size to just two fold degeneracy by adding terms on the hinge that are symmetric and commuting with the bulk Hamiltonian. Suppose $\ket{\tilde{\Psi}'}$ is a ground state of the above interface Hamiltonian. It is straightforward to see that 
\begin{subequations}
\begin{align}
    \mathcal{P}^{z,a}_{xy}\ket{\tilde{\Psi}'}&=\ket{\tilde{\Psi}'}\,,\quad\mathcal{P}^{y,a}_{xz}\ket{\tilde{\Psi}'}=\ket{\tilde{\Psi}'}\,,\quad\mathcal{P}^{x,a}_{yz}\ket{\tilde{\Psi}'}=\ket{\tilde{\Psi}'}\,\text{ for }\qquad a=r,b\\
    \eta_{a}\ket{\tilde{\Psi}'}&=\ket{\tilde{\Psi}'}\,\quad \eta_{\hat{a}}\ket{\tilde{\Psi}'}=\ket{\tilde{\Psi}'}\,\quad(\mathbf{D}^{(3)}_{\rm pln})^2\ket{\tilde{\Psi}'}
=\frac{2^{\text{dim}\,\text{Ker}\,\sigma+\text{dim}\,\text{Ker}\,\sigma^T}}{2^{2|\rm E|-|\mathbb{V}|-|\hat{\mathbb{V}}|}}\ket{\tilde{\Psi}'}\,.
\end{align}
\end{subequations}
The extra terms we added in the Hamiltonian have the effect of removing localization of symmetry, so we have perfect projection acting on the ground state. Hence on the ground space of $\tilde{\mathrm{H}}_{\text{3D-cluster}|\text{blue}}^{'\mathcal{G}}$, we have
\begin{align}
    \mathbf{D}^{(3)}_{\rm pln}\sim \prod_{\substack{v\in\text{Hinge}\\
    v=(x,y,z)}}Z_{(x,y,z)}\,.
\end{align}
We note that we still have projective algebra 
\begin{align}
    \{\mathbf{D}^{(3)}_{\rm pln},\mathcal{X}_{(x,y,z)}\}\ket{\tilde{\Psi}'}=0\,\quad \forall (x,y,z)\in\text{Hinge}\,.
\end{align}
This explains the two fold ground state degeneracy for $\tilde{\mathrm{H}}_{\text{3D-cluster}|\text{blue}}^{'\mathcal{G}}$. On the hinge, the effective $\mathbf{D}^{(3)}_{\rm pln}$ is spontaneously broken. This is similar to the two fold degeneracy we obtain when the $\mathbb{Z}_2$ symmetry is spontaneously broken in $1+1$D.}
\end{widetext}
\color{black}
\section{Stability analysis of Interface modes}\label{sec:stability}
In this section, we perform a stability analysis of interface modes between two different SPTs under local and symmetric perturbations. We analyze this in the case of interface between 1) $1+1$D $\mathbb{Z}_2\times\mathbb{Z}_2$ cluster state and trivial SPT 2) $1+1$D $\mathbb{Z}_2\times\mathbb{Z}_2$ cluster state and the noninvertible SPT (odd state) 3) $2+1$D $\mathbb{Z}_2\times\mathbb{Z}_2$ cluster state and higher-order noninvertible SSPT (blue state).
\subsection{Interface between \texorpdfstring{$1+1$}{Lg}D \texorpdfstring{$\mathbb{Z}_2\times\mathbb{Z}_2$}{Lg} cluster SPT and trivial SPT}
Let us consider a ring with $N$ sites and label them from $1,...,N$. We consider trivial SPT on the first $l-1$ sites and the cluster state on the rest of the sites. Explicitly, the Hamiltonian is given by 
\begin{align}
    \mathrm{H}_{\text{trivial}|\text{1D-cluster}}=-\sum_{i=1}^{l-1}X_i-\sum_{i=l+1}^{L-1}Z_{i-1}X_iZ_{i+1}\, .
\end{align}
Let us add a local and $\mathbb{Z}_2\times\mathbb{Z}_2$ symmetric perturbation to this Hamiltonian supported around sites $l$ and $L$
\begin{align}
    \mathrm{H}_{\text{trivial}|\text{1D-cluster}}^{\text{pert}}=\mathrm{H}_{\text{trivial}|\text{1D-cluster}}-\mathcal{O}^{(l)}-\mathcal{O}^{(L)}\, .
\end{align}
We note that $\mathcal{O}^{(l)}$ and $\mathcal{O}^{(L)}$ could be sum of many terms supported near the site $l$ and $L$. The only restriction on their support is that there should be at least one odd and one even site that is not contained in the union of their support and lying between $l$ and $L$.
We claim that the Hamiltonian $\mathrm{H}_{\text{trivial}|\text{1D-cluster}}^{\text{pert}}$ still has four-fold degeneracy. We prove this by exhibiting two pairs of operators that commute with the Hamiltonian and satisfy a projective algebra. Suppose $2k$ is an even site that is not contained in the support of $\mathcal{O}^{(l)}$ and $\mathcal{O}^{(L)}$ and $l<2k<L$, then one can consider a string of Pauli-$X$ on odd sites followed by Pauli-$Z$ at $2k$ as shown in Figure~\ref{fig:stringquantumnocluster}. This string would commute with $\mathrm{H}_{\text{trivial}|\text{1D-cluster}}^{\text{pert}}$ since $\mathcal{O}^{(l)}$ is a symmetric perturbation, and on the support of $\mathcal{O}^{(l)}$, the string operator and the symmetry generator are same. This string operator would anti-commute with the symmetry generator on even sites. This gives a pair of operators (the string operator and the symmetry generator) that commute with $\mathrm{H}_{\text{trivial}|\text{1D-cluster}}^{\text{pert}}$ and satisfy the projective algebra. Similarly, one could repeat the same for a string operator starting with $Z$ on the odd sites, followed by Pauli-$X$ on even sites. Again, this string would anti-commute with the symmetry generator on odd sites. Hence, we have two pairs of operators commuting with $\mathrm{H}_{\text{trivial}|\text{1D-cluster}}^{\text{pert}}$ that satisfy the projective algebra. This implies that there should be at least four degenerate ground states for $\mathrm{H}_{\text{trivial}|\text{1D-cluster}}^{\text{pert}}$.
\begin{figure*}[hbt!]
    \centering
    \includegraphics[scale=1]{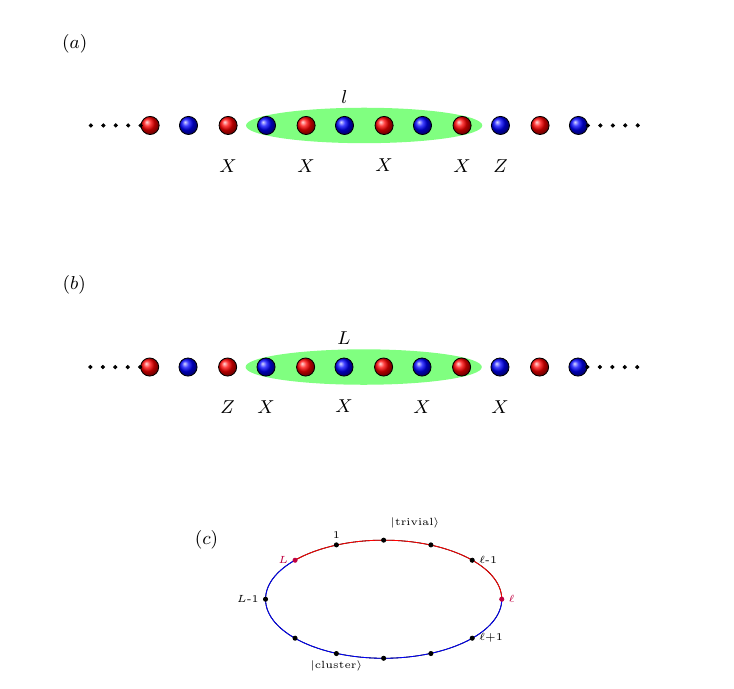}
    \caption{The figures illustrate the way to construct string operators. In Figure (a), the green colored region denotes the support of $\mathcal{O}^{(l)}$ around site $l$ colored blue. The string operator commute with $\mathrm{H}_{\text{trivial}|\text{1D-cluster}}^{\text{pert}}$ and anti-commute with $\prod_{v_b}X_{v_b}$. Similarly, in Figure (b), the green colored region denotes the support of $\mathcal{O}^{(L)}$ around site $L$ colored blue. The string operator commute with $\mathrm{H}_{\text{trivial}|\text{1D-cluster}}^{\text{pert}}$ and anti-commute with $\prod_{v_r}X_{v_r}$. Figure (c) shows the interface between the trivial state (on red arc) and the cluster state (on blue arc) with edge modes around site $l$ and $L$. }
    \label{fig:stringquantumnocluster}
\end{figure*}
\subsection{Interface between \texorpdfstring{$\rm H_{\text{1D-cluster}}$}{Lg} and \texorpdfstring{$\rm H_{\text{odd}}$}{Lg}}
Let us consider the interface Hamiltonian between $\mathrm{H}_{\text{1D-cluster}}$ and $\mathrm{H}_{\text{odd}}$. 
\begin{widetext}
\begin{align}
    \mathrm{H}_{\text{1D-cluster}|\text{odd}}=-\sum_{i=1}^{l-1}Z_{i-1}X_iZ_{i+1}+\sum_{i=\frac{l}{2}}^{\frac{L}{2}-1}Z_{2i}X_{2i+1}Z_{2i+2}-\sum_{i=\frac{l}{2}+1}^{\frac{L}{2}-1}Y_{2i-1}X_{2i}Y_{2i+1}+\sum_{i=\frac{l}{2}+1}^{\frac{L}{2}-1}Z_{2i-2}Z_{2i-1}X_{2i}Z_{2i+1}Z_{2i+2}
    \label{eq:1Dinterfaceham}
\end{align}
\end{widetext}
This Hamiltonian has a four-fold ground state degeneracy coming from edge modes located at sites $l$ and $L$. We consider adding symmetric (under $\mathbb{Z}_2\times\mathbb{Z}_2$ 0-form and $\mathrm{D}^{(1)}$) and local perturbations to this interface Hamiltonian. 
\begin{align}
    \mathrm{H}^{\text{pert}}_{\text{1D-cluster}|\text{odd}}=\mathrm{H}_{\text{1D-cluster}|\text{odd}}-\mathcal{O}^{(l)}-\mathcal{O}^{(L)}\,.
\end{align}
We assume the support of this perturbation is contained inside a green region and located around sites $l$ and $L$ as in Figure~\ref{fig:edge_region}. The two green regions do not overlap. The perturbations in general need not commute with $\mathrm{H}_{\text{1D-cluster}|\text{odd}}$. 
\begin{figure*}[h!]
    \centering
    \includegraphics[scale=1]{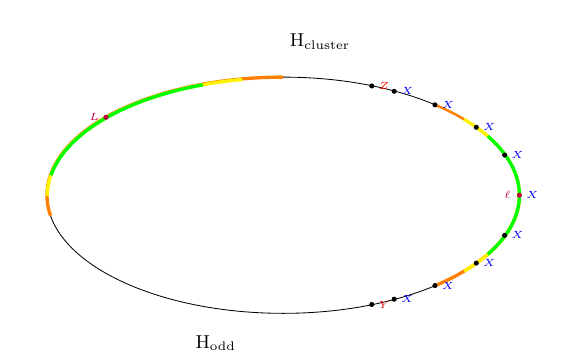}
    \caption{The figure illustrate the interface between $\mathrm{H}_{\text{odd}}$ and $\mathrm{H}_{\text{1D-cluster}}$ with added perturbations. $\mathrm{H}_{\text{1D-cluster}}$ and $\mathrm{H}_{\text{odd}}$ are supported in between sites $L$ and $l$ on the upper and lower arc respectively and as in \eqref{eq:1Dinterfaceham}. We add symmetric local perturbations near the sites $l$ and $L$ supported inside the green region. These perturbations might not commute with $\mathrm{H}_{\text{1D-cluster}|\text{odd}}$. However, all the non-commuting are inside the yellow regions. Hence, to find the ground space of $\mathrm{H}_{\text{1D-cluster}|\text{odd}}^{\text{pert}}$, we can minimize all the terms outside the yellow region individually and the terms inside the yellow region as a whole. The string operator around the site $l$ commutes with $\mathrm{H}_{\text{1D-cluster}|\text{odd}}^{\text{pert}}$ as this operator is equivalent to the symmetry operator inside the green region where the perturbations are supported. After conjugating with a unitary, the yellow regions are enlarged to the orange regions, and outside the orange region, we have a trivial-like Hamiltonian.}
    \label{fig:edge_region}
\end{figure*}

We observe that we can choose a bigger region colored yellow for which terms outside the yellow region commute with terms inside. Other words, all the non-commuting terms are contained inside the yellow region. This is because we consider 1) local Hamiltonian, 2) terms in $\mathrm{H}_{\text{1D-cluster}|\text{odd}}$ commute with each other, and 3) the perturbations are contained in the green region. The ground state is obtained by minimizing each individual term outside the yellow region and minimizing the Hamiltonian restricted to the yellow region.

We can find a string operator (see Figure~\ref{fig:edge_region}) that commutes with $\mathrm{H}_{\text{1D-cluster}|\text{odd}}^{\text{pert}}$ and anti-commutes with $\mathrm{D}^{(1)}$ in the ground space of $\mathrm{H}_{\text{1D-cluster}|\text{odd}}^{\text{pert}}$. The string operator satisfies
\begin{widetext}
\begin{align}
    \begin{array}{ccccccccc}
        \color{red}{Z} & \color{blue}{X} &  & \color{blue}{X} & ... & \color{blue}{X} & & \color{blue}{X}& \color{red}{Y} \end{array}\mathrm{D}^{(1)}\ket{\Psi}&=\mathrm{D}^{(1)}\begin{array}{ccccccccc}
        \color{red}{Z} & \color{blue}{X} &  & \color{blue}{X} & ... & \color{blue}{X} & & \color{blue}{X}& \color{red}{Y} \end{array}\times\left(\begin{array}{ccc}
         \color{blue}{Z}& \color{red}{X} & \color{blue}{Z}
    \end{array}\right)\ket{\Psi}\nonumber\\
    &=-\mathrm{D}^{(1)}\begin{array}{ccccccccc}
        \color{red}{Z} & \color{blue}{X} &  & \color{blue}{X} & ... & \color{blue}{X} & & \color{blue}{X}& \color{red}{Y} \end{array}\ket{\Psi}\, ,
        \label{eq:stringopanticommute}
\end{align}
\end{widetext}
where $\ket{\Psi}$ is a ground state of $\mathrm{H}^{\text{pert}}_{\text{1D-cluster}|\text{odd}}$. For the second equality, we used the fact that the value of the cluster-like term in the bracket is $-1$ in the ground space. This is true even for the perturbed Hamiltonian ground space, provided the support of the perturbations is contained inside the green regions, as we mentioned before, and the term inside the bracket on the R.H.S of the first equality is outside the yellow region. The above anti-commutation implies that the ground space of $\mathrm{H}^{\text{pert}}_{\text{1D-cluster}|\text{odd}}$ should be at least two-fold degenerate.

We now conjugate the Hamiltonian $\mathrm{H}^{\text{pert}}_{\text{1D-cluster}|\text{odd}}$ to disentangle most of the region into a trivial-like Hamiltonian (unique ground state). By this unitary conjugation, the yellow region extends to the orange region as in Figure~\ref{fig:edge_region}. Outside the orange region, we have the trivial-like Hamiltonian. Explicitly, the unitary is 
\begin{align}
    \mathcal{U}^{(1)}_{\text{disent}}=\prod_{i=1}^LCZ_{i,i+1}\prod_{i=\frac{l}{2}}^{\frac{L}{2}-1}CZ_{2i,2i+2}
\end{align}
After conjugating with $\mathcal{U}^{(1)}_{\text{disent}}$, $\mathrm{H}_{\text{1D-cluster}|\text{odd}}$ is transformed to
\begin{widetext}
\begin{align}
   \mathcal{U}^{(1)}_{\text{disent}} \mathrm{H}_{\text{1D-cluster}|\text{odd}}(\mathcal{U}^{(1)}_{\text{disent}})^{\dagger}&=-\sum_{i=1}^{l-1}X_i+\sum_{i=\frac{l}{2}}^{\frac{L}{2}-1}X_{2i+1}+\sum_{i=\frac{l}{2}+1}^{\frac{L}{2}-1}X_{2i}\left(1+X_{2i-1}X_{2i+1}\right)\, .
\end{align}
\end{widetext}
We note that the ground state of this Hamiltonian is obtained by setting $X_i=1$ for $i=1,...,l-1$ and $X_i=-1$ for $i=\frac{l}{2}+1,...,\frac{L}{2}-1$. The sites at $l$ and $L$ are completely decoupled and contribute to the four-fold degeneracy of the ground states. The conjugation by $\mathcal{U}^{(1)}_{\text{disent}}$ does not change the spectrum of the Hamiltonian. We know that a trivial-like Hamiltonian does not contribute to the ground-state degeneracy. Hence, the ground-state degeneracy should come from the two orange regions. We claim that diagonalizing the Hamiltonian in each orange region should give at least two-fold degeneracy. Suppose that one of the orange regions gives a unique ground state, then let us consider the perturbations that are associated with that orange region. Now, we consider the same (but mirror reflected) perturbation on the other interface edge. For this Hamiltonian, there is a unique ground state. However, this is inconsistent with \eqref{eq:stringopanticommute} and the fact that the Hamiltonian should have at least two-fold degeneracy. Hence, both orange regions should give at least two-fold degeneracy. This implies that $\mathrm{H}^{\text{pert}}_{\text{1D-cluster}|\text{odd}}$ should be at least four-fold degenerate.
\subsection{Interface between \texorpdfstring{$\rm H_{\text{2D-cluster}}$}{Lg} and \texorpdfstring{$\rm H_{\text{blue}}$}{Lg}}\label{sec:stability2dclusterblue}
Let us consider a rectangular interface between $\mathrm{H}_{\text{2D-cluster}}$ and $\mathrm{H}_{\text{blue}}$ placed on a torus. Here, we do a symmetric perturbation of the interface Hamiltonian and study the stability of the degenerate ground space. As before, we choose the interface line to run along the blue sublattice with corners at $(i_0+\frac{1}{2},j_0+\frac{1}{2})$, $(i_0+\frac{1}{2},j_1+\frac{1}{2})$, $(i_1+\frac{1}{2},j_0+\frac{1}{2})$ and $(i_1+\frac{1}{2},j_1+\frac{1}{2})$ as given in Figure~\ref{fig:rectangularinterface}. The interface Hamiltonian is given in \eqref{eq:Htilde2Dclusterblue}.
Now we add local and symmetric (both subsystem symmetric as well as $\mathrm{D}^{(2)}$ symmetric) perturbations to the Hamiltonian
\begin{align}
    \mathrm{H}_{\text{2D-cluster}|\text{blue}}^{\text{pert}}&=\tilde{\mathrm{H}}_{\text{2D-cluster}|\text{blue}}-\mathcal{O}^{TL}-\mathcal{O}^{TR}\nonumber\\
    &\qquad-\mathcal{O}^{BL}-\mathcal{O}^{BR}\, ,
\end{align}
where the local perturbations are supported near the four corners. We assume the support of the four perturbations does not overlap with any other.  We take a square region around each of the corners where the support of the perturbations is contained and color it green (see Figure~\ref{fig:corner-region}). This leaves a cross-shaped region inside the rectangular interface, as in Figure~\ref{fig:corner-region}, that does not intersect with the support of any of the perturbations. 

\begin{figure*}[h!]
\centering
    \includegraphics[scale=1]{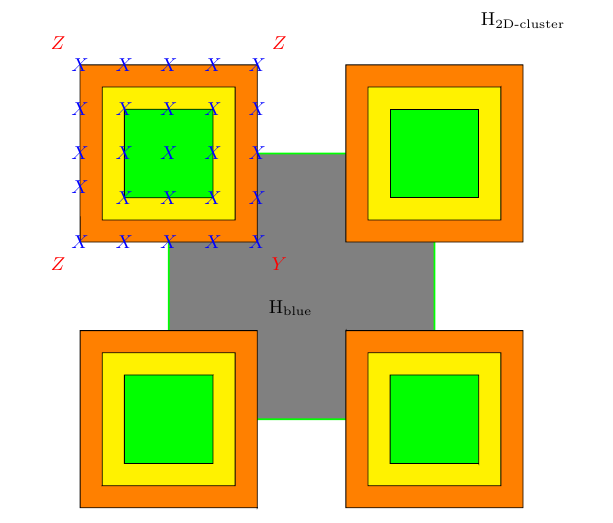}
    \caption{The figure illustrates the interface between $\mathrm{H}_{\text{blue}}$ and $\mathrm{H}_{\text{2D-cluster}}$ with added perturbations. The gray region denotes the support of $\mathrm{H}_{\text{blue}}$ while $\mathrm{H}_{\text{2D-cluster}}$ is supported on the exterior. We also add symmetric terms along the edges of the interface green line everywhere except at the corners. Now, around the corners, we add symmetric local perturbations that are supported inside the green squares on all four corners. These perturbations, in general, might violate the frustration-free condition. However, we can consider the yellow region around the four corners, which contains all the non-commuting terms. Hence, to find the ground space of the perturbed Hamiltonian $\mathrm{H}^{\text{pert}}_{\text{2D-cluster}|\text{blue}}$, we can minimize all the terms outside the yellow region individually and then minimize the total terms in the yellow squares. The membrane operator shown around the top-left corner commutes with $\mathrm{H}^{\text{pert}}_{\text{2D-cluster}|\text{blue}}$, since, inside the green region, the membrane operator is equivalent to a symmetry operator. Finally, after conjugating with the unitary~\eqref{eq:unitarytoconj}, the yellow region gets enlarged to the orange region, and outside of the orange region we have a trivial-like Hamiltonian.}
    \label{fig:corner-region}
\end{figure*}
We observe that one can choose a bigger square region near the four corners colored yellow for which the terms outside the yellow region commute with the terms inside the yellow region. This is because we consider 1) local Hamiltonian, 2) terms in $\tilde{\rm H}_{\text{2D-cluster}|\text{blue}}$ commute with each other, and 3) the perturbations are contained in the green region. This way, all the frustration is inside the four yellow regions. The ground state is obtained by simultaneously minimizing each individual Hamiltonian term outside the yellow squares and minimizing the Hamiltonian restricted to the four yellow squares. 

We can find a membrane operator that commutes with $\mathrm{H}^{\text{pert}}_{\text{2D-cluster}|\text{blue}}$ and anti-commute with $\mathrm{D}^{(2)}$ on the ground space of $\mathrm{H}^{\text{pert}}_{\text{2D-cluster}|\text{blue}}$. Schematically, the membrane operator is given in Figure~\ref{fig:corner-region}. This membrane operator satisfies the following
\begin{widetext}
 \begin{align}
\pgfsetlayers{main,nodelayer,edgelayer,background}
\begin{tikzpicture}
	\begin{pgfonlayer}{nodelayer}
		\node [style=none] (0) at (5, -2) {};
		\node [style=none] (1) at (5.5, -2.5) {};
		\node [style=none] (2) at (5.5, -2.5) {\color{red}{$Y$}};
		\node [style=none] (3) at (5, -2) {};
		\node [style=none] (4) at (5, -2) {\color{blue}{$X$}};
		\node [style=none] (5) at (5, -1) {\color{blue}{$X$}};
		\node [style=none] (6) at (4, -2) {\color{blue}{$X$}};
		\node [style=none] (7) at (4, -1) {\color{blue}{$X$}};
		\node [style=none] (8) at (5, 1) {\color{blue}{$X$}};
		\node [style=none] (9) at (1, 1) {\color{blue}{$X$}};
		\node [style=none] (10) at (1, -2) {\color{blue}{$X$}};
		\node [style=none] (11) at (5.5, 1.5) {\color{red}{$Z$}};
		\node [style=none] (12) at (0.5, 1.5) {\color{red}{$Z$}};
		\node [style=none] (13) at (0.5, -2.5) {\color{red}{$Z$}};
		\node [style=none] (14) at (1, -1) {};
		\node [style=none] (15) at (1, 0) {};
		\node [style=none] (16) at (2.5, 1) {};
		\node [style=none] (17) at (3.5, 1) {};
		\node [style=none] (18) at (5, 0.5) {};
		\node [style=none] (19) at (5, -0.5) {};
		\node [style=none] (20) at (1.5, 0.5) {};
		\node [style=none] (21) at (3.5, -0.5) {};
		\node [style=none] (22) at (3.5, -0.5) {};
		\node [style=none] (23) at (2, -2) {};
		\node [style=none] (24) at (3, -2) {};
		\node [style=none] (26) at (7, -0.25) {};
		\node [style=none] (27) at (7, -0.25) {$\mathrm{D}^{(2)}$};
		\node [style=none] (28) at (8.75, -0.25) {};
		\node [style=none] (29) at (8.75, -0.25) {};
		\node [style=none] (30) at (8.75, -0.25) {$|\Psi\rangle$};
		\node [style=none] (31) at (10, -0.25) {};
		\node [style=none] (32) at (10, -0.25) {$=$};
		\node [style=none] (33) at (7, -0.25) {};
		\node [style=none] (34) at (11.25, -0.25) {};
		\node [style=none] (35) at (11.25, -0.25) {$\mathrm{D}^{(2)}$};
		\node [style=none] (36) at (11.25, -0.25) {};
		\node [style=none] (37) at (16.75, -2) {};
		\node [style=none] (38) at (17.25, -2.5) {};
		\node [style=none] (39) at (17.25, -2.5) {\color{red}{$Y$}};
		\node [style=none] (40) at (16.75, -2) {};
		\node [style=none] (41) at (16.75, -2) {\color{blue}{$X$}};
		\node [style=none] (42) at (16.75, -1) {\color{blue}{$X$}};
		\node [style=none] (43) at (15.75, -2) {\color{blue}{$X$}};
		\node [style=none] (44) at (15.75, -1) {\color{blue}{$X$}};
		\node [style=none] (45) at (16.75, 1) {\color{blue}{$X$}};
		\node [style=none] (46) at (12.75, 1) {\color{blue}{$X$}};
		\node [style=none] (47) at (12.75, -2) {\color{blue}{$X$}};
		\node [style=none] (48) at (17.25, 1.5) {\color{red}{$Z$}};
		\node [style=none] (49) at (12.25, 1.5) {\color{red}{$Z$}};
		\node [style=none] (50) at (12.25, -2.5) {\color{red}{$Z$}};
		\node [style=none] (51) at (12.75, -1) {};
		\node [style=none] (52) at (12.75, 0) {};
		\node [style=none] (53) at (14.25, 1) {};
		\node [style=none] (54) at (15.25, 1) {};
		\node [style=none] (55) at (16.75, 0.5) {};
		\node [style=none] (56) at (16.75, -0.5) {};
		\node [style=none] (57) at (13.25, 0.5) {};
		\node [style=none] (58) at (15.25, -0.5) {};
		\node [style=none] (59) at (15.25, -0.5) {};
		\node [style=none] (60) at (13.75, -2) {};
		\node [style=none] (61) at (14.75, -2) {};
		\node [style=none] (62) at (18.5, -2.5) {};
		\node [style=none] (63) at (18.25, -2.5) {$\times\left(\right.$};
		\node [style=none] (64) at (19.5, -2.5) {\color{red}{$X$}};
		\node [style=none] (65) at (19, -2) {\color{blue}{$Z$}};
		\node [style=none] (66) at (20, -2) {\color{blue}{$Z$}};
		\node [style=none] (67) at (19, -3) {\color{blue}{$Z$}};
		\node [style=none] (68) at (20, -3) {\color{blue}{$Z$}};
		\node [style=none] (70) at (20.75, -2.5) {$\left.\right)$};
		\node [style=none] (72) at (20, -0.25) {$|\Psi\rangle$};
		\node [style=none] (73) at (21.75, -0.25) {$=$};
		\node [style=none] (74) at (23, -0.25) {};
		\node [style=none] (75) at (23.75, -0.25) {$\mathrm{D}^{(2)}$};
		\node [style=none] (76) at (28.75, -1.75) {};
		\node [style=none] (77) at (29.25, -2.25) {};
		\node [style=none] (78) at (29.25, -2.25) {\color{red}{$Y$}};
		\node [style=none] (79) at (28.75, -1.75) {};
		\node [style=none] (80) at (28.75, -1.75) {\color{blue}{$X$}};
		\node [style=none] (81) at (28.75, -0.75) {\color{blue}{$X$}};
		\node [style=none] (82) at (27.75, -1.75) {\color{blue}{$X$}};
		\node [style=none] (83) at (27.75, -0.75) {\color{blue}{$X$}};
		\node [style=none] (84) at (28.75, 1.25) {\color{blue}{$X$}};
		\node [style=none] (85) at (24.75, 1.25) {\color{blue}{$X$}};
		\node [style=none] (86) at (24.75, -1.75) {\color{blue}{$X$}};
		\node [style=none] (87) at (29.25, 1.75) {\color{red}{$Z$}};
		\node [style=none] (88) at (24.25, 1.75) {\color{red}{$Z$}};
		\node [style=none] (89) at (24.25, -2.25) {\color{red}{$Z$}};
		\node [style=none] (90) at (24.75, -0.75) {};
		\node [style=none] (91) at (24.75, 0.25) {};
		\node [style=none] (92) at (26.25, 1.25) {};
		\node [style=none] (93) at (27.25, 1.25) {};
		\node [style=none] (94) at (28.75, 0.75) {};
		\node [style=none] (95) at (28.75, -0.25) {};
		\node [style=none] (96) at (25.25, 0.75) {};
		\node [style=none] (97) at (27.25, -0.25) {};
		\node [style=none] (98) at (27.25, -0.25) {};
		\node [style=none] (99) at (25.75, -1.75) {};
		\node [style=none] (100) at (26.75, -1.75) {};
		\node [style=none] (101) at (22.5, -0.25) {$-$};
		\node [style=none] (102) at (31, -0.25) {};
		\node [style=none] (103) at (31, -0.25) {$|\Psi\rangle$};
	\end{pgfonlayer}
	\begin{pgfonlayer}{edgelayer}
		\draw [thick, style=dotted, in=270, out=90] (14.center) to (15.center);
		\draw [thick, style=dotted] (16.center) to (17.center);
		\draw [thick, style=dotted] (18.center) to (19.center);
		\draw [thick, style=dotted] (20.center) to (22.center);
		\draw [thick, style=dotted] (23.center) to (24.center);
		\draw [thick, style=dotted, in=270, out=90] (51.center) to (52.center);
		\draw [thick, style=dotted] (53.center) to (54.center);
		\draw [thick, style=dotted] (55.center) to (56.center);
		\draw [thick, style=dotted] (57.center) to (59.center);
		\draw [thick, style=dotted] (60.center) to (61.center);
		\draw [thick, style=dotted, in=270, out=90] (90.center) to (91.center);
		\draw [thick, style=dotted] (92.center) to (93.center);
		\draw [thick, style=dotted] (94.center) to (95.center);
		\draw [thick, style=dotted] (96.center) to (98.center);
		\draw [thick, style=dotted] (99.center) to (100.center);
	\end{pgfonlayer}
\end{tikzpicture}
\, ,
\label{eq:membrane_operator}
     \end{align}
 \end{widetext}   
 where $\ket{\Psi}$ is a ground state of $\mathrm{H}^{\text{pert}}_{\text{2D-cluster}|\text{blue}}$. For the second equality, we used the fact that the value of the cluster-like term in the bracket is $-1$ on the ground space. This is true even for the perturbed Hamiltonian ground space, provided the support of the perturbations is contained inside a green region, as we mentioned before. Since $\mathrm{D}^{(2)}$ also commute with $\mathrm{H}^{\text{pert}}_{\text{2D-cluster}|\text{blue}}$, we find that the ground space should be at least two-fold degenerate.
 
We note that we could apply a unitary that could disentangle the region outside the orange squares at the four corners to a trivial-like Hamiltonian (unique ground state). After conjugating with the unitary on the Hamiltonian terms contained in the yellow region, the new terms are contained inside the orange region. Otherwise, the orange region is given by the unitary conjugation of the yellow region. Explicitly, the unitary is 
\begin{widetext}
\begin{align}
    \mathcal{U}^{(2)}_{\text{disent}}=\prod_{v_r}\prod_{v_b
    \in
    \partial (p_b=v_r)}CZ_{v_r,v_b}\prod_{\substack{v_b=(i+\frac{1}{2},j+\frac{1}{2})\\
    i_0\leq i<i_1, j_0\leq j<j_1}}CZ_{v_b,v_b+(1,1)}\prod_{\substack{v_b=(i+\frac{1}{2},j+\frac{1}{2})\\
    i_0< i\leq i_1, j_0\leq j< j_1}}CZ_{v_b,v_b+(-1,1)}\,.
\label{eq:unitarytoconj}
\end{align}
\end{widetext}
After conjugating with $\mathcal{U}^{(2)}_{\text{disent}}$, $\tilde{\mathrm{H}}_{\text{2D-cluster}|\text{blue}}$ is transformed to
\begin{widetext}
\begin{align}
    \mathcal{U}_{\text{disent}}^{(2)}\tilde{\mathrm{H}}_{\text{2D-cluster}|\text{blue}}(\mathcal{U}_{\text{disent}}^{(2)})^{\dagger}&=-\sum_{v_r\in A}X_{v_r}-\sum_{v_b\in A}X_{v_b}+\sum_{v_r\in B}X_{v_r}+\sum_{v_b\in B}X_{v_b}+\sum_{v_b\in B}\begin{array}{ccc}
        X_{v_r} & & X_{v_r}\\
         & X_{v_b} & \\
        X_{v_r} & & X_{v_r}
    \end{array}\nonumber\\
    &+ \sum_{\substack{v_b=(i+\frac{1}{2},j_0+\frac{1}{2})\\
    i_0<i<i_1}}\left(X_{v_b}+\begin{array}{ccc}
        & X_{v_b} & \\
       X_{v_r} & & X_{v_r}
    \end{array}\right)+\sum_{\substack{v_b=(i+\frac{1}{2},j_1+\frac{1}{2})\\
    i_0<i<i_1}}\left(X_{v_b}+\begin{array}{ccc}
     X_{v_r} & &X_{v_r} \\  
      & X_{v_b} &
      \end{array}\right)\nonumber\\
      &+\sum_{\substack{v_b=(i_0+\frac{1}{2},j+\frac{1}{2})\\
    j_0<j<j_1}}\left(X_{v_b}+\begin{array}{ccc}
      X_{v_r} & \\
      & X_{v_b}\\
     X_{v_r} &
    \end{array}\right)+\sum_{\substack{v_b=(i_1+\frac{1}{2},j+\frac{1}{2})\\
    j_0<j<j_1}}\left(X_{v_b}+\begin{array}{ccc}
      & X_{v_r}\\
     X_{v_b} & \\
      & X_{v_r}
    \end{array}\right)\, .
\end{align}
\end{widetext}
where regions $A$ and $B$ are as defined in \eqref{eq:regionABrectangular}.
Conjugation by this unitary should not change the spectrum of the Hamiltonian. Since a trivial-like Hamiltonian has a unique ground state, the two-fold degeneracy should come from any four orange square regions. We claim that all four orange square regions give at least two-fold degeneracy. Suppose not, there is one orange square region that has a unique ground state, then we can consider the same perturbation coming from this orange square region around all four corners. For such a perturbed Hamiltonian, there would be a unique ground state. However, this contradicts the fact that the membrane operator anti-commutes with  $\mathrm{D}^{(2)}$. Hence, all four orange square regions should give at least two-fold ground state degeneracy, and in total, the Hamiltonian $\mathrm{H}^{\text{pert}}_{\text{2D-cluster}|\text{blue}}$ should have at least 16-fold degeneracy.

\red{Now consider adding symmetric local perturbations throughout the fattened green region shown in Figure.~\ref{fig:2Dinterface-pert}. We may consider $k$-local symmetric perturbations. Suppose the side length of the interface $i_1-i_0\sim j_0-j_1\sim L$, then we need to consider $k<<L$. We can then construct the membrane operator that anticommutes with $\mathrm{D}^{(2)}$ within the ground-state subspace, as in~\eqref{eq:membrane_operator}. This membrane operator commutes with all such symmetric local perturbations because, on the support of each local perturbation, it coincides with the subsystem symmetry generators. Consequently, the ground state necessarily remains at least two-fold degenerate, independent of any crystalline symmetries that may or may not be imposed.}
\begin{figure*}[h!]
    \centering   \includegraphics[scale=1]{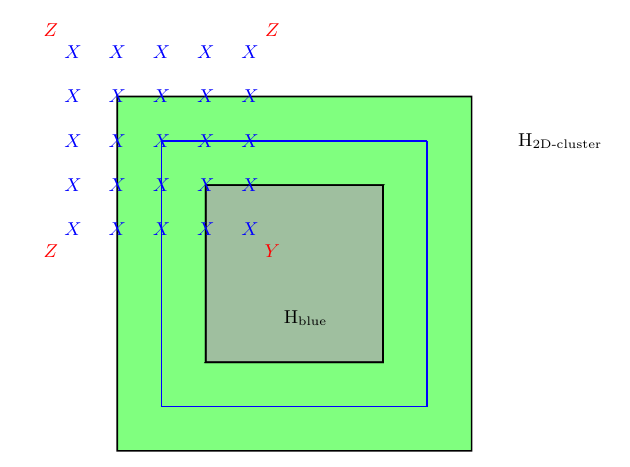}
    \caption{The figure illustrate an interface between $\mathrm{H}_{\text{blue}}$ and $\mathrm{H}_{\text{2D-cluster}}$ with symmetric local perturbations (symmetric under both subsystem symmetries and the noninvertible symmetry $\mathrm{D}^{(2)}$) along a fattened green region along the interface. A membrane operator that anti-commute with $\mathrm{D}^{(2)}$ on the ground space of the interface Hamiltonian with added perturbations is also shown. This membrane operator commutes with the added symmetric local perturbations. If we consider $k$-local Hamiltonian(each term in the Hamiltonian has support on at most $k$ sites) then we take the length of the sides of the interface to be $L$ such that $k<<L$.}
\label{fig:2Dinterface-pert}
\end{figure*}
\subsection{Interface between \texorpdfstring{$\mathrm{H}_{\text{3D-cluster}}^{\mathcal{G}}$}{Lg} and \texorpdfstring{$\mathrm{H}_{\text{blue}}^{(3)\mathcal{G}}$}{Lg}}\label{sec:stabilityhingemode}
\red{Let us consider a cubic interface between $\mathrm{H}_{\text{3D-cluster}}^{\mathcal{G}}$ and $\mathrm{H}_{\text{blue}}^{(3)\mathcal{G}}$ placed on a 3-torus. We add a symmetric perturbation of the interface Hamiltonian to study the stability of degenerate ground space. Let us consider the same setup of interface as in Appendix~\ref{sec:Interfaceanalysis3D}. We call the Hamiltonian $\tilde{\mathrm{H}}^{'\mathcal{G}}_{\text{3D-cluster}|\text{blue}}$ the free Hamiltonian. We add the symmetric perturbation $\rm V$ (symmetric with respect to all the planar subsystem symmetries and the noninvertible symmetry). We assume that the perturbation is supported inside a green region around the hinge as shown in Figure~\ref{fig:3Dhingemodeperturbation}. We note that for this choice of region, the perturbation need not be small. Now we argue that even after the perturbation the ground state degeneracy is not lifted.}

\red{The ground states of $\tilde{\mathrm{H}}^{'\mathcal{G}}_{\text{3D-cluster}|\text{blue}}$ can be labeled by the eigenvalues of the effective $\mathbf{D}^{(3)}_{\rm pln}$ operator on the ground space. We recall that the operator $\mathbf{D}^{(3)}_{\rm pln}$ anti-commute with $\mathcal{X}_{(x,y,z)}$ as in \eqref{eq:fractionalizedsymopalgebra3D}. Suppose that we denote the ground state where the product $\prod\limits_{v\in\text{Hinge}}Z_{v}$ take the value $+1$ to be $\ket{\Omega^{(0)}}$, then the other ground state is obtained by $\ket{\Omega^{(1)}}\equiv \prod\limits_{v=(x,y,z)\in\text{Hinge}}\mathcal{X}_{(x,y,z)}^{i_{(x,y,z)}}\ket{\Omega^{(0)}}$ such that $\sum\limits_{v=(x,y,z)\in\text{Hinge}}\,i_{(x,y,z)}=1$ mod $2$ where $i_{(x,y,z)}$ could be $0$ or $1$. Hence, the ground states can be represented by $\ket{\red{\Omega^{(i)}}}\equiv \prod\limits_{v=(x,y,z)\in\text{Hinge}}\mathcal{X}_{(x,y,z)}^{i_{(x,y,z)}}\ket{\Omega^{(0)}}$.}

\red{We note that we could multiply the stabilizers of the free Hamiltonian with $\mathcal{X}_{(x,y,z)}$ to obtain a membrane/volume operator that commutes with the perturbation $\rm V$. Let us denote the membrane / volume operator that contains the site $(x,y,z)$ on the hinge by $\mathfrak{X}_{(x,y,z)}$ (see figure \ref{fig:3Dhingemodeperturbation} for an illustration of membrane/volume operator). Then,
\begin{subequations}
\begin{align}    \mathcal{X}_{(x,y,z)}\ket{\Omega^{(i)}}&=\mathfrak{X}_{(x,y,z)}\ket{\Omega^{(i)}}\,,\\
[\mathrm{V},\mathfrak{X}_{(x,y,z)}]&=0\,.
\end{align}
\end{subequations}
We note that
\begin{align}
    \{\mathbf{D}^{(3)}_{\rm pln},\mathfrak{X}_{(x,y,z)}\}\ket{\Omega^{(i)}}=\ket{\Omega^{(i)}}\,.
\end{align}
Since $\tilde{\mathrm{H}}^{'\mathcal{G}}_{\text{3D-cluster}|\text{blue}}$ and $\rm V$ commutes with both $\mathbf{D}^{(3)}_{\rm pln}$ and $\mathfrak{X}_{(x,y,z)}$,  we still have at least two fold ground state degeneracy even after adding the perturbation $\rm V$.}
\begin{figure*}[h!]
    \centering
    \includegraphics[scale=1]{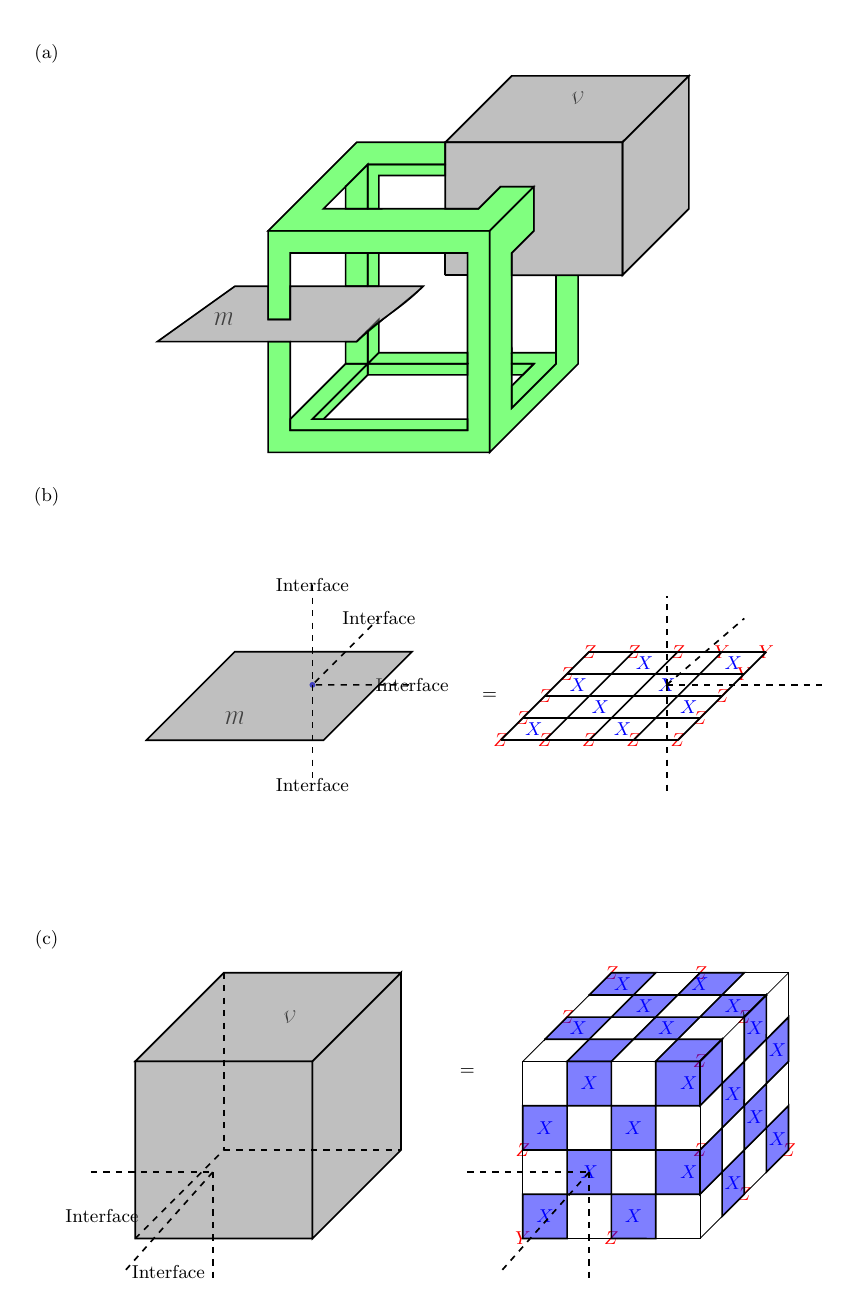}
    \caption{The green region in figure (a) indicate the support of the perturbation $\lambda\rm V$ near the 12 edges. The membrane and cube colored gray represent the membrane operator $\mathcal{M}$ and volume operator $\mathcal{V}$ that is denoted combinedly as $\mathfrak{X}_{(x,y,z)}$ for $(x,y,z)$ located near the centre of an edge and one of the corners respectively. Figure (b) illustrate the membrane operator projected to the plane. Figure (c) illustrate the volume operator. The dashed lines are used to indicate the interface surface cutting through the membrane operator and the volume operator.}
    \label{fig:3Dhingemodeperturbation}
\end{figure*}
\end{document}